\documentclass[12pt]{amsart}

\usepackage{fullpage}

\usepackage{graphicx}
\usepackage{amssymb}
\usepackage{bbm}
\usepackage{comment}
\usepackage{epstopdf}
\usepackage{enumerate}
\usepackage{mathrsfs}
\usepackage{mathtools}  
\usepackage{booktabs}   
\usepackage{multirow}
\usepackage{stmaryrd}
\usepackage{xcolor}
\usepackage{float}
\usepackage{tikz}  
\usetikzlibrary{decorations.markings}
\usepackage{rotating} 
\numberwithin{equation}{section}

\newcommand{\trans}{\intercal} 

\newcommand{\dd}{\mathrm{d}}
 
\newcommand{\R}{\mathbb{R}}
\newcommand{\Q}{\mathbb{Q}}

\newcommand{\eps}{\epsilon}

\newcommand{\bigo}[1]{ \mathcal{O}\!\left( #1 \right) }
\newcommand{\littleo}[1]{ {o}\!\left( #1 \right) }

\newcommand{\ie}{\emph{i.e.}}
\newcommand{\solsplit}{\mathsf{Y}}
\newcommand{\solfuse}{\protect{\mathpalette\rotsfy\relax}}
\newcommand{\rotsfy}[2]{\rotatebox[origin=c]{180}{$#1\mathsf{Y}$}}
\def\eq{\begin{equation}}
\def\endeq{\end{equation}}

\newcommand{\pd}[3][ ]{\frac{\partial^{#1} #2}{\partial #3^{#1} } }

\DeclareMathOperator{\imag}{Im}
\DeclareMathOperator{\re}{Re}

\DeclareMathOperator{\sech}{sech}

\theoremstyle{plain}
\newtheorem{thm}{Theorem}

\newtheorem{prop}[thm]{Proposition}
\newtheorem{myrhp}{Riemann-Hilbert Problem}
\theoremstyle{definition}

\newtheorem{definition}{Definition}

\theoremstyle{remark}

\newtheoremstyle{RHP}  		
  {3pt}					
  {3pt}					
  {\itshape}				
  {}						
  {\bfseries}				
  {}						
  {.5em}					
  {Riemann-Hilbert Problem #2 #3} 

\theoremstyle{RHP}

\newcommand{\Qpot}{\begin{pmatrix} 0 & Q^{[3]}(x,t) & -\gamma^{[1]} \gamma^{[3]} Q^{[2]}(x,t)^* \\ \gamma^{[1]}\gamma^{[2]} Q^{[3]}(x,t)^* & 0 & Q^{[1]}(x,t) \\ Q^{[2]}(x,t) & \gamma^{[2]} \gamma^{[3]} Q^{[1]}(x,t)^* & 0 \end{pmatrix}}

\title{Semiclassical soliton ensembles for the \\ three-wave resonant interaction equations}
\author{R.~J.~Buckingham}
\address{Department of Mathematical Sciences\\University of Cincinnati}
\email{buckinrt@uc.edu}

\author{R.~M.~Jenkins}
\address{Department of Mathematics\\University of Arizona}
\email{rjenkins@math.arizona.edu}

\author{P.~D.~Miller}
\address{Department of Mathematics\\University of Michigan}
\email{millerpd@umich.edu}

\date{\today}
\begin{document}
\begin{abstract}
The three-wave resonant interaction equations are a non-dispersive system of 
partial differential equations with quadratic coupling describing the time 
evolution of the complex amplitudes of three resonant wave modes.  Collisions 
of wave packets induce energy transfer between different modes via pumping 
and decay.  We analyze the collision of two or three packets in the 
semiclassical limit by applying the inverse-scattering transform.  Using WKB 
analysis, we construct an associated semiclassical soliton ensemble, a family 
of reflectionless solutions defined through their scattering data, intended 
to accurately approximate the initial data in the semiclassical limit.  
The map from the initial packets to the soliton ensemble is explicit and 
amenable to asymptotic and numerical analysis.  Plots of the soliton 
ensembles indicate the space-time plane is partitioned into regions 
containing either quiescent, slowly varying, or rapidly oscillatory waves. 
This behavior resembles the well-known generation of dispersive shock waves 
in equations such as the Korteweg-de Vries and nonlinear Schr\"{o}dinger 
equations, although the physical mechanism must be different in the absence 
of dispersion.
\end{abstract}
\maketitle

\tableofcontents

\section{Introduction}
\subsection{The three-wave resonant interaction (TWRI) equations}
A primary obstacle to the effective description of small-amplitude dispersive 
waves by linear theory is the presence of \emph{resonant triads}, in which 
two wave modes conspire to generate a third mode that grows until the 
small-amplitude assumption is violated.  A weakly nonlinear model for this 
process (see Appendix \ref{appA-derivation}) is that the complex amplitudes 
of these modes satisfy the 
\emph{three-wave resonant interaction (TWRI) equations}:
\begin{equation}\label{3wave} \begin{split}
	\eps \left(\pd{q^{[1]}}{t} + c^{[1]}\pd{q^{[1]}}{x}\right) &= \gamma^{[1]} q^{[2]*} q^{[3]*}, \\	
	\eps \left(\pd{q^{[2]}}{t} + c^{[2]}\pd{q^{[2]}}{x}\right) &= \gamma^{[2]} q^{[3]*} q^{[1]*}, \\	
	\eps \left(\pd{q^{[3]}}{t} + c^{[3]}\pd{q^{[3]}}{x}\right) &= \gamma^{[3]} q^{[1]*} q^{[2]*},
\end{split} \end{equation}	
where the wave speeds $c^{[k]}$ are distinct real constants and 
$(\gamma^{[k]})^2 =1$ for $k=1,2,3$. 
The unknowns are three complex-valued fields $q^{[k]}=q^{[k]}(x,t)$ that we will sometimes refer to as \emph{channels} or, when localized, \emph{packets}.  
The TWRI equations have a wide variety of physical applications, following from the 
fact that resonant wave coupling is such a basic nonlinear phenomenon.  These 
include 
waves in plasmas \cite{Sagdeev:1969, Stenflo:1994},
Rossby waves \cite{Newell:1969}, 
capillary-gravity waves \cite{McGoldrick:1965}, 
buckling of cylindrical shells \cite{Lange:1973}, 
Bose-Einstein condensates \cite{Sun:2005}, 
and a variety of applications to nonlinear optics, such as  
stimulated Raman and Brillouin scattering \cite{Baldwin:1969}, 
resonant Bragg reflection \cite{Mak:1998}, 
parametric amplification \cite{Ibragimov:1999}, 
and information storage and processing \cite{Baronio:2009}.
See \cite{Kaup76,Kaup:2010} for more references and discussion.
As the equations \eqref{3wave} are invariant to permutations of the indices, without 
loss of generality we will assume that 
\begin{equation}
	c^{[1]} > c^{[2]} > c^{[3]}.
	\label{eq:c-inequalities}
\end{equation} 
We also assume throughout this work that
\begin{equation}
\epsilon>0 
\label{eq:epsilon-inequality}
\end{equation}
(see \S\ref{subsec-intro-semiclassical} for further discussion).

The TWRI system \eqref{3wave} admits a reduction in which the fields $q^{[k]}(x,t)$ are real-valued, in which case \eqref{3wave} fits into a family of systems including Broadwell models \cite{Broadwell64}, i.e., approximations of the Boltzmann equation of kinetic theory in which the independent velocity variable of the phase space probability density function $q(x,v,t)$ is sampled at the discrete velocities $v=c^{[k]}$, $k=1,\dots,N$, and quadratic interaction terms representing a corresponding discretization of the Boltzmann collision operator are included.  Such models have been studied by many authors, and many qualitative features of the dynamics have been determined using methods from the theory of hyperbolic systems.  These studies show that the type of interaction, i.e., the choice of signs $\gamma^{[k]}$ in the context of \eqref{3wave}, strongly influences the long-term behavior.  See \cite{Beale86} and the references therein.  The key existence and uniqueness results for such systems carry over also to the complex case \eqref{3wave}; for example in \cite[Theorem 9.2.3]{Rauch12} it is shown\footnote{The statement of Theorem 9.2.3 in \cite{Rauch12} appears to pertain to the real reduction of \eqref{3wave}, but from the proof it is obvious that it applies more generally to the complex case, a fact that has been confirmed directly to us by the author.} that if the signs $\gamma^{[k]}$, $k=1,2,3$, are all the same, then there exist smooth compactly supported Cauchy data for which the solution of \eqref{3wave} blows up in $L^\infty$ in finite time,
and that otherwise Cauchy data with $q^{[k]}(\cdot,0)\in H^s(\mathbb{R})$ for all $s>0$ yield a unique global solution in which each field is a $C^s$ function of time with values in $H^s(\mathbb{R})$ for every $s>0$.

Therefore, in order to ensure that solutions remain bounded, unless otherwise indicated we assume throughout that the signs $\gamma^{[k]}$ satisfy 
\begin{equation}
\gamma^{[1]}\gamma^{[2]}=\gamma^{[2]}\gamma^{[3]}=-1,\quad \gamma^{[1]}\gamma^{[3]}=1.
\label{eq:gamma-assumption}
\end{equation}
It is not hard to see where this condition arises.  Indeed, multiplying the equation governing 
$q^{[k]}(x,t)$ by $q^{[k]}(x,t)^*$ and adding the result to its complex 
conjugate gives
\begin{equation}
\eps\left(\pd{|q^{[k]}|^2}{t} + c^{[k]}\pd{|q^{[k]}|^2}{x}\right)=2\gamma^{[k]}\re\{q^{[1]}q^{[2]}q^{[3]}\},\quad k=1,2,3.
\end{equation}
For classical solutions that decay as $x\to \pm\infty$, we therefore obtain
\begin{equation}
\eps\frac{\dd}{\dd t}\int_\mathbb{R}|q^{[k]}(x,t)|^2\,\dd x = 2\gamma^{[k]}\int_{\mathbb{R}}\re\{q^{[1]}(x,t)q^{[2]}(x,t)q^{[3]}(x,t)\}\,\dd x,\quad k=1,2,3.
\end{equation}
Such solutions therefore obey the \emph{Manley-Rowe relations} \cite{ManleyR56}:
\begin{equation}
\begin{split}
\int_\mathbb{R}\left(|q^{[1]}(x,t)|^2-\gamma^{[1]}\gamma^{[2]}|q^{[2]}(x,t)|^2\right)\,\dd x &= K_{12},\\
\int_\mathbb{R}\left(|q^{[2]}(x,t)|^2-\gamma^{[2]}\gamma^{[3]}|q^{[3]}(x,t)|^2\right)\,\dd x &= K_{23},\\
\int_\mathbb{R}\left(|q^{[3]}(x,t)|^2-\gamma^{[3]}\gamma^{[1]}|q^{[1]}(x,t)|^2\right)\,\dd x &= K_{31},
\end{split}
\label{eq:Manley-Rowe}
\end{equation}
where $K_{jk}$ are constants.  These relations show immediately that \emph{provided the signs $\gamma^{[1]}$, $\gamma^{[2]}$, and $\gamma^{[3]}$ are not all the same}, the $L^2(\mathbb{R}_x)$ norms of all three fields $q^{[k]}(x,t)$ are uniformly bounded in time $t$.  Indeed under this condition exactly two of the three signs $-\gamma^{[1]}\gamma^{[2]}$, $-\gamma^{[2]}\gamma^{[3]}$, and $-\gamma^{[3]}\gamma^{[1]}$ will be positive, and therefore the corresponding constants $K_{jk}$ will be nonnegative and will bound both terms on the left-hand side.  For example, if $\gamma^{[1]}\gamma^{[2]}=\gamma^{[2]}\gamma^{[3]}=-1$, then $K_{12}$ and $K_{23}$ will be nonnegative
and we will have the inequalities:
\begin{equation}
\|q^{[1]}(\cdot,t)\|_2^2\le K_{12},\quad \|q^{[2]}(\cdot,t)\|_2^2\le\min\{K_{12},K_{23}\},\quad
\|q^{[3]}(\cdot,t)\|_2^2\le K_{23}.
\end{equation}
On the other hand, if all three signs are the same:  $\gamma^{[1]}=\gamma^{[2]}=\gamma^{[3]}$, then the Manley-Rowe relations all involve indefinite functionals that do not control the growth of any $L^2(\mathbb{R}_x)$ norms of the fields.  In addition to the finite-time blow-up result of Rauch \cite{Rauch12} mentioned above, it has been well-known since the work of Zakharov and Manakov \cite{Zakharov:1973} that in this case there exist explicit solutions corresponding to Schwartz-class initial conditions that exhibit blow-up in finite time (the elementary solitons of types $\solsplit$ and $\solfuse$; see \S\ref{subsec:lax-pair} and Appendix~\ref{sec:solitons} for details).  For these reasons, the case that all signs are the same is therefore known as the \emph{explosive instability} case of the TWRI system \eqref{3wave}, while the case that one of the signs $\gamma^{[k]}$ differs from the other two is called the \emph{decay instability} case.  By associating the distinguished sign $\gamma^{[2]}$ with the intermediate velocity $c^{[2]}$ (see \eqref{eq:c-inequalities}), we are studying a particular kind of decay instability in this paper.  See the last paragraph of \S\ref{subsec:zs-matrix} for a justification of this choice.

The system \eqref{3wave} is completely integrable in the sense that it admits 
a Lax pair representation (see \S\ref{subsec:lax-pair} and Appendix 
\ref{appB-inverse-scattering} for details).  The Lax pair was 
discovered by Zakharov and Manakov \cite{Zakharov:1973}, and the 
inverse-scattering procedure was first developed by them \cite{Zakharov:1975} 
and Kaup \cite{Kaup76}.  For early qualitative results see Kaup, Reiman, and 
Bers \cite{Kaup:1979}.  More recently, Shchesnovich and Yang have studied 
higher-order TWRI solitons \cite{Shchesnovich:2003}.  There are a number of 
generalizations of the TWRI equations that are also completely integrable, 
including the TWRI equations in a spatially inhomogeneous medium 
\cite{Reiman:1979}, the $N$-wave equations in $1+1$ (spatial+time) dimensions 
\cite{Ablowitz:1975,Gerdjikov:1981, BealsC84, BealsC85, Novikov:1984}, the 
TWRI equations in $2+1$ dimensions \cite{Novikov:1984}, and a coupled system 
of PDEs with a $4\times 4$ Lax pair recently derived by Biondini and Wang 
\cite{Biondini:2015}.

The discovery that the TWRI system is integrable provides a starting point 
for understanding solutions, but a tremendous amount of work remains to be 
done.  The main objective is to find an effective way of describing the map 
from the initial conditions to the solution profile at a later time.  Except 
for a very few special cases, it is not possible to carry out either the 
direct or inverse scattering maps explicitly.  The most promising techniques 
for studying the qualitative behavior of solutions are asymptotic analysis 
(which we pursue here), numerics, and series expansions.  For a numerical 
approach to the direct-scattering problem see Degasperis et 
al.\@ \cite{Degasperis:2011};  for recent results on the series approach 
(avoiding the inverse-scattering machinery), see the interesting work of 
Martin and Segur \cite{Martin:2016}.

\subsection{The semiclassical limit}
\label{subsec-intro-semiclassical}
The TWRI equations pose an interesting challenge for asymptotic analysis.  For most 
nonlinear wave equations, the long-time limit provides a effective way of 
understanding post-collision dynamics as the solution profiles often 
simplify in this limit.  The situation is completely different for the 
TWRI equations, where the propagation speed in each channel is fixed.  For 
example, suppose we start with three disjointly-supported packets (as we will see below, this significantly 
simplifies the direct-scattering calculation).  Given the distinct constant propagation speeds, it is reasonable to expect\footnote{To our knowledge a suitably precise statement to this effect has yet to be rigorously proven for the TWRI system, although the hypothesis is prevalent in the literature, see e.g.,  \cite{Kaup76}.  However, for Broadwell-type systems similar to the real reduction of \eqref{3wave}, see \cite{Beale86}.} that after some time the packets will 
have passed each other and interaction will effectively cease.  The profiles the packets 
have at this moment, no matter how complicated, will be preserved as-is 
for all future time, and the long-time limit provides no additional 
information.  This essentially non-dispersive character of the TWRI system distinguishes it from other well-known integrable nonlinear wave equations such as the Korteweg-de Vries or nonlinear Schr\"odinger equations.

We propose here that it is more fruitful to study instead the semiclassical
limit $\epsilon\downarrow 0$, as the post-collision packet profiles can have 
asymptotic expansions in $\epsilon$ with simpler leading terms.  
We consider the TWRI equations \eqref{3wave} with initial data 
\begin{equation}
q^{[k]}(x,0)=e^{i\theta^{[k]}}H^{[k]}(x)e^{i\kappa^{[k]}x/\epsilon},\quad k=1,2,3,
\label{eq:initial-data}
\end{equation}
where $\{\theta^{[k]}\}_{k=1}^3$ and $\{\kappa^{[k]}\}_{k=1}^3$ are real 
($\epsilon$-independent) constants and $\{H^{[k]}\}_{k=1}^3$ are real-valued 
($\epsilon$-independent) non-negative and rapidly decaying functions on $\mathbb{R}$.  
Semiclassical limits have been investigated via the inverse-scattering method 
for a variety of scalar equations, including the 
Korteweg-de Vries 
\cite{Lax:1983a,Lax:1983b,Lax:1983c,Claeys:2009,Claeys:2010a,Claeys:2010b}, 
nonlinear Schr\"odinger 
\cite{Miller02,Kamvissis:2003,LyngM07,Bertola:2010,Bertola:2013,Jenkins:2014}, 
modified nonlinear Schr\"odinger \cite{DiFranco:2008,DiFranco:2011,DiFranco:2013}, 
and sine-Gordon equations 
\cite{Buckingham:2008,Buckingham:2012a,Buckingham:2012b,Buckingham:2013}.  
The interplay between the short length scale 
$\epsilon$ and the long length scale set by the initial condition typically 
leads to a situation in which slowly-varying waves develop large gradients that are then regularized by the generation of \emph{dispersive shock waves}, \ie, slowly modulated trains of highly oscillatory waves.    
These dispersive shock waves occupy space-time domains that become fixed as $\epsilon\downarrow 0$, and the mathematical goal is to describe the limiting boundaries of the regions and the asymptotic (in $\epsilon$) behavior of the waves in each region.
Furthermore, for these scalar equations it has been shown that, at least for certain classes 
of initial conditions, the semiclassical limit can be \emph{universal} near certain critical points in the space-time, in the sense that the leading-order behavior has a fixed form with dependence on the initial conditions entering only parametrically.  A recent collection of papers on this topic that includes several review articles is the special issue \cite{PhysicaDIssue}.

One way to understand the semiclassical limit for integrable equations is to note that the mass of 
each soliton is $\mathcal{O}(\epsilon)$.  Since the total mass of the initial 
condition is independent of $\epsilon$, one might expect that the initial 
profile is a nonlinear condensate of $\mathcal{O}(1/\epsilon)$ solitons, 
plus some amount of radiation.  Furthermore, the small width of the 
constituent solitons allows them to approximate the initial profile 
increasingly well as $\epsilon$ decreases, which suggests radiation 
generically plays a small role in the semiclassical limit.  This large mass 
of thin solitons provides a 
mechanism for the generation of oscillations from the non-oscillatory 
initial data, as the constituent solitons move with respect to one another either through having different velocities or via incurring phase shifts due to interactions with other solitons.    
Note that this integrable interpretation of the semiclassical limit can apparently apply to systems such as the TWRI equations for which dispersion is absent and hence a physical explanation for the generation of structures resembling dispersive shock waves (say as a dispersive response to wave steepening) is less clear.
On the spectral side, the scattering data is expected to 
include $\mathcal{O}(1/\epsilon)$ exceptional points\footnote{Exceptional points are values of the spectral parameter at which the fundamental matrix solution involved in the Riemann-Hilbert problem of inverse scattering fails to exist.  In the Zakharov-Shabat spectral problem and other similar $2\times 2$ direct scattering problems, these exceptional points are exactly the $L^2(\mathbb{R})$ eigenvalues of the problem, but as will be seen in Appendix~\ref{appB-inverse-scattering}, the direct scattering problem for the TWRI system does not have any eigenvalues per se.}, typically accumulating 
on fixed complex contours, and a vanishingly small reflection coefficient.  

This suggests that a natural way to study the semiclassical limit for a 
given initial condition is to apply a WKB (small $\epsilon$) approximation to the direct 
scattering map, obtaining estimated discrete scattering data 
that better approximate the exact scattering data as $\epsilon$ tends to 
zero.  For further simplification the ``continuous" part of the scattering data can be disregarded.  
Thus, for each $\epsilon$, the given original initial data 
$q^{[k]}(x,0)$ can be associated with a set of purely discrete approximate scattering data, which 
in turn corresponds to a bona fide solution of the 
TWRI equations.  We call this collection of exact solutions, one for each small $\epsilon$, a 
\emph{semiclassical soliton ensemble} for the given initial data, and it is these solutions that we 
study in the semiclassical limit.  This procedure has been used previously 
to study the semiclassical limit of scalar wave equations.  The expectation 
is that the time evolution of the semiclassical soliton ensemble provides a 
good approximation of the solution $q^{[k]}(x,t)$ to the original problem.  
In fact, for the TWRI equations we will demonstrate convergence at time $t=0$ for  
certain initial data $q^{[k]}(x,0)$ (although, interestingly, we also 
identify some initial data where our soliton ensembles do \emph{not} 
converge to $q^{[k]}(x,0)$).  
In a forthcoming work \cite{Buckingham:inprep} we translate the algebraic equations
determining the semiclassical soliton ensemble into suitable jump conditions for
a $3\times 3$ matrix Riemann-Hilbert problem in the spirit of \cite{Buckingham:2013,Kamvissis:2003,LyngM07,Miller02,Miller08}, and then apply the Deift-Zhou steepest descent technique
to prove convergence at $t=0$ in the absence of a central packet.  We also give a conditional
convergence result in the presence of three packets.  These analytical results mirror the
numerical observations we make in this paper, and they also provide a description of the interesting 
interactions we report here for $t>0$.

\subsection{Outline of the paper}  
We begin in \S\ref{sec2-plots} by illustrating the dynamics of some semiclassical soliton ensembles
for the TWRI system constructed using the methodology developed in this paper.  The plots shown in \S\ref{sec2-plots} require the numerical solution of numerous poorly-conditioned linear algebraic systems, but the solutions plotted are otherwise exact; they are not numerical simulations of the TWRI system via any time-marching scheme.  These plots demonstrate the emergence of phenomena similar to those familiar from studies of other integrable equations in the semiclassical limit.

In \S \ref{sec3-exact-scattering} we summarize the key points of the inverse-scattering transform solution of the Cauchy problem for the TWRI system, and then we use this theory to compute the scattering data for three 
disjointly supported packets.  We assume the packet in channel 1 is to the 
left of the packet in channel 2, which is to the left of the packet in 
channel 3, guaranteeing collision of the packets in finite time.  
The disjoint support condition means that for any fixed $x$-value the 
TWRI scattering problem reduces to the Zakharov-Shabat eigenvalue 
problem associated to the focusing nonlinear Schr\"odinger equation with a potential related to just the one packet nonzero at $x$.  This 
greatly simplifies the direct scattering problem, although it is more involved 
than the corresponding nonlinear Schr\"odinger analysis since it is 
necessary to stitch together solutions of multiple Zakharov-Shabat problems.  
This procedure for computing scattering data for disjointly-supported 
packets is well known (see, for instance, Kaup \cite{Kaup76}).  However, the 
formulae \eqref{eq:type-1-beta32}, \eqref{eq:type-1-beta23}, 
\eqref{eq:type-3-beta21}, \eqref{eq:type-3-beta12}, 
\eqref{eq:type-solsplit-beta31}, \eqref{eq:type-solsplit-beta32}, 
\eqref{eq:type-solsplit-beta13}, and \eqref{eq:type-solsplit-beta12}, 
expressing the connection coefficients (part of the discrete scattering data) in terms of quantities that can 
be computed without analytic continuation from the solutions of the individual Zakharov-Shabat problems, are new.  

We require 
the aforementioned exact formulae for the connection coefficients to consider the semiclassical limit and effectively construct the soliton ensemble associated with the disjointly-supported initial data.  This construction is described in   
\S\ref{sec4-semiclassical-approx}.
Here we assume that for $k=1,2,3$, $H^{[k]}$ is a single-peaked or Klaus-Shaw potential and use corresponding semiclassical formulae for the Zakharov-Shabat problem to specify the (discrete) scattering data for the soliton ensemble.
We then study the accuracy of the semiclassical soliton ensemble approximation, by comparing some ensembles with $t=0$ to the corresponding specified Cauchy data that generated them.  These plots illustrate the expected convergence to the Cauchy data in the situation that the central packet with speed $c^{[2]}$ is absent, and they also show that if the central packet is present then both convergence and divergence are possible.  The fact that semiclassical soliton ensembles for three disjointly-supported packets may not converge to the Cauchy data at $t=0$ is evidently related to the ad-hoc neglect of terms that are small beyond all orders but that are difficult to calculate accurately.  

In the Appendix we collect some background material related to our study of the TWRI system.
Appendix \ref{appA-derivation} contains a brief derivation of the TWRI system in the context of triad resonances in a family of semilinear dispersive equations.  Appendix \ref{appB-inverse-scattering} is a self-contained account of the treatment of the Cauchy problem for the TWRI system by the inverse-scattering transform. 
Finally, Appendix \ref{appC-nonselfadjoint-zs} summarizes necessary information about the Zakharov-Shabat scattering problem.

\

\noindent
\textbf{Acknowledgements.}  We thank D. J. Kaup, J. Rauch, and H. Segur for useful discussions.  The collaboration of all three authors began at the conference ``Scattering and Inverse-Scattering in Multidimensions'' held in May 2014 at the University of Kentucky and funded by the National Science Foundation grant DMS-1408891.  We also thank the National Science Foundation for support on research grants DMS-1312458 and DMS-1615718 (Buckingham) and DMS-0807653, DMS-1206131, and DMS-1513054 (Jenkins and Miller).

\section{Plots of Semiclassical Soliton Ensembles}
\label{sec2-plots}
\label{subsec:numerical-results}

While our end goal is to obtain rigorous asymptotic 
expansions valid for wide classes of initial data, numerical plots 
are useful guides to the qualitative behavior of soliton ensembles.   
Given initial data $q^{[k]}(x,0)$, $k=1,2,3$, the calculation of the 
scattering data for the 
approximating semiclassical soliton ensembles, which we label $\widetilde{q}^{[k]}(x,t)$, $k=1,2,3$, 
is explicit in terms of integrals that can, in principle, be computed 
numerically.  Due to rapid oscillations with frequencies growing as 
$\epsilon \downarrow 0$, the computation time for the numerical quadratures 
increases as $\epsilon$ decreases.  Once scattering data is 
obtained (possibly to a desired numerical precision), the 
inverse-scattering transformation can be carried out explicitly through 
the solution of the linear systems of equations \eqref{eq:b-plus-1}--\eqref{eq:c-plus-23} (with coefficients depending 
on $x$ and $t$). In practice the computations quickly 
become unwieldy as $\epsilon$ decreases and the system size $N(\epsilon)\times N(\epsilon)$, with $N(\epsilon)=\bigo{\epsilon^{-1}}$, increases.  
However, one remarkable fact about the inverse-scattering procedure is that the solution at 
any given values of $x$ and $t$ can be obtained \emph{without} 
calculating the solution at any other values of $x$ and $t$.  To 
generate plots of the exact semiclassical soliton ensembles we start by
choosing an appropriate space-time grid\footnote{ Rapid oscillations in 
$\widetilde{q}^{[k]}(x,t)$ require 
picking $\mathcal{O}(\epsilon^{-1})$ points in both the $x$ and $t$ 
directions to obtain decent resolution.} $\{(x_j,t_j)\}$.  
Then, at each grid point, we compute the coefficients in the linear system 
\eqref{eq:b-plus-1}--\eqref{eq:c-plus-23}.
This linear system (consisting of $N(\epsilon) \times N(\epsilon)$ 
equations with numerical, not functional, coefficients) can then be inverted.
However, the system can be poorly conditioned; 
typically we require $N(\epsilon)$ decimal places of precision in the coefficients of the linear system to obtain accurate solutions. 

\subsection{Dynamics of ensembles associated with two colliding packets}

For illustration we consider here initial data consisting of two disjointly 
supported colliding packets with semicircular profiles (each multiplied by its own scale factor).
We choose semicircles because the scattering data of the semiclassical soliton ensemble can be computed \emph{exactly}, 
avoiding the numerical evaluation of many integrals, which significantly reduces 
the computation time.  
Specifically, we consider initial data of the type \eqref{eq:initial-data}
with initial envelopes
\begin{equation}
H^{[2]}(x) \equiv 0, \qquad 
H^{[k]}(x) :=\frac{2H^{[k]}_\text{max}\chi_{(a^{[k]},b^{[k]})}(x)}{b^{[k]}-a^{[k]}}\sqrt{(x-a^{[k]})(b^{[k]}-x)}, \quad k=1,3,
\label{plot-Hk}
\end{equation}
where the positive square root is meant
and the condition 
\begin{equation}
\label{endpoints-two-packets}
a^{[1]} < b^{[1]} < a^{[3]} < b^{[3]}
\end{equation}
on the support endpoints (together with \eqref{eq:c-inequalities}) ensures collision of the packets for some positive $t$.  

We fix the parameters
\eq
\begin{gathered}
	\{c^{[1]}, c^{[2]}, c^{[3]}\}=\{1,0,-1\}, \quad 
	\{\gamma^{[1]}, \gamma^{[2]}, \gamma^{[3]}\} = \{1,-1,1\}, \quad 
	\{H_\text{max}^{[1]},H_\text{max}^{[3]}\} = \{1,1\}, \\
	\{\theta^{[1]},\theta^{[3]}\} = \{0,0\}, \quad 
	\{(a^{[1]},b^{[1]}), (a^{[3]},b^{[3]}) \} 
	= \left\{\left(-\frac{3}{2},-\frac{1}{2}\right),\left(\frac{1}{2},\frac{3}{2}\right)\right\}, \quad
	\kappa^{[1]}=0.
\label{plot-vals1}
\end{gathered}
\endeq
For these parameters, the number of solitons associated to each channel are given 
in Table \ref{N-for-plot-vals1} for each value of $\epsilon$ we plot.
\begin{table}[H]
\renewcommand{\arraystretch}{1.2}
\centering
\begin{tabular}{@{}l@{\qquad}c@{\qquad \qquad}c@{\qquad \qquad}c@{}}
\toprule
$\epsilon$ & $N^{[1]}(\epsilon)$ & $N^{[2]}(\epsilon)$ & $N^{[3]}(\epsilon)$ \\
\midrule
1/20 & 4 & 0 & 4 \\
1/40 & 7 & 0 & 7 \\
1/80 & 14 & 0 & 14 \\
1/160 & 28 & 0 & 28 \\
1/320 & 57 & 0 & 57 \\
\bottomrule
\end{tabular}
\caption{The number of solitons of each type assuming initial data 
defined by \eqref{eq:initial-data}, \eqref{plot-Hk}, and \eqref{plot-vals1} for the 
values of $\epsilon$ used in Figures \ref{fig-sc13-2d_k3-1}, 
\ref{fig-sc13-2d_e0p025}, \ref{fig-region-map}, \ref{fig-sc13-1d-collisions}, 
\ref{fig-sc13-1d-t1p2}, \ref{fig-sc13-1d-t2}, \ref{fig-sc13-1d-t0}, and 
\ref{fig-sc13-1d-t0-error}.  The 
$N^{[k]}(\epsilon)$ are defined in \eqref{eq:N-of-epsilon} and independent of 
$\kappa^{[1]}$, $\kappa^{[2]}$, and $\kappa^{[3]}$.}
\label{N-for-plot-vals1}
\end{table}
Varying $\kappa^{[1]}$ and $\kappa^{[3]}$ does not appear to 
affect the plots of $|\widetilde{q}^{[k]}(x,t)|$, $k=1,2,3$, as long as 
$|\kappa^{[3]}-\kappa^{[1]}|$ remains fixed.  Therefore, we fix $\kappa^{[1]}=0$ and vary 
$\kappa^{[3]}>0$.  Further details regarding the construction of the following plots can be found in 
\S\ref{subsec-data-for-plots}.

We first consider dynamics as the system evolves in time.  
In Figure 1 we fix $\kappa^{[3]}=1$ and show how the semiclassical soliton ensemble behaves as $\epsilon$ is decreased.  We clearly see the emergence in the limit $\epsilon\downarrow 0$ of fixed space-time regions containing qualitatively different types of waves.  In particular, the excitation of the initially absent packet in the channel with intermediate speed $c^{[2]}$ appears to be confined to a diamond-shaped region of space-time.  

\begin{figure}[H]
\begin{center}
\includegraphics[height=5.9in]{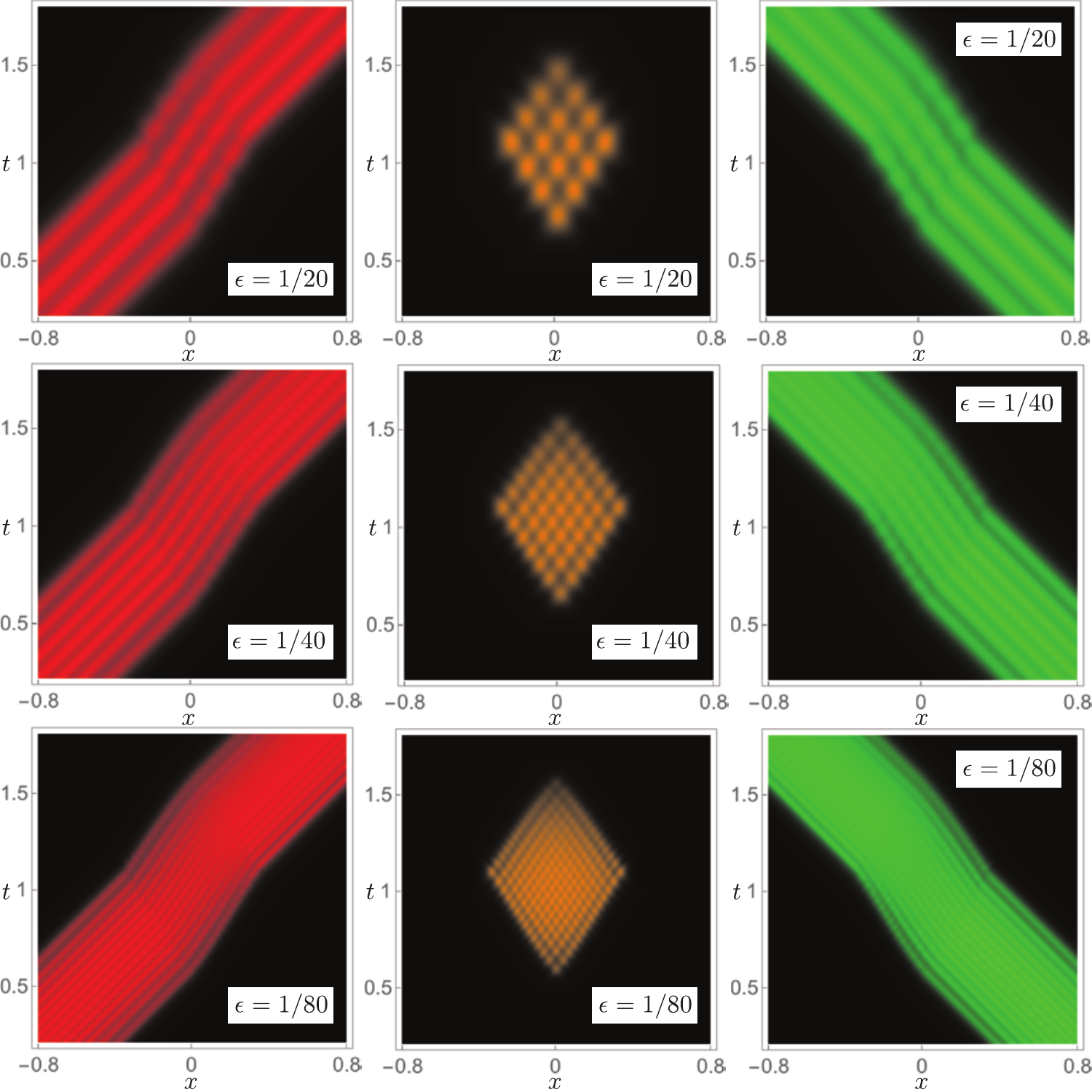}
\caption{Space-time plots of $|\widetilde{q}^{[k]}(x,t)|$, as defined by 
\eqref{eq:initial-data} and \eqref{plot-Hk} with parameters 
\eqref{plot-vals1} and $\kappa^{[3]}=1$ for 
$k=1,2,3$ (left-to-right) and $\epsilon\in\{1/20,1/40,1/80\}$ (top-to-bottom, 
as indicated). The bottom row of plots with $\epsilon=1/80$ is the same as 
the second row of plots in Figure \ref{fig-sc13-2d_e0p025}.}
\label{fig-sc13-2d_k3-1}
\end{center}
\end{figure}

Next, Figure \ref{fig-sc13-2d_e0p025} shows the effect of 
varying the phase gradient $\kappa^{[3]}$ in the initial conditions, holding $\epsilon$ small but fixed.  
First, note from Figure \ref{fig-sc13-2d_e0p025} that for $\kappa^{[3]}$ much larger than 
$\kappa^{[1]}$ (i.e. $\kappa^{[3]}=|\kappa^{[3]}-\kappa^{[1]}|=5$), the packets appear 
to pass through each other without much interaction or change between the 
pre-collision and post-collision profiles.  On the other hand, the closer 
$\kappa^{[3]}$ is to $\kappa^{[1]}=0$, the more interaction there is between packets 
during collision and the more perturbed the profiles are post-collision.  
Also note the shape of the overlap region changes as 
$\kappa^{[3]}$ is adjusted.

\begin{figure}[H]
\begin{center}
\includegraphics[height=5.9in]{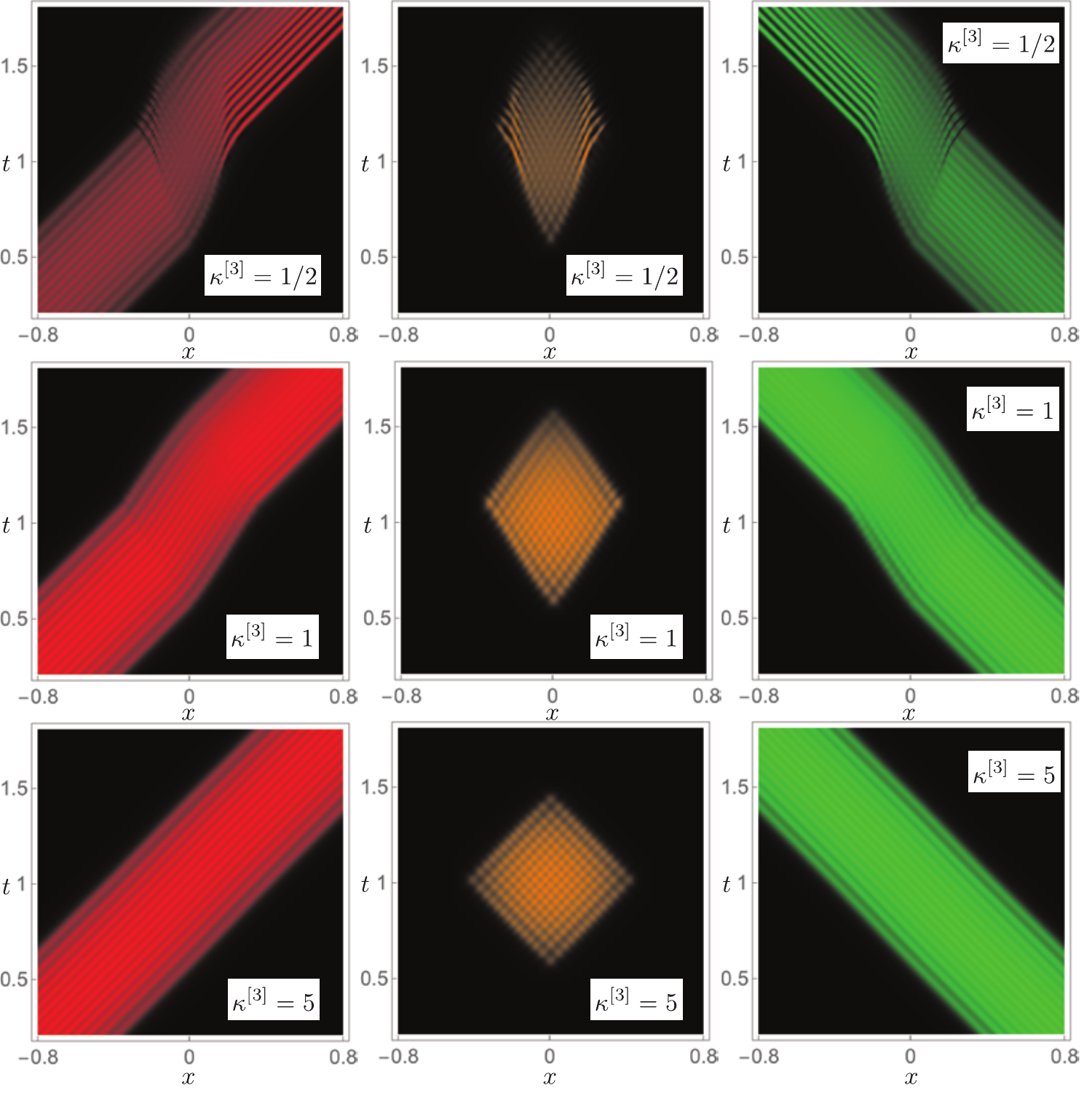}
\caption{Space-time plots of $|\widetilde{q}^{[k]}(x,t)|$, as defined by 
\eqref{eq:initial-data} and \eqref{plot-Hk} with parameters 
\eqref{plot-vals1} and $\epsilon=1/80$ for $k=1,2,3$ (left-to-right) and 
$\kappa^{[3]}\in\{1/2,1,5\}$ (top-to-bottom, as indicated).  The middle row 
of plots with $\kappa^{[3]}=1$ is the same as the bottom row 
of plots in Figure \ref{fig-sc13-2d_k3-1}.}
\label{fig-sc13-2d_e0p025}
\end{center}
\end{figure}

In Figure \ref{fig-region-map}, we reproduce the top-left panel from Figure 
\ref{fig-sc13-2d_e0p025} with various regions marked indicating the (apparent) 
small-$\epsilon$ behavior.  Table~\ref{table:regions_for_fig} describes the qualitative 
behavior in each region.  The descriptor ``zero'' refers to the expected 
semiclassical limit;  of course since the semiclassical soliton ensemble is a pure multisoliton solution, the fields $\widetilde{q}^{[k]}(x,t)$ are typically non-zero 
everywhere for any non-zero value of $\epsilon$.  Similarly, ``oscillatory'' refers to microstructure with $\bigo{\epsilon}$ wavelengths and frequencies but with amplitude that does not vanish with $\epsilon$, while ``non-oscillatory'' regions may contain oscillations whose amplitude apparently decays as $\epsilon\downarrow 0$.  
Oscillations that are damped out in finite time are called ``transient"; otherwise 
they are called ``persistent".

\begin{table}
\renewcommand{\arraystretch}{1.1}
\begin{tabular}{@{}llll@{}}
\toprule
Region & Channel 1 & Channel 2 & Channel 3 \\
\midrule
A & Zero & Zero & Zero \\
B & Non-zero, Non-oscillatory & Zero & Zero \\
$\text{B}^{\prime}$ & Zero & Zero & Non-zero, Non-oscillatory \\
C & Non-zero, Non-oscillatory & Non-zero, Non-oscillatory & Non-zero, Non-oscillatory \\
D & Oscillatory (Persistent) & Oscillatory (Transient) & Oscillatory (Transient) \\
$\text{D}^\prime$ & Oscillatory (Transient) & Oscillatory (Transient) & Oscillatory (Persistent) \\
E & Oscillatory (Persistent) & Zero & Zero \\
$\text{E}^\prime$ & Zero & Zero & Oscillatory (Persistent) \\
\bottomrule
\end{tabular}
\caption{Qualitative description of the behavior of each wave channel in each of the regions defined in Figure~\ref{fig-region-map}.
\label{table:regions_for_fig}
}
\end{table}

\begin{figure}[H]
\begin{center}
\includegraphics[height=2.5in]{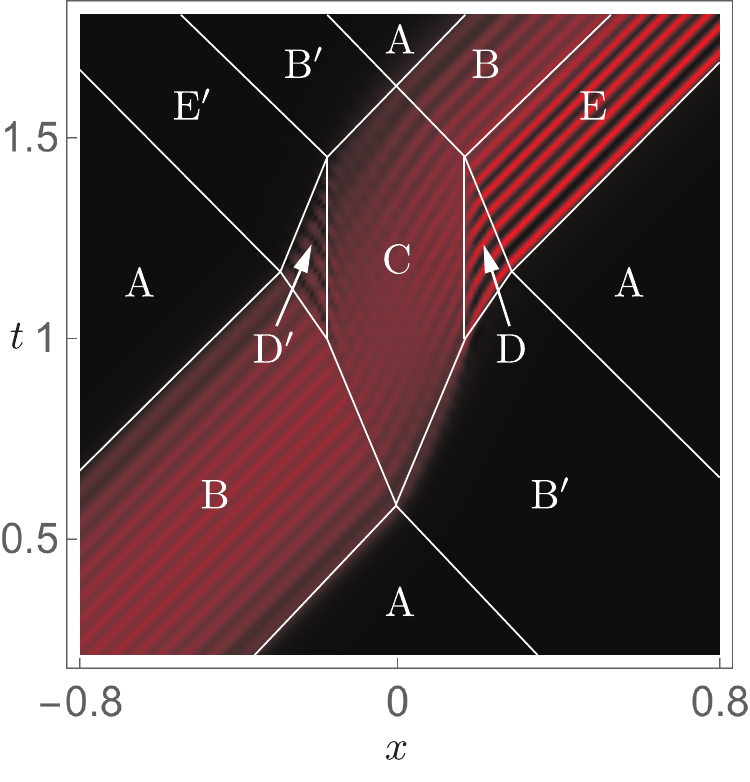}
\caption{Space-time plot of $|\widetilde{q}^{[1]}(x,t)|$, as defined by 
\eqref{eq:initial-data} and \eqref{plot-Hk} with parameters \eqref{plot-vals1}, 
$\epsilon=1/80$, and  $\kappa^{[3]}=1/2$, along with the regions A, B, 
B$^\prime$, C, D, D$^\prime$, E, and E$^\prime$.
The behavior in each region is described in Table~\ref{table:regions_for_fig}.  
The plot is the same as the top-left plot in Figure \ref{fig-sc13-2d_k3-1}.}
\label{fig-region-map}
\end{center}
\end{figure}

We now illustrate the dynamics through several representative time slices.  As this 
requires solving the (poorly conditioned 
$\mathcal{O}(\epsilon^{-1})\times\mathcal{O}(\epsilon^{-1})$ system) at merely 
$\mathcal{O}(\epsilon^{-1})$ points, as opposed to $\mathcal{O}(\epsilon^{-2})$ 
points for the space-time plots, we can decrease $\epsilon$ significantly.
In Figure \ref{fig-sc13-1d-collisions} we plot horizontal (fixed-$t$) slices for 
three values of $t$ and three values of $\kappa^{[3]}$ (note $\epsilon$ is decreased 
by a factor of two compared to Figure \ref{fig-sc13-2d_e0p025}).  
The three chosen times illustrate the evolution of the system in the overlap 
region.  Here it becomes clearer that when $\kappa^{[3]}$ is relatively large 
compared to $\kappa^{[1]}=0$ the packets with speeds $c^{[1]}$ and $c^{[3]}$ pass 
through each other with relatively little disturbance (and little excitation of 
channel 2).  On the other hand, as $|\kappa^{[3]}-\kappa^{[1]}|$ decreases, there is 
greater excitation of channel 2 in the overlap region, and a corresponding 
increase in the perturbation of channels 1 and 3 post-collision.  Comparing the top 
row of plots with $\kappa^{[3]}=1/2$ in Figure \ref{fig-sc13-1d-collisions} with the 
labeled regions in Figure \ref{fig-region-map}, we see the plot at time $t=0.8$ 
cuts horizontally through regions A, B, C, B$^\prime$, and A;  the plot at time 
$t=1$ cuts through regions A, B, D$^\prime$, C, D, B$^\prime$, and A;  and the 
plot at time $t=1.2$ cuts through regions A, D$^\prime$, C, D, and A.

\begin{figure}[H]
\begin{center}
\includegraphics[height=4.2in]{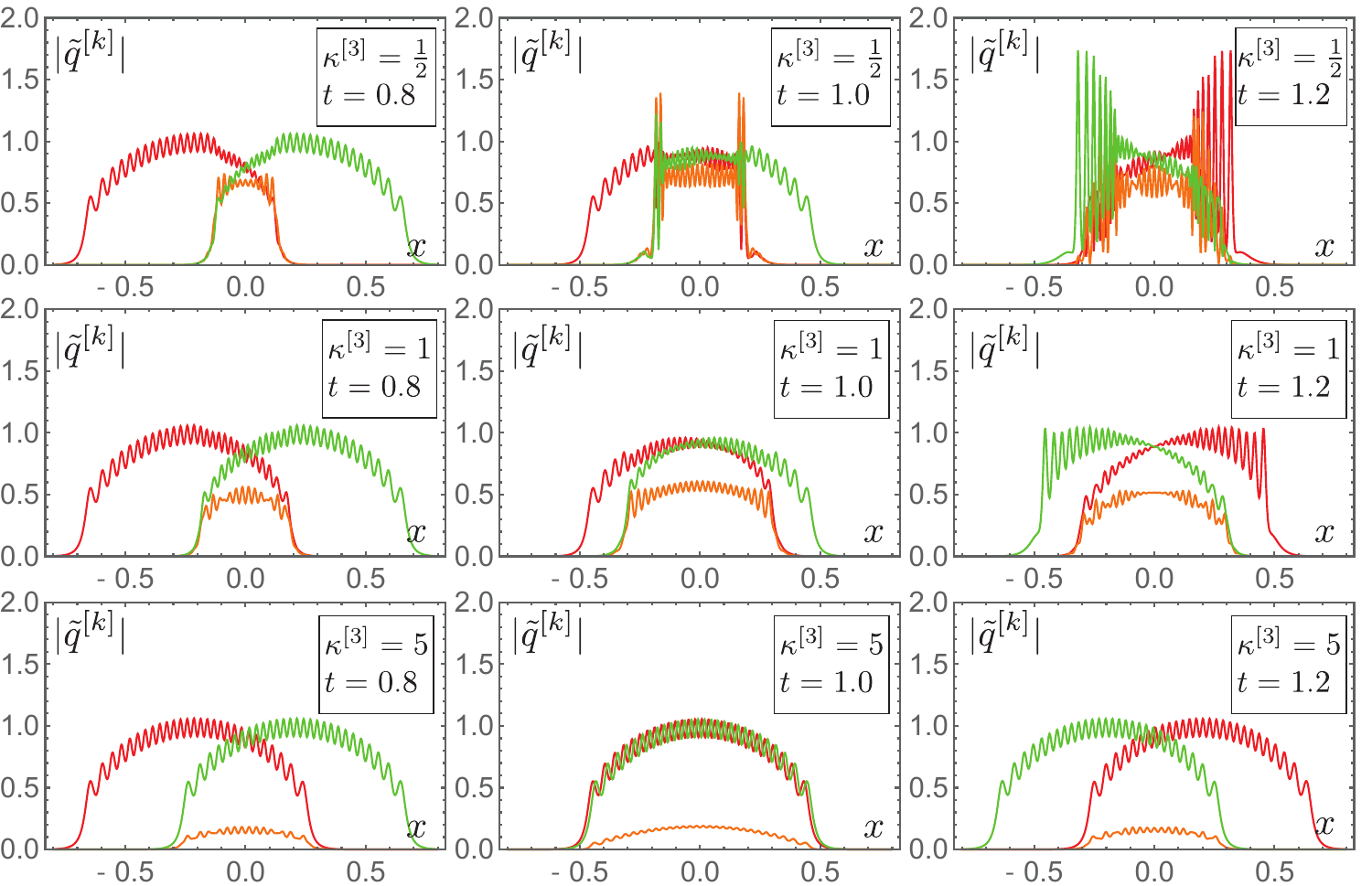}
\caption{Plots of $|\widetilde{q}^{[k]}(x,t)|$, as defined by 
\eqref{eq:initial-data} and \eqref{plot-Hk} with parameters 
\eqref{plot-vals1} and $\epsilon=1/160$ for $t \in\{0.8,1,1.2\}$ 
(left-to-right, as indicated) and $\kappa^{[3]}\in\{1/2,1,5\}$ (top-to-bottom, as 
indicated). 
Red:  $|\widetilde{q}^{[1]}(x,t)|$. 
Orange:  $|\widetilde{q}^{[2]}(x,t)|$. 
Green:  $|\widetilde{q}^{[3]}(x,t)|$. 
}
\label{fig-sc13-1d-collisions}
\end{center}
\end{figure}

In Figure \ref{fig-sc13-1d-t1p2} we fix $t=1.2$ and $\kappa^{[3]}=1/2$ as 
in the top-right plot in Figure \ref{fig-sc13-1d-collisions} and plot 
$|\widetilde{q}^{[1]}(x,1.2)|$ for various values of $\epsilon$.  We restrict to 
the solution in channel 1 for clarity.  The plots pass through the regions 
labeled A, D$^\prime$, C, D, and A in Figure \ref{fig-region-map}.  The plots 
suggest that $|\widetilde{q}^{[1]}|$, in the small-$\epsilon$ limit, has 
$\mathcal{O}(\epsilon)$-period oscillations within an $\epsilon$-independent 
envelope 
in regions D$^\prime$ and D, and is non-oscillatory but non-zero in region C.  
From the plots in Figure \ref{fig-sc13-1d-collisions}, it appears the height of the 
envelopes in regions D$^\prime$ and D decreases as $\kappa^{[3]}$ increases (recall 
we are restricting ourselves to $\kappa^{[1]}=0$ and $\kappa^{[3]}>0$).   

\begin{figure}[H]
\begin{center}
\includegraphics[height=1.35in]{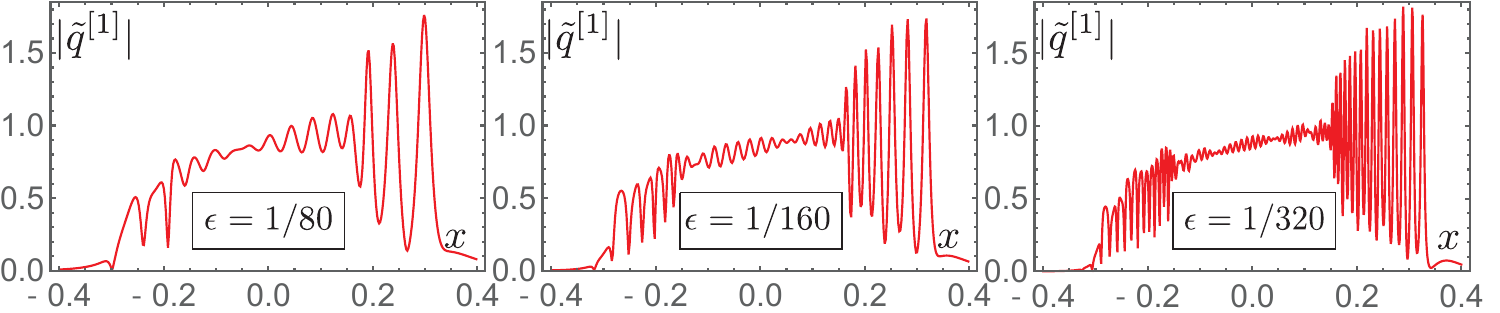}
\caption{Plots of $|\widetilde{q}^{[1]}(x,1.2)|$, as defined by 
\eqref{eq:initial-data} and \eqref{plot-Hk} with parameters 
\eqref{plot-vals1} and $\kappa^{[3]}=1/2$, for $\epsilon\in\{1/80,1/160,1/320\}$ 
(left-to-right, as indicated).  The plot with $\epsilon=1/160$ corresponds to 
the top-right plot in Figure \ref{fig-sc13-1d-collisions}.  Note that in all 
plots $t=1.2$.}
\label{fig-sc13-1d-t1p2}
\end{center}
\end{figure}

In Figure \ref{fig-sc13-1d-t2}, we keep $\kappa^{[3]}=1/2$ and plot 
$|q^{[1]}(x,2)|$ for $0\leq x\leq 1.5$ and various values of 
$\epsilon$.  These 
plots show the final wave profile in channel 1 after interaction (the 
profiles in channels 2 and 3 are indistinguishably close to zero at 
this scale).  The 
corresponding regions in Figure \ref{fig-sc13-1d-collisions} are A, B, E, 
and A.

\begin{figure}[H]
\begin{center}
\includegraphics[height=1.35in]{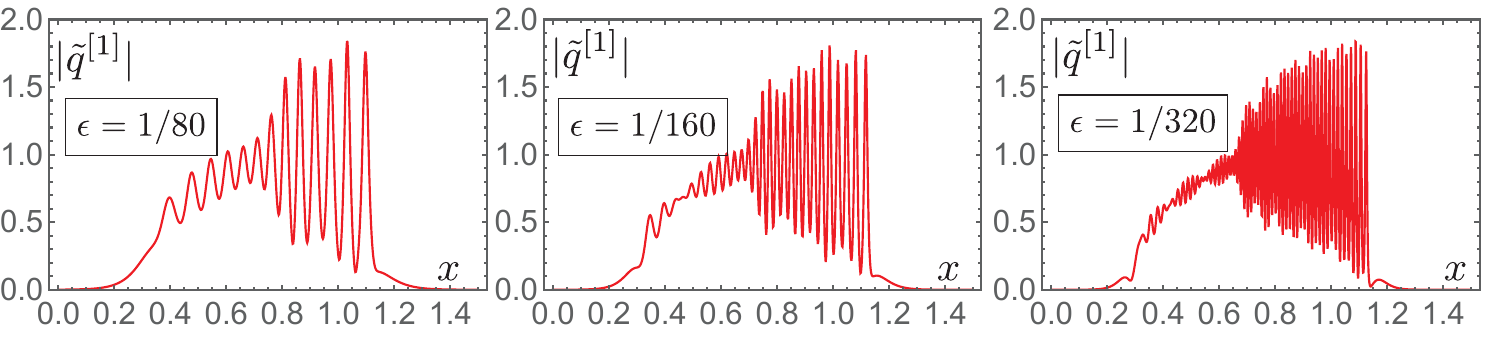}
\caption{Plots of $|\widetilde{q}^{[1]}(x,2)|$, as defined by 
\eqref{eq:initial-data} and \eqref{plot-Hk} with parameters 
\eqref{plot-vals1} and $\kappa^{[3]}=1/2$, for $\epsilon\in\{1/80,1/160,1/320\}$ 
(left-to-right, as indicated).  
Note that in all plots $t=2$.}
\label{fig-sc13-1d-t2}
\end{center}
\end{figure}

\subsection{Dynamics of ensembles associated with three colliding packets}
We now illustrate a semiclassical soliton ensemble associated with three 
semicircles, one in each channel.  Given the large number of parameters, instead of 
attempting a comprehensive study we present one example that gives a flavor the 
type of phenomena seen.  Specifically, we choose initial data of the type 
\eqref{eq:initial-data}
with initial envelopes
\begin{equation}
H^{[k]}(x) :=\frac{2H^{[k]}_\text{max}\chi_{(a^{[k]},b^{[k]})}(x)}{b^{[k]}-a^{[k]}}\sqrt{(x-a^{[k]})(b^{[k]}-x)}, \quad k=1,2,3,
\label{plot-Hk-sc123}
\end{equation}
with parameters 
\eq
\begin{split}
\{c^{[1]}, c^{[2]}, c^{[3]}\}=\{1,0,-1\}, \quad \{\gamma^{[1]}, \gamma^{[2]}, \gamma^{[3]}\} = \{1,-1,1\}, \quad \{\kappa_1,\kappa_2,\kappa_3\} = \left\{-\frac{1}{2},0,\frac{1}{2}\right\}, \\
\{H_\text{max}^{[1]},H_\text{max}^{[2]},H_\text{max}^{[3]}\} = \left\{\frac{3}{4},1,\frac{5}{4}\right\}, \quad \{\theta^{[1]},\theta^{[2]},\theta^{[3]}\} = \{0,0,0\}, \hspace{.9in}\\
\{(a^{[1]},b^{[1]}), (a^{[2]},b^{[2]}), (a^{[3]},b^{[3]}) \} = \left\{\left(-2,-\frac{3}{4}\right),\left(-\frac{1}{2},\frac{1}{2}\right),\left(\frac{5}{4},\frac{7}{4}\right)\right\}. \hspace{.5in}
\end{split}
\label{plot-vals2}
\endeq
Table \ref{N-for-plot-vals2} gives the number of solitons of each type for each 
value of $\epsilon$ used in a plot.
\begin{table}[H]
\renewcommand{\arraystretch}{1.1}
\centering
\begin{tabular}{@{}l@{\qquad}c@{\qquad \qquad}c@{\qquad \qquad}c@{}}
\toprule
$\epsilon$ & $N^{[1]}(\epsilon)$ & $N^{[2]}(\epsilon)$ & $N^{[3]}(\epsilon)$ \\
\midrule
1/20 & 3 & 5 & 2 \\
1/40 & 7 & 10 & 4 \\
1/80 & 13 & 20 & 9 \\
1/160 & 27 & 40 & 18 \\
\bottomrule
\end{tabular}
\caption{The number of solitons of each type assuming initial data 
defined by \eqref{eq:initial-data}, \eqref{plot-Hk}, and \eqref{plot-vals2} for the 
values of $\epsilon$ used in Figures \ref{fig-sc123-2d-e0p0125}, 
\ref{fig-sc123-1d-tpos}, and \ref{fig-sc123-1d-t0}.}
\label{N-for-plot-vals2}
\end{table}
It should be cautioned that for three packets the 
semiclassical soliton ensembles generated by our procedure do \emph{not} 
necessarily converge to the desired initial data (see 
\S\ref{subsec:three-packets-convergence} for more details).  
Figure \ref{fig-sc123-1d-t0} in \S\ref{subsec:three-packets-convergence} suggests 
that the initial data \eqref{plot-vals2} avoid this issue.

The spatiotemporal dynamics of three colliding packets is shown in Figure 
\ref{fig-sc123-2d-e0p0125}.  As with two packets, we observe various spacetime 
regions with different qualitative behaviors.
\begin{figure}[h]
\begin{center}
\includegraphics[height=1.9in]{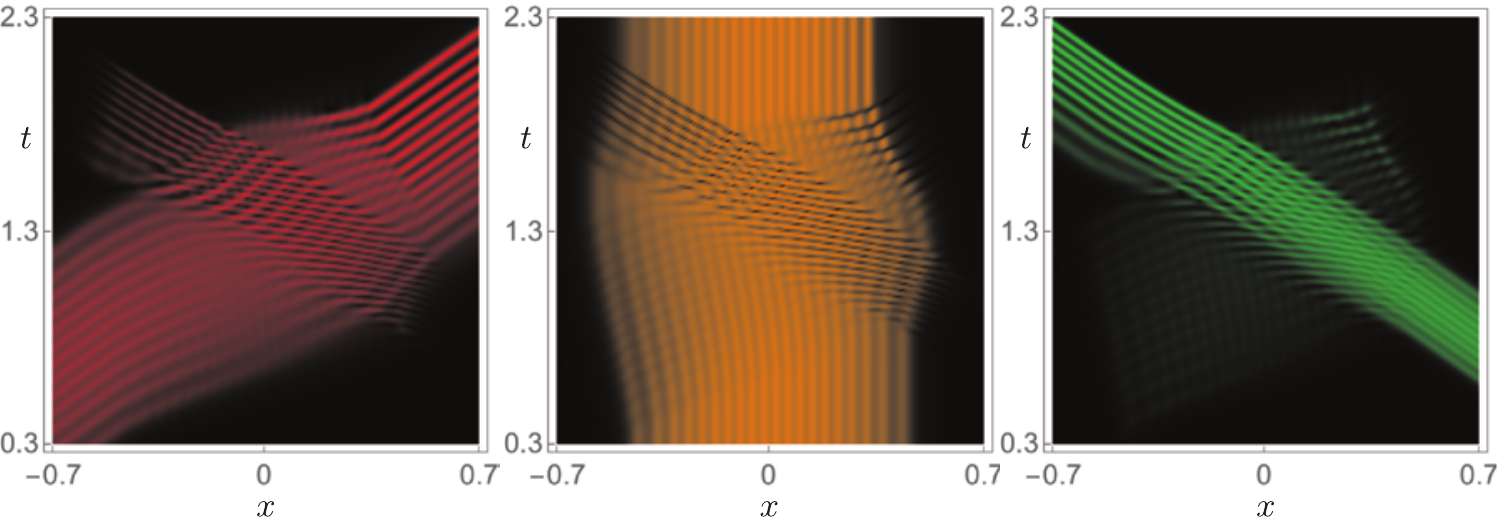}
\caption{Space-time plots of $|\widetilde{q}^{[k]}(x,t)|$, as defined by 
\eqref{eq:initial-data} and \eqref{plot-Hk} with parameters \eqref{plot-vals2} and 
$\epsilon=1/80$, for $k=1,2,3$ (left-to-right).}
\label{fig-sc123-2d-e0p0125}
\end{center}
\end{figure}

\noindent
Three representative time slices are shown in Figure \ref{fig-sc123-1d-tpos}.
\begin{figure}[H]
\begin{center}
\includegraphics[height=1.35in]{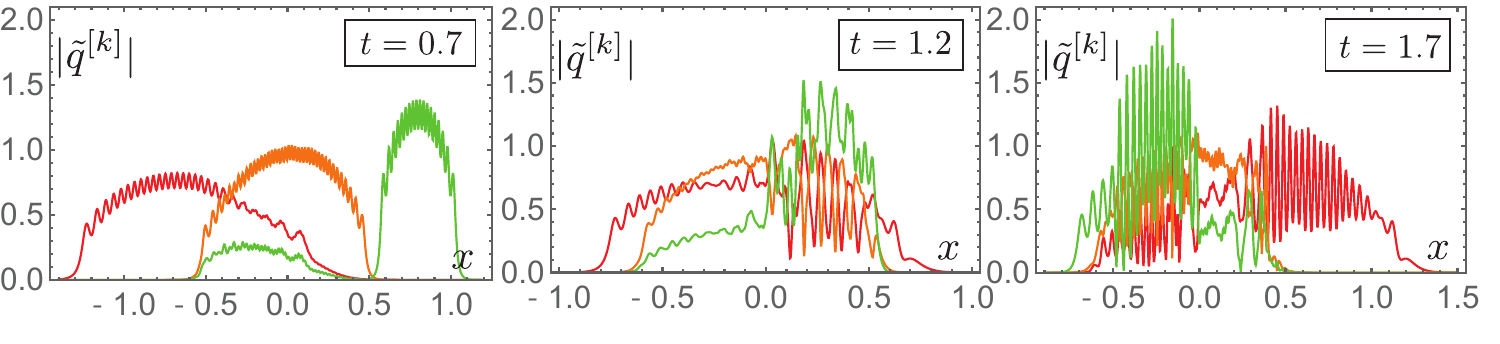}
\caption{Plots of $|\widetilde{q}^{[k]}(x,t)|$, as defined by 
\eqref{eq:initial-data} and \eqref{plot-Hk} with parameters 
\eqref{plot-vals2} and $\epsilon=1/160$, for $t\in\{0.7,1.2,1.7\}$ 
(left-to-right, as indicated).  
Red:  $|\widetilde{q}^{[1]}(x,t)|$. 
Orange:  $|\widetilde{q}^{[2]}(x,t)|$. 
Green:  $|\widetilde{q}^{[3]}(x,t)|$. 
}
\label{fig-sc123-1d-tpos}
\end{center}
\end{figure}

\section{Exact Scattering Data for Disjointly Supported Fields}
\label{sec3-exact-scattering}
\subsection{Summary of the inverse-scattering transform for the TWRI equations}
\label{subsec:lax-pair}
The TWRI equations \eqref{3wave} have the following Lax pair:
\begin{align}
\label{eq:LaxPair-x}	\eps \pd{\Phi}{x} &= \mathcal{L} \Phi, \\
\label{eq:LaxPair-t}	\eps \pd{\Phi}{t} &= \mathcal{B} \Phi,
\end{align}
where we define the $3\times 3$ matrices 
\begin{align}
	\mathcal{L} =\mathcal{L}(x,t;\lambda)&:= -i\lambda \mathbf{C} - \mathbf{Q}(x,t) \label{x flow},\\
	\mathcal{B}=\mathcal{B}(x,t;\lambda) &:= -i\lambda |\mathbf{C}| \mathbf{C}^{-1} + |\mathbf{C}| \mathbf{C}^{-1} \mathbf{Q}(x,t) \mathbf{C}^{-1}. \label{t flow}
\end{align}
Here 
\begin{gather}\label{CQ}
\begin{split}
\mathbf{C}&:=\mathrm{diag}(c^{[1]},c^{[2]},c^{[3]}),\quad |\mathbf{C}|:=\det\mathbf{C},\quad \text{and}\\
	\mathbf{Q}(x,t) &:= \Qpot, 
\end{split}
\intertext{where each $Q^{[k]}$ is a renormalization of $q^{[k]}$ given by}
\label{Q normalization}
Q^{[k]}(x,t):=\sqrt{\frac{\Delta^{[k]}}{\Delta^{[1]}\Delta^{[2]}\Delta^{[3]}}}\,\gamma^{[k]}q^{[k]}(x,t)
\intertext{with}
\Delta^{[1]}:=c^{[2]}-c^{[3]}>0,\quad
\Delta^{[2]}:=c^{[1]}-c^{[3]}>0,\quad\text{and}\quad
\Delta^{[3]}:=c^{[1]}-c^{[2]}>0.
\label{eq:Delta123}
\end{gather} 
Of course $\mathbf{C}^{-1}$ only exists if none of the wave speeds $c^{[k]}$ vanishes, but it is easy to check the apparent singularity of the matrix $\mathcal{B}$ associated with the vanishing of a single wave speed is in fact removable due to the fact that the diagonal elements of $\mathbf{Q}(x,t)$ are all zero.

Let $[\mathbf{A},\mathbf{B}]$ denote the matrix commutator $[\mathbf{A},\mathbf{B}]= \mathbf{A}\mathbf{B}-\mathbf{B}\mathbf{A}$. 
The system \eqref{eq:LaxPair-x}--\eqref{eq:LaxPair-t} has a simultaneous fundamental solution for a given $\lambda\in\mathbb{C}$ if and only if the matrices $\mathcal{L}$ and $\mathcal{B}$ satisfy the (zero-curvature) compatibility condition:
\begin{equation}
\eps\frac{\partial\mathcal{L}}{\partial t}-\eps\frac{\partial\mathcal{B}}{\partial x}
+[\mathcal{L},\mathcal{B}]=\mathbf{0}.
\label{eq:zero-curvature}
\end{equation}
The latter is easily seen to be equivalent (for all $\lambda\in\mathbb{C}$) to 
a matrix reformulation of \eqref{3wave} in terms of the renormalized potentials \eqref{Q normalization}.
This important observation is originally due to Zakharov and Manakov \cite{Zakharov:1973}.

We now summarize the salient points of the inverse-scattering transform based on the Lax representation \eqref{eq:zero-curvature}.  The inverse-scattering transform for the TWRI system was first described by Zakharov and Manakov \cite{Zakharov:1975} and Kaup \cite{Kaup76} with additional important contributions by Beals and Coifman \cite{BealsC84,BealsC85,BealsC87}.
Full details and proofs of the following results can be found in Appendix~\ref{appB-inverse-scattering}.
Let $\sigma=\pm$ be a fixed sign and suppose that $\imag\{\lambda\}\neq 0$.  Given initial data with $q^{[k]}\in L^1(\mathbb{R})$ and $q^{[k]}_x\in L^1(\mathbb{R})$, seek a fundamental matrix solution $\Phi^\sigma$ of \eqref{eq:LaxPair-x} defined by the following conditions on $\mathbf{M}^\sigma(x;\lambda):=\Phi^\sigma(x;\lambda)e^{i\lambda\mathbf{C}x/\epsilon}$:
\begin{equation}
\lim_{x\to \sigma\infty}\mathbf{M}^\sigma(x;\lambda)=\mathbb{I}\quad\text{and}\quad
\sup_{x\in\mathbb{R}}\|\mathbf{M}^\sigma(x;\lambda)\|<\infty.
\end{equation}
These conditions are equivalent to a certain Fredholm integral equation; see \eqref{eq:Fredholm-system} in Appendix~\ref{appB-inverse-scattering}.
The Fredholm equation has a unique solution $\mathbf{M}^\sigma(x;\lambda)$ in $L^\infty(\mathbb{R})$ except for isolated points $\lambda\in\mathbb{C}\setminus\mathbb{R}$ called \emph{exceptional points}.  The solution satisfies the Schwarz-symmetry condition
\begin{equation}
\mathbf{M}^\sigma(x;\lambda^*)=\mathbf{E}
\mathbf{M}^\sigma(x;\lambda)^{-\dagger}\mathbf{E},\quad\imag\{\lambda\}\neq 0,\quad
\mathbf{E}:=\mathrm{diag}(\gamma^{[1]},-\gamma^{[2]},\gamma^{[3]}).
\label{eq:TWRI-Schwarz-symmetry}
\end{equation}
It can be shown that 
\begin{gather}
\lim_{x\to -\infty}\mathbf{M}^+(x;\lambda)=\mathbf{D}(\lambda)\quad\text{and}\quad
\lim_{x\to +\infty}\mathbf{M}^-(x;\lambda)=\mathbf{D}(\lambda)^{-1},
\intertext{and therefore}
\mathbf{M}^+(x;\lambda)=\mathbf{M}^-(x;\lambda)\mathbf{D}(\lambda),\quad \imag\{\lambda\}\neq 0,
\label{eq:Mplus-Mminus-intro}
\intertext{with}
\mathbf{D}(\lambda)=\begin{cases}\mathrm{diag}(1/u(\lambda),u(\lambda)/v(\lambda),v(\lambda)),\quad &\imag\{\lambda\}>0,\\
\mathrm{diag}(u(\lambda),v(\lambda)/u(\lambda),1/v(\lambda)),\quad &\imag\{\lambda\}<0,
\end{cases}
\label{eq:D-representation}
\end{gather}
where $u$ and $v$ are functions analytic in $\mathbb{C}\setminus\mathbb{R}$ with $u(\lambda^*)=u(\lambda)^*$ and $v(\lambda^*)=v(\lambda)^*$.  Their zeros are exactly the singularities of $\mathbf{M}^\sigma(x;\lambda)$, and this actually implies that the exceptional points in $\mathbb{C}\setminus \mathbb{R}$ are all poles.  If $\lambda_0\in\mathbb{C}_+$ is a simple pole of $\mathbf{M}^\sigma(x;\lambda)$, then there exist corresponding nonzero constant matrices
\begin{equation}
\mathbf{N}^+=\begin{pmatrix}0 & 0 & 0\\\beta_{21} & 0 & 0\\\beta_{31} & \beta_{32} & 0\end{pmatrix},\;\beta_{21}\beta_{32}=0,\;\text{and}\;
\mathbf{N}^-=\begin{pmatrix}0 & \beta_{12} & \beta_{13}\\0 & 0 & \beta_{23} \\ 0 & 0 & 0\end{pmatrix},\; \beta_{12}\beta_{23}=0,
\label{eq:Nplus-Nminus-form}
\end{equation}
such that 
\begin{equation}
\mathop{\mathrm{Res}}_{\lambda=\lambda_0}\mathbf{M}^\sigma(x;\lambda)=
\lim_{\lambda\to \lambda_0}\mathbf{M}^\sigma(x;\lambda)e^{-i\lambda_0\mathbf{C}x/\epsilon}
\mathbf{N}^\sigma e^{i\lambda_0\mathbf{C}x/\epsilon}.
\label{eq:M-sigma-residue-summary}
\end{equation}
We call the quantities $\beta_{jk}$ \emph{connection coefficients} for the pole $\lambda_0$.
There is also a \emph{jump matrix} $\mathbf{V}_0^\sigma(\lambda)$ defined for $\lambda\in\mathbb{R}$ such that
\begin{equation}
\mathbf{M}_+^\sigma(x;\lambda)=\mathbf{M}_-^\sigma(x;\lambda)e^{-i\lambda\mathbf{C}x/\epsilon}
\mathbf{V}_0^\sigma(\lambda)e^{i\lambda\mathbf{C}x/\epsilon},\quad \lambda\in\mathbb{R}
\label{eq:M-sigma-jump-intro}
\end{equation}
where $\mathbf{M}^\sigma_\pm(x;\lambda)$ are the boundary values taken by $\mathbf{M}^\sigma(x;\lambda')$ as $\lambda'\to \lambda\in\mathbb{R}$ from $\mathbb{C}_\pm$.

The \emph{scattering data} consists of the jump matrix $\mathbf{V}_0^\sigma(\lambda)$ for $\lambda\in\mathbb{R}$ and the discrete data consisting generically of finitely many pairs $(\lambda_0\in\mathbb{C}_+,\mathbf{N}^\sigma)$.  This data evolves in time in a simple way:
\begin{equation}
\mathbf{V}_0^\sigma(\lambda;t)=e^{-i\lambda\det(\mathbf{C})\mathbf{C}^{-1}t/\epsilon}
\mathbf{V}_0^\sigma(\lambda)e^{i\lambda\det(\mathbf{C})\mathbf{C}^{-1}t/\epsilon},\quad\lambda\in\mathbb{R},
\label{eq:jump-evolution-intro}
\end{equation}
and for the residue matrix $\mathbf{N}^\sigma$ associated to a pole $\lambda_0\in\mathbb{C}_+$:
\begin{equation}
\mathbf{N}^\sigma(t)=e^{-i\lambda_0\det(\mathbf{C})\mathbf{C}^{-1}t/\epsilon}\mathbf{N}^\sigma
e^{i\lambda_0\det(\mathbf{C})\mathbf{C}^{-1}t/\epsilon}.
\label{eq:pole-evolution-intro}
\end{equation}
The poles $\lambda_0\in\mathbb{C}_+$ are independent of $t$.

The solution $q^{[k]}(x,t)$ of the initial-value problem with the given initial data can be recovered from generic scattering data\footnote{Modifications necessary in the case of nongeneric data are outlined in Appendix \ref{subapp-Zhou-transform}.} by solving a $3\times 3$ matrix Riemann-Hilbert problem for an unknown $\mathbf{M}^\sigma(\lambda)=\mathbf{M}^\sigma(x,t;\lambda)$ (on the inverse side we view $x$ and $t$ as real parameters and consider $\lambda$ as the main independent complex variable); see Riemann-Hilbert Problem~\ref{rhp:M-sigma}.  
In the particular decay instability case that the signs $\gamma^{[k]}$ satisfy \eqref{eq:gamma-assumption}, there exists a unique classical solution of Riemann-Hilbert Problem~\ref{rhp:M-sigma} for all $(x,t)\in\mathbb{R}^2$.  The matrix $\mathbf{M}^\sigma(x;\lambda)$ in the direct scattering theory is related to the solution of the Riemann-Hilbert problem by $\mathbf{M}^\sigma(x;\lambda)=\mathbf{M}^\sigma(x,0;\lambda)$.

From the solution $\mathbf{M}^\sigma(x,t;\lambda)$ we may extract a matrix coefficient $\mathbf{F}^\sigma(x,t)$ by
\begin{equation}
\mathbf{M}^\sigma(x,t;\lambda)=\mathbb{I}+\mathbf{F}^\sigma(x,t)\lambda^{-1} + o(\lambda^{-1}),\quad\lambda\to\infty
\end{equation}
and the solution of the initial-value problem for the TWRI system is:
\begin{equation}
\begin{split}
q^{[1]}(x,t)&=-i\gamma^{[1]}\Delta^{[1]}\sqrt{\Delta^{[2]}\Delta^{[3]}}F^\sigma_{23}(x,t)\\
q^{[2]}(x,t)&=i\gamma^{[2]}\Delta^{[2]}\sqrt{\Delta^{[1]}\Delta^{[3]}}F_{31}^\sigma(x,t)\\
q^{[3]}(x,t)&=-i\gamma^{[3]}\Delta^{[3]}\sqrt{\Delta^{[1]}\Delta^{[2]}}F_{12}^\sigma(x,t).
\end{split}
\end{equation}

By contrast with $\mathbf{M}^\sigma(x;\lambda)$, the \emph{Jost solutions} $\mathbf{M}^\sigma_\mathrm{J}(x;\lambda)=\Phi^\sigma_\mathrm{J}(x;\lambda)e^{i\lambda\mathbf{C}x/\epsilon}$ for the problem are generally defined for $\lambda\in\mathbb{R}$ only.  They satisfy Volterra equations; see \eqref{eq:TWRI-Jost-Volterra} in Appendix~\ref{appB-inverse-scattering}.
A \emph{scattering matrix} $\mathbf{S}(\lambda)$ is then defined from the Jost solutions by
\begin{equation}
\mathbf{M}_\mathrm{J}^+(x;\lambda)=\mathbf{M}_\mathrm{J}^-(x;\lambda)e^{-i\lambda\mathbf{C}x/\epsilon}\mathbf{S}(\lambda)e^{i\lambda\mathbf{C}x/\epsilon},\quad\lambda\in\mathbb{R}.
\label{eq:scattering-relation}
\end{equation}
We have the following ``LDU'' and ``UDL'' factorizations:
\begin{equation}
\begin{split}
\mathbf{S}(\lambda)^{-1}&=\begin{pmatrix}1&0&0\\
T_{21}^+(\lambda)&1&0\\
T_{31}^+(\lambda)&T_{32}^+(\lambda)&1\end{pmatrix}\mathbf{D}_+(\lambda)^{-1}
\begin{pmatrix}1&T_{12}^+(\lambda) & T_{13}^+(\lambda)\\
0 & 1 & T_{23}^+(\lambda)\\
0&0&1\end{pmatrix}\\
\mathbf{S}(\lambda)&=
\begin{pmatrix}1&T_{12}^-(\lambda) & T_{13}^-(\lambda)\\
0 & 1 & T_{23}^-(\lambda)\\ 0 & 0 & 1\end{pmatrix}\mathbf{D}_+(\lambda)
\begin{pmatrix}1 & 0 & 0\\
T_{21}^-(\lambda)&1&0\\
T_{31}^-(\lambda)&T_{32}^-(\lambda)&1\end{pmatrix}.
\end{split}
\label{eq:LDU-UDL}
\end{equation}
Here, $\mathbf{D}_+(\lambda)$ refers to the boundary value from the upper half-plane taken by the diagonal matrix $\mathbf{D}(\lambda)$.
Then, the jump matrices are expressed in terms of the triangular factors by
\begin{equation}
\begin{split}
\mathbf{V}_0^+(\lambda)&:=\begin{pmatrix}1&-\gamma^{[1]}\gamma^{[2]}T_{21}^+(\lambda)^* & 
\gamma^{[1]}\gamma^{[3]}T_{31}^+(\lambda)^* \\
0 & 1 & -\gamma^{[2]}\gamma^{[3]}T_{32}^+(\lambda)^*\\0&0&1\end{pmatrix}
\begin{pmatrix}1&0&0\\T_{21}^+(\lambda) & 1 & 0\\T_{31}^+(\lambda)& T_{32}^+(\lambda) & 1
\end{pmatrix}\quad\text{and}\\
\mathbf{V}_0^-(\lambda)&:=\begin{pmatrix}1 & 0 & 0\\
-\gamma^{[1]}\gamma^{[2]}T_{12}^-(\lambda)^* & 1 & 0\\
\gamma^{[1]}\gamma^{[3]}T_{13}^-(\lambda)^* & -\gamma^{[2]}\gamma^{[3]}T_{23}^-(\lambda)^* & 1
\end{pmatrix}
\begin{pmatrix} 1 & T_{12}^-(\lambda)& T_{13}^-(\lambda)\\
0 & 1 & T_{23}^-(\lambda)\\ 0 & 0 & 1\end{pmatrix}.
\end{split}
\label{eq:TWRI-Jump}
\end{equation}
Also, the boundary values taken on $\mathbb{R}$ by the scalar functions $u$ and $v$ analytic in $\mathbb{C}_+$ are
\begin{equation}
u_+(\lambda)=[\mathbf{S}(\lambda)^{-1}]_{11}=S_{11}(\lambda)^*\quad\text{and}\quad
v_+(\lambda)=S_{33}(\lambda),\quad\lambda\in\mathbb{R}.
\end{equation}

If $\lambda_0\in\mathbb{C}_+$ is a simple zero of $v(\lambda)$ for which $u(\lambda_0)\neq 0$, then we say that $\lambda_0$ is a \emph{simple pole of type $1$}, and the residue matrices take the form
\begin{gather}
\label{eq:type-1-residues}
\mathbf{N}^+=\begin{pmatrix}0&0&0\\0&0&0\\0&\beta^{[1]}_{32}&0\end{pmatrix}\quad\text{and}\quad
\mathbf{N}^-=\begin{pmatrix}0&0&0\\0&0&\beta^{[1]}_{23}\\0&0&0\end{pmatrix},
\intertext{where the connection coefficients are related by}
\label{eq:type-1-beta-relations}
\beta^{[1]}_{23}\beta^{[1]}_{32}=u(\lambda_0)/v'(\lambda_0)^2.
\end{gather}
If $\mathbf{S}(\lambda)$ is diagonal and there are no other poles in $\mathbb{C}_+$, the corresponding solution is a pure soliton of $q^{[1]}(x,t)$ with $q^{[2]}(x,t)=q^{[3]}(x,t)\equiv 0$; see \eqref{eq:pure-1}.  It is a traveling wave, a function of $x-c^{[1]}t$ only.
Likewise, if $\lambda_0\in\mathbb{C}_+$ is a simple zero of $u(\lambda)$ for which $v(\lambda_0)\neq 0$, 
then we say that $\lambda_0$ is a \emph{simple pole of type $3$}, and the residue matrices take the form
\begin{gather}
\label{eq:type-3-residues}
\mathbf{N}^+=\begin{pmatrix}0&0&0\\\beta^{[3]}_{21}&0&0\\0&0&0\end{pmatrix}\quad\text{and}\quad
\mathbf{N}^-=\begin{pmatrix}0&\beta^{[3]}_{12}&0\\0&0&0\\0&0&0\end{pmatrix},
\intertext{where the connection coefficients are related by}
\label{eq:type-3-beta-relations}
\beta^{[3]}_{12}\beta^{[3]}_{21}=v(\lambda_0)/u'(\lambda_0)^2.
\end{gather}
If $\mathbf{S}(\lambda)$ is diagonal and there are no other poles in $\mathbb{C}_+$, the corresponding solution is a pure soliton of $q^{[3]}(x,t)$ with $q^{[1]}(x,t)=q^{[2]}(x,t)\equiv 0$; see \eqref{eq:pure-3}.  It is a traveling wave, a function of $x-c^{[3]}t$ only.

If $\lambda_0\in\mathbb{C}_+$ is a simultaneous simple zero of both $u(\lambda)$ and $v(\lambda)$ there are three possibilities:
\begin{itemize}
\item $\lambda_0$ is a \emph{simple pole of type $2$} if
\begin{gather}
\label{eq:type-2-residues}
\mathbf{N}^+=\begin{pmatrix}0&0&0\\0&0&0\\\beta^{[2]}_{31}&0&0\end{pmatrix}\quad\text{and}\quad
\mathbf{N}^-=\begin{pmatrix}0&0&\beta^{[2]}_{13}\\0&0&0\\0&0&0\end{pmatrix},
\intertext{where the connection coefficients are related by}
\label{eq:type-2-beta-relations}
\beta^{[2]}_{13}\beta^{[2]}_{31}=1/(u'(\lambda_0)v'(\lambda_0)).
\end{gather} 
In the absence of other scattering data, this is a pure soliton of $q^{[2]}(x,t)$, with $q^{[1]}(x,t)=q^{[3]}(x,t)\equiv 0$, a traveling wave depending only on $x-c^{[2]}t$; see \eqref{eq:pure-2}.
\item $\lambda_0$ is a \emph{simple pole of type $\solsplit$} if for $\beta_{31}^{[\solsplit]}$, $\beta_{32}^{[\solsplit]}$, $\beta_{12}^{[\solsplit]}$, and $\beta_{13}^{[\solsplit]}$ all nonzero,
\begin{gather}
\label{eq:type-solsplit-residues}
\mathbf{N}^+=\begin{pmatrix}0&0&0\\0&0&0\\\beta^{[\solsplit]}_{31}&\beta^{[\solsplit]}_{32}&0\end{pmatrix}\quad\text{and}\quad\mathbf{N}^-=\begin{pmatrix}0&\beta^{[\solsplit]}_{12}&\beta^{[\solsplit]}_{13}\\0&0&0\\0&0&0\end{pmatrix},
\intertext{where the connection coefficients are related by}
\label{eq:type-solsplit-beta-relations}
\beta^{[\solsplit]}_{13}\beta^{[\solsplit]}_{31}=1/(u'(\lambda_0)v'(\lambda_0))\quad\text{and}\quad\beta^{[\solsplit]}_{12}=v'(\lambda_0)\beta^{[\solsplit]}_{32}/(u'(\lambda_0)^2\beta^{[\solsplit]}_{31}).
\end{gather}
In absence of other scattering data, this represents a ``splitting'' of a type $2$ soliton into a sum of solitons of types $1$ and $3$ as $t$ increases; see \eqref{eq:split-1}--\eqref{eq:split-3}.
\item $\lambda_0$ is a \emph{simple pole of type $\solfuse$} if for $\beta_{21}^{[\solfuse]}$, $\beta_{31}^{[\solfuse]}$, $\beta_{13}^{[\solfuse]}$, and $\beta_{23}^{[\solfuse]}$ all nonzero,
\begin{gather}
\label{eq:type-solfuse-residues}
\mathbf{N}^+=\begin{pmatrix}0&0&0\\\beta^{[\solfuse]}_{21}&0&0\\\beta^{[\solfuse]}_{31}&0&0\end{pmatrix}\quad\text{and}\quad
\mathbf{N}^-=\begin{pmatrix}0&0&\beta^{[\solfuse]}_{13}\\0&0&\beta^{[\solfuse]}_{23}\\0&0&0\end{pmatrix},
\intertext{where the connection coefficients are related by}
\label{eq:type-solfuse-beta-relations}
\beta^{[\solfuse]}_{13}\beta^{[\solfuse]}_{31}=1/(u'(\lambda_0)v'(\lambda_0))\quad\text{and}\quad\beta^{[\solfuse]}_{23}=-u'(\lambda_0)\beta^{[\solfuse]}_{21}/(v'(\lambda_0)^2\beta^{[\solfuse]}_{31}).
\end{gather}
In absence of other scattering data, this represents the reverse process; a ``fusion'' of solitons of types $1$ and $3$ into a single soliton of type $2$ as $t$ increases; see \eqref{eq:fuse-1}--\eqref{eq:fuse-3}.
\end{itemize}
These soliton solutions all exist globally in the Schwartz space $\mathscr{S}(\mathbb{R})$ provided the signs $\gamma^{[k]}$ satisfy the conditions \eqref{eq:gamma-assumption}.  For all other choices of signs, four of the five elementary types of solitons blow up in $L^\infty_\mathrm{loc}(\mathbb{R})$ for at least some real $t$; for all decay instability cases except \eqref{eq:gamma-assumption} the singularities are persistent for all $t\in\mathbb{R}$ which simply rules out these as physically relevant solutions, while in the case of explosive instability there exists a blow-up time $t_*\in\mathbb{R}$ such that the soliton of type $\solsplit$ is in $\mathscr{S}(\mathbb{R})$ for all $t<t_*$ but is not in $L^\infty_\mathrm{loc}(\mathbb{R})$ for $t\ge t_*$ and such that the soliton of type $\solfuse$ is in $\mathscr{S}(\mathbb{R})$ for all $t>t_*$ but is not in $L^\infty_\mathrm{loc}(\mathbb{R})$ for $t\le t_*$.  (This observation goes back at least to \cite{Zakharov:1973}.)  See Figure~\ref{fig:SolitonSingularities}.
\begin{figure}[h]
\begin{center}
\includegraphics{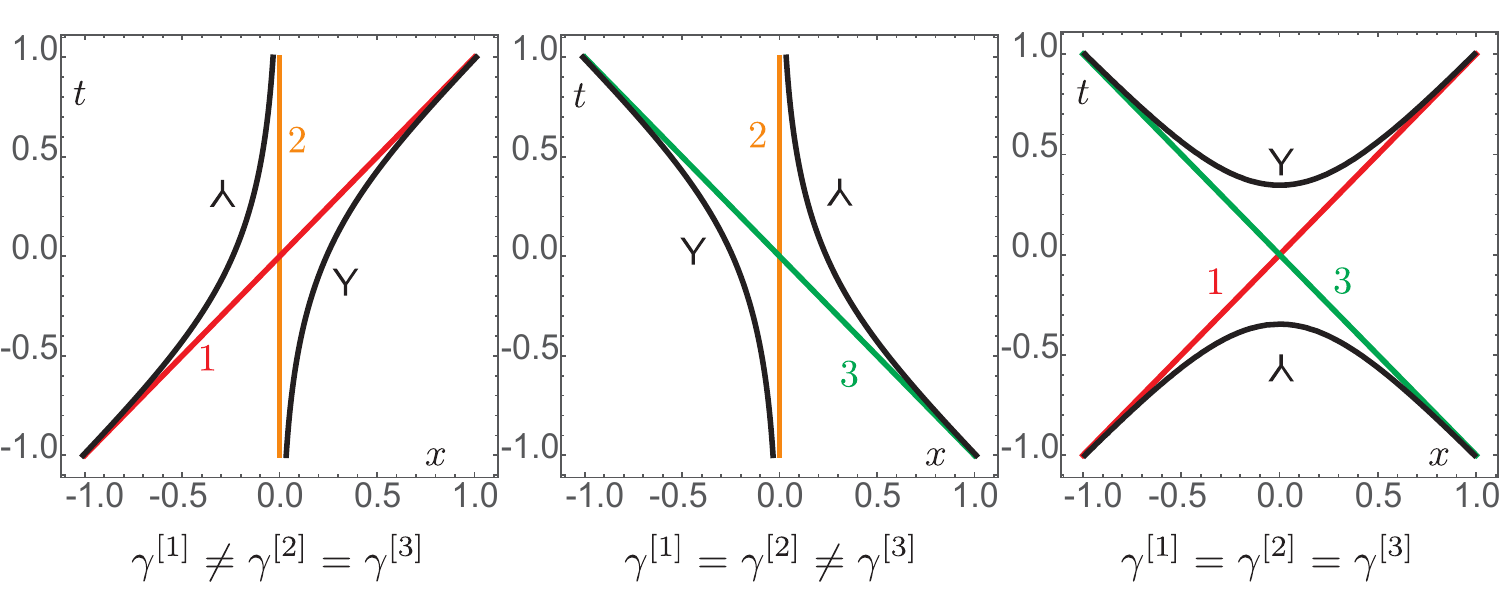}
\end{center}
\caption{The blow-up locus in the $(x,t)$-plane for the solitons of types $1$, $2$, $3$, $\mathsf{Y}$, and \protect\rotatebox[origin=c]{180}{$\mathsf{Y}$}.  In the decay instability cases $\gamma^{[1]}\neq\gamma^{[2]}=\gamma^{[3]}$ (left) and $\gamma^{[1]}=\gamma^{[2]}\neq\gamma^{[3]}$ (center), only the type $3$ (respectively, type $1$) soliton is bounded for all $(x,t)$.  In the explosive instability case (right), only the type $2$ soliton is bounded for all $(x,t)$, and the solitons of types  $\mathsf{Y}$ and \protect\rotatebox[origin=c]{180}{$\mathsf{Y}$} exhibit finite-time blowup.  Here $\epsilon=1$ and the speeds are $c^{[1]}=1$, $c^{[2]}=0$, and $c^{[3]}=-1$, while the pole data is $\imag\{\lambda_0\}=1$, $\beta_{jk}^{[\mathrm{type}]}=2i$.}
\label{fig:SolitonSingularities}
\end{figure}
When there are several simple poles of different types, one has in the usual way a kind of nonlinear superposition of the associated elementary solitons combined with a radiation field generated by the jump condition across the real line.  However, the non-dispersive character of the TWRI system means that unlike the solitons of other familiar integrable equations such as the Korteweg-de Vries or nonlinear Schr\"odinger equation, the fundamental TWRI solitons of types 1, 2, and 3 do not represent a balance between nonlinear and dispersive effects because \emph{any} initial condition in which two of the fields are identically zero will propagate without change of form at the constant velocity of the third excited channel.  These solitons (and their nonlinear superpositions provided only one of the three types is included) therefore are more of a mathematical artifact of the inverse-scattering solution method than coherent structures in the usual sense.  The situation is of course different for the fundamental solitons of type $\solsplit$ and $\solfuse$ for which all three fields are typically nonzero at each $(x,t)\in\mathbb{R}^2$, and each of which describes a basic physical mode of the three-wave interaction.

In the fundamental papers \cite{Zakharov:1975,Kaup76}, the constants $\beta_{jk}$ associated with a simple pole $\lambda_0\in\mathbb{C}_+$ are characterized in terms of off-diagonal elements of the scattering matrix $\mathbf{S}(\lambda)$ evaluated for $\lambda=\lambda_0\in\mathbb{C}_+$.  There are two main difficulties with such formulae:
\begin{itemize}
\item They require sufficient decay of the initial data in $x$ to admit analytic continuation of $\mathbf{S}(\lambda)$ from $\mathbb{R}$ to a neighborhood of $\lambda_0\in\mathbb{C}_+$.  For example,  $\mathbf{S}(\lambda)$ is entire for compactly supported initial data, but for general Schwartz-class data $\mathbf{S}(\lambda)$ is nowhere analytic on $\mathbb{R}$.
\item Even if the required analyticity is present, it is difficult to accurately compute the matrix $\mathbf{S}(\lambda)$ in the complex plane.  This difficulty becomes worse as $\epsilon\downarrow 0$.
\end{itemize}

Therefore, before proceeding to study the semiclassical limit $\epsilon\downarrow 0$, it is necessary to first develop suitable exact formulae for the connection coefficients $\beta_{jk}$ associated to a pole $\lambda_0\in\mathbb{C}_+$.  This will be done in \S\ref{section:3wave-residues} in the special case of initial data with disjoint supports:  
\begin{equation}
\mathrm{supp}(q^{[1]}(x,0))<x_{12}<\mathrm{supp}(q^{[2]}(x,0))<x_{23}<\mathrm{supp}(q^{[3]}(x,0)).
\label{eq:disjoint-support-assumption}
\end{equation}
Thus, the fastest-moving packet is initially supported furthest to the left, the slowest-moving packet is initially supported furthest to the right, and the packet with the intermediate speed is initially supported in between.
There is a strictly positive time $T$ before which the quadratic nonlinear terms on the right-hand side of the TWRI system \eqref{3wave} vanish identically, and hence the solution is exactly (and independently of $\eps>0$)
\begin{equation}
q^{[k]}(x,t)=q^{[k]}(x-c^{[k]}t,0),\quad k=1,2,3,\quad t<T.
\end{equation}
Since $c^{[1]}>c^{[2]}>c^{[3]}$, only after time $T$ will the packets collide and 
the nonlinear interaction become important.
In this situation, the scattering matrix obtained from the disjointly supported initial data has a natural factorization $\mathbf{S}(\lambda)=\mathbf{S}^{[1]}(\lambda)\mathbf{S}^{[2]}(\lambda)\mathbf{S}^{[3]}(\lambda)$ in which $\mathbf{S}^{[k]}(\lambda)$ only depends on $q^{[k]}(x,0)$ (see Proposition~\ref{prop-scattering-factorize}).  
The matrix factors $\mathbf{S}^{[k]}(\lambda)$ take the following form:
\begin{equation}
\mathbf{S}^{[1]}=\begin{pmatrix}1&0&0\\0 & S_{22}^{[1]} & S_{23}^{[1]}\\
0 & S_{32}^{[1]} & S_{33}^{[1]}\end{pmatrix},\;
\mathbf{S}^{[2]}=\begin{pmatrix}S_{11}^{[2]} & 0 & S_{13}^{[2]}\\
0 & 1 & 0\\
S_{31}^{[2]} & 0 & S_{33}^{[2]}\end{pmatrix},\;
\mathbf{S}^{[3]}=\begin{pmatrix}S_{11}^{[3]} & S_{12}^{[3]} & 0\\
S_{21}^{[3]} & S_{22}^{[3]} & 0\\ 0 & 0 & 1\end{pmatrix}.
\label{eq:S-factors-special-form}
\end{equation}
Therefore, 
\begin{equation}
u_+(\lambda)=S_{33}^{[2]}(\lambda)S_{22}^{[3]}(\lambda)\quad\text{and}\quad
v_+(\lambda)=S_{33}^{[1]}(\lambda)S_{33}^{[2]}(\lambda),\quad \lambda\in\mathbb{R}.
\label{eq:u-plus-v-plus-partials}
\end{equation}
It follows that if $S_{33}^{[1]}$, $S_{33}^{[2]}$, and $S_{22}^{[3]}$ have distinct simple zeros in $\mathbb{C}_+$, then
\begin{itemize}
\item The roots of $S_{33}^{[1]}$ in $\mathbb{C}_+$ --- determined from $q^{[1]}(x,0)$ alone --- are those of $v$ that are not zeros of $u$, i.e., simple poles of type $1$.
\item The roots of $S_{22}^{[3]}$ in $\mathbb{C}_+$ --- determined from $q^{[3]}(x,0)$ alone --- are those of $u$ that are not zeros of $v$, i.e., simple poles of type $3$.
\item The roots of $S_{33}^{[2]}$  in $\mathbb{C}_+$ --- determined from $q^{[2]}(x,0)$ alone --- are  simultaneous zeros of both $u$ and $v$, i.e., simple poles of type $2$, $\solsplit$, or $\solfuse$.
\end{itemize}
In this way, the soliton content of the solution is determined separately and individually by the three initial waves $q^{[k]}(x,0)$.  Note that, although the locations of poles generated by $q^{[k]}(x,0)$ are completely independent of $q^{[j]}(x,0)$ for $j\neq k$, it will be shown in \S\ref{section:3wave-residues} that the connection coefficients for those poles actually depend on all three potentials.  

A key point emphasized by Kaup \cite{Kaup76} is that under the support assumption \eqref{eq:disjoint-support-assumption}, the partial scattering matrices $\mathbf{S}^{[k]}(\lambda)$ can be calculated from the analysis of three different $2\times 2$ first-order systems of Zakharov-Shabat type.  To recount the relevant formulas and also to derive exact (and later asymptotic) formulas for the connection coefficients given the assumption \eqref{eq:disjoint-support-assumption}, we now need to recall certain definitions related to the nonselfadjoint Zakharov-Shabat spectral problem.  Further details can be found in Appendix~\ref{appC-nonselfadjoint-zs}.
Let $A:\mathbb{R}\to\mathbb{R}$ be a real function in $L^1(\mathbb{R})$.  The Jost solutions $\mathbf{W}^\pm=(\mathbf{w}^{\pm,1},\mathbf{w}^{\pm,2})$ for the Zakharov-Shabat problem are defined by
\begin{equation}
\epsilon\frac{d\mathbf{W}^\pm}{dx}=\begin{pmatrix}-i\zeta & A(x)\\-A(x) & i\zeta\end{pmatrix}\mathbf{W}^\pm,\;
\lim_{x\to\pm\infty}\mathbf{W}^\pm(x;\zeta)e^{i\zeta\sigma_3x/\epsilon}=\mathbb{I},\;\zeta\in\mathbb{R}.
\label{eq:ZS-Jost}
\end{equation}
The scattering matrix $\mathbf{S}^\mathrm{ZS}(\zeta)$ is defined by $\mathbf{W}^+(x;\zeta)=\mathbf{W}^-(x;\zeta)\mathbf{S}^\mathrm{ZS}(\zeta)$ and takes the form
\begin{equation}
\mathbf{S}^\mathrm{ZS}(\zeta)=\begin{pmatrix}a(\zeta)^* & b(\zeta)^*\\-b(\zeta) & a(\zeta)\end{pmatrix},
\quad |a(\zeta)|^2+|b(\zeta)|^2=1,\quad\zeta\in\mathbb{R}.
\label{eq:ZS-scat-form}
\end{equation}
The vector solutions $\mathbf{w}^{-,1}(x;\zeta)$ and $\mathbf{w}^{+,2}(x;\zeta)$ are boundary values of functions analytic for $\imag\{\zeta\}>0$, and $a(\zeta)=\det(\mathbf{w}^{-,1}(x;\zeta),\mathbf{w}^{+,2}(x;\zeta))$ is also.  Its zeros in $\mathbb{C}_+$ are $L^2(\mathbb{R})$ eigenvalues for \eqref{eq:ZS-Jost}; there is for each such $\zeta_0\in\mathbb{C}_+$ a nonzero proportionality constant $\tau$ such that
\begin{equation}
\mathbf{w}^{-,1}(x;\zeta_0)=\tau\mathbf{w}^{+,2}(x;\zeta_0),\quad a(\zeta_0)=0,\quad\imag\{\zeta_0\}>0,
\label{eq:tau-def-1}
\end{equation}
and since $\imag\{\zeta_0\}>0$ the left- and right-hand sides of this equation are nonzero vector solutions of \eqref{eq:ZS-Jost} exhibiting exponential decay in opposite directions as $|x|\to\infty$.
The proportionality constant $\tau$ is the Zakharov-Shabat analogue of the connection coefficients $\beta_{jk}$ in the TWRI system.

\subsection{Expression of the TWRI jump matrix in terms of Zakharov-Shabat data}
\label{subsec:zs-matrix}
Recall the form \eqref{eq:initial-data} of the initial data, as well as the support assumption \eqref{eq:disjoint-support-assumption}.
By definition, the scattering matrix $\mathbf{S}^{[1]}(\lambda)$ is that of the TWRI spectral problem with potential $\mathbf{Q}(x)$ replaced by the ``cutoff'' potential $\mathbf{Q}^{[1]}(x):=\mathbf{Q}(x)\chi_{(-\infty,x_{12})}(x)$.  Since according to \eqref{eq:disjoint-support-assumption} only the 23 and 32 entries of the matrix $\mathbf{Q}^{[1]}(x)$ are nonzero for $x\in\mathbb{R}$, the corresponding TWRI Jost solutions $\Phi_{\mathrm{J}}^{[1]\pm}(x;\lambda)$ have the block-diagonal form
\begin{equation}
\Phi_{\mathrm{J}}^{[1]\pm}(x;\lambda)=\begin{pmatrix}e^{-i\lambda c^{[1]}x/\epsilon} & \mathbf{0}^\trans\\
\mathbf{0} & \omega_1^{\sigma_3}e^{-i\lambda (c^{[2]}+c^{[3]})x/(2\epsilon)}e^{i\kappa^{[1]}\sigma_3x/(2\epsilon)}\mathbf{W}^{[1]\pm}(x;\lambda)\omega_1^{-\sigma_3}\end{pmatrix},
\end{equation}
where $\omega_1^2=-\gamma^{[1]}e^{i\theta^{[1]}}$ and where $\mathbf{W}^{[1]\pm}(x;\lambda)$ are matrix solutions of the non-selfadjoint Zakharov-Shabat system \eqref{eq:ZS-Jost} with potential $A(x)=A^{[1]}(x)$ and spectral parameter $\zeta=\zeta^{[1]}(\lambda)$, 
where 
\begin{equation}
A^{[1]}(x):=\frac{H^{[1]}(x)}{\sqrt{\Delta^{[2]}\Delta^{[3]}
}}\quad\text{and}\quad\zeta^{[1]}(\lambda):=\frac{\Delta^{[1]}}{2}\lambda + \frac{\kappa^{[1]}}{2},
\label{eq:A1-zeta1}
\end{equation}
and where we have used \eqref{eq:gamma-assumption} and reality of $H^{[1]}$.  Note that due to \eqref{eq:c-inequalities}, $\zeta^{[1]}$ is an affine transformation of the upper half-plane onto itself that preserves $\mathbb{R}$.  Since $\lambda\in\mathbb{R}$ for $\Phi_{\mathrm{J}}^{[1]\pm}(x;\lambda)$ we also have $\zeta^{[1]}(\lambda)\in\mathbb{R}$, and the boundary conditions satisfied by $\Phi_{\mathrm{J}}^{[1]\pm}(x;\lambda)$ as $x\to\pm\infty$ show that
$\mathbf{W}^{[1]\pm}(x;\lambda)$ are precisely the Jost matrices of the Zakharov-Shabat system \eqref{eq:ZS-Jost}, linked by a corresponding Zakharov-Shabat scattering matrix $\mathbf{S}^{\mathrm{ZS}[1]}(\zeta^{[1]}(\lambda))$ whose elements involve functions $a^{[1]}(\zeta^{[1]}(\lambda))$ and $b^{[1]}(\zeta^{[1]}(\lambda))$ according to \eqref{eq:ZS-scat-form}.  The essential elements of the scattering matrix $\mathbf{S}^{[1]}(\lambda)$ are therefore given in terms of Zakharov-Shabat data by
\begin{equation}
\begin{split}
\begin{pmatrix} S^{[1]}_{22}(\lambda) & S^{[1]}_{23}(\lambda)\\S^{[1]}_{32}(\lambda) & S^{[1]}_{33}(\lambda)\end{pmatrix}&= \omega_1^{\sigma_3}\mathbf{S}^{\mathrm{ZS}[1]}(\zeta^{[1]}(\lambda))\omega_1^{-\sigma_3} \\ &= \begin{pmatrix}a^{[1]}(\zeta^{[1]}(\lambda))^* & -\gamma^{[1]}e^{i\theta^{[1]}}b^{[1]}(\zeta^{[1]}(\lambda))^*\\
\gamma^{[1]}e^{-i\theta^{[1]}}b^{[1]}(\zeta^{[1]}(\lambda)) & a^{[1]}(\zeta^{[1]}(\lambda))\end{pmatrix},\quad \lambda\in\mathbb{R}.
\end{split}
\label{eq:S1-pieces}
\end{equation}

Similarly, to calculate $\mathbf{S}^{[2]}(\lambda)$, the TWRI scattering matrix for the cutoff potential $\mathbf{Q}^{[2]}(x):=\mathbf{Q}(x)\chi_{(x_{12},x_{23})}(x)$ for which only the $13$ and $31$ entries are nonzero for $x\in\mathbb{R}$, observe that the Jost solutions $\Phi_{\mathrm{J}}^{[2]\pm}(x;\lambda)$ corresponding to this compactly supported potential have the form
\begin{multline}
\Phi_{\mathrm{J}}^{[2]\pm}(x;\lambda)=\\
\begin{pmatrix}
e^{-i[\lambda(c^{[1]}+c^{[3]})+\kappa^{[2]}]x/(2\epsilon)}W_{11}^{[2]\pm}(x;\lambda) & 0 & \omega_2^2e^{-i[\lambda(c^{[1]}+c^{[3]})+\kappa^{[2]}]x/(2\epsilon)}W_{12}^{[2]\pm}(x;\lambda)\\
0 & e^{-i\lambda c^{[2]}x/\epsilon} & 0\\
\omega_2^{-2}e^{-i[\lambda(c^{[1]}+c^{[3]})-\kappa^{[2]}]x/(2\epsilon)}W_{21}^{[2]\pm}(x;\lambda) & 0 & 
e^{-i[\lambda(c^{[1]}+c^{[3]})-\kappa^{[2]}]x/(2\epsilon)}W_{22}^{[2]\pm}(x;\lambda)\end{pmatrix},
\end{multline}
where $\omega_2^2=-\gamma^{[2]}e^{-i\theta^{[2]}}$ and where $\mathbf{W}^{[2]\pm}(x;\lambda)$ are exactly the Jost matrices of the non-selfadjoint Zakharov-Shabat system \eqref{eq:ZS-Jost} with potential $A(x)=A^{[2]}(x)$ and spectral parameter $\zeta=\zeta^{[2]}(\lambda)$,
where
\begin{equation}
A^{[2]}(x):=\frac{H^{[2]}(x)}{\sqrt{\Delta^{[3]}\Delta^{[1]}}}\quad\text{and}\quad\zeta^{[2]}(\lambda):=\frac{\Delta^{[2]}}{2}\lambda-\frac{\kappa^{[2]}}{2}.
\label{eq:A2-zeta2}
\end{equation}
Again $\zeta^{[2]}$ preserves the upper half-plane taking $\mathbb{R}$ to itself.  Letting $a^{[2]}(\zeta^{[2]}(\lambda))$ and $b^{[2]}(\zeta^{[2]}(\lambda))$ denote the independent elements of the Zakharov-Shabat scattering matrix $\mathbf{S}^{\mathrm{ZS}[2]}(\zeta^{[2]}(\lambda))$ linking $\mathbf{W}^{[2]\pm}(x;\lambda)$,
we obtain
\begin{equation}
\begin{split}
\begin{pmatrix}S^{[2]}_{11}(\lambda) & S^{[2]}_{13}(\lambda)\\S^{[2]}_{31}(\lambda) & S^{[2]}_{33}(\lambda)\end{pmatrix}&=
\omega_2^{\sigma_3}\mathbf{S}^{\mathrm{ZS}[2]}(\zeta^{[2]}(\lambda))\omega_2^{-\sigma_3}\\
&=\begin{pmatrix}a^{[2]}(\zeta^{[2]}(\lambda))^* & -\gamma^{[2]}e^{-i\theta^{[2]}}b^{[2]}(\zeta^{[2]}(\lambda))^*\\
\gamma^{[2]}e^{i\theta^{[2]}}b^{[2]}(\zeta^{[2]}(\lambda)) & a^{[2]}(\zeta^{[2]}(\lambda))\end{pmatrix},\quad\lambda\in\mathbb{R}.
\end{split}
\label{eq:S2-pieces}
\end{equation}

Finally, to find $\mathbf{S}^{[3]}(\lambda)$ associated with the cutoff potential $\mathbf{Q}^{[3]}(x):=\mathbf{Q}(x)\chi_{(x_{23},+\infty)}(x)$ for which only the $12$ and $21$ entries are nonzero for $x\in\mathbb{R}$, note that the Jost solutions $\Phi_{\mathrm{J}}^{[3]\pm}(x;\lambda)$ for this potential have the form
\begin{equation}
\Phi_{\mathrm{J}}^{[3]\pm}(x;\lambda)=\begin{pmatrix}\omega_3^{\sigma_3}e^{-i\lambda(c^{[1]}+c^{[2]})x/(2\epsilon)}e^{i\kappa^{[3]}\sigma_3x/(2\epsilon)}\mathbf{W}^{[3]\pm}(x;\lambda)\omega_3^{-\sigma_3} & \mathbf{0}\\
\mathbf{0}^\trans & e^{-i\lambda c^{[3]}x/\epsilon}\end{pmatrix},
\end{equation}
where $\omega_3^2=-\gamma^{[3]}e^{i\theta^{[3]}}$ and where $\mathbf{W}^{[3]\pm}(x;\lambda)$ are the Jost solutions of the non-selfadjoint Zakharov-Shabat system \eqref{eq:ZS-Jost} with potential 
$A(x)=A^{[3]}(x)$ and spectral parameter $\zeta=\zeta^{[3]}(\lambda)$,
where
\begin{equation}
A^{[3]}(x):=\frac{H^{[3]}(x)}{\sqrt{\Delta^{[1]}\Delta^{[2]}}}\quad\text{and}\quad\zeta^{[3]}(\lambda):=\frac{\Delta^{[3]}}{2}\lambda + \frac{\kappa^{[3]}}{2}.
\label{eq:A3-zeta3}
\end{equation}
Once again $\zeta^{[3]}$ preserves $\mathbb{C}_+$ taking $\mathbb{R}$ onto $\mathbb{R}$.  Letting $a^{[3]}(\zeta^{[3]}(\lambda))$ and $b^{[3]}(\zeta^{[3]}(\lambda))$ denote the independent elements of the Zakharov-Shabat scattering matrix $\mathbf{S}^{\mathrm{ZS}[3]}(\zeta^{[3]}(\lambda))$ linking $\mathbf{W}^{[3]\pm}(x;\lambda)$, we obtain
\begin{equation}
\begin{split}
\begin{pmatrix} S^{[3]}_{11}(\lambda) & S^{[3]}_{12}(\lambda)\\S^{[3]}_{21}(\lambda) & S^{[3]}_{22}(\lambda)\end{pmatrix}&=
\omega_3^{\sigma_3}\mathbf{S}^{\mathrm{ZS}[3]}(\zeta^{[3]}(\lambda))\omega_3^{-\sigma_3}\\
&=\begin{pmatrix}a^{[3]}(\zeta^{[3]}(\lambda))^* & -\gamma^{[3]}e^{i\theta^{[3]}}b^{[3]}(\zeta^{[3]}(\lambda))^*\\
\gamma^{[3]}e^{-i\theta^{[3]}}b^{[3]}(\zeta^{[3]}(\lambda)) & a^{[3]}(\zeta^{[3]}(\lambda))\end{pmatrix},\quad\lambda\in\mathbb{R}.
\end{split}
\label{eq:S3-pieces}
\end{equation}

With these results, we clearly have sufficient information to construct the jump matrix $\mathbf{V}_0^\sigma(\lambda)$ for $\lambda\in\mathbb{R}$, needed to formulate the Riemann-Hilbert problem of inverse scattering for the TWRI initial-value problem, in terms of the Zakharov-Shabat spectral functions $a^{[k]}$ and $b^{[k]}$ for $k=1,2,3$.  Indeed, one simply applies LDU or UDL factorization (cf., \eqref{eq:LDU-UDL}) to $\mathbf{S}(\lambda)=\mathbf{S}^{[1]}(\lambda)\mathbf{S}^{[2]}(\lambda)\mathbf{S}^{[3]}(\lambda)$ and obtains $\mathbf{V}^\sigma_0(\lambda)$ from \eqref{eq:TWRI-Jump}.

We remark that an important implication of the assumption \eqref{eq:gamma-assumption} singling out a particular variety of the decay instability case for the TWRI system \eqref{3wave} is that the three Zakharov-Shabat spectral problems that arise from the assumption \eqref{eq:disjoint-support-assumption} of disjoint supports are all of nonselfadjoint or focusing type.  If \eqref{eq:gamma-assumption} were not to hold, at least one of the $2\times 2$ spectral problems would be of selfadjoint or defocusing type.  The nonselfadjoint version of the problem is preferable for our purposes because it leads in the semiclassical limit to purely discrete spectrum which allows for effective construction of semiclassical soliton ensembles as exact solutions via finite-dimensional linear algebra.  

\subsection{Expression of the TWRI discrete scattering data in terms of Zakharov-Shabat data}
\label{section:3wave-residues}
The expression of the jump matrix factors $\mathbf{S}^{[k]}(\lambda)$ in terms of the Zakharov-Shabat spectral functions $a^{[k]}$ and $b^{[k]}$ for the potentials $A^{[k]}(x)$ as above also determines the
location of the poles of $\mathbf{M}^\sigma(\lambda)=\mathbf{M}^\sigma(x,t;\lambda)$ in $\mathbb{C}_+$ in terms of the eigenvalues of the three different Zakharov-Shabat problems.
Indeed, using \eqref{eq:S1-pieces}, \eqref{eq:S2-pieces}, and \eqref{eq:S3-pieces} in \eqref{eq:u-plus-v-plus-partials} shows that in terms of Zakharov-Shabat scattering matrix elements we have
\begin{equation}
u_+(\lambda)=a^{[2]}(\zeta^{[2]}(\lambda))a^{[3]}(\zeta^{[3]}(\lambda))\quad\text{and}\quad
v_+(\lambda)=a^{[1]}(\zeta^{[1]}(\lambda))a^{[2]}(\zeta^{[2]}(\lambda)).
\label{eq:u-v-ZS}
\end{equation}
Analytically continuing these into the upper half $\lambda$-plane gives the functions $u(\lambda)$ and $v(\lambda)$ whose roots are the poles of the matrix $\mathbf{M}^\sigma(x,t;\lambda)$, and these are obviously the pre-images under the mappings $\zeta=\zeta^{[k]}(\lambda)$ of the zeros of $a^{[k]}(\zeta)$, i.e., the eigenvalues of the three different Zakharov-Shabat problems corresponding to the separate potentials $A^{[1]}(x)$, $A^{[2]}(x)$, and $A^{[3]}(x)$.  We now consider in more detail how generic (simple) Zakharov-Shabat eigenvalues lead to poles of $\mathbf{M}^\sigma(\lambda)=\mathbf{M}^\sigma(x,t;\lambda)$, and also how the corresponding connection coefficients may be computed explicitly in terms of Zakharov-Shabat data \emph{without analytic continuation of any functions $b^{[k]}$ from the real axis}.

\subsubsection{Poles arising from simple zeros of $a^{[1]}(\zeta^{[1]}(\lambda))$.}
\label{section:disjoint-type-1}
If $\lambda_0\in\mathbb{C}_+$ is a simple zero of $a^{[1]}(\zeta^{[1]}(\lambda))$ (i.e., $\zeta=\zeta^{[1]}(\lambda_0)\in\mathbb{C}_+$ is a simple eigenvalue of the Zakharov-Shabat problem \eqref{eq:ZS-Jost} with potential $A(x)=A^{[1]}(x)$) for which $a^{[2]}(\zeta^{[2]}(\lambda_0))$ and $a^{[3]}(\zeta^{[3]}(\lambda_0))$ are nonzero, then $\lambda_0$ is a simple pole of $\mathbf{M}^\sigma(\lambda)$ of type 1.  To complete the scattering data corresponding to $\lambda_0$, it is sufficient to calculate the nonzero connection coefficient $\beta_{32}^{[1]}$ appearing in the residue condition at $\lambda_0$ for $\mathbf{M}^+(\lambda)$ (see \eqref{eq:type-1-residues}), since the corresponding connection coefficient $\beta_{23}^{[1]}$ needed to describe the residue of $\mathbf{M}^-(\lambda)$ can be obtained from $\beta_{32}^{[1]}$ using \eqref{eq:type-1-beta-relations}.  The constant $\beta_{32}^{[1]}$ is characterized in terms of the columns $\mathbf{m}^{+,j}(x;\lambda)$, $j=1,2,3$, of the matrix $\mathbf{M}^+(x;\lambda)=\mathbf{M}^+(x,0;\lambda)$ by the condition 
\begin{equation}
\mathop{\mathrm{Res}}_{\lambda=\lambda_0}\mathbf{m}^{+,2}(x;\lambda) = \beta_{32}^{[1]}e^{i\lambda_0\Delta^{[1]}x/\epsilon}\mathbf{m}^{+,3}(x;\lambda_0).
\label{eq:beta32-definition}
\end{equation}
A general construction of Beals and Coifman \cite{BealsC87} based on the exterior algebra and described in Appendix~\ref{section:Volterra} shows that the simple pole in $\mathbf{m}^{+,2}(x;\lambda)$ arises because the latter can be expressed as a cross product of analytic vector functions divided by $v(\lambda)$.  Therefore \eqref{eq:beta32-definition} can be equivalently written in the form
\begin{equation}
\frac{1}{v'(\lambda_0)}\mathbf{n}^{-,3}(x;\lambda_0)\times\mathbf{n}^{+,1}(x;\lambda_0)=\beta_{32}^{[1]}e^{i\lambda_0\Delta^{[1]}x/\epsilon}\mathbf{m}^{+,3}(x;\lambda_0).
\label{eq:beta32-equation-1}
\end{equation}
Here $\mathbf{m}^{+,3}(x;\lambda_0)$ satisfies \eqref{eq:m-columns-ODE} for $j=3$ and $\lambda=\lambda_0\in\mathbb{C}_+$ with the boundary condition\footnote{Here and below, $\mathbf{e}^k$, $k=1,2,3$, denote the standard coordinate basis of unit vectors in $\mathbb{C}^3$.} $\mathbf{m}^{+,3}(x;\lambda_0)\to \mathbf{e}^3$ as $x\to +\infty$, $\mathbf{n}^{-,3}(x;\lambda_0)$ satisfies \eqref{eq:first-wedge-1}--\eqref{eq:first-wedge-2} for $\lambda=\lambda_0\in\mathbb{C}_+$ with the boundary condition $\mathbf{n}^{-,3}(x;\lambda_0)\to\mathbf{e}^3$ as $x\to -\infty$, and $\mathbf{n}^{+,1}(x;\lambda_0)$ satisfies \eqref{eq:second-wedge} for $\lambda=\lambda_0\in\mathbb{C}_+$ with the boundary condition $\mathbf{n}^{+,1}(x;\lambda_0)\to \mathbf{e}^1$ as $x\to +\infty$.  Importantly, \emph{all three of these vectors are analytic functions of $\lambda$ at $\lambda=\lambda_0$, being given as solutions of appropriate Volterra equations}.  Also, from \eqref{eq:A1-zeta1} and \eqref{eq:u-v-ZS} we have $v'(\lambda_0)=\tfrac{1}{2}\Delta^{[1]}a^{[1]\prime}(\zeta^{[1]}(\lambda_0))a^{[2]}(\zeta^{[2]}(\lambda_0))$.

We may consider \eqref{eq:beta32-equation-1} for any $x\in\mathbb{R}$, and we choose the value $x=x_{23}$.  First observe that since only the $12$ and $21$ elements of $\mathbf{Q}(x)$ are nonzero when $x\ge x_{23}$, it is easy to see from the Volterra equations for $\mathbf{m}^{+,3}(x;\lambda_0)$ and $\mathbf{n}^{+,1}(x;\lambda_0)$ (see \eqref{eq:Fredholm-system} and \eqref{eq:mplus-2cross3-Im-lambda-positive} respectively) that 
\begin{equation}
\mathbf{m}^{+,3}(x;\lambda_0)\equiv \mathbf{e}^3\quad\text{and}\quad n^{+,1}_3(x;\lambda_0)=0,\quad\text{for all } x\ge x_{23}.
\label{eq:type-1-exact-at-x23}
\end{equation}
The differential equation \eqref{eq:second-wedge} satisfied by $\mathbf{n}^{+,1}(x;\lambda)$ for $\lambda=\lambda_0$ implies that the two-component vector defined by 
\begin{equation}
\mathbf{w}^{[3]+,2}(x;\lambda_0):=\begin{pmatrix}\gamma^{[3]}e^{-i\theta^{[3]}}e^{-i\kappa^{[3]}x/\epsilon} & 0\\ 0 & 1\end{pmatrix}e^{i\zeta^{[3]}(\lambda_0)x/\epsilon}\sigma_1\begin{pmatrix}n_1^{+,1}(x;\lambda_0)\\n_2^{+,1}(x;\lambda_0)\end{pmatrix}
\label{eq:w32plus-define}
\end{equation}
is a vector solution of the nonselfadjoint Zakharov-Shabat system \eqref{eq:ZS-Jost} with potential $A=A^{[3]}$ and spectral parameter $\zeta=\zeta^{[3]}(\lambda_0)$ given by \eqref{eq:A3-zeta3} for $x\ge x_{23}$, and the normalization condition $\mathbf{n}^{+,1}(x;\lambda_0)\to \mathbf{e}^1$ as $x\to +\infty$ implies that $\mathbf{w}^{[3]+,2}(x;\lambda_0)$ is asymptotic to $e^{i\zeta^{[3]}(\lambda_0)x/\epsilon}(0,1)^\trans$ as $x\to +\infty$, uniquely identifying this solution with the second column of the Zakharov-Shabat Jost matrix $\mathbf{W}^{+}(x;\zeta)$ for the potential $A=A^{[3]}$ with $\zeta=\zeta^{[3]}(\lambda_0)$.  Therefore, 
\begin{equation}
\mathbf{n}^{+,1}(x;\lambda_0)=\begin{pmatrix}
e^{-i\zeta^{[3]}(\lambda_0)x/\epsilon}w_2^{[3]+,2}(x;\lambda_0)\\
\gamma^{[3]}e^{i\theta^{[3]}}e^{i\kappa^{[3]}x/\epsilon}e^{-i\zeta^{[3]}
(\lambda_0)x/\epsilon}w_1^{[3]+,2}(x;\lambda_0)\\
0\end{pmatrix},\quad x\ge x_{23},
\label{eq:n-plus-1-x-greater-x23}
\end{equation}
and so all components of $\mathbf{m}^{+,3}(x_{23};\lambda_0)$ and $\mathbf{n}^{+,1}(x_{23};\lambda_0)$ have been explicitly written in terms of Zakharov-Shabat data for \eqref{eq:ZS-Jost} subject to $A=A^{[3]}$ and $\zeta=\zeta^{[3]}$ given by \eqref{eq:A3-zeta3}.

To evaluate $\mathbf{n}^{-,3}(x;\lambda_0)$ (the unique solution of the Volterra equation \eqref{eq:mminus-1cross2-Im-lambda-positive}) at $x=x_{23}$, first note that 
since only the $23$ and $32$ elements of $\mathbf{Q}(x)$ are nonzero when $x\le x_{12}$,
\begin{equation}
n_1^{-,3}(x;\lambda_0)=0,\quad\text{for all } x\le x_{12}.
\end{equation}
So far, we have not used the fact that $\lambda_0\in\mathbb{C}_+$ is a simple zero of $v(\lambda)$.  This fact gives additional information about $\mathbf{n}^{-,3}(x;\lambda_0)$ at $x=x_{12}$.  Indeed, for $x\le x_{12}$, the remaining two elements of the vector $\mathbf{n}^{-,3}(x;\lambda_0)$ are easily related to the Zakharov-Shabat system \eqref{eq:ZS-Jost} with potential $A(x)=A^{[1]}(x)$ and spectral parameter $\zeta=\zeta^{[1]}(\lambda)$ given by \eqref{eq:A1-zeta1}; setting
\begin{equation}
\mathbf{w}^{[1]-,1}(x;\lambda_0):=\begin{pmatrix}1 & 0\\0 & \gamma^{[1]}e^{i\theta^{[1]}}e^{i\kappa^{[1]}x/\epsilon}\end{pmatrix}
e^{-i\zeta^{[1]}(\lambda_0)x/\epsilon}\sigma_1\begin{pmatrix}n_2^{-,3}(x;\lambda_0)\\n_3^{-,3}(x;\lambda_0)\end{pmatrix},
\label{eq:w-11-definition}
\end{equation}
one sees that $\mathbf{w}^{[1]-,1}(x;\lambda_0)$ is exactly the first column of the solution $\mathbf{W}^-$ of \eqref{eq:ZS-Jost} (asymptotic to the exponentially decaying vector $e^{-i\zeta^{[1]}(\lambda_0)x/\epsilon}(1,0)^\trans$ as $x\to -\infty$) with potential $A=A^{[1]}$ and spectral parameter $\zeta=\zeta^{[1]}$ given by \eqref{eq:A1-zeta1}.  Now, since by hypothesis $\lambda_0$ is a simple zero of $a^{[1]}(\zeta^{[1]}(\lambda))$, producing a corresponding simple zero of $v(\lambda)$, the Jost solution $\mathbf{w}^{[1]-,1}(x;\lambda_0)$ must coincide with $\tau^{[1]}\mathbf{w}^{[1]+,2}(x;\lambda_0)$, where $\tau^{[1]}$ is the proportionality constant associated with the Zakharov-Shabat eigenvalue $\zeta^{[1]}(\lambda_0)$ (see \eqref{eq:tau-def-1}) and $\mathbf{w}^{[1]+,2}(x;\lambda_0)$ is the second column of the matrix solution $\mathbf{W}^+$ of \eqref{eq:ZS-Jost} (asymptotic to the exponentially decaying vector $e^{i\zeta^{[1]}(\lambda_0)x/\epsilon}(0,1)^\trans$ as $x\to +\infty$) with $A=A^{[1]}$ and $\zeta=\zeta^{[1]}$ given by \eqref{eq:A1-zeta1}.  Since $A^{[1]}(x)\equiv 0$ for $x\ge x_{12}$, we have $\mathbf{w}^{[1]+,2}(x;\lambda_0)=e^{i\zeta^{[1]}(\lambda_0)x/\epsilon}(0,1)^\trans$ exactly for all $x\ge x_{12}$.  Hence $\mathbf{w}^{[1]-,1}(x_{12};\lambda_0)=\tau^{[1]}e^{i\zeta^{[1]}(\lambda_0)x_{12}/\epsilon}(0,1)^\trans$, and so \eqref{eq:w-11-definition} gives
\begin{equation}
n_2^{-,3}(x_{12};\lambda_0)=\gamma^{[1]}e^{-i\theta^{[1]}}\tau^{[1]}e^{i\lambda_0\Delta^{[1]}x_{12}/\epsilon}\quad\text{and}\quad n_3^{-,3}(x_{12};\lambda_0)=0,
\end{equation}
where we have also used \eqref{eq:A1-zeta1} to eliminate $\zeta^{[1]}(\lambda_0)$.  As only the second component of $\mathbf{n}^{-,3}(x;\lambda_0)$ is nonzero at $x=x_{12}$, and since according to \eqref{eq:first-wedge-1}--\eqref{eq:first-wedge-2} in the interval $[x_{12},x_{23}]$ (where only the potential $A^{[2]}(x)$ is nonzero) we have exactly 
\begin{equation}
n_2^{-,3}(x;\lambda_0)=e^{i\lambda_0\Delta^{[1]}(x-x_{12})/\epsilon}n_2^{-,3}(x_{12};\lambda_0),
\quad x_{12}\le x\le x_{23},
\end{equation}
we deduce the exact value of $\mathbf{n}^{-,3}(x;\lambda_0)$ at $x=x_{23}$ to be
\begin{equation}
\mathbf{n}^{-,3}(x_{23};\lambda_0)=\begin{pmatrix}0\\
\gamma^{[1]}e^{-i\theta^{[1]}}e^{i\lambda_0\Delta^{[1]}x_{23}/\epsilon}\tau^{[1]}\\
0\end{pmatrix}.
\label{eq:type-1-n-minus-3-at-x23}
\end{equation}

Combining \eqref{eq:n-plus-1-x-greater-x23} with \eqref{eq:type-1-n-minus-3-at-x23} gives
\eq
\begin{split}
\mathbf{n}^{-,3}(x_{23};\lambda_0)&\times\mathbf{n}^{+,1}(x_{23};\lambda_0)=\\
&\begin{pmatrix}0\\0\\
-\gamma^{[1]}e^{-i\theta^{[1]}}\tau^{[1]}e^{-i\zeta^{[3]}(\lambda_0)x_{23}/\epsilon}e^{i\lambda_0\Delta^{[1]}x_{23}/\epsilon}
w_{2}^{[3]+,2}(x_{23};\lambda_0)
\end{pmatrix},
\end{split}
\endeq
which is obviously proportional to $\mathbf{m}^{+,3}(x_{23};\lambda_0)=\mathbf{e}^3$, and therefore from \eqref{eq:beta32-equation-1} at $x=x_{23}$ we obtain
\begin{equation}
\beta^{[1]}_{32}=-\frac{2\gamma^{[1]}e^{-i\theta^{[1]}}\tau^{[1]}w_{2}^{[3]+,2}(x_{23};\lambda_0)e^{-i\zeta^{[3]}(\lambda_0)x_{23}/\epsilon}}{\Delta^{[1]}a^{[1]\prime}(\zeta^{[1]}(\lambda_0))a^{[2]}(\zeta^{[2]}(\lambda_0))}.
\label{eq:type-1-beta32}
\end{equation}
Using \eqref{eq:type-1-beta-relations} and \eqref{eq:u-v-ZS} then gives
\begin{equation}
\beta_{23}^{[1]}=-\frac{2\gamma^{[1]}e^{i\theta^{[1]}}a^{[3]}(\zeta^{[3]}(\lambda_0))}{\Delta^{[1]}a^{[1]\prime}(\zeta^{[1]}(\lambda_0))\tau^{[1]}w_2^{[3]+,2}(x_{23};\lambda_0)e^{-i\zeta^{[3]}(\lambda_0)x_{23}/\epsilon}}
\label{eq:type-1-beta23}
\end{equation}
as the connection coefficient needed to characterize the residue of $\mathbf{M}^-(\lambda)$ at $\lambda=\lambda_0$.  

Observe that even though the poles of type $1$ are determined from the initial packet $q^{[1]}(x,0)$ alone (essentially as a consequence of the assumption \eqref{eq:disjoint-support-assumption} of disjoint supports), the corresponding connection coefficients encode information about all three packets.  Indeed,
from \eqref{eq:type-1-beta32} and \eqref{eq:type-1-beta23} we see that $\beta^{[1]}_{32}$ and $\beta^{[1]}_{23}$ depend on the Zakharov-Shabat spectral function $a^{[2]}(\zeta^{[2]}(\lambda))$ associated with $q^{[2]}(x,0)$ and the spectral function $a^{[3]}(\zeta^{[3]}(\lambda))$ and Jost solution $\mathbf{w}^{[3]+,2}(x;\lambda)e^{-i\zeta^{[3]}(\lambda)x/\epsilon}$ associated with $q^{[3]}(x,0)$.  This fact can be understood at an intuitive level because the connection coefficients encode information about the average position of the soliton corresponding to the pole $\lambda_0$ in the moving frame with velocity $c^{[1]}$.  The position of the soliton in this frame will shift in time $t$ due to nonlinear interactions with other solitons and radiation.  But due to the orderings \eqref{eq:c-inequalities} and \eqref{eq:disjoint-support-assumption} of the velocities $c^{[k]}$ and supports $\mathrm{supp}(q^{[k]}(x,0))$ respectively this interaction can only happen in the future, for $t\ge T >0$.   Thus in order to get the effect of the interaction right for $t\gg T$ it is necessary for the solitons to be positioned at $t=0$ with the future interaction in mind, implying that the scattering data necessary to generate the solution by inverse-scattering for all $t>0$ has to link all three packets even though by the method of characteristics applied directly to \eqref{3wave} there can be no interaction at all for $t<T$.  

\subsubsection{Poles arising from simple zeros of $a^{[3]}(\zeta^{[3]}(\lambda))$.}
\label{section:disjoint-type-3}
If $\lambda_0\in\mathbb{C}_+$ is a simple zero of $a^{[3]}(\zeta^{[3]}(\lambda))$ (i.e., $\zeta=\zeta^{[3]}(\lambda_0)\in\mathbb{C}_+$ is a simple eigenvalue of the Zakharov-Shabat problem \eqref{eq:ZS-Jost} with potential $A(x)=A^{[3]}(x)$) but $a^{[1]}(\zeta^{[1]}(\lambda_0))$ and $a^{[2]}(\zeta^{[2]}(\lambda_0))$ are nonzero, then $\lambda_0$ is a simple pole of $\mathbf{M}^\sigma(\lambda)$ of type 3.  To obtain the corresponding residue matrices $\mathbf{N}^\sigma$ of the form \eqref{eq:type-3-residues} it is sufficient to calculate the nonzero connection coefficient $\beta_{21}^{[3]}$, because $\beta_{12}^{[3]}$ is then known via \eqref{eq:type-3-beta-relations}.  The constant $\beta_{21}^{[3]}$ is characterized in terms of the columns $\mathbf{m}^{+,j}(x;\lambda)$ of $\mathbf{M}^+(x;\lambda)=\mathbf{M}^+(x,0;\lambda)$ by the relation
\begin{equation}
\mathop{\mathrm{Res}}_{\lambda=\lambda_0}\mathbf{m}^{+,1}(x;\lambda)=\beta_{21}^{[3]}e^{i\lambda_0\Delta^{[3]}x/\epsilon}\mathbf{m}^{+,2}(x;\lambda_0),
\end{equation}
or, equivalently, since from the first column of \eqref{eq:Mplus-Mminus-intro} it follows that the simple pole of $\mathbf{m}^{+,1}(x;\lambda)$ arises from the simple zero of $u(\lambda)$ at $\lambda_0$,
\begin{equation}
\frac{1}{u'(\lambda_0)}\mathbf{m}^{-,1}(x;\lambda_0)=\frac{\beta_{21}^{[3]}e^{i\lambda_0\Delta^{[3]}x/\epsilon}}{v(\lambda_0)}\mathbf{n}^{-,3}(x;\lambda_0)\times\mathbf{n}^{+,1}(x;\lambda_0).
\label{eq:beta21-type-3-first}
\end{equation}
Here $\mathbf{m}^{-,1}(x;\lambda_0)$ satisfies \eqref{eq:m-columns-ODE} for $j=1$ and $\lambda=\lambda_0\in\mathbb{C}_+$ with the boundary condition $\mathbf{m}^{-,1}(x;\lambda_0)\to \mathbf{e}^1$ as $x\to -\infty$, and $\mathbf{n}^{-,3}(x;\lambda)$ and $\mathbf{n}^{+,1}(x;\lambda)$ are exactly as in \S\ref{section:disjoint-type-1}.  All three of these are analytic functions of $\lambda$ evaluated at $\lambda=\lambda_0$ via Volterra theory, and from \eqref{eq:A3-zeta3} and \eqref{eq:u-v-ZS} we have $u'(\lambda_0)=\tfrac{1}{2}\Delta^{[3]}a^{[3]\prime}(\zeta^{[3]}(\lambda_0))a^{[2]}(\zeta^{[2]}(\lambda_0))$ and 
$v(\lambda_0)=a^{[1]}(\zeta^{[1]}(\lambda_0))a^{[2]}(\zeta^{[2]}(\lambda_0))$.

Now we will calculate $\beta_{21}^{[3]}$ by evaluating \eqref{eq:beta21-type-3-first} for $x=x_{12}$.  Arguing as in \S\ref{section:disjoint-type-1}, we again have 
\begin{equation}
\mathbf{m}^{-,1}(x;\lambda_0)=\mathbf{e}^1 \quad\text{and}\quad
n_1^{-,3}(x;\lambda_0)=0,\quad\text{for all } x\le x_{12},
\label{eq:type-3-exact-at-x12}
\end{equation}
and $\mathbf{w}^{[1]-,1}(x;\lambda_0)$ defined by \eqref{eq:w-11-definition} is the Jost solution normalized to the decaying exponential $e^{-i\zeta^{[1]}(\lambda_0)x/\epsilon}(1,0)^\trans$ as $x\to -\infty$ of the nonselfadjoint
Zakharov-Shabat system \eqref{eq:ZS-Jost} with potential $A(x)=A^{[1]}(x)$ and spectral parameter $\zeta=\zeta^{[1]}(\lambda)$ given by \eqref{eq:A1-zeta1}.  Therefore,
\begin{equation}
\mathbf{n}^{-,3}(x;\lambda_0)=\begin{pmatrix}
0\\
\gamma^{[1]}e^{-i\theta^{[1]}}e^{-i\kappa^{[1]}x/\epsilon}e^{i\zeta^{[1]}(\lambda_0)x/\epsilon}w_2^{[1]-,1}(x;\lambda_0)\\
e^{i\zeta^{[1]}(\lambda_0)x/\epsilon}w_1^{[1]-,1}(x;\lambda_0)
\end{pmatrix},\quad x\le x_{12},
\label{eq:n-minus-3-x-less-x12}
\end{equation}
which completes our characterization of $\mathbf{m}^{-,1}(x;\lambda_0)$ and $\mathbf{n}^{-,3}(x;\lambda_0)$ for $x=x_{12}$ in terms of Zakharov-Shabat data.

To obtain $\mathbf{n}^{+,1}(x;\lambda_0)$ for $x=x_{12}$, first note that as in \S~\ref{section:disjoint-type-1},
\begin{equation}
n_3^{+,1}(x;\lambda_0)=0,\quad\text{for all } x\ge x_{23},
\label{eq:n1plus3-type-3}
\end{equation}
and $\mathbf{w}^{[3]+,2}(x;\lambda_0)$ defined by \eqref{eq:w32plus-define} is the Jost solution normalized to the decaying exponential $e^{i\zeta^{[3]}(\lambda_0)x/\epsilon}(0,1)^\trans$ as $x\to +\infty$ for the nonselfadjoint Zakharov-Shabat system \eqref{eq:ZS-Jost} with $A(x)=A^{[3]}(x)$ and $\zeta=\zeta^{[3]}(\lambda)$ given by \eqref{eq:A3-zeta3}.  Since $a^{[3]}(\zeta^{[3]}(\lambda_0))=0$, there exists a nonzero proportionality constant $\tau^{[3]}$ associated with the Zakharov-Shabat eigenvalue $\zeta^{[3]}(\lambda_0)\in\mathbb{C}_+$ such that
$\mathbf{w}^{[3]-,1}(x;\lambda_0)=\tau^{[3]}\mathbf{w}^{[3]+,2}(x;\lambda_0)$, where $\mathbf{w}^{[3]-,1}(x;\lambda_0)$ is the Jost solution of the same system (\eqref{eq:ZS-Jost} with \eqref{eq:A3-zeta3}) normalized to the decaying exponential $e^{-i\zeta^{[3]}(\lambda_0)x/\epsilon}(1,0)^\trans$ as $x\to -\infty$.  Since $A^{[3]}(x)$ is supported on $x\ge x_{23}$, we have the identity $\mathbf{w}^{[3]-,1}(x;\lambda_0)=e^{-i\zeta^{[3]}(\lambda_0)x/\epsilon}(1,0)^\trans$ holding for $x\le x_{23}$, so in particular $\mathbf{w}^{[3]+,2}(x_{23};\lambda_0)=(\tau^{[3]})^{-1}e^{-i\zeta^{[3]}(\lambda_0)x_{23}/\epsilon}(1,0)^\trans$.  Using this information in \eqref{eq:w32plus-define} along with \eqref{eq:n1plus3-type-3} shows that only the second component of $\mathbf{n}^{+,1}(x_{23};\lambda_0)$ is nonzero and determines the value of the latter.  For $x_{12}\le x\le x_{23}$, the second component of $\mathbf{n}^{+,1}(x;\lambda_0)$ decouples from the rest (because $A^{[1]}(x)=A^{[3]}(x)=0$) and satisfies, according to \eqref{eq:second-wedge}, 
\begin{equation}
n_2^{+,1}(x;\lambda_0)=e^{-i\lambda_0\Delta^{[3]}(x-x_{23})/\epsilon}n_2^{+,1}(x_{23};\lambda_0),
\quad x_{12}\le x\le x_{23}.
\end{equation}
We therefore obtain
\begin{equation}
\mathbf{n}^{+,1}(x_{12};\lambda_0)=\begin{pmatrix}0\\
\gamma^{[3]}e^{i\theta^{[3]}}e^{-i\lambda_0\Delta^{[3]}x_{12}/\epsilon}\tau^{[3]-1}\\
0
\end{pmatrix}.
\label{eq:n-plus-1-at-x12}
\end{equation}

As in \S\ref{section:disjoint-type-1}, we are now in a position to calculate the cross product $\mathbf{n}^{-,3}\times\mathbf{n}^{+,1}$, this time at $x=x_{12}$ by combining \eqref{eq:n-minus-3-x-less-x12} with \eqref{eq:n-plus-1-at-x12}.  We obtain
\eq
\begin{split}
\mathbf{n}^{-,3}(x_{12};\lambda_0)&\times\mathbf{n}^{+,1}(x_{12};\lambda_0)=\\
&\begin{pmatrix}
-\gamma^{[3]}e^{i\theta^{[3]}}\tau^{[3]-1}e^{i\zeta^{[1]}(\lambda_0)x_{12}/\epsilon}e^{-i\lambda_0\Delta^{[3]}x_{12}/\epsilon}w_1^{[1]-,1}(x_{12};\lambda_0)
\\
0\\
0
\end{pmatrix},
\end{split}
\endeq
which is obviously exactly proportional to $\mathbf{m}^{-,1}(x_{12};\lambda_0)=\mathbf{e}^1$, and therefore evaluating \eqref{eq:beta21-type-3-first} for $x=x_{12}$ gives
\begin{equation}
\beta_{21}^{[3]}=-\frac{2\gamma^{[3]}e^{-i\theta^{[3]}}a^{[1]}(\zeta^{[1]}(\lambda_0))\tau^{[3]}}{\Delta^{[3]}a^{[3]\prime}(\zeta^{[3]}(\lambda_0))w_1^{[1]-,1}(x_{12};\lambda_0)e^{i\zeta^{[1]}(\lambda_0)x_{12}/\epsilon}}.
\label{eq:type-3-beta21}
\end{equation}
Using \eqref{eq:type-3-beta-relations} and \eqref{eq:u-v-ZS} then gives
\begin{equation}
\beta_{12}^{[3]}=-\frac{2\gamma^{[3]}e^{i\theta^{[3]}}w_1^{[1]-,1}(x_{12};\lambda_0)e^{i\zeta^{[1]}(\lambda_0)x_{12}/\epsilon}}{\Delta^{[3]}a^{[3]\prime}(\zeta^{[3]}(\lambda_0))a^{[2]}(\zeta^{[2]}(\lambda_0))\tau^{[3]}}
\label{eq:type-3-beta12}
\end{equation}
as the connection coefficient needed to characterize the residue of $\mathbf{M}^-(\lambda)$ at $\lambda=\lambda_0$.  Again one observes that these connection coefficients for a pole $\lambda_0$ generated independently by the field $q^{[3]}(x,0)$ contain information about $q^{[1]}(x,0)$ and $q^{[2]}(x,0)$ as well.

\subsubsection{Poles arising from simple zeros of $a^{[2]}(\zeta^{[2]}(\lambda))$.}
If $\lambda_0\in\mathbb{C}_+$ is a simple zero of $a^{[2]}(\zeta^{[2]}(\lambda))$ (i.e., $\zeta=\zeta^{[2]}(\lambda_0)\in\mathbb{C}_+$ is a simple eigenvalue of the Zakharov-Shabat problem \eqref{eq:ZS-Jost} with potential $A(x)=A^{[2]}(x)$) for which $a^{[1]}(\zeta^{[1]}(\lambda_0))$ and $a^{[3]}(\zeta^{[3]}(\lambda_0))$ are nonzero, then $\lambda_0$ is a simple pole of $\mathbf{M}^{\sigma}(\lambda)$ of either type $2$, type $\solsplit$, or type $\solfuse$.  To determine which type it is, and compute the corresponding constants $\beta_{jk}^{[\mathrm{type}]}$, it suffices to calculate the residues of the first two columns of $\mathbf{M}^+(x;\lambda)=\mathbf{M}^+(x,0;\lambda)$.  We will calculate them, along with the value of the third (analytic) column, at the point $x=x_{23}$.

To begin, we recall \eqref{eq:type-1-exact-at-x23} and \eqref{eq:n-plus-1-x-greater-x23}, which specify the values of $\mathbf{m}^{+,3}(x_{23};\lambda_0)=\mathbf{e}^3$ and $\mathbf{n}^{+,1}(x_{23};\lambda_0)$, the latter in terms of the Jost solution $\mathbf{w}^{[3]+,2}(x;\lambda_0)$ of the nonselfadjoint Zakharov-Shabat system \eqref{eq:ZS-Jost} with potential $A(x)=A^{[3]}(x)$ and spectral parameter $\zeta=\zeta^{[3]}(\lambda)$ given by \eqref{eq:A3-zeta3} normalized to the decaying exponential $e^{i\zeta^{[3]}(\lambda_0)x/\epsilon}(0,1)^\trans$ as $x\to +\infty$.  From the first column of \eqref{eq:Mplus-Mminus-intro} and comparing \eqref{eq:beta32-definition} with \eqref{eq:beta32-equation-1}, the desired residues can be expressed in the form
\begin{equation}
\mathop{\mathrm{Res}}_{\lambda=\lambda_0}\mathbf{m}^{+,1}(x;\lambda)=\frac{1}{u'(\lambda_0)}\mathbf{m}^{-,1}(x;\lambda_0)\quad\text{and}\quad
\mathop{\mathrm{Res}}_{\lambda=\lambda_0}\mathbf{m}^{+,2}(x;\lambda)=\frac{1}{v'(\lambda_0)}\mathbf{n}^{-,3}(x;\lambda_0)\times\mathbf{n}^{+,1}(x;\lambda_0).
\end{equation}
Also, from \eqref{eq:u-v-ZS} we have the relations $u'(\lambda_0)=\tfrac{1}{2}\Delta^{[2]}a^{[2]\prime}(\zeta^{[2]}(\lambda_0))a^{[3]}(\zeta^{[3]}(\lambda_0))$ as well as $v'(\lambda_0)=\tfrac{1}{2}\Delta^{[2]}a^{[1]}(\zeta^{[1]}(\lambda_0))a^{[2]\prime}(\zeta^{[2]}(\lambda_0))$.
It therefore remains to compute $\mathbf{m}^{-,1}(x_{23};\lambda_0)$ and $\mathbf{n}^{-,3}(x_{23};\lambda_0)$.

To calculate these, first note that according to \eqref{eq:type-3-exact-at-x12} and \eqref{eq:n-minus-3-x-less-x12} we have the values $\mathbf{m}^{-,1}(x_{12};\lambda_0)=\mathbf{e}^1$ and $\mathbf{n}^{-,3}(x_{12};\lambda_0)$, the latter in terms of the Jost solution $\mathbf{w}^{[1]-,1}(x;\lambda_0)$ of the nonselfadjoint Zakharov-Shabat system \eqref{eq:ZS-Jost} with potential $A(x)=A^{[1]}(x)$ and spectral parameter $\zeta=\zeta^{[1]}(\lambda)$ given by \eqref{eq:A1-zeta1}.  We now propagate these values through the interval $x_{12}\le x\le x_{23}$ containing the support of $A^{[2]}$ and outside the supports of $A^{[1]}$ and $A^{[3]}$.  Because $A^{[1]}(x)=A^{[3]}(x)=0$ on this interval, the second component of both vectors decouples in the differential equations \eqref{eq:m-columns-ODE} (for $j=1$) governing $\mathbf{m}^{-,1}(x;\lambda_0)$ and \eqref{eq:first-wedge-1}--\eqref{eq:first-wedge-2} governing $\mathbf{n}^{-,3}(x;\lambda_0)$.  This implies that
\begin{equation}
m_2^{-,1}(x_{23};\lambda_0)=0
\end{equation}
and
\begin{equation}
\begin{split}
n_2^{-,3}(x_{23};\lambda_0)&=e^{i\lambda_0\Delta^{[1]}(x_{23}-x_{12})/\epsilon}n_2^{-,3}(x_{12};\lambda_0)
\\ &=\gamma^{[1]}e^{-i\theta^{[1]}}e^{i\lambda_0\Delta^{[1]}x_{23}/\epsilon}e^{-i\zeta^{[1]}(\lambda_0)x_{12}/\epsilon}w_2^{[1]-,1}(x_{12};\lambda_0),
\end{split}
\end{equation}
where we have used \eqref{eq:n-minus-3-x-less-x12}.  Furthermore, again using the differential equations \eqref{eq:m-columns-ODE} for $j=1$ and \eqref{eq:first-wedge-1}--\eqref{eq:first-wedge-2} shows that on the interval $x_{12}\le x\le x_{23}$, the two-component vector defined either by
\begin{equation}
\mathbf{w}^{[2]-,1}(x;\lambda_0)=\begin{pmatrix}1 & 0\\ 0 & \gamma^{[2]}e^{-i\theta^{[2]}}e^{-i\kappa^{[2]}x/\epsilon}\end{pmatrix}
e^{-i\zeta^{[2]}(\lambda_0)x/\epsilon}\begin{pmatrix}m_1^{-,1}(x;\lambda_0)\\
m_3^{-,1}(x;\lambda_0)\end{pmatrix}
\label{eq:m-minus-1-to-w-2-minus-1}
\end{equation} 
or by 
\begin{equation}
\mathbf{w}^{[2]-,1}(x;\lambda_0)=\begin{pmatrix}1 & 0\\ 0 & -\gamma^{[2]}e^{-i\theta^{[2]}}e^{-i\kappa^{[2]}x/\epsilon}\end{pmatrix}
e^{-i\zeta^{[2]}(\lambda_0)x/\epsilon}\sigma_1\begin{pmatrix}n_1^{-,3}(x;\lambda_0)\\
n_3^{-,3}(x;\lambda_0)\end{pmatrix}
\end{equation}
coincides in each case with the (same) Jost solution of the nonselfadjoint Zakharov-Shabat system \eqref{eq:ZS-Jost} with potential $A(x)=A^{[2]}(x)$ and spectral parameter $\zeta=\zeta^{[2]}(\lambda)$ given by \eqref{eq:A2-zeta2} and asymptotic to the decaying exponential $e^{-i\zeta^{[2]}(\lambda_0)x/\epsilon}(1,0)^\trans$ as $x\to -\infty$.  Since $a^{[2]}(\zeta^{[2]}(\lambda_0))=0$, there exists a nonzero proportionality constant $\tau^{[2]}$ for the Zakharov-Shabat eigenvalue $\zeta^{[2]}(\lambda_0)$ such that $\mathbf{w}^{[2]-,1}(x;\lambda_0)=\tau^{[2]}\mathbf{w}^{[2]+,2}(x;\lambda_0)$, where $\mathbf{w}^{[2]+,2}(x;\lambda_0)$ is the Jost solution of \eqref{eq:ZS-Jost} with potential and spectral parameter given by \eqref{eq:A2-zeta2} asymptotic to the decaying exponential $e^{i\zeta^{[2]}(\lambda_0)x/\epsilon}(0,1)^\trans$ as $x\to +\infty$.  Since $x=x_{23}$ lies to the right of the support of $A^{[2]}$, we have $\mathbf{w}^{[2]+,2}(x_{23};\lambda_0)=e^{i\zeta^{[2]}(\lambda_0)x_{23}/\epsilon}(0,1)^\trans$, and therefore
\begin{equation}
\mathbf{m}^{-,1}(x_{23};\lambda_0)=\begin{pmatrix}
0\\0\\\gamma^{[2]}e^{i\theta^{[2]}}e^{i\lambda_0\Delta^{[2]}x_{23}/\epsilon}\tau^{[2]}
\end{pmatrix}
\label{eq:type-2-m-minus-1-at-x23}
\end{equation}
and
\begin{equation}
\mathbf{n}^{-,3}(x_{23};\lambda_0)=\begin{pmatrix}
-\gamma^{[2]}e^{i\theta^{[2]}}e^{i\lambda_0\Delta^{[2]}x_{23}/\epsilon}\tau^{[2]}\\
\gamma^{[1]}e^{-i\theta^{[1]}}e^{i\lambda_0\Delta^{[1]}x_{23}/\epsilon}
e^{-i\zeta^{[1]}(\lambda_0)x_{12}/\epsilon}w_2^{[1]-,1}(x_{12};\lambda_0)\\
0\end{pmatrix}.
\label{eq:type-2-n-minus-3-at-x23}
\end{equation}

Obviously $\mathbf{m}^{-,1}(x_{23};\lambda_0)$ is proportional to $\mathbf{m}^{+,3}(x_{23};\lambda_0)=\mathbf{e}^3$, and therefore $\beta_{21}^{[\mathrm{type}]}=0$ while 
\begin{equation}
\beta_{31}^{[\mathrm{type}]}=\frac{2\gamma^{[2]}e^{i\theta^{[2]}}\tau^{[2]}}{\Delta^{[2]}a^{[2]\prime}(\zeta^{[2]}(\lambda_0))a^{[3]}(\zeta^{[3]}(\lambda_0))}.
\label{eq:type-solsplit-beta31}
\end{equation}
\emph{Since $\beta_{21}^{[\mathrm{type}]}$ vanishes, this immediately proves that the pole of $\mathbf{M}^\sigma(\lambda)$ at $\lambda=\lambda_0$ cannot in fact be of type $\solfuse$.}  The pole will be of type $2$ if $\mathbf{n}^{-,3}(x_{23};\lambda_0)\times\mathbf{n}^{+,1}(x_{23};\lambda_0)=\mathbf{0}$ and of type $\solsplit$ otherwise.  Computing this cross product using \eqref{eq:n-plus-1-x-greater-x23} and \eqref{eq:type-2-n-minus-3-at-x23} gives
\begin{equation}
\mathbf{n}^{-,3}(x_{23};\lambda_0)\times\mathbf{n}^{+,1}(x_{23};\lambda_0)=
\begin{pmatrix}
0\\0\\
e^{i\lambda_0\Delta^{[1]}x_{23}/\epsilon}\varphi
\end{pmatrix},
\end{equation}
where
\begin{multline}
\varphi:=e^{i(\theta^{[2]}+\theta^{[3]})}\tau^{[2]}e^{i\zeta^{[3]}(\lambda_0)x_{23}/\epsilon}w_1^{[3]+,2}(x_{23};\lambda_0)\\
{}-\gamma^{[1]}e^{-i\theta^{[1]}}e^{-i\zeta^{[1]}(\lambda_0)x_{12}/\epsilon}w_2^{[1]-,1}(x_{12};\lambda_0)e^{-i\zeta^{[3]}(\lambda_0)x_{23}/\epsilon}w_2^{[3]+,2}(x_{23};\lambda_0).
\label{eq:varphi-define}
\end{multline}
Observe that the cross product is clearly proportional to $\mathbf{m}^{+,3}(x_{23};\lambda_0)=\mathbf{e}^3$.  Dividing by $v'(\lambda_0)$ then yields
\begin{equation}
\beta^{[\mathrm{type}]}_{32}=\frac{2\varphi}{\Delta^{[2]}a^{[1]}(\zeta^{[1]}(\lambda_0))a^{[2]\prime}(\zeta^{[2]}(\lambda_0))}.
\label{eq:type-solsplit-beta32}
\end{equation}
Applying \eqref{eq:type-2-beta-relations} or \eqref{eq:type-solsplit-beta-relations} (depending on whether $\varphi=0$ or not) to \eqref{eq:type-solsplit-beta31} and \eqref{eq:type-solsplit-beta32} yields the connection coefficients $\beta_{12}^{[\mathrm{type}]}$ and $\beta_{13}^{[\mathrm{type}]}$ needed to characterize the residues of $\mathbf{M}^-(\lambda)$ at $\lambda=\lambda_0$:
\begin{equation}
\beta_{13}^{[\mathrm{type}]}=\frac{2\gamma^{[2]}e^{-i\theta^{[2]}}}{\Delta^{[2]}a^{[1]}(\zeta^{[1]}(\lambda_0))
a^{[2]\prime}(\zeta^{[2]}(\lambda_0))\tau^{[2]}}
\label{eq:type-solsplit-beta13}
\end{equation}
and
\begin{equation}
\beta_{12}^{[\mathrm{type}]}=-\frac{2\gamma^{[2]}e^{-i\theta^{[2]}}\varphi}{\Delta^{[2]}a^{[2]\prime}(\zeta^{[2]}(\lambda_0))a^{[3]}(\zeta^{[3]}(\lambda_0))}.
\label{eq:type-solsplit-beta12}
\end{equation}
The type of the pole is $2$ if $\varphi=0$ and $\solsplit$ if $\varphi\neq 0$.  The connection coefficients contain information about all three fields even though the pole $\lambda_0$ is generated from the field $q^{[2]}(x,0)$ independently.

The fact that the pole of $\mathbf{M}^\sigma(\lambda)$ corresponding to a Zakharov-Shabat eigenvalue of $A^{[2]}(x)$ cannot be of type $\solfuse$ can also be explained intuitively because, as a consequence of the ordering \eqref{eq:c-inequalities} and \eqref{eq:disjoint-support-assumption} of the velocities and initial supports respectively, we have $q^{[2]}(x,t)=q^{[2]}(x-c^{[2]}t,0)$ for all $t\le 0$.  Therefore all solitons traveling with velocity $c^{[2]}$ at time $t=0$ remain unchanged backwards in time (the packets only interact for $t\ge t_0>0$) and there is no mechanism for them to divide into solitons with velocities $c^{[1]}$ and $c^{[3]}$ in the negative $t$ direction.

\section{Semiclassical Approximation}
\label{sec4-semiclassical-approx}
Assuming that the zeros of $a^{[k]}(\zeta^{[k]}(\lambda))$ are distinct, simple, and finite in number,
all data needed to formulate the Riemann-Hilbert problem of inverse scattering for initial conditions of the form \eqref{eq:initial-data} subject to the support assumption \eqref{eq:disjoint-support-assumption} has been determined
in terms of quantities computable from $2\times 2$ systems of nonselfadjoint Zakharov-Shabat type.  Now we wish to approximate this data in the limit $\epsilon\downarrow 0$, and to do so we shall rely on formulae arising from the semiclassical asymptotic theory of the Zakharov-Shabat problem \eqref{eq:ZS-Jost} that is summarized in Appendix~\ref{subapp-semiclassical-ZS}.

In order to apply the asymptotic formulae from Appendix~\ref{subapp-semiclassical-ZS}, we now assume that the disjointly supported non-negative amplitude functions $H^{[k]}(x)$, $k=1,2,3$, are all continuous \emph{Klaus-Shaw} potentials, i.e., having exactly one peak.  Then it follows \cite{KlausS02} that the discrete eigenvalues $\zeta$ of the Zakharov-Shabat system \eqref{eq:ZS-Jost} with potential $A(x)=A^{[k]}(x)$ (also Klaus-Shaw by simple rescaling) and spectral parameter $\zeta=\zeta^{[k]}(\lambda)$ given by \eqref{eq:A1-zeta1}, \eqref{eq:A2-zeta2}, or \eqref{eq:A3-zeta3} are purely imaginary numbers, and hence in the $\lambda$-plane the corresponding points lie on vertical lines with real parts equal to $\ell^{[k]}$, $k=1,2,3$, where
\begin{equation}
\ell^{[1]}:=-\frac{\kappa^{[1]}}{\Delta^{[1]}},\quad\ell^{[2]}:=\frac{\kappa^{[2]}}{\Delta^{[2]}}, \quad\text{and}\quad
\ell^{[3]}:=-\frac{\kappa^{[3]}}{\Delta^{[3]}}.  
\label{eq:ell-k}
\end{equation}
To make sure that the functions $a^{[k]}(\zeta^{[k]}(\lambda))$ have no common zeros for $\lambda\in\mathbb{C}_+$,
we assume further that the values $\ell^{[k]}$, $k=1,2,3$, are all distinct.

\subsection{Quantization of amplitudes}  
\label{subsec:quantization}
To use the asymptotic formulae in Appendix~\ref{subapp-semiclassical-ZS}, we also need to relate the semiclassical limit $\epsilon\downarrow 0$ to a discrete limit by ensuring that for $k=1,2,3$, the condition
\begin{equation}
\epsilon = \frac{1}{N^{[k]}\pi}\int_\mathbb{R}A^{[k]}(x)\,\dd x
\label{eq:ZS-epsilon-assumption-first}
\end{equation}
holds for $\epsilon>0$ and some positive integers $N^{[k]}$ (cf., \eqref{eq:ZS-epsilon-assumption}).
Indeed, \eqref{eq:ZS-epsilon-assumption-first} implies that the nonselfadjoint Zakharov-Shabat system 
\eqref{eq:ZS-Jost} with Klaus-Shaw potential $A(x)=A^{[k]}(x)$ has precisely $N^{[k]}$ purely imaginary
eigenvalues in the upper half-plane \cite{KlausS02}.
Here, the (Klaus-Shaw) potentials $A^{[k]}(x)$, $k=1,2,3$, are defined by \eqref{eq:A1-zeta1}, \eqref{eq:A2-zeta2}, and \eqref{eq:A3-zeta3}.  Thus the limit $\epsilon\downarrow 0$ corresponds to the integers $N^{[k]}$ growing without bound.

If the integrals of $A^{[k]}$ are commensurate in the sense that there exist positive integers $M^{[k]}$, $k=1,2,3$, and some positive real number $E$ such that
\begin{equation}
\int_\mathbb{R}A^{[k]}(x)\,\dd x = M^{[k]}E,\quad k=1,2,3,
\end{equation}
then the condition \eqref{eq:ZS-epsilon-assumption-first} will hold for each $A=A^{[k]}$ with a common value of $\epsilon>0$ provided we also choose the integers $N^{[k]}$ in the form
\begin{equation}
N^{[k]}=M^{[k]}N,\quad N_0=0,1,2,3,\dots.
\label{eq:N-commensurate-quantize}
\end{equation}
We may therefore define a discrete sequence $\{\epsilon_N\}_{N=1}^\infty$ such that
$\epsilon_N\downarrow 0$ as $N\to\infty$ by setting 
\begin{equation}
\epsilon_N:=\frac{E}{N\pi},
\end{equation}
so that replacing $\epsilon$ with $\epsilon_N$ and using \eqref{eq:N-commensurate-quantize}, \eqref{eq:ZS-epsilon-assumption-first} holds for all $N$.  This allows the semiclassical limit to be explored with $\epsilon\downarrow 0$ along the discrete sequence $\{\epsilon_N\}_{N=1}^\infty$.

We wish to avoid the assumption that the integrals of $A^{[k]}$ are commensurate, and we would also like to be able to let $\epsilon\downarrow 0$ along any sequence, or indeed as a continuous variable, but we still
require \eqref{eq:ZS-epsilon-assumption-first} to hold.
We therefore proceed differently, by first approximating the given potentials $A^{[k]}$ by commensurate potentials obtained by multiplying each $A^{[k]}$ by a scale factor $f^{[k]}(\epsilon)$ such that $f^{[k]}(\epsilon)=1+\littleo{1}$ as $\epsilon\downarrow 0$.  The procedure we will now describe can also be applied in the commensurate case, and it allows $\epsilon$ to tend to zero as a continuous variable.  We call it \emph{quantization of amplitudes}.

Given the disjointly supported Klaus-Shaw potentials $A^{[k]}(x)$, $k=1,2,3$, we define three integer-valued functions of $\epsilon>0$ by
\begin{equation}
N^{[k]}=N^{[k]}(\epsilon):=\left\llbracket \frac{1}{\epsilon\pi}\int_\mathbb{R}A^{[k]}(x)\,\dd x\right\rrbracket,
\quad\epsilon>0,
\label{eq:N-of-epsilon}
\end{equation}
where $\llbracket\cdot\rrbracket$ denotes the nearest integer function defined for concreteness to satisfy $\llbracket n+\tfrac{1}{2}\rrbracket=n+1$ for $n\in\mathbb{Z}$.  Obviously $N^{[k]}(\epsilon)$ is nondecreasing as $\epsilon>0$ decreases, and $N^{[k]}(\epsilon)\to\infty$ as $\epsilon\downarrow 0$.  Then define a renormalization factor by 
\begin{equation}
f^{[k]}(\epsilon):=N^{[k]}(\epsilon)\pi\epsilon\left[\int_\mathbb{R}A^{[k]}(x)\,\dd x\right]^{-1},\quad\epsilon>0.
\label{eq:renormalization-factor}
\end{equation}
It is easy to see that $f^{[k]}(\epsilon)=1+\bigo{\epsilon}$ as $\epsilon\downarrow 0$, and moreover we have the obvious identity
\begin{equation}
\epsilon = \frac{1}{N^{[k]}(\epsilon)\pi}\int_\mathbb{R}f^{[k]}(\epsilon)A^{[k]}(x)\,\dd x
\end{equation}
which should be compared with \eqref{eq:ZS-epsilon-assumption-first}.  Therefore, by quantization of amplitudes we mean the following procedure:  given the $\epsilon$-independent Klaus-Shaw functions $H^{[k]}(x)$ (and hence $A^{[k]}(x)$ via \eqref{eq:A1-zeta1}, \eqref{eq:A2-zeta2}, and \eqref{eq:A3-zeta3}) and $\epsilon>0$ we will \emph{replace} $A^{[k]}(x)$ by 
\begin{equation}
A_\epsilon^{[k]}(x):=f^{[k]}(\epsilon)A^{[k]}(x), \quad k=1,2,3, 
\label{eq:rescaled-potentials}
\end{equation}
and study the Cauchy problem with the corresponding rescaled initial data $q^{[k]}_\epsilon(x,0):=f^{[k]}(\epsilon)q^{[k]}(x,0)$, $k=1,2,3$.  The number of Zakharov-Shabat eigenvalues in $\mathbb{C}_+$ generated by the (barely, when $\epsilon>0$ is small) rescaled potential $A(x)=A_\epsilon^{[k]}(x)$ in \eqref{eq:ZS-Jost} will then be exactly $N^{[k]}(\epsilon)$ for each $\epsilon>0$, and we may use all of the asymptotic formulae recorded in Appendix~\ref{subapp-semiclassical-ZS} in the continuous limit $\epsilon\downarrow 0$.  In particular, these formula prescribe $N^{[k]}(\epsilon)$ positive imaginary approximate eigenvalues in one-to-one correspondence with the actual (positive imaginary) eigenvalues.  Since the modification of the initial data involves multiplying by $x$-independent factors of the form $1+\bigo{\epsilon}$, it is clear that it corresponds to a relative error of $\bigo{\epsilon}$ in any norm.

\subsection{The semiclassical soliton ensemble}
\label{section:semiclassical-soliton-ensemble}
We now explain how we will approximate in the limit $\epsilon\downarrow 0$ the scattering data for the Cauchy problem for the TWRI system \eqref{3wave} with initial data of the form \eqref{eq:initial-data} having quantized disjointly supported Klaus-Shaw amplitudes.  

According to Appendix~\ref{subapp-semiclassical-ZS}, under the assumptions in force $b^{[k]}(\zeta^{[k]}(\lambda))=\littleo{1}$ as $\epsilon\downarrow 0$ uniformly for $\lambda\in\mathbb{R}$, for $k=1,2,3$.  According to \eqref{eq:S1-pieces}, \eqref{eq:S2-pieces}, and \eqref{eq:S3-pieces}, the partial scattering matrices $\mathbf{S}^{[k]}(\lambda)$, $k=1,2,3$, are all approximately diagonal, and hence the same is true for their product $\mathbf{S}(\lambda)=\mathbf{S}^{[1]}(\lambda)\mathbf{S}^{[2]}(\lambda)\mathbf{S}^{[3]}(\lambda)$.  Applying LDU or UDL factorization to $\mathbf{S}(\lambda)$ (cf., \eqref{eq:LDU-UDL}), one sees that the quantities $T_{j\ell}^\sigma(\lambda)$, $j\neq\ell$, are all uniformly small for $\lambda\in\mathbb{R}$, from which it follows via \eqref{eq:TWRI-Jump} that the jump matrix $\mathbf{V}_0^\sigma(\lambda)$ may be approximated with the identity matrix in the semiclassical limit $\epsilon\downarrow 0$. As part of the definition of the semiclassical soliton ensemble, we therefore replace the jump matrix $\mathbf{V}_0^\sigma(\lambda)$ by $\widetilde{\mathbf{V}}_0^\sigma(\lambda):=\mathbb{I}$ for $\lambda\in\mathbb{R}$, making the approximating inverse-scattering problem reflectionless and therefore reducible to finite-dimensional linear algebra.  

The exact poles in $\mathbb{C}_+$ for the problem are the $N^{[1]}+N^{[2]}+N^{[3]}$ distinct values
$\{\lambda^{[k]}_n\}_{n=0}^{N^{[k]}-1}$, $k=1,2,3$, where the $\lambda^{[k]}_n$ are the necessarily simple roots in $\mathbb{C}_+$ of $a^{[k]}(\zeta^{[k]}(\lambda))$ calculated for the quantized amplitude $A(x)=A_\epsilon^{[k]}(x)$ .  We take the poles $\lambda_0^{[k]},\dots,\lambda_{N^{[k]}-1}^{[k]}$ to be ordered along the line $\re\{\lambda\}=\ell^{[k]}$ in the downward direction toward $\mathbb{R}$.  By the WKB theory described in Appendix~\ref{subapp-semiclassical-ZS}, we define approximations to the poles denoted $\{\widetilde{\lambda}^{[k]}_{n}\}_{n=0}^{N^{[k]}-1}$ as follows:
\begin{equation}
\widetilde{\lambda}^{[k]}_{n}:=(\zeta^{[k]})^{-1}(i\widetilde{s}^{[k]}_{n}),\quad n=0,1,\dots,N^{[k]}-1,
\label{eq:approx-eigenvals}
\end{equation}
in which $(\zeta^{[k]})^{-1}$ denotes the (linear) inverse of the function $\zeta^{[k]}$ and the positive numbers $\widetilde{s}^{[k]}_{n}$ are determined from the Bohr-Sommerfeld quantization rule
\begin{equation}
\Psi^{[k]}(i\widetilde{s}^{[k]}_{n})=(n+\tfrac{1}{2})\epsilon\pi = \frac{2n+1}{2N^{[k]}}\int_\mathbb{R}A_\epsilon^{[k]}(x)\,\dd x,\quad n=0,1,\dots,N^{[k]}-1,
\label{eq:ensemble-eigenvalues}
\end{equation}
where the phase integral is given by 
\begin{equation}
\Psi^{[k]}(is):=\int_{x^{[k]}_-(s)}^{x^{[k]}_+(s)}\sqrt{A_\epsilon^{[k]}(x)^2-s^2}\,\dd x,\quad 0<s<\max_{x\in\mathbb{R}}A_\epsilon^{[k]}(x),
\label{eq:phase-integral}
\end{equation}
with $x^{[k]}_-(s)<x^{[k]}_+(s)$ being the two roots of the equation $A_\epsilon^{[k]}(x)=s$.  Note that the approximations $\widetilde{\lambda}^{[k]}_0,\dots,\widetilde{\lambda}^{[k]}_{N^{[k]}-1}$ are also ordered along the line $\re\{\lambda\}=\ell^{[k]}$ in the direction toward the real axis.  

To complete the definition of the scattering data for the semiclassical soliton ensemble we need to approximate the connection coefficients for the simple poles.  Aside from explicit known constants that do not require approximation, the exact formulae \eqref{eq:type-1-beta32}, \eqref{eq:type-1-beta23}, \eqref{eq:type-3-beta21}, \eqref{eq:type-3-beta12}, \eqref{eq:type-solsplit-beta31}, \eqref{eq:varphi-define}--\eqref{eq:type-solsplit-beta32}, \eqref{eq:type-solsplit-beta13}, and \eqref{eq:type-solsplit-beta12} involve (i) the functions $a^{[k]}(\zeta^{[k]}(\lambda))$ and their derivatives evaluated at the poles, (ii) proportionality constants $\tau^{[k]}$ associated with eigenfunctions of \eqref{eq:ZS-Jost}, and (iii) Jost solutions of \eqref{eq:ZS-Jost} evaluated when $\zeta$ is not an eigenvalue at values of $x$ outside the support of the quantized amplitude $A(x)$.  All of these quantities require approximation in the semiclassical limit $\epsilon\downarrow 0$, and we consider them in turn.

Having neglected the functions $b^{[k]}(\zeta^{[k]}(\lambda))$ and approximated the zeros of the functions $a^{[k]}(\zeta^{[k]}(\lambda))$ to yield $\{\widetilde{\lambda}_n^{[k]}\}_{n=0}^{N^{[k]}-1}$ for $k=1,2,3$, according to Appendix~\ref{subapp-semiclassical-ZS} the latter functions are themselves approximated for $\lambda\in\mathbb{C}_+$ by
\begin{equation}
a^{[k]}(\zeta^{[k]}(\lambda))\approx\widetilde{a}^{[k]}(\zeta^{[k]}(\lambda)),\quad \widetilde{a}^{[k]}(\zeta):=\prod_{n=0}^{N^{[k]}-1}\frac{\zeta-i\widetilde{s}^{[k]}_{n}}{\zeta+i\widetilde{s}^{[k]}_{n}}.
\label{eq:tilde-a-ensemble}
\end{equation}
The derivative of $a^{[k]}(\zeta^{[k]}(\lambda))$ with respect to $\zeta$ at the pole $\lambda_m^{[k]}$ approximated by $\widetilde{\lambda}^{[k]}_{m}$ is approximated by
\begin{equation}
a^{[k]\prime}(\zeta^{[k]}(\lambda))\approx\widetilde{a}^{[k]\prime}(i\widetilde{s}^{[k]}_{m})=\frac{1}{i}
\frac{\prod_{n\neq m}(\widetilde{s}^{[k]}_{m}-\widetilde{s}^{[k]}_{n})}{\prod_{n=0}^{N^{[k]}-1}(\widetilde{s}^{[k]}_{m}+\widetilde{s}^{[k]}_{n})}.
\label{eq:tilde-a-prime}
\end{equation}

Again referring to Appendix~\ref{subapp-semiclassical-ZS}, the proportionality constant $\tau^{[k]}$ associated with the Zakharov-Shabat eigenvalue approximated by $\zeta^{[k]}(\widetilde{\lambda}^{[k]}_{n})$ is approximated, \emph{mutatis mutandis}, by \eqref{eq:ZS-proportionality-constant}--\eqref{eq:mu-define}; namely for any integer $K$,
\begin{equation}
\tau^{[k]}\approx\widetilde{\tau}^{[k]}_n:=i(-1)^{K}\exp\left(\frac{i}{\epsilon}(2K+1)\Psi^{[k]}(i\widetilde{s}^{[k]}_{n})\right)\exp\left(\frac{1}{\epsilon}\mu^{[k]}(i\widetilde{s}^{[k]}_{n})\right).
\label{eq:tilde-tau}
\end{equation}
Here the function $\mu^{[k]}(is)$ is defined by \eqref{eq:mu-define}, replacing $A(x)$ and $x_\pm(s)$ by $A_\epsilon^{[k]}(x)$ and $x^{[k]}_\pm(s)$, respectively.  The formula \eqref{eq:tilde-tau} is independent\footnote{Keeping the integer $K$ arbitrary will be useful in applying steepest descent asymptotic analysis to the Riemann-Hilbert formulation of the semiclassical soliton ensemble, as it allows the use of multiple analytic interpolants of the same quantities $\widetilde{\tau}^{[k]}_n$ at the corresponding approximate poles $\widetilde{\lambda}_n^{[k]}$.  See \cite{Buckingham:inprep}.} of the choice of $K\in\mathbb{Z}$ due to the definition \eqref{eq:ensemble-eigenvalues}.

In \eqref{eq:type-1-beta32}--\eqref{eq:type-1-beta23} the quantity $w^{[3]+,2}_2(x_{23};\lambda_0)e^{-i\zeta^{[3]}(\lambda_0)x_{23}/\epsilon}$ appears in which $\lambda_0$ is a simple zero of $a^{[1]}(\zeta^{[1]}(\lambda))$, while in \eqref{eq:type-3-beta21}--\eqref{eq:type-3-beta12} the quantity $w^{[1]-,1}_1(x_{12};\lambda_0)e^{i\zeta^{[1]}(\lambda_0)x_{12}/\epsilon}$ appears in which $\lambda_0$ is a simple zero of $a^{[3]}(\zeta^{[3]}(\lambda))$.  These expressions involve solutions of \eqref{eq:ZS-Jost} for potentials $A=A_\epsilon^{[1]}$ and $A=A_\epsilon^{[3]}$ respectively, normalized at $x=-\infty$ and $x=+\infty$ respectively, and evaluated at a point on the other side of the support of the potential.
According to Appendix~\ref{subapp-semiclassical-ZS}, in the limit $\epsilon\downarrow 0$ these expressions can be analyzed using the WKB method in the absence of turning points, and they satisfy
\begin{equation}
\begin{split}
w^{[3]+,2}_2(x_{23};\lambda_0)e^{-i\zeta^{[3]}(\lambda_0)x_{23}/\epsilon}&=e^{-L^{[3]}(\zeta^{[3]}(\lambda))/\epsilon}(1+\littleo{1}),\\
w^{[1]-,1}_1(x_{12};\lambda_0)e^{i\zeta^{[1]}(\lambda_0)x_{12}/\epsilon}&=e^{-L^{[1]}(\zeta^{[1]}(\lambda))/\epsilon}(1+\littleo{1}),
\end{split}
\label{eq:Jost-1-approx}
\end{equation}
where for $\zeta$ in the upper half-plane omitting the segment $0\le -i\zeta\le \max_{x\in\mathbb{R}}A_\epsilon^{[k]}(x)$,
\begin{equation}
L^{[k]}(\zeta):=\int_\mathbb{R}\left[(-\zeta^2-A_\epsilon^{[k]}(y)^2)^{1/2}+i\zeta\right]\,\dd y.
\label{eq:L-quantized}
\end{equation}
See \eqref{eq:approximate-Josts}--\eqref{eq:Lfunc-def}.

The above approximations are sufficient to determine approximations of the connection coefficients for poles of types 1 and 3.  Indeed, replacing the functions $a^{[k]}(\zeta^{[k]}(\lambda))$ by $\widetilde{a}^{[k]}(\zeta^{[k]}(\lambda))$ according to \eqref{eq:tilde-a-ensemble}--\eqref{eq:tilde-a-prime}, replacing $\tau^{[k]}$ by $\widetilde{\tau}_n^{[k]}$ given by \eqref{eq:tilde-tau}, and neglecting the error terms in the
approximations \eqref{eq:Jost-1-approx} and substituting into \eqref{eq:type-1-beta32}--\eqref{eq:type-1-beta23}, we define
\begin{gather}
\widetilde{\beta}_{n,32}^{[1]}:=
-\frac{2\gamma^{[1]}e^{-i\theta^{[1]}}\widetilde{\tau}^{[1]}_ne^{-L^{[3]}(\zeta^{[3]}(\widetilde{\lambda}_n^{[1]}))/\epsilon}}{\Delta^{[1]}\widetilde{a}^{[1]\prime}(i\widetilde{s}^{[1]}_{n})\widetilde{a}^{[2]}(\zeta^{[2]}(\widetilde{\lambda}^{[1]}_n))}
\label{eq:beta32-approx}
\intertext{and} 
\widetilde{\beta}_{n,23}^{[1]}:=
-\frac{2\gamma^{[1]}e^{i\theta^{[1]}}\widetilde{a}^{[3]}(\zeta^{[3]}(\widetilde{\lambda}^{[1]}_n))}{
\Delta^{[1]}\widetilde{a}^{[1]\prime}(i\widetilde{s}^{[1]}_{n})\widetilde{\tau}^{[1]}_ne^{-L^{[3]}(\zeta^{[3]}(\widetilde{\lambda}_n^{[1]}))/\epsilon}}
\label{eq:beta23-approx}
\end{gather}
as the approximate connection coefficients for the approximate type 1 pole $\widetilde{\lambda}_n^{[1]}$ for $n=0,\dots,N^{[1]}(\epsilon)$.  Likewise substituting into 
\eqref{eq:type-3-beta21}--\eqref{eq:type-3-beta12}, we define
\begin{gather}
\widetilde{\beta}^{[3]}_{n,21}:=
-\frac{2\gamma^{[3]}e^{-i\theta^{[3]}}\widetilde{a}^{[1]}(\zeta^{[1]}(\widetilde{\lambda}^{[3]}_n))\widetilde{\tau}^{[3]}_n}{\Delta^{[3]}\widetilde{a}^{[3]\prime}(i\widetilde{s}^{[3]}_{n})
e^{-L^{[1]}(\zeta^{[1]}(\widetilde{\lambda}^{[3]}_n))/\epsilon}}
\label{eq:beta21-approx}
\intertext{and}
\widetilde{\beta}^{[3]}_{n,12}:=
-\frac{2\gamma^{[3]}e^{i\theta^{[3]}}e^{-L^{[1]}(\zeta^{[1]}(\widetilde{\lambda}^{[3]}_n))/\epsilon}}{\Delta^{[3]}\widetilde{a}^{[3]\prime}(i\widetilde{s}^{[3]}_{n})\widetilde{a}^{[2]}(\zeta^{[2]}(\widetilde{\lambda}^{[3]}_n))\widetilde{\tau}^{[3]}_n}
\label{eq:beta12-approx}
\end{gather}
as the approximate connection coefficients for the approximate type 3 pole $\widetilde{\lambda}_n^{[3]}$ for $n=0,\dots,N^{[3]}(\epsilon)$.

To approximate the connection coefficients for the simple poles contributed by the central packet $q^{[2]}(x,0)$, we first need to determine whether these are poles of type $2$ or poles of type $\solsplit$, i.e., whether $\varphi$ given by \eqref{eq:varphi-define} vanishes or not.  The two terms in $\varphi$ are proportional to the expressions $e^{L^{[3]}(\zeta^{[3]}(\lambda))/\epsilon}e^{-i\zeta^{[3]}(\lambda_0)x_{23}/\epsilon}w_1^{[3]+,2}(x_{23};\lambda_0)$ and
$e^{L^{[1]}(\zeta^{[1]}(\lambda))/\epsilon}e^{i\zeta^{[1]}(\lambda_0)x_{12}/\epsilon}w_2^{[1]-,1}(x_{12};\lambda_0)$ respectively, and according to the discussion at the end of Appendix~\ref{subapp-semiclassical-ZS} these quantities are small as $\epsilon\downarrow 0$, but it is difficult to ascertain exactly how small they are; indeed in the case that the potentials are all infinitely differentiable, these have complete WKB expansions in powers of $\epsilon$ each term of which vanishes, making them small beyond all orders.  Rather than deal with the question of approximating these small quantities, for the purposes of this paper we will replace them with zero, i.e., we will approximate $\varphi$ with $\widetilde{\varphi}=0$ for each pole generated by the central packet.  In other words, \emph{for the purposes of semiclassical approximation, all poles generated by the central packet $q^{[2]}(x,0)$ are taken to be poles of type 2 rather than type $\solsplit$.}  
It therefore only remains to approximate the connection coefficients associated with poles (of type 2) having approximations $\{\widetilde{\lambda}^{[2]}_n\}_{n=0}^{N^{[2]}-1}$ all lying on the vertical line $\re\{\widetilde{\lambda}_n^{[2]}\}=\ell^{[2]}$.  Substituting into \eqref{eq:type-solsplit-beta31}, we define 
\begin{gather}
\widetilde{\beta}^{[2]}_{n,31}:=\frac{2\gamma^{[2]}e^{i\theta^{[2]}}\widetilde{\tau}^{[2]}_n}{\Delta^{[2]}\widetilde{a}^{[2]\prime}(i\widetilde{s}^{[2]}_{n})
\widetilde{a}^{[3]}(\zeta^{[3]}(\widetilde{\lambda}^{[2]}_n))}
\label{eq:tilde-type-solsplit-beta31}
\intertext{and, substituting into \eqref{eq:type-solsplit-beta13},}
\widetilde{\beta}^{[2]}_{n,13}:=
\frac{2\gamma^{[2]}e^{-i\theta^{[2]}}}{\Delta^{[2]}\widetilde{a}^{[1]}(\zeta^{[1]}(\widetilde{\lambda}^{[2]}_n))\widetilde{a}^{[2]\prime}(i\widetilde{s}^{[2]}_{n})\widetilde{\tau}^{[2]}_n}
\label{eq:tilde-type-solsplit-beta13}
\end{gather}
as the approximate connection coefficients for the type 2 approximate pole $\widetilde{\lambda}_n^{[2]}$ for $n=0,\dots,N^{[2]}(\epsilon)-1$.

We therefore arrive at the following formal definition.
\begin{definition}[Semiclassical soliton ensembles for the TWRI equations]
\label{def:SSE}
Let Cauchy initial data of the form \eqref{eq:initial-data} subject to the disjoint-support condition \eqref{eq:disjoint-support-assumption} be
given such that the amplitudes $H^{[k]}(\cdot)$ are all continuous Klaus-Shaw functions.  For each $\epsilon>0$, the corresponding semiclassical soliton ensemble is the solution $\{\widetilde{q}^{[k]}(x,t)\}_{k=1}^3$ of the TWRI equations \eqref{3wave} generated from the following scattering data:
\begin{itemize}
\item Trivial jump:  the jump matrix across $\mathbb{R}$ is $\widetilde{\mathbf{V}}^\sigma_0(\lambda):=\mathbb{I}$.  
\item $N^{[1]}(\epsilon)$ simple poles of type 1 in $\mathbb{C}_+$ denoted $\{\widetilde{\lambda}_n^{[1]}\}_{n=0}^{N^{[1]}(\epsilon)-1}$ and defined by \eqref{eq:approx-eigenvals}--\eqref{eq:phase-integral} with connection coefficients $\{\widetilde{\beta}_{n,32}^{[1]}\}_{n=0}^{N^{[1]}(\epsilon)-1}$ and $\{\widetilde{\beta}_{n,23}^{[1]}\}_{n=0}^{N^{[1]}(\epsilon)-1}$ defined by \eqref{eq:beta32-approx}--\eqref{eq:beta23-approx}.
\item $N^{[2]}(\epsilon)$ simple poles of type 2 in $\mathbb{C}_+$ denoted $\{\widetilde{\lambda}_n^{[2]}\}_{n=0}^{N^{[2]}(\epsilon)-1}$ and defined by \eqref{eq:approx-eigenvals}--\eqref{eq:phase-integral} with connection coefficients $\{\widetilde{\beta}_{n,31}^{[2]}\}_{n=0}^{N^{[2]}(\epsilon)-1}$ and $\{\widetilde{\beta}_{n,13}^{[2]}\}_{n=0}^{N^{[2]}(\epsilon)-1}$ defined by \eqref{eq:tilde-type-solsplit-beta31}--\eqref{eq:tilde-type-solsplit-beta13}.
\item $N^{[3]}(\epsilon)$ simple poles of type 3 in $\mathbb{C}_+$ denoted $\{\widetilde{\lambda}_n^{[3]}\}_{n=0}^{N^{[3]}(\epsilon)-1}$ and defined by \eqref{eq:approx-eigenvals}--\eqref{eq:phase-integral} with connection coefficients $\{\widetilde{\beta}_{n,21}^{[1]}\}_{n=0}^{N^{[3]}(\epsilon)-1}$ and $\{\widetilde{\beta}_{n,12}^{[3]}\}_{n=0}^{N^{[3]}(\epsilon)-1}$ defined by \eqref{eq:beta21-approx}--\eqref{eq:beta12-approx}.
\item No simple poles of types $\solsplit$ or $\solfuse$, or any other singularities.
\end{itemize}
\end{definition}
The semiclassical soliton ensemble for initial data \eqref{eq:initial-data} is intended to be a good approximation to the actual solution of the Cauchy initial value problem and in particular should be an accurate approximation at $t=0$.

We wish to emphasize that the approximation of the poles contributed by the central packet $q^{[2]}(x,0)$ as poles of type 2 only by neglecting $\varphi$ is of a particularly uncontrolled nature.  Indeed, while each term of $\varphi$ contains a factor that is hard to approximate but small, there are also other factors that can be exponentially large in the limit $\epsilon \to 0$.  It is worth observing that this difficulty disappears entirely if one assumes that $H^{[2]}(x)\equiv 0$, i.e., there is nothing in the central channel at $t=0$, as then one has $a^{[2]}(\zeta^{[2]}(\lambda))\equiv 1$ and there are no controversial poles at all.  The accuracy of the semiclassical soliton ensemble in the case $H^{[2]}(x)\equiv 0$ will be confirmed numerically in \S\ref{subsec:two-packets-convergence}, and it will be shown in \S\ref{subsec:three-packets-convergence} that under some additional conditions the approximation $\varphi\approx 0$ is valid when $H^{[2]}(x)$ does not vanish identically, although it can also fail.

\subsection{Semiclassical soliton ensembles for colliding semicircular packets}
\label{subsec-data-for-plots}
Initial data of the form \eqref{eq:initial-data} for which the amplitude functions are semicircular profiles,
\begin{equation}
H^{[k]}(x):=\frac{2H_\mathrm{max}^{[k]}\chi_{(a^{[k]},b^{[k]})}(x)}{b^{[k]}-a^{[k]}}\sqrt{(x-a^{[k]})(b^{[k]}-x)},\quad k=1,2,3,
\label{eq:semicircular-H}
\end{equation}
are particularly convenient for the study of semiclassical soliton ensembles.  Here,
$H^{[k]}_\mathrm{max}\ge 0$ is the maximum value of $H^{[k]}(x)$ and the support endpoints satisfy $a^{[1]}<b^{[1]}<a^{[2]}<b^{[2]}<a^{[3]}<b^{[3]}$ for consistency with \eqref{eq:disjoint-support-assumption}.  Obviously the functions $H^{[k]}(x)$ all satisfy the Klaus-Shaw condition.  The main convenience of semicircular amplitude profiles stems from the fact that the quantities $\Psi^{[k]}(is)$, $L^{[k]}(\zeta)$, and $\mu^{[k]}(is)$, which are normally given by integral transforms of the rescaled and quantized amplitude functions $A_\epsilon^{[k]}(x)$ (see \eqref{eq:A1-zeta1}, \eqref{eq:A2-zeta2}, \eqref{eq:A3-zeta3}, \eqref{eq:renormalization-factor} and the formulae \eqref{eq:ensemble-eigenvalues}, \eqref{eq:L-quantized}, and \eqref{eq:mu-define} with $A$ and $x_\pm$ replaced by $A^{[k]}_\epsilon$ and $x^{[k]}_\pm$ respectively), can be computed in closed-form.  This dramatically speeds up subsequent numerical calculations.  

Indeed, since \eqref{eq:semicircular-H} implies that
\begin{equation}
\int_\mathbb{R} A^{[k]}(x)\,\dd x=\frac{\sqrt{\Delta^{[k]}}}{\sqrt{\Delta^{[1]}\Delta^{[2]}\Delta^{[3]}}}
\int_\mathbb{R} H^{[k]}(x)\,\dd x = \frac{\sqrt{\Delta^{[k]}}(b^{[k]}-a^{[k]})\pi H^{[k]}_\mathrm{max}}{4\sqrt{\Delta^{[1]}\Delta^{[2]}\Delta^{[3]}}},
\end{equation}
$N^{[k]}(\epsilon)$ and $f^{[k]}(\epsilon)$ can be calculated from \eqref{eq:N-of-epsilon} and \eqref{eq:renormalization-factor} without computing any integrals.  Let $A_\mathrm{max}^{[k]}(\epsilon)$ denote the maximum value of $A^{[k]}_\epsilon(x)$:
\begin{equation}
A^{[k]}_\mathrm{max}(\epsilon):=\frac{\sqrt{\Delta^{[k]}}f^{[k]}(\epsilon)H^{[k]}_\mathrm{max}}{\sqrt{\Delta^{[1]}\Delta^{[2]}\Delta^{[3]}}},\quad k=1,2,3.
\end{equation}
Then, 
a residue calculation shows the phase integral \eqref{eq:phase-integral} is explicitly given for the semicircular initial data by
\begin{equation}
\Psi^{[k]}(is)=\frac{\pi(b^{[k]}-a^{[k]})}{4A^{[k]}_\mathrm{max}(\epsilon)}[(A^{[k]}_\mathrm{max}(\epsilon))^2-s^2], \quad 0<s<A_\mathrm{max}^{[k]}(\epsilon),\quad k=1,2,3.
\end{equation}
Similarly, $L^{[k]}(\zeta)$ defined by \eqref{eq:L-quantized} can be calculated explicitly:
\begin{equation}
L^{[k]}(\zeta) =
	\frac{a^{[k]}-b^{[k]}}{4A_\mathrm{max}^{[k]}(\epsilon)}
	\left(2i\zeta A_\mathrm{max}^{[k]}(\epsilon) + (A_\mathrm{max}^{[k]}(\epsilon)^2+\zeta^2)
	\log\left(\frac{i\zeta+A_\mathrm{max}^{[k]}(\epsilon)}{i\zeta-A_\mathrm{max}^{[k]}(\epsilon)}\right)\right), 
	\quad k=1,2,3,
\end{equation}
where $\zeta$ lies in $\mathbb{C}_+$ with the imaginary interval $0\le -i\zeta\le A_\mathrm{max}^{[k]}(\epsilon)$ omitted and the principal branch of the complex logarithm is meant.
Finally, since $A_\epsilon^{[k]}(x)$ is in each case a function that is even about its centroid $\tfrac{1}{2}(a^{[k]}+b^{[k]})$, the formula \eqref{eq:mu-define} gives
\begin{equation}
\mu^{[k]}(is)=\frac{1}{2}(a^{[k]}+b^{[k]})s,\quad 0<s<A^{[k]}_\mathrm{max}(\epsilon),\quad k=1,2,3.
\end{equation}

Another advantage of the semicircular amplitudes \eqref{eq:semicircular-H} is that since the phase integral in each case is quadratic, the Bohr-Sommerfeld quantization rule \eqref{eq:ensemble-eigenvalues} that normally \emph{implicitly} determines the values $\{\widetilde{s}_n^{[k]}\}_{n=0}^{N^{[k]}-1}$ becomes a completely \emph{explicit} formula:
\begin{equation}
\widetilde{s}_n^{[k]} = \sqrt{A^{[k]}_\text{max}(\epsilon)^2-\frac{4\epsilon A^{[k]}_\text{max}(\epsilon)}{b^{[k]}-a^{[k]}}\left(n+\frac{1}{2}\right)}, \quad n=0,\dots,N^{[k]}(\epsilon)-1, \quad k=1,2,3,
\end{equation}
and then the poles of types $1$, $2$, and $3$ for the semiclassical soliton ensemble are explicitly
\begin{equation}
\widetilde{\lambda}_n^{[k]}=\ell^{[k]}+\frac{2\widetilde{s}_n^{[k]}}{\Delta^{[k]}}i,\quad n=0,\dots,N^{[k]}(\epsilon)-1,\quad k=1,2,3,
\end{equation}
where $\ell^{[k]}$ is given by \eqref{eq:ell-k} for $k=1,2,3$.

Applying these explicit results to Definition~\ref{def:SSE} in \S\ref{section:semiclassical-soliton-ensemble} completes the specification of the semiclassical soliton ensemble in the case of initial data \eqref{eq:initial-data} with semicircular amplitude profiles \eqref{eq:semicircular-H}.  To find $\widetilde{q}^{[k]}(x,t)$ one now formulates and solves the corresponding inverse problem.
In this reflectionless setting described in detail in Appendix~\ref{sec:solitons},
the inverse scattering transform amounts to inserting the ansatz \eqref{eq:Mplus-ansatz} for $\mathbf{M}^+(\lambda)$ into the residue conditions \eqref{eq:RHP-pole-C-plus} and \eqref{eq:RHP-pole-C-minus}.
This yields the square linear system \eqref{eq:b-plus-1}--\eqref{eq:b-plus-3}
for the $\mathbb{C}^3$ vector unknowns $\mathbf{b}_n^{[1]+}$ and $\mathbf{c}_n^{[1]+}$ for $n=0,\dots,N^{[1]}(\epsilon)-1$, 
$\mathbf{a}_n^{[2]+}$ and $\mathbf{c}_n^{[2]+}$ for $n=0,\dots,N^{[2]}(\epsilon)-1$, and $\mathbf{a}_n^{[3]+}$ and 
$\mathbf{b}_n^{[3]+}$ for $n=0,\dots,N^{[3]}(\epsilon)-1$
(in all these equations we take\footnote{If we want to consider a semiclassical soliton ensemble for semicircular amplitudes with $H^{[2]}_\mathrm{max}=0$, we also take $N^{[2]}=0$ and omit the equations
\eqref{eq:a-plus-2}--\eqref{eq:c-plus-2} along with the unknowns $\mathbf{a}_n^{[2]+}$ and $\mathbf{c}_n^{[2]+}$ for $n=0,\dots,N^{[2]}(\epsilon)-1$ to obtain a smaller square linear system.} $N^{[\solsplit]} = N^{[\solfuse]}=0$).  Given $(x,t)\in\mathbb{R}^2$ we evaluate the coefficients of this system to high precision and solve for the unknowns numerically.
The three fields $\widetilde{q}^{[k]}(x,t)$, $k=1,2,3$, of the semiclassical soliton ensemble of the TWRI system,
whose scattering data is exactly given by Definition~\ref{def:SSE}, 
are then recovered at the specified $(x,t)\in\mathbb{R}^2$ using \eqref{eq:q-formulas}.  
Repeating this process for a sequence of decreasing values of $\epsilon$, 
we may expect that $\widetilde{q}^{[k]}(x,t)\to q^{[k]}(x,t)$ as $\epsilon\downarrow 0$.

\subsection{Convergence of semiclassical soliton ensembles at $t=0$ without a central packet}  
\label{subsec:two-packets-convergence}
We now illustrate the convergence as $\epsilon\downarrow 0$ of the 
semiclassical soliton ensemble at $t=0$ to the original initial data with 
initial packets in channels 1 and 3 (but not channel 2).  We choose the 
initial condition given by \eqref{eq:initial-data} and \eqref{plot-Hk} with  
parameters \eqref{plot-vals1}.  
The procedure for computing the semiclassical soliton ensembles is 
detailed in \S\ref{subsec:numerical-results}.  
The convergence can be seen visually in Figure \ref{fig-sc13-1d-t0}, where 
we plot $|q^{[k]}(x,0)|=H^{[k]}(x)$ and $|\widetilde{q}^{[k]}(x,0)|$, 
$k=1,3$, for various values of $\epsilon$.  Note that before the collision time of the two semicircles 
(for these parameters the collision occurs at $t=T=1/2$), we have 
$q^{[2]}(x,t)\equiv 0$.  While $\widetilde{q}^{[2]}(x,0)$ is not identically 
zero in these calculations, it is sufficiently close to zero to not be 
distinguishable in the plots.
\begin{figure}[H]
\begin{center}
\includegraphics[height=1.35in]{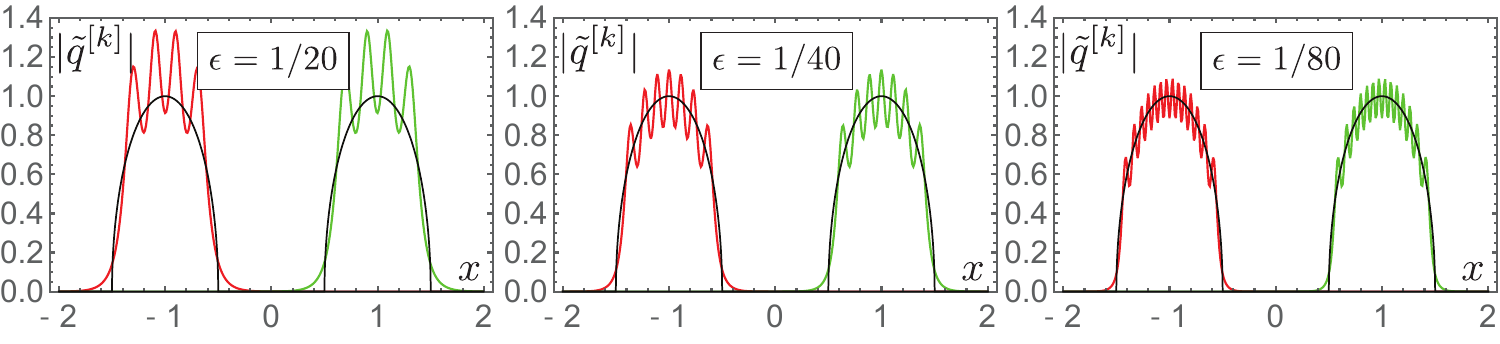}
\caption{Accuracy of the semiclassical soliton ensemble at time $t=0$.  Plots of 
$|q^{[k]}(x,0)|=H^{[k]}(x)$, $k=1,3$,  as defined by 
\eqref{eq:initial-data} and \eqref{plot-Hk} with parameters 
\eqref{plot-vals1} and $\kappa^{[3]}=1$, along 
with the modulus of $\widetilde{q}^{[k]}(x,0)$ obtained from the associated semiclassical soliton ensemble, $k=1,3$, for various $\epsilon$, as 
indicated.  Black:  $H^{[k]}(x)$, $k=1,3$.  Red:  
$|\widetilde{q}^{[1]}(x,0)|$. Green:  $|\widetilde{q}^{[3]}(x,0)|$.}
\label{fig-sc13-1d-t0}
\end{center}
\end{figure}

In Figure \ref{fig-sc13-1d-t0-error} 
we plot the absolute value of the difference between the original data 
$q^{[1]}(x,0)$ and the approximating semiclassical soliton ensemble 
$\widetilde{q}^{[1]}(x,0)$ for various values of $\epsilon$.  The plot with 
$\epsilon=1/80$ suggests the rate 
of convergence is slowest near the semicircle edges (where the initial 
data is nondifferentiable).  The solution in channel 3 displays 
qualitatively similar behavior.
\begin{figure}[H]
\begin{center}
\includegraphics[height=1.35in]{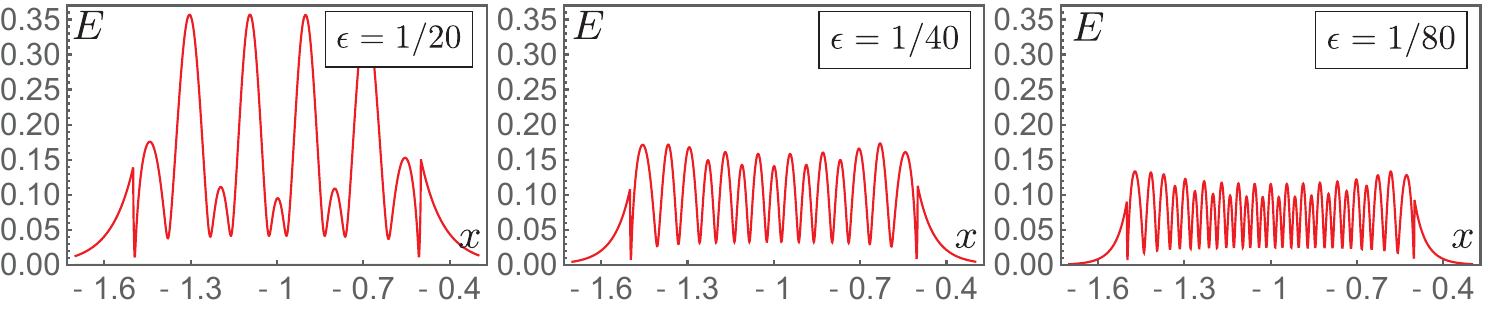}
\caption{Plot of the approximation error 
$E:=|q^{[1]}(x,0)-\widetilde{q}^{[1]}(x,0)|$, where $q^{[1]}$ is defined by 
\eqref{eq:initial-data} and \eqref{plot-Hk} with parameters 
\eqref{plot-vals1} and $\kappa^{[3]}=1$, and 
$\widetilde{q}^{[1]}$ is the corresponding field from the semiclassical soliton ensemble for various 
values of $\epsilon$, as indicated.}
\label{fig-sc13-1d-t0-error}
\end{center}
\end{figure}

Figures \ref{fig-sc13-1d-t0} and \ref{fig-sc13-1d-t0-error} suggest but do not prove 
convergence of the semiclassical soliton ensemble profile at $t=0$ to the original 
initial condition.  However, it is possible to show the convergence for two packets 
rigorously by analyzing the associated Riemann-Hilbert problem.  This will be 
reported elsewhere \cite{Buckingham:inprep}.  Interestingly, while studying the 
Riemann-Hilbert problem it became apparent that convergence does not always hold for 
three packets, as will be shown numerically in 
\S\ref{subsec:three-packets-convergence}.

\subsection{Conditional convergence of semiclassical soliton ensembles at $t=0$ for three 
packets}  
\label{subsec:three-packets-convergence}
We now consider the effect of including a packet in channel 2.  We return to the semicircular amplitudes
considered in \S\ref{subsec-data-for-plots} but now take $H^{[2]}_\mathrm{max}>0$.
For some parameter values the semiclassical soliton ensembles appear to 
converge as $\epsilon\downarrow 0$ at $t=0$ to the given initial data.  For 
example, if we choose the parameters \eqref{plot-vals2}, 
the convergence can be seen visually in Figure \ref{fig-sc123-1d-t0}.  There 
we plot $|q^{[k]}(x,0)|=H^{[k]}(x)$ and $|\widetilde{q}^{[k]}(x,0)|$, $k=1,2,3$, 
for various values of $\epsilon$, where $\widetilde{q}^{[k]}(x,0)$, $k=1,2,3$, are the 
fields of the associated semiclassical soliton ensemble.  
\begin{figure}[h]
\begin{center}
\includegraphics[height=1.35in]{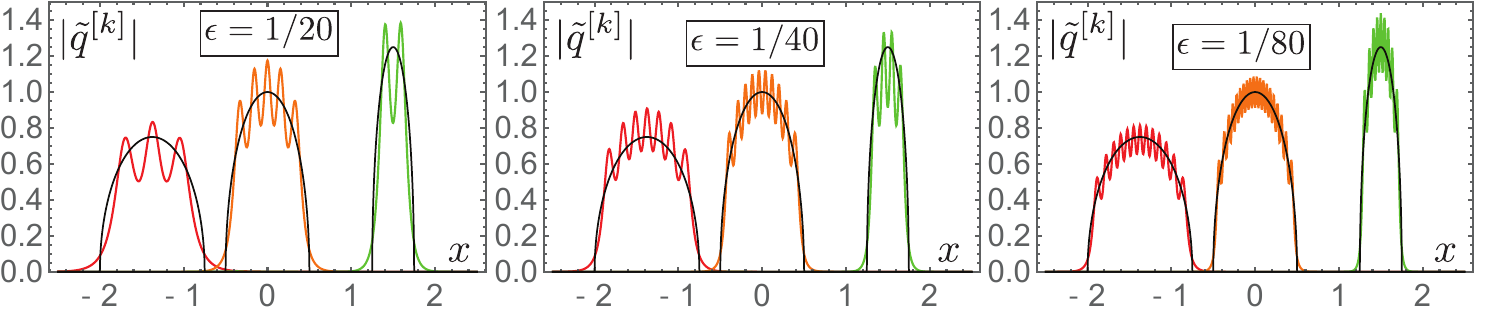}
\caption{Convergence of the semiclassical soliton ensemble at $t=0$ with three 
packets.  Plots of 
$|q^{[k]}(x,0)|=H^{[k]}(x)$ as defined by \eqref{eq:initial-data} and 
\eqref{plot-Hk-sc123} with parameters \eqref{plot-vals2}, along with 
the modulus of the fields $\widetilde{q}^{[k]}(x,0)$, $k=1,2,3$, of the associated semiclassical soliton ensemble for 
various $\epsilon$.  
Black:  $H^{[k]}(x)$, $k=1,2,3$.  Red:  $|\widetilde{q}^{[1]}(x,0)|$.  Orange:  
$|\widetilde{q}^{[2]}(x,0)|$.  Green:  $|\widetilde{q}^{[3]}(x,0)|$.}
\label{fig-sc123-1d-t0}
\end{center}
\end{figure}

However, the situation is different if we choose the parameters 
\eq
\begin{split}
\{c^{[1]}, c^{[2]}, c^{[3]}\}=\left\{1,0,-\frac{1}{5}\right\}, \quad \{\gamma^{[1]}, \gamma^{[2]}, \gamma^{[3]}\} = \{1,-1,1\}, \hspace{1in}\\
\{\kappa_1,\kappa_2,\kappa_3\}=\left\{-1,5,0\right\}, \quad \{H_\text{max}^{[1]},H_\text{max}^{[2]},H_\text{max}^{[3]}\} = \{1,1,1\}, \quad \{\theta^{[1]},\theta^{[2]},\theta^{[3]}\} = \{0,0,0\},\\ 
\{(a^{[1]},b^{[1]}), (a^{[2]},b^{[2]}), (a^{[3]},b^{[3]}) \} = \left\{(-2,-1),\left(-\frac{3}{4},\frac{3}{4}\right),(1,2)\right\}. \hspace{.7in}
\label{plot-vals3}
\end{split}
\endeq
Table \ref{N-for-plot-vals3} gives the associated number of solitons of each type 
for selected $\epsilon$.
\begin{table}[H]
\renewcommand{\arraystretch}{1.1}
\begin{center}
\begin{tabular}{@{}l@{\qquad}c@{\qquad \qquad}c@{\qquad \qquad}c@{}}
\toprule
$\epsilon$ & $N^{[1]}(\epsilon)$ & $N^{[2]}(\epsilon)$ & $N^{[3]}(\epsilon)$ \\
\midrule
0.085 & 3 & 10 & 6 \\
0.075 & 3 & 11 & 7 \\
0.065 & 4 & 13 & 8 \\
0.055 & 4 & 15 & 9 \\
0.045 & 5 & 19 & 11 \\
0.035 & 7 & 24 & 15 \\
\bottomrule
\end{tabular}
\end{center}
\caption{The number of solitons of each type assuming initial data 
defined by \eqref{eq:initial-data}, \eqref{plot-Hk}, and \eqref{plot-vals3} for the 
values of $\epsilon$ used in Figure \ref{fig-sc123-1d-t0-nonconv}.}
\label{N-for-plot-vals3}
\end{table}
As shown in Figure \ref{fig-sc123-1d-t0-nonconv}, the fields $\widetilde{q}^{[k]}(x,t)$ of the semiclassical soliton ensemble 
evidently do \emph{not} converge to $q^{[k]}(x,t)$ at $t=0$.  
\begin{figure}[H]
\begin{center}
\includegraphics[height=2.9in]{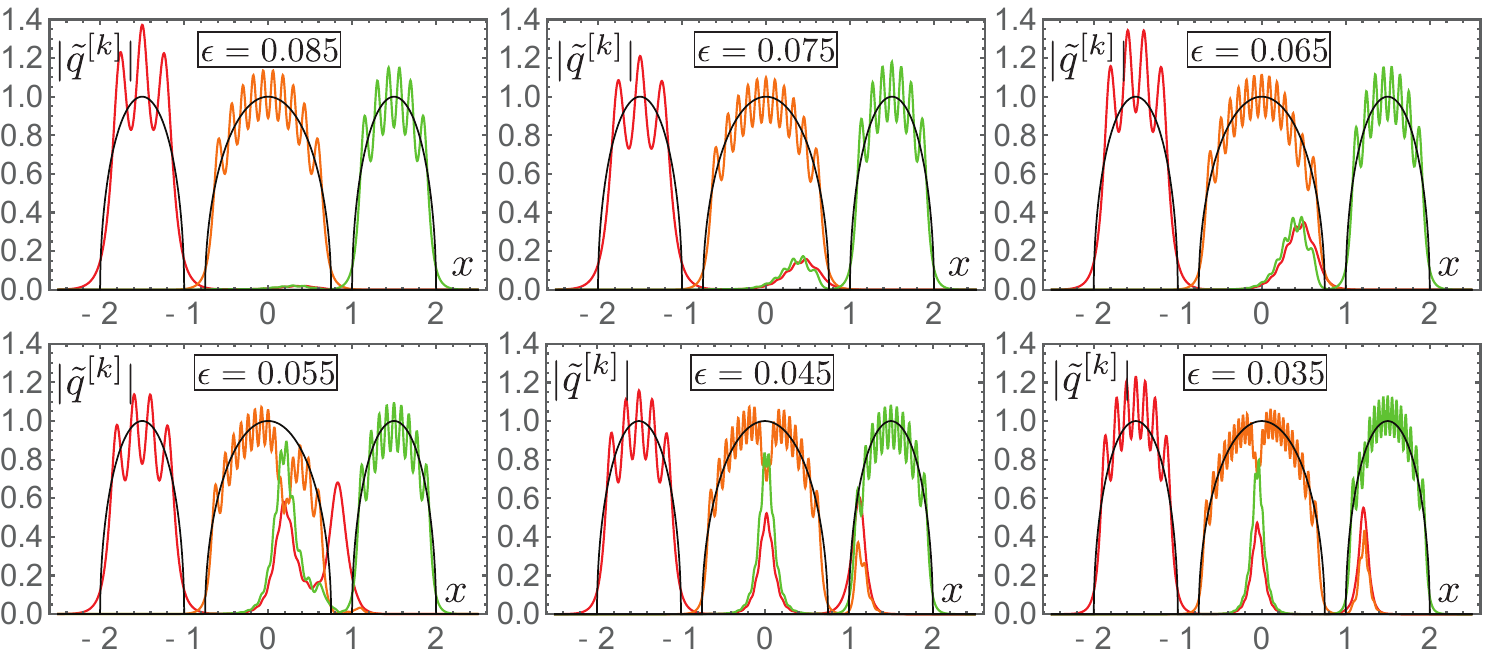}
\caption{Failure of convergence of the semiclassical soliton ensemble at $t=0$ with three packets.  Plots of 
$|q^{[k]}(x,0)|=H^{[k]}(x)$ as defined by 
\eqref{eq:initial-data} and \eqref{plot-Hk-sc123}
with parameters \eqref{plot-vals3}, along with 
the modulus of the fields $\widetilde{q}^{[k]}(x,0)$ of the associated semiclassical soliton ensemble 
for $k=1,2,3$ and various values of $\epsilon$.  
Black:  $H^{[k]}(x)$, $k=1,2,3$.  Red:  $|\widetilde{q}^{[1]}(x,0)|$.  Orange:  
$|\widetilde{q}^{[2]}(x,0)|$.  Green:  $|\widetilde{q}^{[3]}(x,0)|$.}
\label{fig-sc123-1d-t0-nonconv}
\end{center}
\end{figure}
As suggested in \S\ref{section:semiclassical-soliton-ensemble}, the lack of convergence may be traced to the inaccuracy of the approximation $\varphi\approx 0$, i.e., in this case one should include solitons of type $\solsplit$ in place of some of those of type $2$.

\appendix

\section{A Derivation of the TWRI Equations}
\label{appA-derivation}
Consider the semilinear partial differential equation 
\eq
u_{TT}+\Omega(-i\partial_X)^2 u = N(u),
\endeq
where $\Omega(k)^2$ is a real, even, nonnegative polynomial with $\Omega(0)^2>0$ 
and $N$ is twice-differentiable function with $N(0)=N'(0)=0$ encoding 
fully nonlinear terms.  Letting $\delta>0$ be a small parameter, the nonlinearity 
becomes weak upon setting $u:=\delta U$:
\eq
\label{semilinear-small-amp}
U_{TT}+\Omega(-i\partial_X)^2U = \frac{1}{2}\delta N''(0)U^2 + \mathcal{O}(\delta^2).
\endeq
When $\delta=0$, this equation has solutions $U(X,T)=Ae^{i(kX-\omega T)}+\mathrm{c.c.}$, 
where $A\in\mathbb{C}$, $k,\omega\in\mathbb{R}$, $\omega^2=\Omega(k)^2$, and $\mathrm{c.c.}$ 
stands for complex conjugate of the preceding expression.  Now assume that a 
resonant triad exists, i.e. wave numbers $k_1,k_2,k_3$ satisfying 
\eq
\label{resonant-condition}
k_1+k_2+k_3=0 \quad \text{and} \quad \omega_1+\omega_2+\omega_3=0, \quad \text{where}\quad \omega_j^2:=\Omega(k_j)^2.
\endeq
To observe the effect of the nonlinearity, assume a slow dependence of the 
amplitudes of these three modes on $X$ and $T$, where we set 
$\chi:=\delta X$ and $\tau:=\delta T$:
\eq
\label{slow-amp-ansatz}
U(X,T) = \sum_{j=1}^3 A_j(\chi,\tau)e^{i(k_j X - \omega_j T)} + \mathrm{c.c.} + \delta E(X,T),
\endeq
where $E(X,T)$ is a lower-order error term.  Substituting \eqref{slow-amp-ansatz} 
into \eqref{semilinear-small-amp} and collecting terms proportional to $\delta$ 
yields the following equation for the error:
\eq
\begin{split}
E_{TT} + \Omega(-i\partial_X)^2E = & \sum_{j=1}^3\left[\frac{N^{\prime\prime}(0)}{2}A_\ell^*A_m^* + 2i\omega_j A_{j\tau} + 2i \omega_j\Omega'(k_j)A_{j\chi}\right]e^{i(k_j X-\omega_j T)} + \mathrm{c.c.} \\
 & \text{+ nonresonant harmonics }+\mathcal{O}(\delta), \quad \text{$j,\ell,m$ distinct}.
\end{split}
\endeq
Here we have repeatedly used \eqref{resonant-condition}, and nonresonant harmonics 
refers to all other exponential terms, none of which (in general) satisfy the linear dispersion relation $\omega^2=\Omega(k)^2$.  Note that the external forcing frequency in 
each of the terms proportional to $\exp(ik_j X-i\omega_j T)$ matches an internal 
system frequency for $j=1,2,3$.  The resulting resonance leads to linear growth 
of the error, and after $\mathcal{O}(\delta^{-1})$ time the error $E$ becomes 
large.  To prevent this, the amplitudes should satisfy 
\eq
\begin{split}
A_{1\tau} + c^{[1]}A_{1\chi} & = -\frac{N^{\prime\prime}(0)}{4i\omega_1}A_2^*A_3^*, \\
A_{2\tau} + c^{[2]}A_{2\chi} & = -\frac{N^{\prime\prime}(0)}{4i\omega_2}A_3^*A_1^*, \\
A_{3\tau} + c^{[3]}A_{3\chi} & = -\frac{N^{\prime\prime}(0)}{4i\omega_3}A_1^*A_2^*,
\end{split}
\label{proto-3wave}
\endeq
where $c^{[j]}:=\Omega^\prime(k_j)$.  Indeed, \eqref{proto-3wave} removes the terms that cause the undesired resonance effect.  We obtain our starting form \eqref{3wave} of 
the TWRI equations by setting 
\eq
x:=\epsilon\chi, \quad t:=\epsilon\tau, \quad q^{[j]}(x,t):=\frac{i|N^{\prime\prime}(0)|\sqrt{|\omega_j|}}{4\sqrt{|\omega_1\omega_2\omega_3|}}A_j\left(\frac{x}{\epsilon},\frac{t}{\epsilon}\right), \quad \text{and} \quad \gamma^{[j]}:=\text{sgn}(N^{\prime\prime}(0)\omega_j).
\endeq
Observe that the resonance condition \eqref{resonant-condition} implies that this system is automatically in the decay instability case.  In deriving the TWRI system from other simple contexts it is possible to arrive at a system exhibiting explosive instability; for example if the nonlinear term $N(u)$ is replaced with $N(u_x)$, then the signs become $\gamma^{[j]}:=-\mathrm{sgn}(N''(0)\omega_jk_\ell k_m)$, for $j,\ell,m$ all distinct.  Again taking into account \eqref{resonant-condition}, this yields explosive instability ($\gamma^{[1]}=\gamma^{[2]}=\gamma^{[3]}$) provided the phase velocities $\omega_j/k_j$ are all positive.

\section{Inverse Scattering Transform for the TWRI Equations}
\label{appB-inverse-scattering}
Here we present a primarily self-contained review of the inverse-scattering 
transform theory necessary for the construction of semiclassical soliton 
ensembles.  Starting from the basic framework constructed by Zakharov and 
Manakov 
\cite{Zakharov:1973,Zakharov:1975} and Kaup \cite{Kaup76}, we review the 
general approaches developed by Beals and Coifman \cite{BealsC84,BealsC87} 
and Zhou \cite{Zhou89,Zhou89b} as applied to the TWRI equations.  In 
Appendices \ref{subapp:direct-scattering1}--\ref{subapp:time-dependence} we 
set up the direct-scattering transformation mapping the initial data to the 
scattering data.  In Appendices \ref{appB5-inverse-scattering} and 
\ref{subapp-Zhou-transform} we formulate the inverse map as a Riemann-Hilbert 
problem.  Finally, as our focus is on pure-soliton solutions, in Appendix 
\ref{sec:solitons} we detail the construction of these functions from the 
Riemann-Hilbert problem and discuss some of their relevant properties.
\subsection{Formulation of an integral equation for direct scattering}
\label{subapp:direct-scattering1}
We develop the direct scattering theory for the equation \eqref{eq:LaxPair-x} by expanding upon a general approach to $N\times N$ first-order systems given by Beals and Coifman in \cite{BealsC87}.  Let $t$ be fixed and write $q^{[k]}(x,t)=q^{[k]}(x)$ for $k=1,2,3$, defining a matrix $\mathbf{Q}(x)$ by \eqref{CQ}--\eqref{Q normalization}.  Suppose that $q^{[k]}(\cdot)\in L^1(\mathbb{R})$ for $k=1,2,3$.  Let $\sigma:=\pm$ be a fixed sign.  We seek a fundamental matrix solution $\Phi=\Phi^\sigma(x;\lambda)$ of \eqref{eq:LaxPair-x} defined for $\imag\{\lambda\}\neq 0$ satisfying the conditions:
\begin{equation}
\lim_{x\to\sigma\infty}\mathbf{M}^\sigma(x;\lambda)=\mathbb{I}\quad\text{and}\quad
\sup_{x\in\mathbb{R}}\|\mathbf{M}^\sigma(x;\lambda)\|<\infty,\quad \mathbf{M}^\sigma(x;\lambda):=
\Phi^\sigma(x;\lambda)e^{i\lambda\mathbf{C}x/\epsilon}.
\label{eq:RHP-M-direct-conditions}
\end{equation}
The substitution $\Phi=\mathbf{M}^\sigma e^{-i\lambda\mathbf{C}x/\epsilon}$ into \eqref{eq:LaxPair-x} results in a coupled system of differential equations on the matrix elements $M^\sigma_{jk}$ of $\mathbf{M}^\sigma$ as follows:
\begin{equation}
\epsilon \frac{\dd M^\sigma_{jk}}{\dd x}=-i\lambda (c^{[j]}-c^{[k]})M^\sigma_{jk} - (\mathbf{Q}\mathbf{M}^\sigma)_{jk}\quad j,k=1,2,3,
\label{eq:Mij-system}
\end{equation}
or equivalently via the introduction of an exponential integrating factor,
\begin{equation}
\epsilon \frac{\dd}{\dd x}\left(e^{i\lambda (c^{[j]}-c^{[k]})x/\epsilon}M^\sigma_{jk}\right)=-e^{i\lambda (c^{[j]}-c^{[k]})x/\epsilon}(\mathbf{Q}\mathbf{M}^\sigma)_{jk},\quad j,k=1,2,3.
\label{eq:integrating-factor}
\end{equation}
According to \eqref{eq:RHP-M-direct-conditions} and the inequalities \eqref{eq:c-inequalities}--\eqref{eq:epsilon-inequality}, the value of $e^{i\lambda (c^{[j]}-c^{[k]}) x/\epsilon}M_{jk}^\sigma(x;\lambda)$ is known in either the limit $x\to-\infty$ or $x\to +\infty$, depending on $\sigma$, the sign of $\imag\{\lambda\}$, and the subscripts $j,k$.  Indeed, for the diagonal elements $k=j$,
\begin{equation}
\lim_{x\to \sigma\infty} e^{i\lambda(c^{[j]}-c^{[k]})x/\epsilon}M_{jk}^\sigma(x;\lambda)=1,\quad k=j=1,2,3
\end{equation}
holds regardless of the sign of $\imag\{\lambda\}$, while otherwise using only the boundedness condition on the elements of $\mathbf{M}^\sigma$ gives
\begin{equation}
\lim_{x\to \sigma_{jk}(\lambda)\infty} e^{i\lambda(c^{[j]}-c^{[k]})x/\epsilon}M^\sigma_{jk}(x;\lambda)=0,\quad
\sigma_{jk}(\lambda):=\mathrm{sgn}(\imag\{\lambda\})\,\mathrm{sgn}(k-j),\quad k\neq j.
\label{eq:sigma-j-k-lambda}
\end{equation}
If we set
\begin{equation}
s_{jk}^\sigma(\lambda):=\begin{cases}\sigma,&\quad k=j=1,2,3,\\
\sigma_{jk}(\lambda),&\quad k\neq j,
\end{cases}
\label{eq:signs}
\end{equation}
then we may integrate \eqref{eq:integrating-factor} incorporating the limiting values to obtain
\begin{equation}
M_{jk}^\sigma(x;\lambda)=\delta_{jk}-\frac{1}{\epsilon}\int_{s_{jk}^\sigma(\lambda)\infty}^x
e^{-i\lambda (c^{[j]}-c^{[k]})(x-y)/\epsilon}(\mathbf{Q}(y)\mathbf{M}^\sigma(y;\lambda))_{jk}\,\dd y,\quad
\imag\{\lambda\}\neq 0.
\label{eq:Fredholm-system}
\end{equation}
The system of integral equations \eqref{eq:Fredholm-system} posed on $L^\infty(\mathbb{R})$ is therefore a necessary condition for $\mathbf{M}^\sigma(x;\lambda)$ to satisfy both the differential equations \eqref{eq:Mij-system} and the auxiliary conditions \eqref{eq:RHP-M-direct-conditions}. To show that \eqref{eq:Fredholm-system} is also sufficient, it remains to establish that $\mathbf{M}^\sigma(x;\lambda)\to\mathbb{I}$ as $x\to\sigma\infty$.  It is rather obvious that this holds for the diagonal elements, as well as for those off-diagonal elements for which $s_{jk}^\sigma=\sigma$ holds.  For the remaining off-diagonal elements, we appeal to a dominated convergence argument using $\mathbf{M}^\sigma(\cdot;\lambda)\in L^\infty(\mathbb{R})$ and $\mathbf{Q}(\cdot)\in L^1(\mathbb{R})$ along with the fact that the exponential factor $e^{-i\lambda (c^{[j]}-c^{[k]})(x-y)/\epsilon}$ is bounded in magnitude by $1$ and decays to zero for each fixed $y\in\mathbb{R}$ as $x\to\sigma\infty$.  
 
The system of integral equations \eqref{eq:Fredholm-system} is formally of Fredholm type and not Volterra type, because the limits of integration are necessarily different for different matrix elements.  In particular, this implies that there may not be a solution for every $\lambda\in\mathbb{C}\setminus\mathbb{R}$.  The integral operator in \eqref{eq:Fredholm-system} is not compact on $L^\infty(\mathbb{R})$, so the analytic Fredholm theorem does not immediately apply to characterize the dependence of $\mathbf{M}^\sigma(x;\lambda)$ on the spectral parameter $\lambda$.  In \cite{BealsC84} it was shown that $\mathbf{M}^\sigma(x;\lambda)$ is meromorphic separately in the domains $\imag\{\lambda\}>0$ and $\imag\{\lambda\}<0$, with isolated singularities that are poles of finite order at locations independent of $x\in\mathbb{R}$ (and these poles are the only values of $\lambda$ for which a solution of \eqref{eq:Fredholm-system} fails to exist).  These results were shown by an induction argument on the $L^1(\mathbb{R})$ norm of the potential $\mathbf{Q}$, in which at each step the norm is reduced by a factor of $1/2$ until it is below a threshold at which the integral equation \eqref{eq:Fredholm-system} can be solved by Neumann series.  Each step of the induction introduces an analytic function in the denominator that contributes possible poles.  In \cite{BealsC87} an alternate approach was developed, a method based on Volterra equations whose solutions are automatically analytic in the appropriate half-planes.   We will describe this latter approach to the theory of \eqref{eq:Fredholm-system} in Appendix~\ref{section:Volterra} below.  

Assuming that $\mathbf{M}^\sigma(x;\lambda)$ solving \eqref{eq:Fredholm-system} exists in $L^\infty(\mathbb{R})$, it follows from a dominated convergence argument similar to the one above proving $\mathbf{M}^\sigma(x;\lambda)\to\mathbb{I}$ as $x\to\sigma\infty$ that  $\mathbf{M}^\sigma(x;\lambda)$ has a limit as $x\to -\sigma\infty$.  Moreover this limit is a diagonal matrix denoted $\mathbf{D}(\lambda)^\sigma$, i.e., $\mathbf{D}(\lambda)$ for $\sigma=+$ or $\mathbf{D}(\lambda)^{-1}$ for $\sigma=-$, and the elements of $\mathbf{D}(\lambda)$ 
are meromorphic functions of $\lambda\in\mathbb{C}\setminus\mathbb{R}$. 
Multiplication of \eqref{eq:Fredholm-system} on the right by either $\mathbf{D}(\lambda)$ or its inverse reveals the key relationship \eqref{eq:Mplus-Mminus-intro} between $\mathbf{M}^+(x;\lambda)$ and $\mathbf{M}^-(x;\lambda)$.

Integration by parts shows that the system \eqref{eq:Fredholm-system} is uniquely solvable if $|\lambda|$ is sufficiently large, and that $\mathbf{M}^\sigma(x;\lambda)\to\mathbb{I}$ as $\lambda\to\infty$ in $\mathbb{C}\setminus\mathbb{R}$.  In more detail, one observes that when $j=k$, the system \eqref{eq:Fredholm-system} explicitly expresses $M_{jj}^\sigma(x;\lambda)$ in terms of the off-diagonal elements of $\mathbf{M}^\sigma(x;\lambda)$ only, because $\mathbf{Q}$ is off-diagonal.  Using this fact to eliminate the diagonal elements of $\mathbf{M}^\sigma(x;\lambda)$ from the remaining equations in \eqref{eq:Fredholm-system} gives a closed linear system of (double) integral equations on the off-diagonal matrix elements of $\mathbf{M}^\sigma(x;\lambda)$.  Integrating by parts under the additional assumption that $\mathbf{Q}'(\cdot)\in L^1(\mathbb{R})$ shows that the operator for this system has a norm on $L^\infty(\mathbb{R})$ that is proportional to $|\lambda|^{-1}$, and that the forcing term is a function in $L^\infty(\mathbb{R})$ with norm also proportional to $|\lambda|^{-1}$.  Hence the off-diagonal elements can be obtained uniquely in this space by convergent Neumann series and have norms proportional to $|\lambda|^{-1}$.  Using this result in \eqref{eq:Fredholm-system} considered for $j=k$ then shows that $M^\sigma_{jj}(\cdot;\lambda)$ converges uniformly to $1$ as $|\lambda|\to\infty$.  (The same conclusion holds true without the need to integrate by parts and with only the assumption that $\mathbf{Q}(\cdot)\in L^1(\mathbb{R})$ provided one takes the limit $\lambda\to\infty$ in the sense that $\imag\{\lambda\}\to\infty$ as well.) For potentials $\mathbf{Q}$ with $\mathbf{Q}$ and $\mathbf{Q}'$ in $L^1(\mathbb{R})$, the singularities of $\mathbf{M}^\sigma(x;\lambda)$ therefore form a bounded set in $\mathbb{C}\setminus\mathbb{R}$.  For a generic subset of off-diagonal potentials $\mathbf{Q}$ in the Schwartz space, the singularities of $\mathbf{M}^\sigma(x;\lambda)$ form a \emph{finite} subset of $\mathbb{C}\setminus\mathbb{R}$, and each is a simple pole.  It also follows from Abel's theorem that $\det(\mathbf{M}^\sigma(x;\lambda))$ is independent of $x$, and then from the conditions \eqref{eq:RHP-M-direct-conditions} that $\det(\mathbf{M}^\sigma(x;\lambda))=1$ holds  for each non-exceptional $\lambda\in\mathbb{C}\setminus\mathbb{R}$ and for each $x\in\mathbb{R}$.  From \eqref{eq:Mplus-Mminus-intro} it then follows that also that the diagonal matrix $\mathbf{D}(\lambda)$ satisfies $\det(\mathbf{D}(\lambda))=1$ for all such $\lambda$.

Finally, we note that if $\mathbf{M}^\sigma(x;\lambda)$ admits an asymptotic representation for large $\lambda$ of the form
\begin{equation}
\mathbf{M}^\sigma(x;\lambda)=\mathbb{I} + \mathbf{F}^\sigma(x)\lambda^{-1} + o(\lambda^{-1}),\quad\lambda\to\infty,
\label{eq:M-lambda-infinity-expand}
\end{equation}
such that also
\begin{equation}
\frac{\dd\mathbf{M}^\sigma}{\dd x}(x;\lambda) = \frac{\dd\mathbf{F}^\sigma}{\dd x}(x)\lambda^{-1} + o(\lambda^{-1}),\quad\lambda\to\infty,
\end{equation}
then it follows by taking the limit $\lambda\to\infty$ in the differential equation \eqref{eq:Mij-system}
that
\begin{equation}
\mathbf{Q}(x)=-i[\mathbf{C},\mathbf{F}^\sigma(x)]
=\begin{pmatrix}
0 & -i\Delta^{[3]}F^\sigma_{12}(x) & -i\Delta^{[2]}F^\sigma_{13}(x)\\
i\Delta^{[3]}F^\sigma_{21}(x) & 0 & -i\Delta^{[1]}F^\sigma_{23}(x)\\
i\Delta^{[2]}F^\sigma_{31}(x) & i\Delta^{[1]}F^\sigma_{32}(x) & 0\end{pmatrix}.
\end{equation}
This identity shows also that the matrices $\mathbf{F}^\sigma(x)$ for $\sigma=\pm$ have the same off-diagonal elements.
Comparing with \eqref{CQ} and \eqref{Q normalization} yields the formula for reconstructing the potentials $q^{[k]}(x)$ from $\mathbf{M}^\sigma$:
\begin{equation}
\begin{split}
q^{[1]}(x)&=\gamma^{[1]}\sqrt{\Delta^{[2]}\Delta^{[3]}}Q_{23}(x)=-i\gamma^{[1]} \Delta^{[1]}\sqrt{\Delta^{[2]}\Delta^{[3]}}F^\sigma_{23}(x),\\
q^{[2]}(x)&=\gamma^{[2]}\sqrt{\Delta^{[3]}\Delta^{[1]}}Q_{31}(x)=i\gamma^{[2]} \Delta^{[2]}\sqrt{\Delta^{[1]}\Delta^{[3]}}F^\sigma_{31}(x),\\
q^{[3]}(x)&=\gamma^{[3]}\sqrt{\Delta^{[1]}\Delta^{[2]}}Q_{12}(x)=-i\gamma^{[3]} \Delta^{[3]}\sqrt{\Delta^{[1]}\Delta^{[2]}}F^\sigma_{12}(x).
\end{split}
\label{eq:q-potentials-reconstruct}
\end{equation}

\subsection{Analysis of the direct-scattering integral equation}
\label{section:Volterra}
We now construct $\mathbf{M}^\sigma(x;\lambda)$ for $\imag\{\lambda\}\neq 0$ 
from analytic functions solving Volterra equations.  
Let $\mathbf{m}^{\sigma,k}(x;\lambda)$ denote the $k^\text{th}$ column of $\mathbf{M}^\sigma(x;\lambda)$.  We will now show that the matrix $\mathbf{M}^\sigma(x;\lambda)$ can be constructed from $\mathbf{m}^{\sigma,1}(x;\lambda)$, $\mathbf{m}^{\sigma,3}(x;\lambda)$, and from the wedge products $\mathbf{m}^{\sigma,1}(x;\lambda)\wedge\mathbf{m}^{\sigma,2}(x;\lambda)$
and $\mathbf{m}^{\sigma,2}(x;\lambda)\wedge\mathbf{m}^{\sigma,3}(x;\lambda)$, and that these quantities are proportional via meromorphic factors to solutions of Volterra integral equations, the latter automatically being analytic for $\imag\{\lambda\}\neq 0$.  As a natural basis of $\mathbb{C}^3$ we take the standard unit vectors $\mathbf{e}^1$, $\mathbf{e}^2$, and $\mathbf{e}^3$.  For the wedge product space $(\mathbb{C}^3)^{\wedge 2}$ we take as a basis $\mathbf{f}^1:=\mathbf{e}^2\wedge\mathbf{e}^3$, $\mathbf{f}^2:=\mathbf{e}^3\wedge\mathbf{e}^1$, and $\mathbf{f}^3:=\mathbf{e}^1\wedge\mathbf{e}^2$.  

We remark that this construction is actually fairly well-known in the literature concerned with stability of waves using Evans function methods \cite{LedouxMT10}.  The Evans function is, by definition, the Wronskian determinant of a subspace of solutions decaying as $x\to -\infty$ with a subspace of solutions decaying as $x\to +\infty$.  Its zeros are therefore exactly the eigenvalues of the problem.  The numerical computation of the Evans function involves the calculation of bases of these subspaces by solving initial-value problems with boundary conditions at $\pm\infty$ and then ``shooting'' toward a common point, say $x=0$, at which the solutions are compared and the Wronskian calculated.  It is well known that if the subspace of solutions decaying as $x\to +\infty$, say, contains solutions with two different asymptotic exponential decay rates, then it is numerically very difficult if not impossible to calculate accurately the solution with the smaller decay rate, because any numerical noise introduced while integrating in from $+\infty$ will contaminate the solution with a bit of the faster decaying (and hence faster growing as $x$ decreases) solution, which will then overtake it.  Hence the numerical calculation of the Wronskian in such situations nearly always gives zero, and it has been a standard technique in this situation to pass to the exterior algebra, where the whole subspace of decaying solutions is represented as a single wedge product that satisfies its own induced differential equation with boundary condition at $x=+\infty$.  Fortunately, the latter corresponds to a Volterra integral equation.  This implies that the numerical solution of the associated initial-value problem with initial condition at infinity is well-conditioned.

\subsubsection{Algebraic construction of $\mathbf{M}^\sigma$ from exterior products}
First, we show that $\mathbf{M}^\sigma(x;\lambda)$ may be explicitly constructed from the specified data.  This is an algebraic fact.  Indeed, we recall that $\det(\mathbf{M}^\sigma)=1$ and suppose that in addition $\mathbf{m}^{\sigma,1}$, $\mathbf{m}^{\sigma,3}$, and the wedge products 
$\mathbf{m}^{\sigma,1}\wedge\mathbf{m}^{\sigma,2}$ and $\mathbf{m}^{\sigma,2}\wedge\mathbf{m}^{\sigma,3}$ are given in terms of their components as:
\begin{equation}
\begin{split}
\mathbf{m}^{\sigma,1}&=m^1_1\mathbf{e}^1 +m^1_2\mathbf{e}^2+m^1_3\mathbf{e}^3,\\
\mathbf{m}^{\sigma,3}&=m^3_1\mathbf{e}^1 +m^3_2\mathbf{e}^2+m^3_3\mathbf{e}^3,
\end{split}
\end{equation}
and
\begin{equation}
\begin{split}
\mathbf{m}^{\sigma,1}\wedge\mathbf{m}^{\sigma,2}&=n^{3}_1\mathbf{f}^1 +n^{3}_2\mathbf{f}^2 +n^{3}_3\mathbf{f}^3,\\
\mathbf{m}^{\sigma,2}\wedge\mathbf{m}^{\sigma,3}&=n^{1}_1\mathbf{f}^1+n^{1}_2\mathbf{f}^2+n^{1}_3\mathbf{f}^3.
\end{split}
\label{eq:wedges-components}
\end{equation}
Obviously it only remains to show that the central column $\mathbf{m}^{\sigma,2}=m^2_1\mathbf{e}^1+m^2_2\mathbf{e}^2+m^2_3\mathbf{e}^3$ can be constructed explicitly from this data, and it suffices to use \eqref{eq:wedges-components} and $\det(\mathbf{M}^\sigma)=1$.  Indeed,
\begin{equation}
\begin{split}
n^{3}_2n^{1}_3-n^{3}_3n^{1}_2 &= (m^1_3m^2_1-m^1_1m^2_3)(m^2_1m^3_2-m^2_2m^3_1)-(m^1_1m^2_2-m^1_2m_1^2)(m^2_3m^3_1-m^2_1m^3_3)\\
&=m_1^2\,\mathrm{det}(\mathbf{M}^\sigma)\\
&=m_1^2,
\\
n^{3}_3n^{1}_1-n^{3}_1n^{1}_3&=
(m_1^1m_2^2-m_2^1m_1^2)(m_2^2m_3^3-m_3^2m_2^3)-
(m_2^1m_3^2-m_3^1m_2^2)(m_1^2m_2^3-m_2^2m_1^3)\\
&=m_2^2\,\mathrm{det}(\mathbf{M}^\sigma)\\&=m_2^2,
\\
n^{3}_1n^{1}_2-n^{3}_2n^{1}_1&=
(m^1_2m^2_3-m^1_3m^2_2)(m^2_3m^3_1-m^2_1m^3_3)-
(m^1_3m^2_1-m^1_1m^2_3)(m^2_2m^3_3-m^2_3m^3_2)\\
&=m^2_3\,\mathrm{det}(\mathbf{M}^\sigma)\\
&=m^2_3.
\end{split}
\end{equation}
In other words, if we identify $(\mathbb{C}^3)^{\wedge 2}$ with $\mathbb{C}^3$ in the usual way by associating $\mathbf{f}^j$ with $\mathbf{e}^j$, $j=1,2,3$, then the vector $\mathbf{m}^{\sigma,2}$ can be recovered explicitly from the cross-products 
\begin{equation}
\begin{split}
\mathbf{n}^{\sigma,3}(x;\lambda)&:=\mathbf{m}^{\sigma,1}(x;\lambda)\times\mathbf{m}^{\sigma,2}(x;\lambda),\\
\mathbf{n}^{\sigma,1}(x;\lambda)&:=\mathbf{m}^{\sigma,2}(x;\lambda)\times\mathbf{m}^{\sigma,3}(x;\lambda)
\end{split}
\label{eq:cross-product-definition}
\end{equation}
by the formula
\begin{equation}
\mathbf{m}^{\sigma,2}(x;\lambda)=\mathbf{n}^{\sigma,3}(x;\lambda)\times
\mathbf{n}^{\sigma,1}(x;\lambda).
\label{eq:m2-cross-product-formula}
\end{equation}

\subsubsection{Representation of $\mathbf{M}^\sigma(x;\lambda)$ in terms of solutions of Volterra equations}
The next observation we may make in introducing the wedge product space is that the system of differential equations \eqref{eq:Mij-system} imply a corresponding linear system of differential equations for the components of the wedge products of columns of $\mathbf{M}^\sigma$.  Indeed,
\eqref{eq:Mij-system} can be rewritten as
\begin{equation}
\epsilon\frac{\dd \mathbf{m}^{\sigma,j}}{\dd x}=-i\lambda(\mathbf{C}-c^{[j]}\mathbb{I})\mathbf{m}^{\sigma,j}-\mathbf{Q}\mathbf{m}^{\sigma,j},\quad j=1,2,3.
\label{eq:m-columns-ODE}
\end{equation}
Therefore,
\begin{multline}
\epsilon\frac{\dd}{\dd x}(\mathbf{m}^{\sigma,1}\wedge\mathbf{m}^{\sigma,2})=
-i\lambda\left([(\mathbf{C}-c^{[1]}\mathbb{I})\mathbf{m}^{\sigma,1}]\wedge\mathbf{m}^{\sigma,2}
+\mathbf{m}^{\sigma,1}\wedge [(\mathbf{C}-c^{[2]}\mathbb{I})\mathbf{m}^{\sigma,2}]\right)\\
{}-\left([\mathbf{Q}\mathbf{m}^{\sigma,1}]\wedge\mathbf{m}^{\sigma,2}+\mathbf{m}^{\sigma,1}\wedge[\mathbf{Q}\mathbf{m}^{\sigma,2}]\right).
\end{multline}
In the basis $\{\mathbf{f}^j\}_{j=1}^3$, the components of $\mathbf{m}^{\sigma,1}\wedge\mathbf{m}^{\sigma,2}$ are $\{n^{3}_j\}_{j=1}^3$ (which in turn form a basis of the antisymmetric quadratic forms of the elements of $\mathbf{m}^{\sigma,1}$ and $\mathbf{m}^{\sigma,2}$), and because the components of the right-hand side are also such antisymmetric quadratic forms, they are necessarily linear in $\{n^{3}_j\}_{j=1}^3$.  This implies that, working in the basis $\{\mathbf{f}^j\}_{j=1}^3$, there is a $3\times 3$ matrix $\mathbf{H}^{3}$ such that
\begin{equation}
\epsilon\frac{\dd\mathbf{n}^{\sigma,3}}{\dd x}=\mathbf{H}^{3}\mathbf{n}^{\sigma,3}.
\label{eq:first-wedge-1}
\end{equation}
A computation shows that
\begin{equation}
\mathbf{H}^{3}=i\lambda(\mathbf{C}-c^{[3]}\mathbb{I}) +\begin{pmatrix}-Q_{22}-Q_{33} &Q_{21} & Q_{31}\\Q_{12} & -Q_{11}-Q_{33} & Q_{32}\\
Q_{13} & Q_{23} & -Q_{11}-Q_{22}\end{pmatrix}=i\lambda(\mathbf{C}-c^{[3]}\mathbb{I})+\mathbf{Q}^\trans,
\label{eq:first-wedge-2}
\end{equation}
with the second equality following because $\mathbf{Q}$ is off-diagonal.  Similarly,
\begin{equation}
\epsilon\frac{\dd\mathbf{n}^{\sigma,1}}{\dd x}=\mathbf{H}^{1}\mathbf{n}^{\sigma,1},\quad
\label{eq:second-wedge}
\mathbf{H}^{1}=i\lambda(\mathbf{C}-c^{[1]}\mathbb{I})+\mathbf{Q}^\trans.
\end{equation}
An important observation following from \eqref{eq:m-columns-ODE}, \eqref{eq:first-wedge-1}, and \eqref{eq:second-wedge} and uniqueness for the corresponding initial-value problems is that for all non-exceptional $\lambda\in\mathbb{C}\setminus\mathbb{R}$, each vector of $\mathbf{m}^{\sigma,j}(x;\lambda)$, $j=1,2,3$, $\mathbf{n}^{\sigma,3}(x;\lambda)$, and $\mathbf{n}^{\sigma,1}(x;\lambda)$ either vanishes for all $x\in\mathbb{R}$ or for no $x\in\mathbb{R}$.

Now, taking into account \eqref{eq:signs}, fixing $k=1$ in \eqref{eq:Fredholm-system} yields a Volterra integral equation with lower integration limit $\sigma\infty$ for the first column $\mathbf{m}^{\sigma,1}(x;\lambda)$ provided that $\mathrm{sgn}(\imag\{\lambda\})=-\sigma$.  Similarly, taking $k=3$ in \eqref{eq:Fredholm-system} yields a Volterra integral equation with lower integration limit $\sigma\infty$ for the third column $\mathbf{m}^{\sigma,3}(x;\lambda)$ provided that $\mathrm{sgn}(\imag\{\lambda\})=\sigma$.  These columns are therefore analytic functions in the indicated half-planes.  This fact suggests writing the diagonal matrix $\mathbf{D}(\lambda)$ with $\det(\mathbf{D}(\lambda))=1$ in the form \eqref{eq:D-representation}
with the functions $u$ and $v$ being analytic in the domain $\mathbb{C}\setminus\mathbb{R}$.  Indeed, taking into account the fact that $\mathbf{M}^\sigma(x;\lambda)\to\mathbb{I}$ as $x\to\sigma\infty$ and the relation \eqref{eq:Mplus-Mminus-intro} we obtain the formulas
\begin{equation}
u(\lambda)=\begin{cases}\displaystyle\lim_{x\to +\infty} M^-_{11}(x;\lambda),&\imag\{\lambda\}>0,\\
\displaystyle\lim_{x\to -\infty}M_{11}^+(x;\lambda),&\imag\{\lambda\}<0
\end{cases}\quad
\text{and}\quad
v(\lambda)=\begin{cases}\displaystyle\lim_{x\to -\infty}M^+_{33}(x;\lambda),&\imag\{\lambda\}>0,\\
\displaystyle\lim_{x\to +\infty}M_{33}^-(x;\lambda),&\imag\{\lambda\}<0.
\end{cases}
\end{equation}

To obtain the columns $\mathbf{m}^{\sigma,1}(x;\lambda)$ and $\mathbf{m}^{\sigma,3}(x;\lambda)$ for $\lambda$ in the opposite half-planes, we simply use \eqref{eq:Mplus-Mminus-intro} to exchange $\sigma$ for $-\sigma$.  Thus:
\begin{itemize}
\item For $\imag\{\lambda\}>0$, we have 
\begin{equation}
\mathbf{m}^{+,1}(x;\lambda)=\frac{1}{u(\lambda)}\mathbf{m}^{-,1}(x;\lambda),
\label{eq:m1-plus-Im-lambda-positive}
\end{equation}
where $u(\lambda)$ is analytic and $\mathbf{m}^{-,1}(x;\lambda)$ is analytic and non-vanishing $\text{for all } x\in\mathbb{R}$ due to the normalization condition $\mathbf{m}^{-,1}\to \mathbf{e}^1$ as $x\to -\infty$. 
Similarly, 
\begin{equation}
\mathbf{m}^{-,3}(x;\lambda)=\frac{1}{v(\lambda)}\mathbf{m}^{+,3}(x;\lambda),
\label{eq:m3-minus-Im-lambda-positive}
\end{equation}
where $v(\lambda)$ is analytic and $\mathbf{m}^{+,3}(x;\lambda)$ is analytic and non-vanishing $\text{for all } x\in\mathbb{R}$ due to the normalization condition $\mathbf{m}^{+,3}\to \mathbf{e}^3$ as $x\to +\infty$.
\item For $\imag\{\lambda\}<0$, we have 
\begin{equation}
\mathbf{m}^{-,1}(x;\lambda)=\frac{1}{u(\lambda)}\mathbf{m}^{+,1}(x;\lambda),
\label{eq:m1-minus-Im-lambda-negative}
\end{equation}
where $u(\lambda)$ is analytic and $\mathbf{m}^{+,1}(x;\lambda)$ is analytic and non-vanishing $\text{for all } x\in\mathbb{R}$ due to the normalization condition $\mathbf{m}^{+,1}\to \mathbf{e}^1$ as $x\to +\infty$.
Similarly,
\begin{equation}
\mathbf{m}^{+,3}(x;\lambda)=\frac{1}{v(\lambda)}\mathbf{m}^{-,3}(x;\lambda),
\label{eq:m3-plus-Im-lambda-negative}
\end{equation}
where $v(\lambda)$ is analytic and $\mathbf{m}^{-,3}(x;\lambda)$ is analytic and non-vanishing $\text{for all } x\in\mathbb{R}$ due to the normalization condition $\mathbf{m}^{-,3}\to\mathbf{e}^3$ as $x\to -\infty$.
\end{itemize}
To obtain the central column $\mathbf{m}^{\sigma,2}(x;\lambda)$, we first consider the cross-products $\mathbf{n}^{\sigma,3}$ and $\mathbf{n}^{\sigma,1}$ defined by \eqref{eq:cross-product-definition},
which satisfy the differential equations \eqref{eq:first-wedge-1}--\eqref{eq:first-wedge-2} and \eqref{eq:second-wedge} respectively.  Since $\mathbf{C}-c^{[3]}\mathbb{I}$ and $\mathbf{C}-c^{[1]}\mathbb{I}$ are diagonal matrices with entries of fixed sign (non-negative and non-positive, respectively), the differential equations can be integrated up to Volterra equations assuming that the sign of $\imag\{\lambda\}$ is correlated with the infinite limit of integration.  Thus, with the use of appropriate exponential integrating factors, taking into account the boundary condition $\mathbf{M}^\sigma(x;\lambda)\to\mathbb{I}$ as $x\to\sigma\infty$, we find that
\begin{itemize}
\item For $\imag\{\lambda\}>0$, 
\begin{equation}
\mathbf{n}^{-,3}(x;\lambda)=\mathbf{e}^3 +\frac{1}{\epsilon}\int_{-\infty}^x
e^{i\lambda(\mathbf{C}-c^{[3]}\mathbb{I})(x-y)/\epsilon}\mathbf{Q}(y)^\trans \mathbf{n}^{-,3}(y;\lambda)\,\dd y
\label{eq:mminus-1cross2-Im-lambda-positive}
\end{equation}
and
\begin{equation}
\mathbf{n}^{+,1}(x;\lambda)=\mathbf{e}^1 +\frac{1}{\epsilon}\int_{+\infty}^x
e^{i\lambda(\mathbf{C}-c^{[1]}\mathbb{I})(x-y)/\epsilon}\mathbf{Q}(y)^\trans\mathbf{n}^{+,1}(y;\lambda)\,\dd y,
\label{eq:mplus-2cross3-Im-lambda-positive}
\end{equation}
defining these two cross-products as analytic functions of $\lambda$.  To obtain $\mathbf{n}^{+,3}$ and $\mathbf{n}^{-,1}$ in the same half-plane,
we use \eqref{eq:Mplus-Mminus-intro} with $\mathbf{D}(\lambda)$ in the representation \eqref{eq:D-representation} and bilinearity of the cross-products in \eqref{eq:cross-product-definition} to obtain
\begin{equation}
\mathbf{n}^{+,3}(x;\lambda)=\frac{1}{v(\lambda)}
\mathbf{n}^{-,3}(x;\lambda),
\label{eq:mplus-1cross2-Im-lambda-positive}
\end{equation}
where $v(\lambda)$ is analytic and $\mathbf{n}^{-,3}(x;\lambda)$ is analytic and non-vanishing $\text{for all } x\in\mathbb{R}$ due to the normalization condition $\mathbf{n}^{-,3}(x;\lambda)\to\mathbf{e}^3$ as $x\to -\infty$.  Similarly,
\begin{equation}
\mathbf{n}^{-,1}(x;\lambda)=\frac{1}{u(\lambda)}
\mathbf{n}^{+,1}(x;\lambda),
\label{eq:mminus-2cross3-Im-lambda-positive}
\end{equation}
where $u(\lambda)$ is analytic and $\mathbf{n}^{+,1}(x;\lambda)$ is analytic and non-vanishing $\text{for all } x\in\mathbb{R}$ due to the normalization condition $\mathbf{n}^{+,1}(x;\lambda)\to\mathbf{e}^1$ as $x\to +\infty$.
\item For $\imag\{\lambda\}<0$,
\begin{equation}
\mathbf{n}^{+,3}(x;\lambda)=\mathbf{e}^3 +\frac{1}{\epsilon}\int_{+\infty}^x
e^{i\lambda(\mathbf{C}-c^{[3]}\mathbb{I})(x-y)/\epsilon}\mathbf{Q}(y)^\trans\mathbf{n}^{+,3}(y;\lambda)\,\dd y
\label{eq:mplus-1cross2-Im-lambda-negative}
\end{equation}
and
\begin{equation}
\mathbf{n}^{-,1}(x;\lambda)=\mathbf{e}^1 +\frac{1}{\epsilon}\int_{-\infty}^x
e^{i\lambda(\mathbf{C}-c^{[1]}\mathbb{I})(x-y)/\epsilon}\mathbf{Q}(y)^\trans\mathbf{n}^{-,1}(y;\lambda)\,\dd y,
\label{eq:mminus-2cross3-Im-lambda-negative}
\end{equation}
defining these two cross-products as analytic functions of $\lambda$.  To obtain $\mathbf{n}^{-,3}$ and $\mathbf{n}^{+,1}$ in the same half-plane,
we use \eqref{eq:Mplus-Mminus-intro} with $\mathbf{D}(\lambda)$ in the representation \eqref{eq:D-representation} and bilinearity of the cross-products in \eqref{eq:cross-product-definition} to obtain
\begin{equation}
\mathbf{n}^{-,3}(x;\lambda)(x;\lambda)=\frac{1}{v(\lambda)}
\mathbf{n}^{+,3}(x;\lambda),
\label{eq:mminus-1cross2-Im-lambda-negative}
\end{equation}
where $v(\lambda)$ is analytic and $\mathbf{n}^{+,3}(x;\lambda)$ is analytic and non-vanishing $\text{for all } x\in\mathbb{R}$ due to the normalization condition $\mathbf{n}^{+,3}(x;\lambda)\to\mathbf{e}^3$ as $x\to +\infty$.  Similarly,
\begin{equation}
\mathbf{n}^{+,1}(x;\lambda)=\frac{1}{u(\lambda)}
\mathbf{n}^{-,1}(x;\lambda),
\label{eq:mplus-2cross3-Im-lambda-negative}
\end{equation}
where $u(\lambda)$ is analytic and $\mathbf{n}^{-,1}(x;\lambda)$ is analytic and non-vanishing $\text{for all } x\in\mathbb{R}$ due to the normalization condition $\mathbf{n}^{-,1}(x;\lambda)\to\mathbf{e}^1$ as $x\to -\infty$.
\end{itemize}
Once $\mathbf{n}^{\sigma,3}(x;\lambda)$ and $\mathbf{n}^{\sigma,1}(x;\lambda)$ have been obtained in this way for $\imag\{\lambda\}\neq 0$, the column $\mathbf{m}^{\sigma,2}(x;\lambda)$ is recovered explicitly from these cross-products via \eqref{eq:m2-cross-product-formula}.

\subsection{Singularities of $\mathbf{M}^\sigma(x;\lambda)$ and associated scattering data}
\label{subappB3-singularities}
\subsubsection{Schwarz symmetry of $\mathbf{M}^\sigma(x;\lambda)$}
Recall $\mathbf{E}:=\mathrm{diag}(\gamma^{[1]},-\gamma^{[2]},\gamma^{[3]})$, and note that $\mathbf{E}=\mathbf{E}^{-1}$.  A simple calculation shows that the coefficient matrix $\mathcal{L}(x,t;\lambda):=-i\lambda\mathbf{C}-\mathbf{Q}(x,t)$ (see \eqref{x flow}) satisfies (for fixed real $x$ and $t$)
\begin{equation}
\mathcal{L}(x,t;\lambda) = -\mathbf{E}\mathcal{L}(x,t;\lambda^*)^\dagger\mathbf{E}
\end{equation}
(dagger denotes the conjugate transpose).  It then follows easily that $\Phi(x;\lambda)$ is an invertible matrix solution of \eqref{eq:LaxPair-x} (suppressing the $t$-dependence) if and only if
\begin{equation}
\Phi(x;\lambda^*):=\mathbf{E}\Phi(x;\lambda)^{-\dagger}\mathbf{E}
\end{equation}
is an invertible matrix solution of the same equation at the complex-conjugate value of the spectral parameter.  Here the superscript $-\dagger$ indicates both Hermitian conjugation and matrix inversion.  Multiplication on the right by $e^{i\lambda^*\mathbf{C}x/\epsilon}$ and taking into account the conditions \eqref{eq:RHP-M-direct-conditions} shows that $\mathbf{M}^\sigma(x;\lambda)$ is a solution of the direct scattering problem for some $\lambda$ with $\imag\{\lambda\}\neq 0$ if and only if $\mathbf{E}\mathbf{M}^\sigma(x;\lambda)^{-\dagger}\mathbf{E}$ is also a solution for $\lambda^*$.  Since (by the Volterra approach to direct scattering) $\mathbf{M}^\sigma(x;\lambda)$ is unique if it exists for some $\lambda$, it follows that the solution $\mathbf{M}^\sigma(x;\lambda)$ of \eqref{eq:Fredholm-system} has the Schwarz symmetry 
\eqref{eq:TWRI-Schwarz-symmetry}.
Therefore, knowledge of $\mathbf{M}^\sigma(x;\lambda)$ for $\imag\{\lambda\}>0$ only determines it also for $\imag\{\lambda\}<0$.  Note that using \eqref{eq:TWRI-Schwarz-symmetry} in \eqref{eq:Mplus-Mminus-intro} with $\mathbf{D}(\lambda)$ written in the form \eqref{eq:D-representation} shows that the analytic functions $u:\mathbb{C}\setminus\mathbb{R}\to\mathbb{C}$ and $v:\mathbb{C}\setminus\mathbb{R}\to\mathbb{C}$ satisfy 
\begin{equation}
u(\lambda^*)^*=u(\lambda)\quad\text{and}\quad v(\lambda^*)^*=v(\lambda).
\end{equation}
\subsubsection{Discontinuity of $\mathbf{M}^\sigma(x;\lambda)$ for $\lambda\in\mathbb{R}$}
To relate the boundary values taken on the real axis by $\mathbf{M}^\sigma(x;\lambda)$ from the upper and lower half $\lambda$-planes, we first introduce solutions of the conditions \eqref{eq:RHP-M-direct-conditions} and the differential equation \eqref{eq:LaxPair-x} assuming now that $\lambda\in\mathbb{R}$ (in which case the boundedness condition on $\mathbf{M}^\sigma(x;\lambda)$ becomes superfluous).  These are the so-called \emph{Jost solutions} of the scattering problem, and we denote them by $\mathbf{M}^\sigma_\mathrm{J}(x;\lambda)$ for $\lambda\in\mathbb{R}$.  They are uniquely characterized by Volterra equations:
\begin{equation}
M^\sigma_{\mathrm{J},jk}(x;\lambda)=\delta_{jk} -\frac{1}{\epsilon}\int_{\sigma\infty}^x e^{-i\lambda(c^{[j]}-c^{[k]})(x-y)/\epsilon}(\mathbf{Q}(y)\mathbf{M}_\mathrm{J}^\sigma(y;\lambda))_{jk}\,\dd y,\quad\lambda\in\mathbb{R}.
\label{eq:TWRI-Jost-Volterra}
\end{equation}
The iterates of these equations converge for $\lambda\in\mathbb{R}$ provided only $\mathbf{Q}\in L^1(\mathbb{R})$.  The corresponding matrix solutions $\Phi_\mathrm{J}^\sigma(x;\lambda)$ of the differential equation \eqref{eq:LaxPair-x} are both fundamental and so there exists a \emph{scattering matrix} $\mathbf{S}(\lambda)$ such that the identity
$\Phi_\mathrm{J}^+(x;\lambda)=\Phi_\mathrm{J}^-(x;\lambda)\mathbf{S}(\lambda)$
holds as an identity in $x\in\mathbb{R}$.  Equivalently, the scattering matrix is determined from $\mathbf{M}^\sigma_\mathrm{J}(x;\lambda)$ by the relation \eqref{eq:scattering-relation}, which in particular implies that 
$\mathbf{S}(\lambda)$ is unimodular:  $\det(\mathbf{S}(\lambda))=1$ holds for all $\lambda\in\mathbb{R}$.

The solution $\mathbf{M}^\sigma(x;\lambda)$ of the Fredholm-type system \eqref{eq:Fredholm-system} for $\imag\{\lambda\}\neq 0$ can be continuously extended to the real axis (possibly excluding a discrete set of real values of $\lambda$; see Appendix~\ref{subapp-Zhou-transform} for how such spectral singularities can be dealt with) from both half-planes.  We denote the boundary values by $\mathbf{M}^\sigma_\pm(x;\lambda)$ defined where possible as follows:
\begin{equation}
\mathbf{M}^\sigma_\pm(x;\lambda):=\lim_{\delta\downarrow 0}\mathbf{M}^\sigma(x;\lambda\pm i\delta),\quad \lambda\in\mathbb{R}.
\end{equation}
The boundary values $\mathbf{M}^\sigma_\pm(x;\lambda)$ continue to satisfy the Fredholm-type system \eqref{eq:Fredholm-system} with the understanding that the sign $\sigma_{jk}(\lambda)$ appearing in \eqref{eq:sigma-j-k-lambda}--\eqref{eq:signs} is reinterpreted simply as $\pm\mathrm{sgn}(k-j)$.
Now, the argument that $\mathbf{M}^\sigma(x;\lambda)\to \mathbb{I}$ as $x\to \sigma\infty$ for $\imag\{\lambda\}\neq 0$ breaks down when $\lambda$ becomes real.  From \eqref{eq:Fredholm-system} it is easy to see that if $\lambda\in\mathbb{R}$ and $\mathbf{M}^\sigma_\pm(x;\lambda)$ solving \eqref{eq:Fredholm-system} exists, then $M^\sigma_{\pm,jk}(x;\lambda)\to \delta_{jk}$ as $x\to\sigma\infty$ provided that $k=j$ or $\pm\mathrm{sgn}(k-j)=\sigma$.  Otherwise, $M^\sigma_{\pm,jk}(x;\lambda)e^{i\lambda(c^{[j]}-c^{[k]})x/\epsilon}$ tends, as $x\to\sigma\infty$, to a (generally nonzero) limit depending on $\lambda\in\mathbb{R}$.  Therefore, for some quantities $T_{jk}^+(\lambda)$, 
\begin{equation}
\lim_{x\to +\infty}e^{i\lambda\mathbf{C}x/\epsilon}\mathbf{M}^+_+(x;\lambda)e^{-i\lambda\mathbf{C}x/\epsilon} = \begin{pmatrix}1 & 0 & 0\\ T^+_{21}(\lambda) & 1 & 0\\ T^+_{31}(\lambda) & T^+_{32}(\lambda) & 1\end{pmatrix},
\label{eq:M-plus-plus-limit}
\end{equation}
while for some other quantities $T_{jk}^-(\lambda)$, 
\begin{equation}
\lim_{x\to -\infty}e^{i\lambda\mathbf{C}x/\epsilon}\mathbf{M}^-_+(x;\lambda)e^{-i\lambda\mathbf{C}x/\epsilon}=\begin{pmatrix}1 & T^-_{12}(\lambda) & T^-_{13}(\lambda)\\0 & 1 & T^-_{23}(\lambda)\\
0 & 0 & 1\end{pmatrix}.
\label{eq:M-minus-plus-limit}
\end{equation}
Letting $\lambda$ tend to the real axis from the upper half-plane in the Schwarz symmetry relation \eqref{eq:TWRI-Schwarz-symmetry} yields the formula
\begin{equation}
\mathbf{M}^\sigma_-(x;\lambda)=\mathbf{E}\mathbf{M}^\sigma_+(x;\lambda)^{-\dagger}\mathbf{E},\quad \lambda\in\mathbb{R}.
\end{equation}
Combining this with \eqref{eq:M-plus-plus-limit} yields
\begin{equation}
\lim_{x\to +\infty} e^{i\lambda\mathbf{C}x/\epsilon}\mathbf{M}_-^+(x;\lambda)e^{-i\lambda\mathbf{C}x/\epsilon}=\begin{pmatrix}1 & -\gamma^{[1]}\gamma^{[2]} T^+_{21}(\lambda)^* & \gamma^{[1]}\gamma^{[3]} T^+_{31}(\lambda)^*\\0 & 1 & -\gamma^{[2]}\gamma^{[3]} T^+_{32}(\lambda)^*\\0 & 0 & 1
\end{pmatrix}^{-1},
\end{equation}
and combining it with \eqref{eq:M-minus-plus-limit} yields
\begin{equation}
\lim_{x\to -\infty}e^{i\lambda\mathbf{C}x/\epsilon}\mathbf{M}^-_-(x;\lambda)e^{-i\lambda\mathbf{C}x/\epsilon}=\begin{pmatrix}
1 & 0 & 0\\
-\gamma^{[1]}\gamma^{[2]} T^-_{12}(\lambda)^* & 1 & 0\\
\gamma^{[1]}\gamma^{[3]} T^-_{13}(\lambda)^* & -\gamma^{[2]}\gamma^{[3]} T^-_{23}(\lambda)^* & 1
\end{pmatrix}^{-1}.
\end{equation}
By uniqueness of the Jost solutions, it then follows that
\begin{equation}
\mathbf{M}_+^+(x;\lambda)=
\mathbf{M}_\mathrm{J}^+(x;\lambda)e^{-i\lambda\mathbf{C}x/\epsilon}\begin{pmatrix}1 & 0 & 0\\T^+_{21}(\lambda) & 1 & 0\\T^+_{31}(\lambda) & T^+_{32}(\lambda) & 1\end{pmatrix}e^{i\lambda\mathbf{C}x/\epsilon},
\label{eq:M-plus-jump-upper}
\end{equation}
\begin{equation}
\mathbf{M}_-^+(x;\lambda)=\mathbf{M}_\mathrm{J}^+(x;\lambda) e^{-i\lambda\mathbf{C}x/\epsilon}
\begin{pmatrix}1 & -\gamma^{[1]}\gamma^{[2]} T^+_{21}(\lambda)^* & \gamma^{[1]}\gamma^{[3]} T^+_{31}(\lambda)^*\\0 & 1 & -\gamma^{[2]}\gamma^{[3]} T^+_{32}(\lambda)^*\\0 & 0 & 1
\end{pmatrix}^{-1}e^{i\lambda\mathbf{C}x/\epsilon},
\label{eq:M-plus-jump-lower}
\end{equation}
\begin{equation}
\mathbf{M}_+^-(x;\lambda)=\mathbf{M}_\mathrm{J}^-(x;\lambda)e^{-i\lambda\mathbf{C}x/\epsilon}
\begin{pmatrix}1 & T^-_{12}(\lambda) & T^-_{13}(\lambda)\\
0 & 1 & T^-_{23}(\lambda)\\
0 & 0 & 1\end{pmatrix}e^{i\lambda\mathbf{C}x/\epsilon},
\label{eq:M-minus-jump-upper}
\end{equation}
and
\begin{equation}
\mathbf{M}^-_-(x;\lambda)=\mathbf{M}_\mathrm{J}^-(x;\lambda)e^{-i\lambda\mathbf{C}x/\epsilon}
\begin{pmatrix}1 & 0 &0\\
-\gamma^{[1]}\gamma^{[2]} T^-_{12}(\lambda)^* & 1 & 0\\
\gamma^{[1]}\gamma^{[3]} T^-_{13}(\lambda)^* & -\gamma^{[2]}\gamma^{[3]} T^-_{23}(\lambda)^* & 1\end{pmatrix}^{-1} e^{i\lambda \mathbf{C}x/\epsilon}.
\label{eq:M-minus-jump-lower}
\end{equation}
Combining \eqref{eq:M-plus-jump-upper}--\eqref{eq:M-plus-jump-lower} or \eqref{eq:M-minus-jump-upper}--\eqref{eq:M-minus-jump-lower}, gives the \emph{jump condition} relating the boundary values of $\mathbf{M}^\sigma(x;\lambda)$:
\begin{equation}
\mathbf{M}^\sigma_+(x;\lambda)=\mathbf{M}^\sigma_-(x;\lambda)e^{-i\lambda\mathbf{C}x/\epsilon}
\mathbf{V}_0^\sigma(\lambda)e^{i\lambda\mathbf{C}x/\epsilon},\quad \lambda\in\mathbb{R},
\end{equation}
where $\mathbf{V}_0^\sigma(\lambda)$ is defined in terms of the quantities $T^\sigma_{\ell m}(\lambda)$ for $\sigma(\ell- m)>0$ by \eqref{eq:TWRI-Jump}.

The uniqueness of the Jost solutions combined with the Schwarz symmetry \eqref{eq:TWRI-Schwarz-symmetry} for $\lambda\in\mathbb{R}$ shows that $\mathbf{M}_\mathrm{J}^\pm(x;\lambda)=
\mathbf{E}\mathbf{M}_\mathrm{J}^\pm(x;\lambda)^{-\dagger}\mathbf{E}$ holds for all real $\lambda$.  From these identities it follows that the scattering matrix satisfies a similar identity:
\begin{equation}
\mathbf{S}(\lambda)=\mathbf{E}\mathbf{S}(\lambda)^{-\dagger}\mathbf{E},\quad\lambda\in\mathbb{R}.
\label{eq:S-symmetry}
\end{equation}
We will now express the elements of the jump matrices $\mathbf{V}_0^\sigma(\lambda)$ in terms of those of the scattering matrix $\mathbf{S}(\lambda)$, which makes them computable directly from the analysis of the Jost solutions alone.  To do this, we observe that $e^{i\lambda (c^{[j]}-c^{[k]})x/\epsilon}M_{\pm,jk}^\sigma(x;\lambda)$ also has a limit as $x\to -\sigma\infty$, namely,
\begin{equation}
\lim_{x\to -\infty}e^{i\lambda\mathbf{C}x/\epsilon}\mathbf{M}_+^+(x;\lambda)e^{-i\lambda\mathbf{C}x/\epsilon}=\begin{pmatrix}1 & T^+_{12}(\lambda) & T^+_{13}(\lambda)\\
0 &1 & T^+_{23}(\lambda)\\
0 & 0 & 1\end{pmatrix}^{-1}\mathbf{D}_+(\lambda),
\label{eq:M-plus-plus-other-limit}
\end{equation}
and similarly,
\begin{equation}
\lim_{x\to +\infty}e^{i\lambda\mathbf{C}x/\epsilon}\mathbf{M}_+^-(x;\lambda)e^{-i\lambda\mathbf{C}x/\epsilon}=\begin{pmatrix} 1 & 0 & 0\\
T^-_{21}(\lambda) & 1& 0\\
T^-_{31}(\lambda) & T^-_{32}(\lambda) & 1\end{pmatrix}^{-1}\mathbf{D}_+(\lambda)^{-1},
\label{eq:M-minus-plus-other-limit}
\end{equation}
where $T^\sigma_{\ell m}(\lambda)$ for $\sigma(\ell-m)<0$ are some additional quantities, and where 
\begin{equation}
\mathbf{D}_+(\lambda):=\lim_{\delta\downarrow 0} \mathbf{D}(\lambda+ i\delta),\quad\lambda\in\mathbb{R}.
\end{equation} 
Again by uniqueness of the Jost solutions, from \eqref{eq:M-plus-plus-other-limit}--\eqref{eq:M-minus-plus-other-limit} we get
\begin{equation}
\mathbf{M}_+^+(x;\lambda)=\mathbf{M}_\mathrm{J}^-(x;\lambda)e^{-i\lambda\mathbf{C}x/\epsilon}
\begin{pmatrix}1& T^+_{12}(\lambda) & T^+_{13}(\lambda)\\
0 & 1 & T^+_{23}(\lambda)\\ 0 & 0 & 1\end{pmatrix}^{-1}\mathbf{D}_+(\lambda)
e^{i\lambda\mathbf{C}x/\epsilon}
\label{eq:M-plus-plus-in-terms-of-Jost-minus}
\end{equation}
and 
\begin{equation}
\mathbf{M}^-_+(x;\lambda)=\mathbf{M}_\mathrm{J}^+(x;\lambda)e^{-i\lambda\mathbf{C}x/\epsilon}
\begin{pmatrix}1 & 0 & 0\\
T^-_{21}(\lambda) & 1 & 0\\
T^-_{31}(\lambda) & T^-_{32}(\lambda) &1\end{pmatrix}^{-1}\mathbf{D}_+(\lambda)^{-1}e^{i\lambda\mathbf{C}x/\epsilon}.
\label{eq:M-minus-plus-in-terms-of-Jost-plus}
\end{equation}
Eliminating $\mathbf{M}_+^+(x;\lambda)$ between \eqref{eq:M-plus-jump-upper} and \eqref{eq:M-plus-plus-in-terms-of-Jost-minus}, substituting from \eqref{eq:scattering-relation} and using the fact that $\mathbf{M}^-_\mathrm{J}(x;\lambda)$ is invertible gives the ``LDU'' factorization of $\mathbf{S}(\lambda)^{-1}$;
 similarly, from \eqref{eq:M-minus-jump-upper} and \eqref{eq:M-minus-plus-in-terms-of-Jost-plus} we get
the ``UDL'' factorization of $\mathbf{S}(\lambda)$.  These factorizations are explicitly given by \eqref{eq:LDU-UDL}.
Therefore, the quantities $T^+_{jk}(\lambda)$ for $j>k$ that enter into the jump matrix $\mathbf{V}_0^+(\lambda)$ can be obtained from the lower triangular factor in the LDU factorization of $\mathbf{S}(\lambda)^{-1}$, while the quantities $T^-_{jk}(\lambda)$ for $j<k$ that enter into the jump matrix $\mathbf{V}_0^-(\lambda)$ can be obtained from the upper triangular factor in the UDL factorization of $\mathbf{S}(\lambda)$.  Clearly, if $\mathbf{S}(\lambda)$ is diagonal, then 
$\mathbf{S}(\lambda)=\mathbf{D}_+(\lambda)$ and the jump matrices $\mathbf{V}_0^\sigma(\lambda)$ both coincide with the identity matrix.  More generally, the diagonal matrix $\mathbf{D}(\lambda)$ for $\imag\{\lambda\}>0$ can be recovered from the scattering matrix $\mathbf{S}(\lambda)$ by meromorphic continuation of $\mathbf{D}_+(\lambda)$, and from the factorization of $\mathbf{S}(\lambda)^{-1}$ given by \eqref{eq:LDU-UDL} and the representation \eqref{eq:D-representation} we have
\begin{equation}
u_+(\lambda)
=\left[\mathbf{S}(\lambda)^{-1}\right]_{11}=S_{11}(\lambda)^*,\quad\lambda\in\mathbb{R},
\label{eq:reciprocal-d1-plus}
\end{equation}
(the second equality follows from \eqref{eq:S-symmetry}) 
while from the factorization of $\mathbf{S}(\lambda)$ given by \eqref{eq:LDU-UDL} we have
\begin{equation}
v_+(\lambda)
=S_{33}(\lambda),\quad\lambda\in\mathbb{R}.
\label{eq:d3-plus}
\end{equation}
We remind the reader that the functions defined for $\lambda\in\mathbb{R}$ by \eqref{eq:reciprocal-d1-plus} and \eqref{eq:d3-plus} are the boundary values of functions analytic in the upper half-plane.

Finally, we note a general property of the scattering matrix $\mathbf{S}(\lambda)$ that is particularly useful in the study of the TWRI equations for initial fields $q^{[k]}(x,0)=q^{[k]}(x)$ having disjoint supports.  
\begin{prop}
Let $\mathbf{Q}(x)$ be a potential with scattering matrix $\mathbf{S}(\lambda)$.  Let $x_0\in\mathbb{R}$, and define potentials $\mathbf{Q}_<(x):=\mathbf{Q}(x)\chi_{(-\infty,x_0)}(x)$ and
$\mathbf{Q}_>(x):=\mathbf{Q}(x)\chi_{(x_0,+\infty)}(x)$ with associated scattering matrices $\mathbf{S}_<(\lambda)$ and $\mathbf{S}_>(\lambda)$ respectively.  Then $\mathbf{S}(\lambda)=\mathbf{S}_<(\lambda)\mathbf{S}_>(\lambda)$ holds for all $\lambda\in\mathbb{R}$.
\label{prop-scattering-factorize}
\end{prop}
\begin{proof}
Let $\Phi_\mathrm{J}^\pm(x;\lambda)$, $\Phi_{<\mathrm{J}}^\pm(x;\lambda)$, and $\Phi_{>\mathrm{J}}^\pm(x;\lambda)$ denote the Jost solutions corresponding to the potentials $\mathbf{Q}(x)$, $\mathbf{Q}_<(x)$, and $\mathbf{Q}_>(x)$ respectively.  Then 
\begin{multline}
\Phi_\mathrm{J}^+(x_0;\lambda)=\Phi_{>\mathrm{J}}^+(x_0;\lambda)=\Phi_{>\mathrm{J}}^-(x_0;\lambda)\mathbf{S}_>(\lambda)=e^{-i\lambda\mathbf{C}x_0/\epsilon}\mathbf{S}_>(\lambda)\\{}=
\Phi_{<\mathrm{J}}^+(x_0;\lambda)\mathbf{S}_>(\lambda)=\Phi_{<\mathrm{J}}^-(x_0;\lambda)
\mathbf{S}_<(\lambda)\mathbf{S}_>(\lambda)=\Phi_\mathrm{J}^-(x_0;\lambda)
\mathbf{S}_<(\lambda)\mathbf{S}_>(\lambda).
\end{multline}
Comparing with $\Phi_\mathrm{J}^+(x;\lambda)=\Phi_\mathrm{J}^-(x;\lambda)\mathbf{S}(\lambda)$
completes the proof.
\end{proof}
By applying this result successively, further decompositions of $\mathbf{S}(\lambda)$ can be achieved; for example if $x_{12}<x_{23}$ and we set $\mathbf{Q}^{[1]}(x):=\mathbf{Q}(x)\chi_{(-\infty,x_{12})}(x)$, $\mathbf{Q}^{[2]}(x):=\mathbf{Q}(x)\chi_{(x_{12},x_{23})}(x)$, and $\mathbf{Q}^{[3]}(x):=\mathbf{Q}(x)\chi_{(x_{23},+\infty)}(x)$, with associated scattering matrices $\mathbf{S}^{[k]}(\lambda)$, $k=1,2,3$, then $\mathbf{S}(\lambda)=\mathbf{S}^{[1]}(\lambda)\mathbf{S}^{[2]}(\lambda)\mathbf{S}^{[3]}(\lambda)$.
This is exactly the generalization needed in the situation described by \eqref{eq:disjoint-support-assumption}.

\subsubsection{Isolated singularities of $\mathbf{M}^\sigma(x;\lambda)$}
\label{section-discrete-data}
We first consider singularities in the upper half-plane $\imag\{\lambda\}>0$.  Since $\mathbf{m}^{+,3}(x;\lambda)$ and $\mathbf{m}^{-,1}(x;\lambda)$ satisfy Volterra equations with analytic kernels and absolutely convergent iterates for $\imag\{\lambda\}>0$, these columns of $\mathbf{M}^\sigma(x;\lambda)$ are automatically analytic in the upper half-plane.  According to \eqref{eq:m1-plus-Im-lambda-positive}, $\mathbf{m}^{+,1}(x;\lambda)$ has a pole of finite order at $\lambda\in\mathbb{C}_+$ if and only if the analytic function $u(\lambda)$ has a zero there of the same order.  Similarly, according to \eqref{eq:m3-minus-Im-lambda-positive}, $\mathbf{m}^{-,3}(x;\lambda)$ has a pole of finite order at $\lambda\in\mathbb{C}_+$ if and only if the analytic function $v(\lambda)$ has a zero there of the same order.  

The conditions under which the central column $\mathbf{m}^{\sigma,2}(x;\lambda)$ exhibits a singularity for $\imag\{\lambda\}>0$ are more subtle.  Combining \eqref{eq:m2-cross-product-formula} with \eqref{eq:mplus-1cross2-Im-lambda-positive} and the analyticity of $\mathbf{n}^{+,3}(x;\lambda)$ shows that it is necessary that $v(\lambda)=0$ for $\mathbf{m}^{+,2}(x;\lambda)$ to be singular (have a pole of finite order) at $\lambda\in\mathbb{C}_+$.
However, the condition $v(\lambda)=0$ alone is not sufficient to generate a singularity because although neither of the analytic factors $\mathbf{n}^{-,3}(x;\lambda)$ nor 
$\mathbf{n}^{+,1}(x;\lambda)$ (appearing in \eqref{eq:m2-cross-product-formula} for $\sigma=+$ along with the scalar factor $v(\lambda)^{-1}$) can vanish for any $x\in\mathbb{R}$, it is indeed possible that $\mathbf{n}^{-,3}(x;\lambda)\times\mathbf{n}^{+,1}(x;\lambda)$ may vanish identically as a function of $x$.  The latter is a vector solution of the differential equation \eqref{eq:m-columns-ODE} for $j=2$ that has the limit $u(\lambda)\mathbf{e}^2$ as $x\to -\infty$ and the limit $v(\lambda)\mathbf{e}^2$ as $x\to +\infty$.  Obviously the limiting value at $x=+\infty$ vanishes under the condition $v(\lambda)=0$ necessary for existence of a singularity of $\mathbf{m}^{+,2}(x;\lambda)$.  It is also clear that if $u(\lambda)\neq 0$ then $\mathbf{n}^{-,3}(x;\lambda)\times\mathbf{n}^{+,1}(x;\lambda)$ can vanish for no $x\in\mathbb{R}$ and thus $\mathbf{m}^{+,2}(x;\lambda)$ has a pole of the same order $z_v$ as the zero of $v(\lambda)$ (and at the same time there can be no singularity of $\mathbf{m}^{+,1}(x;\lambda)$).  More generally, if $u(\lambda)$ vanishes to some non-negative order $z_u<z_v$ then $\mathbf{m}^{+,2}(x;\lambda)$ has a pole of order at least $1$ and at most $z_v$.
If $u(\lambda)$ vanishes to order $z_u\ge z_v>0$ at some $\lambda\in\mathbb{C}_+$, then either $\mathbf{n}^{-,3}(x;\lambda)\times\mathbf{n}^{+,1}(x;\lambda)$ vanishes identically in $x$ to order at least $z_v$ or not.  In the former case 
the singularity of $\mathbf{m}^{+,2}(x;\lambda)$ is removable, while in the latter case 
(if $z_v=1$ this implies the existence of a nontrivial eigenfunction of \eqref{eq:m-columns-ODE} for $j=2$, i.e., a nonzero solution decaying to zero as $|x|\to\infty$)
 $\mathbf{m}^{+,2}(x;\lambda)$ has a pole of order at least $1$ and at most $z_v$ at $\lambda\in\mathbb{C}_+$.  Similar analysis shows that $u(\lambda)$ must have a zero of order $z_u>0$ at $\lambda\in\mathbb{C}_+$ in order that $\mathbf{m}^{-,2}(x;\lambda)$ have a pole of order at most $z_u$ at $\lambda\in\mathbb{C}_+$, but the singularity is non-removable only if either $v(\lambda)$ vanishes to non-negative order $z_v<z_u$ or $z_v\ge z_u$ and also there is some $x\in\mathbb{R}$ for which $\mathbf{n}^{-,3}(x;\lambda)\times\mathbf{n}^{+,1}(x;\lambda)$ does not vanish to order at least $z_u$ at $\lambda\in\mathbb{C}_+$.
 
 We therefore see that $\mathbf{M}^\sigma(x;\lambda)$ has a \emph{simple} pole at a point $\lambda_0\in\mathbb{C}_+$ if one of the following three distinct cases holds.
 \begin{itemize}
 \item If $\lambda_0\in\mathbb{C}_+$ is a simple zero of $u(\lambda)$ but $v(\lambda_0)\neq 0$, then $\mathbf{m}^{+,1}(x;\lambda)$ and $\mathbf{m}^{-,2}(x;\lambda)$ have simple poles at $\lambda_0$ while all other columns of $\mathbf{M}^\sigma(x;\lambda)$ are analytic at $\lambda_0$.  
 Since both $\mathbf{m}^{-,1}\times\mathbf{m}^{-,2}=\mathbf{n}^{-,3}$ and $\mathbf{m}^{+,1}\times\mathbf{m}^{+,2}=v^{-1}\mathbf{m}^{-,1}\times\mathbf{m}^{-,2}=v^{-1}\mathbf{n}^{-,3}$ are known to be analytic at $\lambda_0$ it follows that the principal parts of their Laurent expansions (proportional to $(\lambda-\lambda_0)^{-1}$) vanish, leading to the identities
 \begin{equation}
 \left[\mathop{\mathrm{Res}}_{\lambda=\lambda_0}\mathbf{m}^{+,1}(x;\lambda)\right]\times\mathbf{m}^{+,2}(x;\lambda_0)=\mathbf{0}\quad\text{and}\quad
 \mathbf{m}^{-,1}(x;\lambda_0)\times\left[\mathop{\mathrm{Res}}_{\lambda=\lambda_0}\mathbf{m}^{-,2}(x;\lambda)\right]=\mathbf{0}.
 \end{equation}
 The residue factors are necessarily nonzero (by the assumptions on $u$ and $v$ near $\lambda_0$) and the analytic columns solve the first-order equation \eqref{eq:m-columns-ODE} with nonzero boundary conditions and hence can vanish for no $x\in\mathbb{R}$.  It follows that there exist nonzero scalars $\beta_{12}(x)$ and $\beta_{21}(x)$ associated with $\lambda_0$ such that
\begin{equation}
\begin{split}
\mathop{\mathrm{Res}}_{\lambda=\lambda_0}\mathbf{M}^+(x;\lambda)&=\lim_{\lambda\to\lambda_0}\mathbf{M}^+(x;\lambda)\begin{pmatrix}0 & 0 & 0\\\beta_{21}(x) & 0 & 0\\
0 & 0 & 0\end{pmatrix},\\
\mathop{\mathrm{Res}}_{\lambda=\lambda_0}\mathbf{M}^-(x;\lambda)&=\lim_{\lambda\to\lambda_0}\mathbf{M}^-(x;\lambda)\begin{pmatrix}0 & \beta_{12}(x) & 0\\0 & 0 & 0\\0 & 0 & 0\end{pmatrix}.
\end{split}
\end{equation}
Using \eqref{eq:Mplus-Mminus-intro} and \eqref{eq:D-representation} for $\imag\{\lambda_0\}>0$ one can easily show that
\begin{equation}
\beta_{12}(x)\beta_{21}(x)=\frac{v(\lambda_0)}{u'(\lambda_0)^2}.
\label{eq:beta12-beta21-relation}
\end{equation}
 \item If $\lambda_0\in\mathbb{C}_+$ is a simple zero of $v(\lambda)$ but $u(\lambda_0)\neq 0$, then $\mathbf{m}^{+,2}(x;\lambda)$ and $\mathbf{m}^{-,3}(x;\lambda)$ have simple poles at $\lambda_0$ while all other columns of $\mathbf{M}^\sigma(x;\lambda)$ are analytic at $\lambda_0$.  Since both $\mathbf{m}^{+,2}\times\mathbf{m}^{+,3}=\mathbf{n}^{+,1}$ and $\mathbf{m}^{-,2}\times\mathbf{m}^{-,3}=u^{-1}\mathbf{m}^{+,2}\times\mathbf{m}^{+,3}=u^{-1}\mathbf{n}^{+,1}$ are known to be analytic at $\lambda_0$ it follows that the principal parts of their Laurent expansions (proportional to $(\lambda-\lambda_0)^{-1}$) vanish, leading to the identities
 \begin{equation}
 \left[\mathop{\mathrm{Res}}_{\lambda=\lambda_0}\mathbf{m}^{+,2}(x;\lambda)\right]\times
 \mathbf{m}^{+,3}(x;\lambda_0)=\mathbf{0}\quad\text{and}\quad
 \mathbf{m}^{-,2}(x;\lambda_0)\times\left[\mathop{\mathrm{Res}}_{\lambda=\lambda_0}
 \mathbf{m}^{-,3}(x;\lambda)\right]=\mathbf{0}.
 \end{equation}
 The residue factors are necessarily nonzero (by the assumptions on $u$ and $v$ near $\lambda_0$) and the analytic columns solve the first-order equation \eqref{eq:m-columns-ODE} with nonzero boundary conditions and hence can vanish for no $x\in\mathbb{R}$.  It follows that there exist nonzero scalars $\beta_{23}(x)$ and $\beta_{32}(x)$ such that
 \begin{equation}
 \begin{split}
 \mathop{\mathrm{Res}}_{\lambda=\lambda_0}\mathbf{M}^+(x;\lambda)=
 \lim_{\lambda\to\lambda_0}\mathbf{M}^+(x;\lambda)\begin{pmatrix}0 & 0 & 0\\
 0 & 0 & 0\\
 0 & \beta_{32}(x) & 0\end{pmatrix},\\
 \mathop{\mathrm{Res}}_{\lambda=\lambda_0}\mathbf{M}^-(x;\lambda)=
 \lim_{\lambda\to\lambda_0}\mathbf{M}^-(x;\lambda)\begin{pmatrix}0 & 0 & 0\\
 0 & 0 & \beta_{23}(x)\\0 & 0 & 0\end{pmatrix}.
 \end{split}
 \end{equation}
 Again, from \eqref{eq:Mplus-Mminus-intro} and \eqref{eq:D-representation} for $\imag\{\lambda_0\}>0$ it follows that
 \begin{equation}
 \beta_{23}(x)\beta_{32}(x)=\frac{u(\lambda_0)}{v'(\lambda_0)^2}.
 \label{eq:type-1-beta-swap}
 \end{equation}
 \item If $\lambda_0$ is a simultaneous simple zero of both $u(\lambda)$ and $v(\lambda)$, we have a subsidiary dichotomy based on whether or not 
 the double cross-product $\mathbf{n}^{-,3}(x;\lambda)\times\mathbf{n}^{+,1}(x;\lambda)$ vanishes identically as a function of $x$ for $\lambda=\lambda_0$.  Note that as a solution of \eqref{eq:m-columns-ODE} for $j=2$, the double cross-product either vanishes for all $x\in\mathbb{R}$ or for no $x\in\mathbb{R}$.
 \begin{itemize}
 \item If $\mathbf{n}^{-,3}(x;\lambda)\times\mathbf{n}^{+,1}(x;\lambda)\neq 0$ for $\lambda=\lambda_0$,
 then only the columns $\mathbf{m}^{+,3}(x;\lambda)$ and $\mathbf{m}^{-,1}(x;\lambda)$ are analytic at $\lambda_0$ and all other columns of $\mathbf{M}^\sigma(x;\lambda)$ have simple poles at $\lambda_0$.  Since $\mathbf{n}^{-,3}(x;\lambda)$ and $\mathbf{n}^{+,1}(x;\lambda)$ are analytic and nonzero at $\lambda=\lambda_0$ it follows that the leading terms of their Laurent expansions (proportional to $(\lambda-\lambda_0)^{-1}$) vanish, so using \eqref{eq:cross-product-definition} yields the identities
\begin{equation}
\mathbf{m}^{-,1}(x;\lambda_0)\times\left[\mathop{\mathrm{Res}}_{\lambda=\lambda_0}
\mathbf{m}^{-,2}(x;\lambda)\right]=\mathbf{0}\quad\text{and}\quad
\left[\mathop{\mathrm{Res}}_{\lambda=\lambda_0}\mathbf{m}^{+,2}(x;\lambda)\right]\times
\mathbf{m}^{+,3}(x;\lambda_0)=\mathbf{0}.
\label{eq:cross-identity-1}
\end{equation}
It also follows that both $\mathbf{m}^{+,1}\times\mathbf{m}^{+,2}=v^{-1}\mathbf{m}^{-,1}\times\mathbf{m}^{-,2}=v^{-1}\mathbf{n}^{-,3}$ and $\mathbf{m}^{-,2}\times\mathbf{m}^{-,3}=u^{-1}\mathbf{m}^{+,2}\times\mathbf{m}^{+,3}=u^{-1}\mathbf{n}^{+,1}$ have simple poles at $\lambda_0$.  This implies that the dominant terms (proportional to $(\lambda-\lambda_0)^{-2}$) of the Laurent expansions of $\mathbf{m}^{+,1}\times\mathbf{m}^{+,2}$ and $\mathbf{m}^{-,2}\times\mathbf{m}^{-,3}$ both vanish, leading to the identities
 \begin{equation}
 \left[\mathop{\mathrm{Res}}_{\lambda=\lambda_0}\mathbf{m}^{+,1}(x;\lambda)\right]\times
 \left[\mathop{\mathrm{Res}}_{\lambda=\lambda_0}\mathbf{m}^{+,2}(x;\lambda)\right]=\mathbf{0}\;\text{and}\;
 \left[\mathop{\mathrm{Res}}_{\lambda=\lambda_0}\mathbf{m}^{-,2}(x;\lambda)\right]\times
 \left[\mathop{\mathrm{Res}}_{\lambda=\lambda_0}\mathbf{m}^{-,3}(x;\lambda)\right]=\mathbf{0}.
 \label{eq:cross-identity-2}
 \end{equation}
 None of the residue factors in \eqref{eq:cross-identity-1}--\eqref{eq:cross-identity-2} can be zero, nor can the analytic columns, and therefore there exist nonzero scalars $\beta_{12}(x)$, $\beta_{13}(x)$, $\beta_{31}(x)$, and $\beta_{32}(x)$ such that
 \begin{equation}
 \begin{split}
 \mathop{\mathrm{Res}}_{\lambda=\lambda_0}\mathbf{M}^+(x;\lambda)&=
 \lim_{\lambda\to\lambda_0}\mathbf{M}^+(x;\lambda)\begin{pmatrix}0&0&0\\0&0&0\\\beta_{31}(x)&\beta_{32}(x) & 0\end{pmatrix},\\
 \mathop{\mathrm{Res}}_{\lambda=\lambda_0}\mathbf{M}^-(x;\lambda)&=
 \lim_{\lambda\to\lambda_0}\mathbf{M}^-(x;\lambda)\begin{pmatrix}0 & \beta_{12}(x) & \beta_{13}(x)\\0 & 0 & 0\\0 & 0 & 0\end{pmatrix}.
 \end{split}
 \end{equation}
 Using \eqref{eq:Mplus-Mminus-intro} and \eqref{eq:D-representation} for $\imag\{\lambda_0\}>0$ shows that
 the residue scalars for $\mathbf{M}^-$ are related to those for $\mathbf{M}^+$ by
 \begin{equation}
 \beta_{12}(x)=\frac{v'(\lambda_0)\beta_{32}(x)}{u'(\lambda_0)^2\beta_{31}(x)}
 \quad\text{and}\quad\beta_{13}(x)=\frac{1}{u'(\lambda_0)v'(\lambda_0)\beta_{31}(x)}.
  \label{eq:type-solsplit-beta-swap}
 \end{equation}
 \item If $\mathbf{n}^{-,3}(x;\lambda)\times\mathbf{n}^{+,1}(x;\lambda)=0$ for $\lambda=\lambda_0$, 
  then only the columns $\mathbf{m}^{+,1}(x;\lambda)$ and $\mathbf{m}^{-,3}(x;\lambda)$ have simple poles at $\lambda_0$ and all other columns of $\mathbf{M}^\sigma(x;\lambda)$ are analytic at $\lambda_0$.  Since $\mathbf{n}^{-,3}(x;\lambda)$ and $\mathbf{n}^{+,1}(x;\lambda)$ are analytic and nonzero at $\lambda_0$, $\mathbf{m}^{+,1}\times\mathbf{m}^{+,2}=v^{-1}\mathbf{m}^{-,1}\times\mathbf{m}^{-,2}=v^{-1}\mathbf{n}^{-,3}$ and $\mathbf{m}^{-,2}\times\mathbf{m}^{-,3}=u^{-1}\mathbf{m}^{+,2}\times\mathbf{m}^{+,3}=u^{-1}\mathbf{n}^{+,1}$ have simple poles at $\lambda_0$.  Therefore, the leading coefficients in their Laurent expansions are nonzero, which implies the \emph{inequalities}
 \begin{equation}
 \left[\mathop{\mathrm{Res}}_{\lambda=\lambda_0}\mathbf{m}^{+,1}(x;\lambda)\right]\times
 \mathbf{m}^{+,2}(x;\lambda_0)\neq\mathbf{0}\quad\text{and}\quad
 \mathbf{m}^{-,2}(x;\lambda_0)\times\left[\mathop{\mathrm{Res}}_{\lambda=\lambda_0}\mathbf{m}^{-,3}(x;\lambda)\right]\neq\mathbf{0}.
 \label{eq:cross-product-inequalities}
 \end{equation}
 On the other hand, the term proportional to $(\lambda-\lambda_0)^{-1}$ in the Laurent expansion of the left-hand side of the identity $\det(\mathbf{M}^\sigma(x;\lambda))=1$ has to vanish, and this implies the identities
 \begin{equation}
 \begin{split}
 \det\left(\mathop{\mathrm{Res}}_{\lambda=\lambda_0}\mathbf{m}^{+,1}(x;\lambda),\mathbf{m}^{+,2}(x;\lambda_0),\mathbf{m}^{+,3}(x;\lambda_0)\right)&=0,\\
 \det\left(\mathbf{m}^{-,1}(x;\lambda_0),\mathbf{m}^{-,2}(x;\lambda_0),\mathop{\mathrm{Res}}_{\lambda=\lambda_0}\mathbf{m}^{-,3}(x;\lambda)\right)&=0.
 \end{split}
 \end{equation}
 From the latter relations it follows that there exist scalars $\beta_{13}(x)$, $\beta_{23}(x)$, $\beta_{21}(x)$, and $\beta_{31}(x)$ such that
 \begin{equation}
 \begin{split}
 \mathop{\mathrm{Res}}_{\lambda=\lambda_0}\mathbf{M}^+(x;\lambda)&=\lim_{\lambda\to\lambda_0}\mathbf{M}^+(x;\lambda)\begin{pmatrix}0 & 0 & 0\\
 \beta_{21}(x) & 0 & 0\\
 \beta_{31}(x) & 0 & 0\end{pmatrix},\\
 \mathop{\mathrm{Res}}_{\lambda=\lambda_0}\mathbf{M}^-(x;\lambda)&=\lim_{\lambda\to\lambda_0}\mathbf{M}^-(x;\lambda)\begin{pmatrix} 0&0&\beta_{13}(x)\\
 0 & 0 & \beta_{23}(x)\\0 & 0 & 0\end{pmatrix}.
 \end{split}
 \label{eq:96}
 \end{equation}
 From the inequalities \eqref{eq:cross-product-inequalities} it follows that both $\beta_{13}(x)$ and $\beta_{31}(x)$ must be nonzero, but it is not possible to exclude the possibility that $\beta_{21}(x)$ or $\beta_{23}(x)$ could vanish.  Using \eqref{eq:Mplus-Mminus-intro} and \eqref{eq:D-representation} for $\imag\{\lambda_0\}>0$ shows that the residue scalars for $\mathbf{M}^-$ are related to those for $\mathbf{M}^+$ by
 \begin{equation}
 \beta_{23}(x)=-\frac{u'(\lambda_0)\beta_{21}(x)}{v'(\lambda_0)^2\beta_{31}(x)}\quad\text{and}\quad
 \beta_{13}(x)=\frac{1}{u'(\lambda_0)v'(\lambda_0)\beta_{31}(x)}.
 \label{eq:type-solfuse-or-2-beta-swap}
 \end{equation}
 \end{itemize}
 \end{itemize}
Because all singularities of $\mathbf{M}^\sigma(x;\lambda)$ for $\lambda\in\mathbb{C}_+$ arise by multiplication of columns that are solutions of \eqref{eq:m-columns-ODE} analytic for $\imag\{\lambda\}>0$ by meromorphic factors independent of $x$, it follows that if $\mathbf{m}^{\sigma,j}(x;\lambda)$ has a simple pole at $\lambda_0$ then its residue satisfies \eqref{eq:m-columns-ODE}.  Then, since the product of $e^{-i\lambda c^{[j]}x/\epsilon}$ with any solution of \eqref{eq:m-columns-ODE} satisfies \eqref{eq:LaxPair-x}, an equation that is the same regardless of the column index $j$, the above residue relations for $\mathbf{M}^\sigma(x;\lambda)$ at $\lambda=\lambda_0$ imply that the scalars $\beta_{jk}(x)$ 
associated to the simple pole $\lambda_0\in\mathbb{C}_+$ necessarily have the form
\begin{equation}
\beta_{jk}(x)=\beta_{jk}e^{-i\lambda_0(c^{[j]}-c^{[k]})x/\epsilon},
\label{eq:gamma-jk-x}
\end{equation}
where $\beta_{jk}$ are complex constants, the \emph{connection coefficients} for the simple pole $\lambda_0$.  

The residue relations may be summarized as follows:  for each simple pole $\lambda_0\in\mathbb{C}_+$ of $\mathbf{M}^\sigma(x;\lambda)$ there exist associated nonzero triangular constant matrices $\mathbf{N}^\sigma$ that are $2$-nilpotent (i.e., $\mathbf{N}^\sigma\mathbf{N}^\sigma=\mathbf{0}$), having the form \eqref{eq:Nplus-Nminus-form}, such that the identity \eqref{eq:M-sigma-residue-summary} holds.  It should be noted that, unlike in many other integrable equations, the poles of $\mathbf{M}^\sigma(x;\lambda)$ are generally not $L^2$ eigenvalues of the equation \eqref{eq:LaxPair-x}; the residue condition implies the existence of a certain subspace of solutions with related behavior in the limits $x\to \pm\sigma\infty$, but this is not a subspace of solutions with exponential decay in both directions.

Assuming that the only singularities of $\mathbf{M}^\sigma(x;\lambda)$ for $\imag\{\lambda\}>0$ are simple poles, the Schwarz-symmetry relation \eqref{eq:TWRI-Schwarz-symmetry} implies that the same is true for $\imag\{\lambda\}<0$ with the poles being the complex conjugates of those in the upper half-plane.  Moreover, relations corresponding to \eqref{eq:M-sigma-residue-summary}
but characterizing instead the residues at the singularity $\lambda_0^*\in\mathbb{C}_-$ can be obtained directly from the latter relations using \eqref{eq:TWRI-Schwarz-symmetry}.  The induced relations all stem from the observation that if $\mathbf{A}=(\mathbf{a}^1,\mathbf{a}^2,\mathbf{a}^3)$ is any $3\times 3$ matrix with unit determinant, then from Cramer's rule,
\begin{equation}
\mathbf{A}^{-\trans}=(\mathbf{a}^2\times\mathbf{a}^3,\mathbf{a}^3\times\mathbf{a}^1,\mathbf{a}^1\times\mathbf{a}^2),
\label{eq:A-inverse-transpose}
\end{equation}
where $\mathbf{A}^{-\trans}$ denotes the inverse transpose matrix.  

Let $\lambda_0\in\mathbb{C}_+$ be a simple pole of $\mathbf{M}^\sigma(x;\lambda)$ for which $\beta_{32}=0$.  According to \eqref{eq:M-sigma-residue-summary}, $\mathbf{M}^+(x;\lambda)$ has a Laurent expansion about $\lambda_0$ of the form  
\begin{equation}
\mathbf{M}^+(x;\lambda)=\frac{(\beta_{21}(x)\mathbf{c}^2+\beta_{31}(x)\mathbf{c}^3,\mathbf{0},\mathbf{0})}{\lambda-\lambda_0} + (\mathbf{c}^1,\mathbf{c}^2,\mathbf{c}^3) + \bigo{\lambda-\lambda_0}.
\end{equation}
Applying \eqref{eq:A-inverse-transpose} gives
\begin{equation}
\mathbf{M}^+(x;\lambda)^{-\trans}=\frac{(\mathbf{0},-\beta_{21}(x)\widetilde{\mathbf{c}}^1,-\beta_{31}(x)\widetilde{\mathbf{c}}^1))}{\lambda-\lambda_0} +(\widetilde{\mathbf{c}}^1,\widetilde{\mathbf{c}}^2,\widetilde{\mathbf{c}}^3) + \bigo{\lambda-\lambda_0},
\end{equation}
where $\widetilde{\mathbf{c}}^1:=\mathbf{c}^2\times\mathbf{c}^3$,
or, evaluating at $\lambda^*$ and taking the complex conjugate,
\begin{equation}
\mathbf{M}^+(x;\lambda^*)^{-\dagger}=\frac{(\mathbf{0},-\beta_{21}(x)^*\widetilde{\mathbf{c}}^{1*},-\beta_{31}(x)^*\widetilde{\mathbf{c}}^{1*})}{\lambda-\lambda_0^*} + (\widetilde{\mathbf{c}}^{1*},\widetilde{\mathbf{c}}^{2*},\widetilde{\mathbf{c}}^{3*})+\bigo{\lambda-\lambda_0^*}.
\end{equation}
The latter implies the residue relation
\begin{equation}
\mathop{\mathrm{Res}}_{\lambda=\lambda_0^*}\mathbf{M}^+(x;\lambda^*)^{-\dagger}=
\lim_{\lambda\to\lambda_0^*}\mathbf{M}^+(x;\lambda^*)^{-\dagger}\begin{pmatrix}0 & -\beta_{21}(x)^* & -\beta_{31}(x)^*\\0 & 0 & 0\\ 0 & 0 & 0\end{pmatrix}.
\end{equation}
Conjugating this formula by $\mathbf{E}=\mathrm{diag}(\gamma^{[1]},-\gamma^{[2]},\gamma^{[3]})=\mathbf{E}^{-1}$ and using \eqref{eq:TWRI-Schwarz-symmetry} gives
\begin{equation}
\mathop{\mathrm{Res}}_{\lambda=\lambda_0^*}\mathbf{M}^+(x;\lambda)=
\lim_{\lambda\to\lambda_0^*}\mathbf{M}^+(x;\lambda)\begin{pmatrix}
0 & \gamma^{[1]}\gamma^{[2]}\beta_{21}(x)^* & -\gamma^{[1]}\gamma^{[3]}\beta_{31}(x)^*\\
0 & 0 & 0\\
0 & 0 & 0\end{pmatrix}.
\label{eq:Mplus-residue-A-Im-lambda-negative}
\end{equation}
Similarly, if $\lambda_0\in\mathbb{C}_+$ is a simple pole of $\mathbf{M}^\sigma(x;\lambda)$ for which $\beta_{12}=0$, 
\begin{equation}
\mathop{\mathrm{Res}}_{\lambda=\lambda_0^*}\mathbf{M}^-(x;\lambda)=\lim_{\lambda\to\lambda_0^*}\mathbf{M}^-(x;\lambda)\begin{pmatrix}0 & 0 & 0\\0 & 0 & 0\\-\gamma^{[1]}\gamma^{[3]}\beta_{13}(x)^* & 
\gamma^{[2]}\gamma^{[3]}\beta_{23}(x)^* & 0\end{pmatrix}.
\end{equation}

Now let $\lambda_0\in\mathbb{C}_+$ be a simple pole of $\mathbf{M}^\sigma(x;\lambda)$ for which $\beta_{21}=0$.  From \eqref{eq:M-sigma-residue-summary}, the Laurent expansion of $\mathbf{M}^+(x;\lambda)$ about $\lambda_0$ has the form (it is necessary to keep track of further terms in the  expansion in order to observe their eventual cancellation)
\begin{equation}
\mathbf{M}^+(x;\lambda)=\frac{(\beta_{31}(x)\mathbf{c}^3,\beta_{32}(x)\mathbf{c}^3,\mathbf{0})}{\lambda-\lambda_0} + (\mathbf{c}^1,\mathbf{c}^2,\mathbf{c}^3) + (\mathbf{d}^1,\mathbf{d}^2,\mathbf{d}^3)(\lambda-\lambda_0)+\bigo{(\lambda-\lambda_0)^2}.
\end{equation}
Therefore using \eqref{eq:A-inverse-transpose},
\begin{equation}
\mathbf{M}^+(x;\lambda)^{-\trans}=\frac{(\mathbf{0},\mathbf{0},\beta_{31}(x)\mathbf{c}^3\times\mathbf{c}^2+\beta_{32}(x)\mathbf{c}^1\times\mathbf{c}^3)}{\lambda-\lambda_0} + 
(\widetilde{\mathbf{c}}^1,\widetilde{\mathbf{c}}^2
,\widetilde{\mathbf{c}}^3) + \bigo{\lambda-\lambda_0},
\label{eq:intermediate-formula-1}
\end{equation}
where
\begin{equation}
\widetilde{\mathbf{c}}^1=\beta_{32}(x)\mathbf{c}^3\times\mathbf{d}^3 +\mathbf{c}^2\times\mathbf{c}^3
\quad\text{and}\quad
\widetilde{\mathbf{c}}^2=\beta_{31}(x)\mathbf{d}^3\times\mathbf{c}^3+\mathbf{c}^3\times\mathbf{c}^1.
\end{equation}
We then observe that the residue term in \eqref{eq:intermediate-formula-1} can be expressed in terms of $\widetilde{\mathbf{c}}^1$ and $\widetilde{\mathbf{c}}^2$ only (the terms proportional to $\mathbf{d}^3$ cancel):  
\begin{equation}
\beta_{31}(x)\mathbf{c}^3\times\mathbf{c}^2+\beta_{32}(x)\mathbf{c}^1\times\mathbf{c}^3 = 
-\beta_{31}(x)\widetilde{\mathbf{c}}^1-\beta_{32}(x)\widetilde{\mathbf{c}}^2.
\end{equation}
Evaluating \eqref{eq:intermediate-formula-1} at $\lambda^*$, complex conjugating, and using \eqref{eq:TWRI-Schwarz-symmetry} then leads to the formula
\begin{equation}
\mathop{\mathrm{Res}}_{\lambda=\lambda_0^*}\mathbf{M}^+(x;\lambda)=
\lim_{\lambda\to\lambda_0^*}\mathbf{M}^+(x;\lambda)\begin{pmatrix}
0 & 0 & -\gamma^{[1]}\gamma^{[3]}\beta_{31}(x)^*\\
0 & 0 & \gamma^{[2]}\gamma^{[3]}\beta_{32}(x)^*\\
0 & 0 & 0
\end{pmatrix}.
\label{eq:Mplus-residue-AB-Im-lambda-negative}
\end{equation}
Similarly, if $\beta_{23}=0$, then
\begin{equation}
\mathop{\mathrm{Res}}_{\lambda=\lambda_0^*}\mathbf{M}^-(x;\lambda)=
\lim_{\lambda\to\lambda_0^*}\mathbf{M}^-(x;\lambda)
\begin{pmatrix}
0 & 0 & 0\\
\gamma^{[1]}\gamma^{[2]}\beta_{12}(x)^* & 0 & 0\\
-\gamma^{[1]}\gamma^{[3]}\beta_{13}(x)^* & 0 & 0
\end{pmatrix}.
\label{eq:Mminus-residue-AB-Im-lambda-negative}
\end{equation}

In summary, the relations corresponding to \eqref{eq:M-sigma-residue-summary} but applying to the simple pole at the conjugate point $\lambda_0^*\in\mathbb{C}_-$ can be written in the universal form
\begin{equation}
\mathop{\mathrm{Res}}_{\lambda=\lambda_0^*}\mathbf{M}^\sigma(x;\lambda)=\lim_{\lambda\to\lambda_0^*}\mathbf{M}^\sigma(x;\lambda)e^{-i\lambda_0^*\mathbf{C}x/\epsilon}(-\mathbf{E}\mathbf{N}^{\sigma\dagger}\mathbf{E})e^{i\lambda_0^*\mathbf{C}x/\epsilon}.
\end{equation}

\subsection{Time dependence of the scattering data}
\label{subapp:time-dependence}
Let $\{q^{[1]},q^{[2]},q^{[3]}\}$ be a classical solution of the TWRI equations \eqref{3wave} for which $q^{[k]}(x,t)\to 0$, $k=1,2,3$, sufficiently rapidly as $x\to\pm\infty$ for each $t$ in some open interval $t_0<t<t_1$.  
This means that the matrix $\mathbf{M}^\sigma(x;\lambda)$ satisfying the conditions \eqref{eq:RHP-M-direct-conditions} and the differential equations \eqref{eq:Mij-system} exists for all $t\in(t_0,t_1)$.  More precisely, for each such $t$ there is an exceptional set consisting of the union of $\mathbb{R}$ and a  discrete set of finite-order pole singularities in $\mathbb{C}\setminus\mathbb{R}$, and $\mathbf{M}(x;\lambda)$ exists uniquely for each $\lambda$ in the complement of this exceptional set.  We denote this matrix function by $\mathbf{M}(x,t;\lambda)$.  

Fix a point $\lambda\in\mathbb{C}\setminus\mathbb{R}$ that is non-exceptional in a neighborhood $U(\tau)$ of some $\tau\in (t_0,t_1)$.  Then for each $t\in U(\tau)$, $\Phi^\sigma(x,t;\lambda):=\mathbf{M}^\sigma(x,t;\lambda)e^{i\lambda\mathbf{C}x/\epsilon}$ is a fundamental solution matrix for the differential equation \eqref{eq:LaxPair-x}, and hence for each invertible matrix $\mathbf{K}(t)$ the same can be said of $\Phi^\sigma(x,t;\lambda)\mathbf{K}(t)$.  Moreover, since $\{q^{[1]},q^{[2]},q^{[3]}\}$ solve \eqref{3wave}, the two equations \eqref{eq:LaxPair-x}--\eqref{eq:LaxPair-t} are compatible in the sense that they admit a common simultaneous solution matrix for all $t\in U(\tau)$ and $x\in\mathbb{R}$ for the chosen value of $\lambda$.  Obviously this simultaneous solution must have the form $\Phi^\sigma(x,t;\lambda)\mathbf{K}(t)$, and by substitution into \eqref{eq:LaxPair-t} it follows that the matrix $\mathbf{M}^\sigma(x,t;\lambda)$ satisfies the differential equation
\begin{equation}
\epsilon\frac{\partial\mathbf{M}^\sigma}{\partial t}(x,t;\lambda)=\mathcal{B}(x,t;\lambda)\mathbf{M}^\sigma(x,t;\lambda)-\epsilon \mathbf{M}^\sigma(x,t;\lambda)e^{-i\lambda\mathbf{C}x/\epsilon}\frac{\dd\mathbf{K}}{\dd t}(t)\mathbf{K}(t)^{-1}e^{i\lambda\mathbf{C}x/\epsilon}.
\label{eq:M-general-t-eqn}
\end{equation}
Since $q^{[k]}(x,t)\to 0$ as $x\to\pm\infty$, $k=1,2,3$, it follows that $\mathcal{B}(x,t;\lambda)\to\mathcal{B}^\infty(\lambda):=-i\lambda|\mathbf{C}|\mathbf{C}^{-1}$ in the same limit.  Also, by definition of $\mathbf{M}^\sigma(x,t;\lambda)$, $\mathbf{M}^\sigma(x,t;\lambda)\to\mathbb{I}$ as $x\to\sigma\infty$ and $\mathbf{M}^\sigma(x,t;\lambda)\to\mathbf{D}(\lambda;t)^\sigma$ as $x\to-\sigma\infty$, where $\mathbf{D}(\lambda;t)$ is a diagonal matrix with unit determinant.  If we suppose that differentiation with respect to $t$ commutes with taking the limits $x\to\pm\infty$, we conclude that $\partial\mathbf{M}^\sigma(x,t;\lambda)/\partial t \to \mathbf{0}$ as $x\to\sigma\infty$ and that the off-diagonal part of $\partial\mathbf{M}^\sigma(x,t;\lambda)/\partial t$ vanishes also as $x\to -\sigma\infty$.
Multiplying \eqref{eq:M-general-t-eqn} on the left by $\mathbf{M}^\sigma(x,t;\lambda)^{-1}$ and considering the limits $x\to\pm\infty$ then shows that $\epsilon\mathbf{K}'(t)\mathbf{K}(t)^{-1}$ must be a diagonal matrix\footnote{Otherwise $e^{-i\lambda\mathbf{C}x/\epsilon}\mathbf{K}'(t)\mathbf{K}(t)^{-1}e^{i\lambda\mathbf{C}x/\epsilon}$ has off-diagonal elements that blow up in one or the other limit --- as $\imag\{\lambda\}\neq 0$ --- and that cannot be compensated for by the remaining terms in \eqref{eq:M-general-t-eqn} which converge to diagonal matrices.}, and from the limit $x\to\sigma\infty$ one can identify this diagonal matrix as $\mathcal{B}^\infty(\lambda)$.  Therefore \eqref{eq:M-general-t-eqn} can in fact be rewritten as
\begin{equation}
\epsilon\frac{\partial\mathbf{M}^\sigma}{\partial t}(x,t;\lambda)=\mathcal{B}(x,t;\lambda)\mathbf{M}^\sigma(x,t;\lambda)-\mathbf{M}^\sigma(x,t;\lambda)\mathcal{B}^\infty(\lambda).
\label{eq:M-t-eqn}
\end{equation}
This equation immediately shows that if $\lambda=\lambda_0(t)\not\in\mathbb{R}$ is an isolated singularity of $\mathbf{M}^\sigma(x,t;\lambda)$ assumed to be differentiable with respect to $t$, necessarily a pole of finite order, then $\lambda_0(t)$ is in fact independent of $t$.  Indeed, differentiating with respect to $t$ the Laurent expansion of $\mathbf{M}^\sigma(x,t;\lambda)$ about a pole $\lambda_0(t)$ produces a term proportional to $\lambda_0'(t)$ that is more singular at $\lambda_0(t)$ than $\mathbf{M}^\sigma(x,t;\lambda)$ itself.  It follows from \eqref{eq:M-t-eqn} that $\lambda_0'(t)=0$.  Therefore, the exceptional set in the complex $\lambda$-plane for $\mathbf{M}^\sigma(x,t;\lambda)$ is independent of $t$ (and, of course $x$).

Suppose that $\lambda_0$ is a simple pole of $\mathbf{M}^\sigma(x,t;\lambda)$, and let $C$ be a small circle with positive orientation centered at $\lambda_0$ with radius sufficiently small that every point of $C$ is non-exceptional and that $\lambda_0$ is the only exceptional point in the interior.  For each $t$ a relation of the general form
\begin{equation}
\mathop{\mathrm{Res}}_{\lambda=\lambda_0}\mathbf{M}^\sigma(x,t;\lambda)=\lim_{\lambda\to\lambda_0}\mathbf{M}^\sigma(x,t;\lambda)\mathbf{N}(x,t)
\label{eq:M-residue-general}
\end{equation}
holds, where $\mathbf{N}(x,t)$ is a $2$-nilpotent matrix whose structure and dependence on $x$ has been explained in 
Appendix~\ref{section-discrete-data}.  We can now
easily deduce the way that $\mathbf{N}(x,t)$ evolves in time $t$.  We begin by rewriting \eqref{eq:M-residue-general} in the form
\begin{equation}
\frac{1}{2\pi i}\oint_C\mathbf{M}^\sigma(x,t;\lambda)\,\dd\lambda = \frac{1}{2\pi i}\oint_C\frac{\mathbf{M}^\sigma(x,t;\lambda)\mathbf{N}(x,t)}{\lambda-\lambda_0}\,\dd\lambda.
\end{equation}
For each $\lambda\in C$ the differential equation \eqref{eq:M-t-eqn} holds, so differentiating with respect to $t$ under the integral sign yields
\begin{multline}
\oint_C\left[\mathcal{B}(x,t;\lambda)\mathbf{M}^\sigma(x,t;\lambda)-\mathbf{M}^\sigma(x,t;\lambda)\mathcal{B}^\infty(\lambda)\right]\,\dd\lambda \\ {}= \oint_C\frac{\mathcal{B}(x,t;\lambda)\mathbf{M}^\sigma(x,t;\lambda)\mathbf{N}(x,t)-\mathbf{M}^\sigma(x,t;\lambda)\mathcal{B}^\infty(\lambda)\mathbf{N}(x,t)+\mathbf{M}^\sigma(x,t;\lambda)\epsilon\mathbf{N}_t(x,t)}{\lambda-\lambda_0}\,\dd\lambda,
\label{eq:loop-t-derivative}
\end{multline}
where $\mathbf{N}_t(x,t):=\partial\mathbf{N}(x,t)/\partial t$.  Let us write the Laurent expansion of $\mathbf{M}^\sigma(x,t;\lambda)$ about the presumed simple pole $\lambda=\lambda_0$ in the form
\begin{equation}
\mathbf{M}^\sigma(x,t;\lambda)=\frac{\mathbf{R}(x,t)}{\lambda-\lambda_0} +\mathbf{S}(x,t) + \cdots,
\label{eq:M-pole-Laurent}
\end{equation}
and then evaluate the integrals in \eqref{eq:loop-t-derivative} by residues at $\lambda_0$.  This yields the identity
\begin{multline}
\mathcal{B}(x,t;\lambda_0)\mathbf{R}(x,t)-\mathbf{R}(x,t)\mathcal{B}^\infty(\lambda_0)\\
{}=\mathcal{B}_\lambda(x,t;\lambda_0)\mathbf{R}(x,t)\mathbf{N}(x,t) + \mathcal{B}(x,t;\lambda_0)\mathbf{S}(x,t)\mathbf{N}(x,t)\\{}-\mathbf{R}(x,t)\mathcal{B}^\infty_\lambda(\lambda_0)\mathbf{N}(x,t)-\mathbf{S}(x,t)\mathcal{B}^\infty(\lambda_0)\mathbf{N}(x,t)+\mathbf{S}(x,t)\epsilon\mathbf{N}_t(x,t),
\label{eq:residue-calculation-1}
\end{multline}
where the subscript $\lambda$ denotes differentiation with respect to $\lambda$.  
Combining \eqref{eq:M-pole-Laurent} with 
\eqref{eq:M-residue-general} gives $\mathbf{R}(x,t)=\mathbf{S}(x,t)\mathbf{N}(x,t)$ and, since $\mathbf{N}(x,t)$ is $2$-nilpotent, $\mathbf{R}(x,t)\mathbf{N}(x,t)=\mathbf{0}$.  Using these in \eqref{eq:residue-calculation-1} gives
\begin{equation}
\mathbf{S}(x,t)\left(\epsilon\frac{\partial\mathbf{N}}{\partial t}(x,t)+[\mathbf{N}(x,t),\mathcal{B}^\infty(\lambda_0)]-\mathbf{N}(x,t)\mathcal{B}_\lambda^\infty(\lambda_0)\mathbf{N}(x,t)\right)=\mathbf{0}.
\label{eq:pre-S-on-the-left}
\end{equation}
It is easy to check that the matrix $\mathbf{N}(x,t)$ satisfies $\mathbf{N}(x,t)\mathbf{D}\mathbf{N}(x,t)=\mathbf{0}$ for every diagonal matrix $\mathbf{D}$, so as $\mathcal{B}^\infty(\lambda_0)$ is diagonal, \eqref{eq:pre-S-on-the-left} can be rewritten as
\begin{equation}
\mathbf{S}(x,t)\left(\epsilon\frac{\partial\mathbf{N}}{\partial t}(x,t)+[\mathbf{N}(x,t),\mathcal{B}^\infty(\lambda_0)]\right)=\mathbf{0}.
\label{eq:S-on-the-left}
\end{equation}
In fact, it can be shown from \eqref{eq:S-on-the-left} that\footnote{The argument is as follows.  Suppose first that the column space of $\mathbf{N}(x,t)$ coincides with $\mathrm{span}(\mathbf{e}^k)$ for some $k=1,2,3$ and $t_0<t<t_1$.  Then also the column space of $\epsilon\mathbf{N}_t+[\mathbf{N},\mathcal{B}^\infty(\lambda_0)]$ is contained in $\mathrm{span}(\mathbf{e}^k)$ because $\mathcal{B}^\infty(\lambda_0)$ is diagonal.  Since $\mathbf{R}\neq\mathbf{0}$, we have $\mathbf{SN}=\mathbf{R}\neq\mathbf{0}$, so $\mathbf{Se}^k\neq\mathbf{0}$.  It therefore follows from \eqref{eq:S-on-the-left} that \eqref{eq:S-gone!} holds.  The only remaining case to consider is if $\mathbf{N}$ has two zero columns and one column with two nonzero elements for some $t\in (t_0,t_1)$.  But in this case we may repeat the argument leading to \eqref{eq:S-on-the-left} working instead with the complex-conjugate pole $\lambda_0^*$, for which the residue matrix $-\mathbf{E}\mathbf{N}^\dagger\mathbf{E}$ has only one nonzero row and hence the argument described above applies.} even though $\mathbf{S}$ is not necessarily invertible, the evolution equation for $\mathbf{N}(x,t)$ is 
\begin{equation}
\epsilon\frac{\partial\mathbf{N}}{\partial t}(x,t)+[\mathbf{N}(x,t),\mathcal{B}^\infty(\lambda_0)]=\mathbf{0},
\label{eq:S-gone!}
\end{equation}
which implies that the matrix $\mathbf{N}(x,t)$ evolves explicitly in time $t$ as follows:
\begin{equation}
\mathbf{N}(x,t)=e^{\mathcal{B}^\infty(\lambda_0)t/\epsilon}\mathbf{N}(x,0)e^{-\mathcal{B}^\infty(\lambda_0)t/\epsilon}=e^{-i\lambda_0|\mathbf{C}|\mathbf{C}^{-1}t/\epsilon}\mathbf{N}(x,0)
e^{i\lambda_0|\mathbf{C}|\mathbf{C}^{-1}t/\epsilon}.
\label{eq:residue-matrices-time-evolution}
\end{equation}
Note also that according to \eqref{eq:gamma-jk-x}, $\mathbf{N}(x,0)=e^{-i\lambda_0\mathbf{C}x/\epsilon}\mathbf{N}_0e^{i\lambda_0\mathbf{C}x/\epsilon}$ holds for some nonzero triangular $2$-nilpotent constant matrix $\mathbf{N}_0$.  This proves \eqref{eq:pole-evolution-intro}.

Letting $\lambda$ approach the real axis from above and below and supposing that \eqref{eq:M-t-eqn} also governs the limiting boundary values taken on $\mathbb{R}$,  differentiating the jump condition 
$\mathbf{M}^\sigma_+(x,t;\lambda)=\mathbf{M}^\sigma_-(x,t;\lambda)e^{-i\lambda\mathbf{C}x/\epsilon}\mathbf{V}^\sigma_0(\lambda;t)e^{i\lambda\mathbf{C}x/\epsilon}$ with respect to $t$ and using \eqref{eq:M-t-eqn} yields the Lax-type equation
\begin{equation}
\epsilon\frac{\dd\mathbf{V}^\sigma_0}{\dd t}(\lambda;t)+[\mathbf{V}^\sigma_0(\lambda;t),\mathcal{B}^\infty(\lambda)]=\mathbf{0},\quad\lambda\in\mathbb{R}.
\label{eq:V0-ODE}
\end{equation}
From this equation it follows that
\begin{equation}
\mathbf{V}^\sigma_0(\lambda;t)=e^{\mathcal{B}^\infty(\lambda)t/\epsilon}\mathbf{V}_0^\sigma(\lambda;0)e^{-\mathcal{B}^\infty(\lambda) t/\epsilon}=
e^{-i\lambda|\mathbf{C}|\mathbf{C}^{-1}t/\epsilon}\mathbf{V}_0^\sigma(\lambda;0)e^{i\lambda|\mathbf{C}|\mathbf{C}^{-1}t/\epsilon},\quad\lambda\in\mathbb{R},
\label{eq:jump-matrix-time-evolution}
\end{equation}
so the jump matrix $\mathbf{V}_0^\sigma(\lambda;t)$ also evolves explicitly in time $t$, which proves \eqref{eq:jump-evolution-intro}.

The above derivations made use of various technical assumptions concerning the nature of the time dependence induced in $\mathbf{M}^\sigma(x,t;\lambda)$ from the fact that the fields $\{q^{[1]},q^{[2]},q^{[3]}\}$ constitute a suitable solution of the TWRI equations \eqref{3wave}.  However, to some degree these technicalities can be avoided in the sense that it is easy to prove that if the residue matrices evolve in time according to \eqref{eq:residue-matrices-time-evolution} and the jump matrix evolves in time according to \eqref{eq:jump-matrix-time-evolution}, then provided that $\mathbf{M}^\sigma(x,t;\lambda)$ can be reconstructed from its explicitly time-dependent scattering data (see Appendix~\ref{section:inverse} below), the fields $\{q^{[1]},q^{[2]},q^{[3]}\}$ extracted from $\mathbf{M}^\sigma(x,t;\lambda)$ via \eqref{eq:M-lambda-infinity-expand} and \eqref{eq:q-potentials-reconstruct} necessarily solve the TWRI equations \eqref{3wave}.  That said, the necessary assumptions may be fulfilled for certain initial data due to a priori well-posedness results for \eqref{3wave} such as that of Rauch \cite[Theorem 9.2.3]{Rauch12} (obtained without the use of complete integrability).

\subsection{The inverse scattering problem}
\label{appB5-inverse-scattering}
\label{section:inverse}
Suppose that the initial data $\{q^{[k]}(x,0)\}_{k=1}^3$ are such that the complex singularities of $\mathbf{M}^\sigma(x,0;\lambda)$ for $\imag\{\lambda\}>0$ are a finite number $N$ of simple poles, and that there exist no spectral singularities on the real axis, i.e., real zeros of $u_+(\lambda)$ or $v_+(\lambda)$.
The \emph{scattering data} for this initial condition consists of:
\begin{itemize}
\item The set $P=\{\lambda_n,n\in \mathcal{N}\}$ of simple poles of $\mathbf{M}^\sigma(x,0;\lambda)$ in $\mathbb{C}_+$, where $\mathcal{N}$ is a finite indexing set.
\item For each point $\lambda_n\in P$, a nonzero $2$-nilpotent strictly lower-triangular matrix $\mathbf{N}_n^+$ (for $\sigma=+$) or a nonzero $2$-nilpotent strictly upper-triangular matrix $\mathbf{N}_n^-$ (for $\sigma=-$).
\item The jump matrix $\mathbf{V}_0^\sigma(\lambda)$ defined for $\lambda\in\mathbb{R}$ by \eqref{eq:TWRI-Jump}.
\end{itemize}
The inverse problem is to recover $\mathbf{M}^\sigma(x,t;\lambda)$ from this scattering data, which evolves explicitly in time $t$ as described in Appendix~\ref{subapp:time-dependence}.  To this end, we formulate the following Riemann-Hilbert problem, which essentially determines both $\mathbf{M}^+(x,t;\lambda)$ and  $\mathbf{M}^-(x,t;\lambda)$.
\begin{myrhp}
Given scattering data and values of the independent variables $(x,t)\in\mathbb{R}^2$, seek a matrix function $\mathbf{M}^\sigma(\lambda)=\mathbf{M}^\sigma(x,t;\lambda)$ with the following properties:
\begin{itemize}
\item[]\textit{\textbf{Analyticity:}} $\mathbf{M}^\sigma(\lambda)$ is analytic for $\lambda\in \mathbb{C}\setminus(\mathbb{R}\cup P\cup P^*)$.
\item[]\textit{\textbf{Jump condition:}} $\mathbf{M}^\sigma(\lambda)$ takes continuous boundary values $\mathbf{M}^\sigma_\pm(\lambda):=\lim_{\delta\downarrow 0}\mathbf{M}^\sigma(\lambda\pm i\delta)$ for $\lambda\in\mathbb{R}$, and the boundary values are related by 
\begin{equation}
\mathbf{M}^\sigma_+(\lambda)=\mathbf{M}^\sigma_-(\lambda)e^{-i\lambda(\mathbf{C}x+|\mathbf{C}|\mathbf{C}^{-1}t)/\epsilon}\mathbf{V}_0^\sigma(\lambda)e^{i\lambda(\mathbf{C}x+|\mathbf{C}|\mathbf{C}^{-1}t)/\epsilon},\quad\lambda\in\mathbb{R}.
\label{eq:RHP-jump-condition}
\end{equation}
\item[]\textit{\textbf{Poles:}}  Each point of $P\cup P^*$ is a simple pole of $\mathbf{M}^\sigma(\lambda)$, and
\begin{equation}
\mathop{\mathrm{Res}}_{\lambda=\lambda_n}\mathbf{M}^\sigma(\lambda)=\lim_{\lambda\to\lambda_n}\mathbf{M}^\sigma(\lambda)
e^{-i\lambda_n(\mathbf{C}x+|\mathbf{C}|\mathbf{C}^{-1}t)/\epsilon}\mathbf{N}_n^\sigma
e^{i\lambda_n(\mathbf{C}x+|\mathbf{C}|\mathbf{C}^{-1}t)/\epsilon},\quad n\in\mathcal{N}
\label{eq:RHP-pole-C-plus}
\end{equation}
and
\begin{equation}
\mathop{\mathrm{Res}}_{\lambda=\lambda_n^*}\mathbf{M}^\sigma(\lambda)=\lim_{\lambda\to\lambda_n^*}
\mathbf{M}^\sigma(\lambda)e^{-i\lambda_n^*(\mathbf{C}x+|\mathbf{C}|\mathbf{C}^{-1}t)/\epsilon}(-\mathbf{E}\mathbf{N}_n^{\sigma\dagger}\mathbf{E})e^{i\lambda_n^*(\mathbf{C}x+|\mathbf{C}|\mathbf{C}^{-1}t)/\epsilon},\quad n\in\mathcal{N}.  
\label{eq:RHP-pole-C-minus}
\end{equation}
\item[]\textit{\textbf{Normalization:}} $\mathbf{M}^\sigma(\lambda)\to\mathbb{I}$ as $\lambda\to\infty$.
\end{itemize}
\label{rhp:M-sigma}
\end{myrhp}

The solution of this Riemann-Hilbert problem in either case $\sigma=\pm$ suffices to determine the solution $\{q^{[k]}(x,t)\}_{k=1}^3$ of the TWRI system \eqref{3wave} corresponding to the initial conditions that generated the scattering data.  Indeed, from $\mathbf{M}^\sigma(x,t;\lambda)$ one simply extracts the coefficient $\mathbf{F}^\sigma(x,t)$ from the Laurent expansion of $\mathbf{M}^\sigma(x,t;\lambda)$ (see \eqref{eq:M-lambda-infinity-expand}) and then obtains $q^{[k]}(x,t)$ for $k=1,2,3$ from \eqref{eq:q-potentials-reconstruct}.  

\begin{prop}
Suppose that $\gamma^{[1]}\gamma^{[2]}=\gamma^{[2]}\gamma^{[3]}=-1$.  Then for each $(x,t)\in\mathbb{R}^2$ there exists a unique classical solution of Riemann-Hilbert Problem~\ref{rhp:M-sigma}.
\end{prop}
\begin{proof}
For each $n\in\mathcal{N}$, let $D_n$ be a small disk centered at the pole $\lambda=\lambda_n\in\mathbb{C}_+$ with the positive radii of the disks chosen sufficiently small that no two disks intersect and no disk intersects $\mathbb{R}$.  Define a new unknown $\widetilde{\mathbf{M}}^\sigma(\lambda)$ by setting 
\begin{equation}
\widetilde{\mathbf{M}}^\sigma(\lambda):=\mathbf{M}^\sigma(\lambda)e^{-i\lambda(\mathbf{C}x+|\mathbf{C}|\mathbf{C}^{-1}t)/\epsilon}
\left(\mathbb{I}-\frac{\mathbf{N}_n^\sigma}{\lambda-\lambda_n}\right)e^{i\lambda(\mathbf{C}x+|\mathbf{C}|\mathbf{C}^{-1}t)/\epsilon},\; \lambda\in D_n,\; n\in\mathcal{N},
\end{equation}
preserving Schwarz symmetry by setting 
\begin{equation}
\widetilde{\mathbf{M}}^\sigma(\lambda):=\widetilde{\mathbf{M}}^\sigma(\lambda^*)^{-\dagger},\quad \lambda\in D_n^*,\quad n\in\mathcal{N},
\end{equation}
(here we used the fact that the stated conditions on $\gamma^{[k]}$, $k=1,2,3$, guarantee that $\mathbf{E}=\mathbb{I}$ or $\mathbf{E}=-\mathbb{I}$), and for all $\lambda\in\mathbb{C}\setminus\mathbb{R}$ exterior to all disks $D_n$ or $D_n^*$, we simply take $\widetilde{\mathbf{M}}^\sigma(\lambda):=\mathbf{M}^\sigma(\lambda)$.  By a simple calculation using the fact that $\mathbf{N}_n^\sigma\mathbf{D}\mathbf{N}_n^\sigma=\mathbf{0}$ for every diagonal matrix $\mathbf{D}$, it follows easily that $\widetilde{\mathbf{M}}^\sigma(\lambda)$ has removable singularities at all of the poles of $\mathbf{M}^\sigma(\lambda)$ in $\mathbb{C}\setminus\mathbb{R}$.  However, $\widetilde{\mathbf{M}}^\sigma(\lambda)$ now has jump discontinuities across the boundaries of all of the disks, circles $\partial D_n$ which we take to be positively-oriented, while $\partial D_n^*$ will be negatively-oriented.  It follows that the jump across $\partial D_n$ is characterized by the jump condition
\begin{equation}
\widetilde{\mathbf{M}}^\sigma_+(\lambda)=\widetilde{\mathbf{M}}^\sigma_-(\lambda)e^{-i\lambda(\mathbf{C}x+|\mathbf{C}|\mathbf{C}^{-1}t)/\epsilon}
\left(\mathbb{I}-\frac{\mathbf{N}_n^\sigma}{\lambda-\lambda_n}\right)e^{i\lambda(\mathbf{C}x+|\mathbf{C}|\mathbf{C}^{-1}t)/\epsilon},\;\lambda\in\partial D_n,\; n\in\mathcal{N},
\label{eq:M-tilde-jump-C-plus-circles}
\end{equation}
and that the jump across $\partial D_n^*$ is given by
\begin{equation}
\widetilde{\mathbf{M}}^\sigma_+(\lambda)=\widetilde{\mathbf{M}}^\sigma_-(\lambda)
e^{-i\lambda(\mathbf{C}x+|\mathbf{C}|\mathbf{C}^{-1}t)/\epsilon}\left(\mathbb{I}-\frac{\mathbf{N}_n^{\sigma\dagger}}{\lambda-\lambda_n^*}\right)e^{i\lambda(\mathbf{C}x+|\mathbf{C}|\mathbf{C}^{-1}t)/\epsilon},\;\lambda\in\partial D_n^*,\;n\in\mathcal{N}.
\label{eq:M-tilde-jump-C-minus-circles}
\end{equation}
In these formulae the subscript ``$+$'' (respectively, ``$-$'') refers to the boundary value taken on the indicated contour from the left (respectively, right) side according to the assigned orientation.
Meanwhile, the jump across the real axis is given simply by \eqref{eq:RHP-jump-condition} with $\mathbf{M}^\sigma$ replaced everywhere by $\widetilde{\mathbf{M}}^\sigma$, because the latter matrices are equal in a deleted neighborhood of $\mathbb{R}$.  (This assumes that there are no spectral singularities; see Appendix~\ref{subapp-Zhou-transform} for how to deal with these.)

The conditions of the equivalent Riemann-Hilbert problem for $\widetilde{\mathbf{M}}^\sigma(\lambda)$ may be translated into a linear system of singular integral equations with Cauchy kernels, and on suitable spaces of boundary values ($L^2(\Sigma)$ or classical H\"older spaces of functions on $\Sigma=\mathbb{R}\cup\{\mathrm{circles}\}$) the relevant singular integral operator is known to be Fredholm with index zero (the value of the index follows from the unimodularity of the jump matrices; see \cite{Zhou89}).
It therefore remains to prove that the kernel is trivial, which is equivalent to ruling out the existence of nonzero solutions of the Riemann-Hilbert problem modified by replacing the normalization condition by $\widetilde{\mathbf{M}}^\sigma(\lambda)\to\mathbf{0}$ as $\lambda\to\infty$.  For Schwarz-symmetric contours $\Sigma$ (with Schwarz-symmetric orientation), Zhou \cite{Zhou89} has proven that no such vanishing solution exists provided that the jump matrix $\widetilde{\mathbf{V}}(\lambda)$, for which $\widetilde{\mathbf{M}}^\sigma_+(\lambda)=\widetilde{\mathbf{M}}^\sigma_-(\lambda)\widetilde{\mathbf{V}}(\lambda)$ holds for each $\lambda\in\Sigma$, has the following properties:
\begin{itemize}
\item $\widetilde{\mathbf{V}}(\lambda^*)=\widetilde{\mathbf{V}}(\lambda)^\dagger$ for $\lambda\in\Sigma\setminus\mathbb{R}$, and
\item $\widetilde{\mathbf{V}}(\lambda)+\widetilde{\mathbf{V}}(\lambda)^\dagger$ is positive definite for $\lambda\in\Sigma\cap\mathbb{R}$.
\end{itemize}
The first property is obviously true as one can see by comparing \eqref{eq:M-tilde-jump-C-plus-circles}--\eqref{eq:M-tilde-jump-C-minus-circles}. For the second property, we note that according to \eqref{eq:TWRI-Jump}, the conditions in force on the signs $\gamma^{[k]}$, $k=1,2,3$, guarantee that for $\lambda\in\mathbb{R}$, $\widetilde{\mathbf{V}}(\lambda)$ has the form $\mathbf{A}(\lambda)^\dagger\mathbf{A}(\lambda)$ with $\det(\mathbf{A}(\lambda))=1$.  This immediately implies the second property.  Hence the Fredholm system has a unique solution, which corresponds to the unique solution of the equivalent Riemann-Hilbert problem for $\widetilde{\mathbf{M}}^\sigma(\lambda)$.  By inverting the relation between $\mathbf{M}^\sigma$ and $\widetilde{\mathbf{M}}^\sigma$ we obtain the existence of a unique solution of the original Riemann-Hilbert problem.
\end{proof}

\subsection{Problems with spectral singularities, higher-order poles and/or infinitely many poles}
\label{subapp-Zhou-transform}
Here we briefly indicate an approach to the inverse-scattering transform due to Zhou \cite{Zhou89b} that allows for a unified treatment of both generic (finitely many simple poles and no real zeros of $u_+(\lambda)$ or $v_+(\lambda)$) and nongeneric scattering data in a simple way.  The basic idea is very simple.  As mentioned in Appendix~\ref{subapp:direct-scattering1}, given initial data encoded in a matrix $\mathbf{Q}(x)$ with $\mathbf{Q}$ and $\mathbf{Q}'$ in $L^1(\mathbb{R})$, the Fredholm equation \eqref{eq:Fredholm-system} governing $\mathbf{M}^\sigma(x;\lambda)$ has a unique solution for $\imag\{\lambda\}\neq 0$ and $|\lambda|$ sufficiently large in the form of a convergent Neumann series, i.e., in this situation \eqref{eq:Fredholm-system} becomes a small-norm problem.  The solution obtained is obviously analytic in the two domains $\mathbb{C}_\pm^\mathrm{out}(R):=\{\lambda\in \mathbb{C}_\pm:|\lambda|>R\}$ for $R$ sufficiently large.  We will now indicate how to obtain \emph{analytic} solutions of the differential equation \eqref{eq:Mij-system} for $\lambda$ in the complementary domains $\mathbb{C}_\pm^\mathrm{in}(R):=\{\lambda\in\mathbb{C}_\pm: |\lambda|<R\}$ in such a way that the jump discontinuity across $|\lambda|=R$ takes a convenient form.  As described in \cite{Zhou89b}, the only thing we need to give up for $\lambda\in \mathbb{C}_\pm^\mathrm{in}(R)$ is the condition that $\mathbf{M}^\sigma(x;\lambda)$ should remain bounded as $x\to -\sigma\infty$; we will retain the condition that $\mathbf{M}^\sigma(x;\lambda)\to\mathbb{I}$ as $x\to\sigma\infty$, as this will ensure a simple time dependence of the scattering data to be introduced.  Thus, instead of using the Fredholm equation \eqref{eq:Fredholm-system} to define $\mathbf{M}^\sigma(x;\lambda)$ when $\lambda\in\mathbb{C}_\pm^\mathrm{in}(R)$, we proceed as follows.

Let $B^\sigma(x)$ be a smooth ``bump'' function with the properties that $0\le B^\sigma(x)\le 1$ and 
$B^\sigma(x)\equiv 0$ for $\sigma x<L-1$ while $B^\sigma(x)\equiv 1$ for $\sigma x>L$.  The ``cutoff'' potential
\begin{equation}
\mathbf{Q}^\sigma_\mathrm{c}(x):=B^\sigma(x)\mathbf{Q}(x)
\label{eq:cutoff-potential}
\end{equation}
than satisfies $\mathbf{Q}^\sigma_\mathrm{c},\mathbf{Q}^{\sigma\prime}_\mathrm{c}\in L^1(\mathbb{R})$, and $\mathbf{Q}^\sigma_\mathrm{c}(x)$ agrees exactly with $\mathbf{Q}(x)$ for $\sigma x>L$.
Because $\mathbf{Q}$ is in $L^1(\mathbb{R})$, we now choose $L>0$ so large that $\|\mathbf{Q}_\mathrm{c}^\sigma\|_{L^1(\mathbb{R})}<\epsilon$.  With this choice, the Fredholm equation \eqref{eq:Fredholm-system}
with $\mathbf{Q}$ replaced by the cutoff potential $\mathbf{Q}_\mathrm{c}^\sigma$ has a unique solution as a Neumann series that is analytic for $\imag\{\lambda\}\neq 0$, and in particular for $\lambda\in\mathbb{C}_\pm^\mathrm{in}(R)$.  Let this solution be denoted $\mathbf{M}_\mathrm{c}^\sigma(x;\lambda)$.  Clearly we have $\mathbf{M}_\mathrm{c}^\sigma(x;\lambda)\to\mathbb{I}$ as $x\to\sigma\infty$.

With $\mathbf{M}^\sigma_\mathrm{c}(x;\lambda)$ defined as an analytic function for $\lambda$ in $\mathbf{C}_\pm^\mathrm{in}(R)$, we notice that $\Phi^\sigma_\mathrm{c}(x;\lambda):=\mathbf{M}^\sigma_\mathrm{c}(x;\lambda)e^{-i\lambda\mathbf{C}x/\epsilon}$ is a solution of the Lax pair equation \eqref{eq:LaxPair-x} for the original potential $\mathbf{Q}(x)$ over the interval $\sigma x>L$, but not for $\sigma x<L$.  However, if we let $\Psi(x;\lambda)$ denote the fundamental solution matrix for \eqref{eq:LaxPair-x} normalized by the initial condition $\Psi(\sigma L;\lambda)=\mathbb{I}$, then for each finite $x\in\mathbb{R}$, $\Psi(x;\lambda)$ is an entire function of $\lambda$.  The product 
\begin{equation}
\Phi^\sigma(x;\lambda):=\Psi(x;\lambda)\Phi^\sigma_\mathrm{c}(\sigma L;\lambda)
\label{eq:product-formula-Zhou}
\end{equation}
is therefore the unique solution of \eqref{eq:LaxPair-x} defined for all $x\in\mathbb{R}$ that agrees with $\Phi^\sigma_\mathrm{c}(x;\lambda)$ for $\sigma x>L$.  The corresponding matrix function
\begin{equation}
\mathbf{M}^\sigma(x;\lambda):=\Phi^\sigma(x;\lambda)e^{i\lambda\mathbf{C}x/\epsilon}=
\Psi(x;\lambda)\mathbf{M}^\sigma_\mathrm{c}(\sigma L;\lambda)e^{i\lambda\mathbf{C}(x-\sigma L)/\epsilon}
\label{eq:product-formula-Zhou-M}
\end{equation}
is then a solution of the differential equations \eqref{eq:Mij-system} for all $x\in\mathbb{R}$ with the following additional properties:
\begin{itemize}
\item Because $\mathbf{M}^\sigma(x;\lambda)$ agrees with $\mathbf{M}^\sigma_\mathrm{c}(x;\lambda)$
for $\sigma x>L$, we have $\mathbf{M}^\sigma(x;\lambda)\to\mathbb{I}$ as $x\to\sigma\infty$.
\item Because according to \eqref{eq:product-formula-Zhou-M}, $\mathbf{M}^\sigma(x;\lambda)$ is a product of factors that are entire functions of $\lambda$ and a central factor that is analytic for $\imag\{\lambda\}\neq 0$, $\mathbf{M}^\sigma(x;\lambda)$ is analytic for $\imag\{\lambda\}\neq 0$.
\end{itemize}
Note however, that for given $\lambda\in\mathbb{C}_\pm$, $\mathbf{M}^\sigma(x;\lambda)$ is not generally bounded as $x\to -\sigma\infty$.  Also $\mathbf{M}^\sigma(x;\lambda)$ given by \eqref{eq:product-formula-Zhou-M} does not tend to $\mathbb{I}$ as $\lambda\to\infty$.  For the latter reason, we agree to use this definition of $\mathbf{M}^\sigma(x;\lambda)$ only for $\lambda\in\mathbb{C}_\pm^\mathrm{in}(R)$.

The matrix $\mathbf{M}^\sigma(x;\lambda)$ defined as before for $\lambda\in\mathbb{C}_\pm^\mathrm{out}(R)$ and by the above modified procedure for $\lambda\in\mathbb{C}_\pm^\mathrm{in}(R)$ is therefore analytic for $\lambda\in\mathbb{C}\setminus\Sigma$, where the contour $\Sigma$ consists of the real axis and the circle of radius $R$ centered at the origin.  Since the definition for $|\lambda|>R$ is the original one, we retain the property that $\mathbf{M}^\sigma(x;\lambda)\to\mathbb{I}$ as $\lambda\to\infty$.  To formulate the appropriate Riemann-Hilbert problem of inverse scattering, it therefore only remains to determine the jump conditions across $\Sigma$.  Clearly on the part of $\Sigma$ with $|\lambda|>R$ the jump condition is exactly as before; see \eqref{eq:M-sigma-jump-intro}, \eqref{eq:jump-evolution-intro}, and \eqref{eq:TWRI-Jump}.  We next find the form of the jump matrix on the real interval $-R<\lambda<R$ as well as the upper and lower semicircles.

For $-R<\lambda<R$, we start with the observation that, as $\mathbf{M}^\sigma_\mathrm{c}(x;\lambda)$ is a solution of the Fredholm equation \eqref{eq:Fredholm-system} that is analytic in $\mathbb{C}_\pm^\mathrm{in}(R)$ and (by the small-norm condition ensuring uniform convergence of the Neumann series up to the real axis) takes continuous boundary values on $-R<\lambda<R$, it satisfies a jump condition analogous to \eqref{eq:M-sigma-jump-intro} with a jump matrix $\mathbf{V}_{\mathrm{c}0}^\sigma(\lambda)$ obtained from the cutoff potential \eqref{eq:cutoff-potential} in exactly the same way as described previously for the original potential $\mathbf{Q}(x)$.  Thus we have
\begin{equation}
\mathbf{M}^\sigma_{\mathrm{c}+}(x;\lambda)=\mathbf{M}^\sigma_{\mathrm{c}-}(x;\lambda)e^{-i\lambda\mathbf{C}x/\epsilon}\mathbf{V}_{\mathrm{c}0}^\sigma(\lambda)e^{i\lambda\mathbf{C}x/\epsilon},\quad
-R<\lambda<R.
\end{equation}
Using this result for $x=\sigma L$ and taking into account that the factors $\Psi(x;\lambda)$ and $e^{i\lambda\mathbf{C}(x-\sigma L)/\epsilon}$ in \eqref{eq:product-formula-Zhou-M} are entire in $\lambda$, 
we obtain the formula
\begin{equation}
\mathbf{M}^\sigma_+(x;\lambda)=\mathbf{M}^\sigma_-(x;\lambda)e^{-i\lambda\mathbf{C}x/\epsilon}\mathbf{V}_{\mathrm{c}0}^\sigma(\lambda)e^{i\lambda\mathbf{C}x/\epsilon},\quad -R<\lambda<R.
\end{equation}
It is easy to see that the matrix $\mathbf{V}_{\mathrm{c}0}^\sigma(\lambda)$ evolves in time exactly as does $\mathbf{V}_0^\sigma(\lambda)$, namely by the explicit conjugation \eqref{eq:jump-evolution-intro}.

For the jump across the upper and lower semicircles, observe that since for both $\lambda\in\mathbb{C}_\pm^\mathrm{out}(R)$ and $\lambda\in\mathbb{C}_\pm^\mathrm{in}(R)$ the matrix $\mathbf{\Phi}^\sigma(x;\lambda)=\mathbf{M}^\sigma(x;\lambda)e^{-i\lambda\mathbf{C}x/\epsilon}$ satisfies the same Lax equation \eqref{eq:LaxPair-x}, a jump condition of the basic form \eqref{eq:M-sigma-jump-intro} holds for some matrix $\mathbf{V}_0^\sigma(\lambda)$ independent of $x$, where the subscript ``$+$'' (respectively, ``$-$'') refers to the boundary value taken on the semicircle from the left (respectively, right), and we choose the semicircles to be oriented from $\lambda=R$ back to $\lambda=-R$. We just have to determine the structure and time dependence of the matrices $\mathbf{V}_0^\sigma(\lambda)$ for the upper and lower semicircles. 
The structure of the matrices comes from noting the agreement of the asymptotic behavior of the two boundary values as $x\to \sigma\infty$.  This actually implies that 
\begin{itemize}
\item For the upper semicircle, $\mathbf{V}_0^+(\lambda)$ is lower triangular and $\mathbf{V}_0^-(\lambda)$ is upper triangular, with ones on the diagonal in both cases.
\item For the lower semicircle, $\mathbf{V}_0^+(\lambda)$ is upper triangular and $\mathbf{V}_0^-(\lambda)$ is lower triangular, with ones on the diagonal in both cases.
\end{itemize}
Moreover, the off-diagonal entries in the jump matrix are analytic functions of $\lambda$ on the two semicircles, which stems from the fact that the radius $R$ is somewhat arbitrary in this construction provided it is sufficiently large to contain all of the isolated exceptional points of the solution of the Fredholm equation \eqref{eq:Fredholm-system}.  Finally, by a modification of the arguments in Appendix~\ref{subapp:time-dependence} it can be shown that the jump matrices on the upper and lower semicircles
evolve in time $t$ by exactly the standard conjugation formula \eqref{eq:jump-evolution-intro}.  This part of the argument uses the fact that both boundary values tend to $\mathbb{I}$ as $x\to\sigma\infty$, by construction.  

We thus arrive at a Riemann-Hilbert problem of inverse scattering whose unknown has no isolated singularities at all, but rather is piecewise analytic in four complementary domains of the complex plane, taking very nice boundary values related by well-defined\footnote{Actually, the jump matrices depend on a number of rather arbitrary choices like the precise nature of how the cutoff potential is constructed, etc., and one can try to ``mod out'' these ambiguities to properly define unambiguous scattering data; see \cite{Zhou89b} for details.  However if one is only interested in formulating a suitable inverse problem from which to construct the solution of the Cauchy problem, this ambiguity is not much of an issue.} jump conditions across the boundary arcs in which $x$ and $t$ appear explicitly by exponential conjugation, and normalized to the identity as $\lambda\to\infty$.  Since $\mathbf{M}^\sigma(x;\lambda)$ is the same outside the circle of radius $R$ in both the original approach and this version, the solution of the Cauchy problem is extracted from the residue term in the Laurent expansion of $\mathbf{M}^\sigma(x;\lambda)$ about $\lambda=\infty$ in the usual way.  The best part about this construction is that it does not require any a priori knowledge of the isolated singularities of the Fredholm equation \eqref{eq:Fredholm-system} in $\mathbb{C}_\pm$ or possible real zeros of $u_+(\lambda)$ or $v_+(\lambda)$.
The nature of these can be quite severe even for very ``nice'' non-generic potentials; for example in \cite[Example 3.3.16]{Zhou89b} it is shown that there exist Schwartz-class potentials $\mathbf{Q}\in\mathscr{S}(\mathbb{R})$ for which there are infinitely many isolated exceptional points for \eqref{eq:Fredholm-system} in $\mathbb{C}_\pm$ that necessarily accumulate at severe spectral singularities on the real axis, i.e., zeros of infinite order for the boundary values $u_+(\lambda)$ and/or $v_+(\lambda)$.  As unified as Zhou's approach to inverse-scattering is, it obscures somewhat the presence of solitons generated from the poles that have been removed from the domain $|\lambda|<R$; see \S\ref{sec:solitons}.  In Zhou's approach the solitons are instead encoded in the jump matrices on the semicircular arcs of $\Sigma$.  Upon analytic/meromorphic continuation of these jumps toward the real axis one discovers the singularities of the solution of \eqref{eq:Fredholm-system} lurking within the circle of radius $R$.

\subsection{Reflectionless potentials and solitons}
\label{sec:solitons}
Potentials $\{q^{[k]}(x)\}_{k=1}^3$ for which the scattering matrix 
$\mathbf{S}(\lambda)$ associated to the Jost solutions is diagonal for all 
$\lambda\in\mathbb{R}$ (and hence the jump matrix in Riemann-Hilbert Problem~\ref{rhp:M-sigma} is the identity) are called \emph{reflectionless potentials}.  In particular, 
semiclassical soliton ensembles are reflectionless potentials (see Definition~\ref{def:SSE}).  
In Appendix~\ref{subsubsec:solitons-rhp}, we show how Riemann-Hilbert 
Problem~\ref{rhp:M-sigma} can be reduced to a problem of finite-dimensional linear 
algebra with the use of partial-fraction expansions for reflectionless potentials.  
In Appendix~\ref{subsubsec:solitons-onepole} we give explicit formulae for the single 
solitons of types $1$, $2$, $3$, $\solsplit$, and $\solfuse$.  In 
Appendix~\ref{subsubsec:solitons-splitfuse} we elaborate on some basic properties of 
solitons of type $\solsplit$ and $\solfuse$.  Finally, in Appendix~\ref{subsubsec:solitons-merging}
we consider the nonlinear superposition of one pole each of types $1$ and $3$ and 
show how it can degenerate to any of the five elementary solitons as its parameters 
vary.

\subsubsection{Solution of the reflectionless Riemann-Hilbert problem}
\label{subsubsec:solitons-rhp}
It is useful to distinguish the poles $\lambda\in P$ according to the various forms that the $2$-nilpotent triangular matrices $\mathbf{N}^\sigma$ can take, and we adopt the terminology explained 
in \S\ref{subsec:lax-pair}, referring to each simple pole in $P$ as being of type $1$, $2$, $3$, $\solsplit$, or $\solfuse$ based on the shape of its residue matrix $\mathbf{N}^\sigma$.  
We consider the total number of simple poles in $P\subset\mathbb{C}_+$ to therefore be partitioned as follows:  $N=N^{[1]}+N^{[2]}+N^{[3]}+N^{[\solsplit]}+N^{[\solfuse]}$ with poles $\lambda_n^{[\mathrm{type}]}\in P$, $n=0,\dots,N^{[\mathrm{type}]}-1$, of type ``type'' having corresponding nonzero connection coefficients $\beta_{n,jk}^{[\mathrm{type}]}$.

Indeed, 
in the reflectionless situation $\mathbf{M}^+(\lambda)=\mathbf{M}^+(x,t;\lambda)$ can be represented by its partial fraction expansion in the form 
\begin{equation}
\label{eq:Mplus-ansatz}
\begin{split}
\mathbf{M}^+(\lambda) = \mathbb{I}&
+\sum_{n=0}^{N^{[1]}-1}\left[\frac{(\mathbf{0},\mathbf{b}_n^{[1]+},\mathbf{0})}{\lambda-\lambda_n^{[1]}}+\frac{(\mathbf{0},\mathbf{0},\mathbf{c}_n^{[1]+})}{\lambda-\lambda_n^{[1]*}}\right]
+ \sum_{n=0}^{N^{[2]}-1}\left[\frac{(\mathbf{a}_n^{[2]+},\mathbf{0},\mathbf{0})}{\lambda-\lambda_n^{[2]}}+\frac{(\mathbf{0},\mathbf{0},\mathbf{c}_n^{[2]+})}{\lambda-\lambda_n^{[2]*}}\right] \\
  & + \sum_{n=0}^{N^{[3]}-1}\left[\frac{(\mathbf{a}_n^{[3]+},\mathbf{0},\mathbf{0})}{\lambda-\lambda_n^{[3]}}+\frac{(\mathbf{0},\mathbf{b}_n^{[3]+},\mathbf{0})}{\lambda-\lambda_n^{[3]*}}\right] 
 + \sum_{n=0}^{N^{[\solsplit]}-1}
\left[\frac{(\mathbf{a}_n^{[\solsplit]+},\mathbf{b}_n^{[\solsplit]+},\mathbf{0})}{\lambda-\lambda_n^{[\solsplit]}}+
\frac{(\mathbf{0},\mathbf{0},\mathbf{c}_n^{[\solsplit]+})}{\lambda-\lambda_n^{[\solsplit]*}}\right] \\
  &+ \sum_{n=0}^{N^{[\solfuse]}-1}\left[\frac{(\mathbf{a}_n^{[\solfuse]+},\mathbf{0},\mathbf{0})}{\lambda-\lambda_n^{[\solfuse]}} + 
\frac{(\mathbf{0},\mathbf{b}_n^{[\solfuse]+},\mathbf{c}_n^{[\solfuse]+})}{\lambda-\lambda_n^{[\solfuse]*}}\right],
\end{split}
\end{equation}
and $\mathbf{M}^-(\lambda)=\mathbf{M}^-(x,t;\lambda)$ can be represented in the form
\begin{equation}
\begin{split}
\mathbf{M}^-(\lambda)=\mathbb{I} & + \sum_{n=0}^{N^{[1]}-1}\left[
\frac{(\mathbf{0},\mathbf{0},\mathbf{c}_n^{[1]-})}{\lambda-\lambda_n^{[1]}}+
\frac{(\mathbf{0},\mathbf{b}_n^{[1]-},\mathbf{0})}{\lambda-\lambda_n^{[1]*}}\right]
+\sum_{n=0}^{N^{[2]}-1}\left[
\frac{(\mathbf{0},\mathbf{0},\mathbf{c}_n^{[2]-})}{\lambda-\lambda_n^{[2]}}+
\frac{(\mathbf{a}_n^{[2]-},\mathbf{0},\mathbf{0})}{\lambda-\lambda_n^{[2]*}}\right]\\
& + \sum_{n=0}^{N^{[3]}-1}\left[\frac{(\mathbf{0},\mathbf{b}_n^{[3]-},\mathbf{0})}{\lambda-\lambda_n^{[3]}} + \frac{(\mathbf{a}_n^{[3]-},\mathbf{0},\mathbf{0})}{\lambda-\lambda_n^{[3]*}}\right]
+\sum_{n=0}^{N^{[\solsplit]}-1}\left[
\frac{(\mathbf{0},\mathbf{b}_n^{[\solsplit]-},\mathbf{c}_n^{[\solsplit]-})}{\lambda-\lambda_n^{[\solsplit]}} +
\frac{(\mathbf{a}_n^{[\solsplit]-},\mathbf{0},\mathbf{0})}{\lambda-\lambda_n^{[\solsplit]*}}\right]\\
& + \sum_{n=0}^{N^{[\solfuse]}-1}\left[
\frac{(\mathbf{0},\mathbf{0},\mathbf{c}_n^{[\solfuse]-})}{\lambda-\lambda_n^{[\solfuse]}}+
\frac{(\mathbf{a}_n^{[\solfuse]-},\mathbf{b}_n^{[\solfuse]-},\mathbf{0})}{\lambda-\lambda_n^{[\solfuse]*}}\right].
\end{split}
\end{equation}
Here, the vector coefficients $\mathbf{a}_n^{[\mathrm{type}]\sigma}$, $\mathbf{b}_n^{[\mathrm{type}]\sigma}$, and $\mathbf{c}_n^{[\mathrm{type}]\sigma}$ are functions of $x$ and $t$ to be determined.
These representations are obtained by imposing the normalization conditions $\mathbf{M}^\sigma(\lambda)\to\mathbb{I}$ as $\lambda\to\infty$ and by setting to $\mathbf{0}$ the residues of those columns of $\mathbf{M}^\sigma(\lambda)$ that are known to be analytic because the corresponding column of the residue matrix $\mathbf{N}^\sigma$ (for a pole at $\lambda_n^{[\mathrm{type}]}$) or $-\mathbf{E}\mathbf{N}^{\sigma\dagger}\mathbf{E}$ (for a pole at $\lambda_n^{[\mathrm{type}]*}$) vanishes.

The constraints imposed by \eqref{eq:RHP-pole-C-plus}--\eqref{eq:RHP-pole-C-minus}
 on the remaining columns of the residue then yield a square inhomogeneous system 
of linear equations on the unknown vector coefficients.\footnote{From this point we 
specialize to $\sigma=+$.  The procedure for $\sigma=-$ is analogous.  While we 
do not write all the details, we note that the equations for $\sigma=-$ are 
important in their own right;  in particular our experience is that numerical 
calculations are more stable using the equations for $\sigma=-$ for $x$-values 
where there are more packets to the right than to the left.}
If we define the set of pole types
\eq
\mathcal{T}:=\{1,2,3,\solsplit,\solfuse\},
\endeq
then these equations are the following:  
\begin{multline}
	\mathbf{b}_n^{[1]+} = 
 	\beta_{n,32}^{[1]}e^{i\lambda_n^{[1]}\Delta^{[1]}(x-c^{[1]}t)/\epsilon}
 	\left[
	\mathbf{e}^3 + \sum_{\substack{k \in \mathcal{T} \\ k \neq 3} }
	\sum_{m=0}^{N^{[k]}-1}\frac{\mathbf{c}_m^{[k]+}}{\lambda_n^{[1]}-\lambda_m^{[k]*}}
	\right],
	\quad
	n=0,\dots,N^{[1]}-1,
\label{eq:b-plus-1}
\end{multline}

\begin{multline}
	\mathbf{c}_n^{[1]+} =
	\gamma^{[2]}\gamma^{[3]}\beta_{n,32}^{[1]*}
	e^{-i\lambda_n^{[1]*}\Delta^{[1]}(x-c^{[1]}t)/\epsilon}
	\left[
	\vphantom{
	\sum_{k \in \{1,\solsplit \}} \!\! \sum_{m=0}^{N^{[k]}-1}
	\frac{\mathbf{b}_m^{[k]+}}{\lambda_n^{[1]*}-\lambda_m^{[k]}} 
	+ \sum_{k \in \{3,\solfuse \} } \!\! \sum_{m=0}^{N^{[k]}-1}
	\frac{\mathbf{b}_m^{[k]+}}{\lambda_n^{[1]*}-\lambda_m^{[k]*}}
	}
	\mathbf{e}^2 + \right. \\
	\left.
	\sum_{k \in \{1,\solsplit \}} \!\! \sum_{m=0}^{N^{[k]}-1}
	\frac{\mathbf{b}_m^{[k]+}}{\lambda_n^{[1]*}-\lambda_m^{[k]}} 
	+ \sum_{k \in \{3,\solfuse \} } \!\! \sum_{m=0}^{N^{[k]}-1}
	\frac{\mathbf{b}_m^{[k]+}}{\lambda_n^{[1]*}-\lambda_m^{[k]*}}
\right],\quad
n=0,\dots,N^{[1]}-1,
\label{eq:c-plus-1}
\end{multline}

\begin{multline}
\mathbf{a}_n^{[2]+}=\beta_{n,31}^{[2]}e^{i\lambda_n^{[2]}\Delta^{[2]}(x-c^{[2]}t)/\epsilon}\left[
\mathbf{e}^3+ \sum_{\substack{k\in\mathcal{T} \\ k\neq 3}} \sum_{m=0}^{N^{[k]}-1}\frac{\mathbf{c}_m^{[k]+}}{\lambda_n^{[2]}-\lambda_m^{[k]*}} \right],\quad
n=0,\dots,N^{[2]}-1,
\label{eq:a-plus-2}
\end{multline}

\begin{multline}
	\mathbf{c}_n^{[2]+}=-\gamma^{[1]}\gamma^{[3]}\beta_{n,31}^{[2]*}
	e^{-i\lambda_n^{[2]*}	\Delta^{[2]}(x-c^{[2]}t)/\epsilon}
	\left[
	\mathbf{e}^1 + \sum_{\substack{k\in\mathcal{T} \\ k\neq 1}} 
	\sum_{m=0}^{N^{[k]}-1}
	\frac{\mathbf{a}_m^{[k]+}}{\lambda_n^{[2]*}-\lambda_m^{[k]}} \right], \\
	n=0,\dots,N^{[2]}-1,
	\label{eq:c-plus-2}
\end{multline}

\begin{multline}
	\mathbf{a}_n^{[3]+} =
	\beta_{n,21}^{[3]}e^{i\lambda_n^{[3]}\Delta^{[3]}(x-c^{[3]}t)/\epsilon}
	\left[
	\vphantom{
	\sum_{k \in \{1,\solsplit \}} \!\! \sum_{m=0}^{N^{[k]}-1}
	\frac{\mathbf{b}_m^{[k]+}}{\lambda_n^{[3]}-\lambda_m^{[k]}} +
	\sum_{k \in \{3,\solfuse \}} \!\! \sum_{m=0}^{N^{[k]}-1}\frac{\mathbf{b}_m^{[k]+}}
	{\lambda_n^{[3]}-\lambda_m^{[k]*}}
	}
	\mathbf{e}^2 + \right. \\
	\left.
	\sum_{k \in \{1,\solsplit \}} \!\! \sum_{m=0}^{N^{[k]}-1}
	\frac{\mathbf{b}_m^{[k]+}}{\lambda_n^{[3]}-\lambda_m^{[k]}} +
	\sum_{k \in \{3,\solfuse \}} \!\! \sum_{m=0}^{N^{[k]}-1}\frac{\mathbf{b}_m^{[k]+}}
	{\lambda_n^{[3]}-\lambda_m^{[k]*}}
	\right]	
\quad
n=0,\dots,N^{[3]}-1,
\label{eq:a-plus-3}
\end{multline}

\begin{multline}
	\mathbf{b}_n^{[3]+}=\gamma^{[1]}\gamma^{[2]}\beta_{n,21}^{[3]*}
	e^{-i\lambda_n^{[3]*}\Delta^{[3]}(x-c^{[3]}t)/\epsilon}
	\left[
	\mathbf{e}^1+ \sum_{\substack{k\in\mathcal{T} \\ k\neq 1}}  
	\sum_{m=0}^{N^{[k]}-1}\frac{\mathbf{a}_m^{[k]+}}{\lambda_n^{[3]*}-\lambda_m^{[k]}}
	\right], \\
	n=0,\dots,N^{[3]}-1,
	\label{eq:b-plus-3}
\end{multline}

\begin{multline}
	\mathbf{a}_n^{[\solsplit]+}=\beta_{n,31}^{[\solsplit]}
	e^{i\lambda_n^{[\solsplit]}\Delta^{[2]}(x-c^{[2]}t)/\epsilon}
	\left[
	\mathbf{e}^3 + \sum_{\substack{k\in\mathcal{T} \\ k\neq 3}} \sum_{m=0}^{N^{[k]}-1}
	\frac{\mathbf{c}_m^{[k]+}}{\lambda_n^{[\solsplit]}-\lambda_m^{[k]*}} 
	\right], \quad
	n=0,\dots,N^{[\solsplit]}-1,
\label{eq:a-plus-12}
\end{multline}

\begin{multline}
	\mathbf{b}_n^{[\solsplit]+}=\beta_{n,32}^{[\solsplit]}
	e^{i\lambda_n^{[\solsplit]}\Delta^{[1]}(x-c^{[1]}t)/\epsilon}
	\left[
	\mathbf{e}^3 + \sum_{\substack{k\in\mathcal{T} \\ k\neq 3}}  \sum_{m=0}^{N^{[k]}-1}
	\frac{\mathbf{c}_m^{[k]+}}{\lambda_n^{[\solsplit]}-\lambda_m^{[k]*}} 
	\right], \quad
	n=0,\dots,N^{[\solsplit]}-1,
\label{eq:b-plus-12}
\end{multline}

\begin{multline}
	\mathbf{c}_n^{[\solsplit]+} =-\gamma^{[1]}\gamma^{[3]}\beta_{n,31}^{[\solsplit]*}
	e^{-i\lambda_n^{[\solsplit]*}\Delta^{[2]}(x-c^{[2]}t)/\epsilon}
	\left[
	\mathbf{e}^1 + \sum_{\substack{k\in\mathcal{T} \\ k\neq 1}}\sum_{m=0}^{N^{[k]}-1}
	\frac{\mathbf{a}_m^{[k]+}}{\lambda_n^{[\solsplit]*}-\lambda_m^{[k]}} 
	\right] \\
	{}+\gamma^{[2]}\gamma^{[3]}\beta_{n,32}^{[\solsplit]*}
	e^{-i\lambda_n^{[\solsplit]*}\Delta^{[1]}(x-c^{[1]}t)/\epsilon}
	\left[ 
	\mathbf{e}^2 + \sum_{k \in \{1, \solsplit\} } \sum_{m=0}^{N^{[k]}-1} 
	\frac{\mathbf{b}_m^{[k]+}}{\lambda_n^{[\solsplit]*}-\lambda_m^{[k]}} +
	\sum_{ k \in \{3, \solfuse \} } \sum_{m=0}^{N^{[k]}-1}
	\frac{\mathbf{b}_m^{[k]+}}{\lambda_n^{[\solsplit]*}-\lambda_m^{[k]*} }
	\right],\\
	\quad
	n=0,\dots,N^{[\solsplit]}-1,
\end{multline}

\begin{multline}
	\mathbf{a}_n^{[\solfuse]+}= \beta_{n,21}^{[\solfuse]}
	e^{i\lambda_n^{[\solfuse]}\Delta^{[3]}(x-c^{[3]}t)/\epsilon}
	\left[
	\mathbf{e}^2+ \sum_{k\in\{1,\solsplit \} } \sum_{m=0}^{N^{[k]}-1}
	\frac{\mathbf{b}_m^{[k]+}}{\lambda_n^{[\solfuse]}-\lambda_m^{[k]}} 
	\right. \\ \left.
	{}+\sum_{k \in \{ 3, \solfuse \} } \sum_{m=0}^{N^{[k]}-1}
	\frac{\mathbf{b}_m^{[k]+}}{\lambda_n^{[\solfuse]}-\lambda_m^{[k]*}}
	\right]
	+\beta_{n,31}^{[\solfuse]}
	e^{i\lambda_n^{[\solfuse]}\Delta^{[2]}(x-c^{[2]}t)/\epsilon}
	\left[ \mathbf{e}^3 +
	\sum_{\substack{k\in\mathcal{T} \\ k\neq 3}}\sum_{m=0}^{N^{[k]}-1}
	\frac{\mathbf{c}_m^{[k]+}}{\lambda_n^{[\solfuse]}-\lambda_m^{[k]*}} 
	\right], \\
	\quad n=0,\dots,N^{[\solfuse]}-1,
\end{multline}

\eq
\begin{split}
\mathbf{b}_n^{[\solfuse]+}=\gamma^{[1]}\gamma^{[2]}\beta_{n,21}^{[\solfuse]*}e^{-i\lambda_n^{[\solfuse]*}\Delta^{[3]}(x-c^{[3]}t)/\epsilon}\left[\mathbf{e}^1 + \sum_{\substack{k\in\mathcal{T} \\ k\neq 1}} \sum_{m=0}^{N^{[k]}-1}\frac{\mathbf{a}_m^{[k]+}}{\lambda_n^{[\solfuse]*}-\lambda_m^{[k]}} \right], \\
n=0,\dots,N^{[\solfuse]}-1,
\label{eq:b-plus-23}
\end{split}
\endeq

\eq
\begin{split}
\mathbf{c}_n^{[\solfuse]+}=-\gamma^{[1]}\gamma^{[3]}\beta_{n,31}^{[\solfuse]*}e^{-i\lambda_n^{[\solfuse]*}\Delta^{[2]}(x-c^{[2]}t)/\epsilon}\left[
\mathbf{e}^1 + \sum_{\substack{k\in\mathcal{T} \\ k\neq 1}} \sum_{m=0}^{N^{[k]}-1}\frac{\mathbf{a}_m^{[k]+}}{\lambda_n^{[\solfuse]*}-\lambda_m^{[k]}} \right], \\
n=0,\dots,N^{[\solfuse]}-1.
\label{eq:c-plus-23}
\end{split}
\endeq

From the solution of this square system of linear equations, the use of  \eqref{eq:M-lambda-infinity-expand} and \eqref{eq:q-potentials-reconstruct} yields the formulas:  
\begin{equation}
\label{eq:q-formulas}
\begin{split}
q^{[1]}(x,t)&=-i\gamma^{[1]}\Delta^{[1]}\sqrt{\Delta^{[2]}\Delta^{[3]}} \sum_{\substack{k\in\mathcal{T} \\ k\neq 3}}\sum_{n=0}^{N^{[k]}-1}c^{[k]+}_{n,2}(x,t), \\ 
q^{[2]}(x,t)&=i\gamma^{[2]}\Delta^{[2]}\sqrt{\Delta^{[3]}\Delta^{[1]}} \sum_{\substack{k\in\mathcal{T} \\ k\neq 1}} \sum_{n=0}^{N^{[k]}-1}a^{[k]+}_{n,3}(x,t), \\
q^{[3]}(x,t)&=-i\gamma^{[3]}\Delta^{[3]}\sqrt{\Delta^{[1]}\Delta^{[2]}} \sum_{\substack{k\in\mathcal{T} \\ k\neq 2}} \sum_{n=0}^{N^{[k]}-1}b^{[k]+}_{n,1}(x,t).
\end{split}
\end{equation}
Situations in which there are only simple poles of a fixed type $k$, $k=1,2,3$,  yield trivial solutions of the TWRI system \eqref{3wave} in the sense that only the field $q^{[k]}(x,t)$ is nonzero and hence propagates in time by translation with velocity $c^{[k]}$.  For example, if $N^{[2]}=N^{[3]}=N^{[\solsplit]}=N^{[\solfuse]}=0$, then obviously we have $q^{[2]}(x,t)\equiv 0$ as the corresponding sums in \eqref{eq:q-formulas} are all empty, and moreover only equations \eqref{eq:b-plus-1} and \eqref{eq:c-plus-1} are relevant. They take the simplified form
\begin{equation}
\begin{split}
\mathbf{b}_n^{[1]+}&=\beta_{n,32}^{[1]}e^{i\lambda_n^{[1]}\xi^{[1]}/\epsilon}\left[\mathbf{e}^3 +\sum_{m=0}^{N^{[1]}-1}\frac{\mathbf{c}_m^{[1]+}}{\lambda_n^{[1]}-\lambda_m^{[1]*}}\right],\\
\mathbf{c}_n^{[1]+}&=\gamma^{[2]}\gamma^{[3]}\beta_{n,32}^{[1]*}e^{-i\lambda_n^{[1]*}\xi^{[1]}/\epsilon}
\left[\mathbf{e}^2 +\sum_{m=0}^{N^{[1]}-1}\frac{\mathbf{b}_m^{[1]+}}{\lambda_n^{[1]*}-\lambda_m^{[1]}}\right],\quad n=0,\dots,N^{[1]}-1,
\end{split}
\end{equation}
where $\xi^{[1]}:=\Delta^{[1]}(x-c^{[1]}t)$ is a traveling-wave variable with velocity $c^{[1]}$.  As the first component of the forcing terms always vanishes, the same is true of $\mathbf{b}_n^{[1]+}$ and $\mathbf{c}_n^{[1]+}$ provided the determinant of the linear system is nonzero, from which it follows that $q^{[3]}(x,t)\equiv 0$, and also that $q^{[1]}$ depends on $(x,t)$ only via the combination $\xi^{[1]}$.  The profile of this traveling wave solution can be arbitrarily complicated by choosing $N^{[1]}$ sufficiently large.

\subsubsection{Exact solutions corresponding to one pole}
\label{subsubsec:solitons-onepole}
Reflectionless solutions of the TWRI equations \eqref{3wave} for which $P$ contains only one simple pole $\lambda^{[\mathrm{type}]}$ are the elementary solitons of the system.  They are the following.
\begin{itemize}
\item The type $1$ soliton is generated from the data $\lambda^{[1]}=a+ib$ with $a\in\mathbb{R}$ and $b>0$, and using parameters $x_1\in\mathbb{R}$ and $\phi_1\in\mathbb{R}$ (mod $2\pi$), the connection coefficient for $\mathbf{M}^+$ is written without loss of generality in the form $\beta_{32}^{[1]}=2ibe^{b\Delta^{[1]}x_1/\epsilon}e^{-ia\Delta^{[1]}x_1/\epsilon}e^{-i\phi_1}$.  Defining a phase variable by $\xi_1:=x-x_1-c^{[1]}t$, the solution is given by
\begin{equation}
q^{[1]}(x,t)=\gamma^{[1]}e^{i\phi_1}
e^{-ia\Delta^{[1]}\xi_1/\epsilon}b\Delta^{[1]}\sqrt{\Delta^{[2]}\Delta^{[3]}}
\begin{cases} \mathrm{sech}(b\Delta^{[1]}\xi_1/\epsilon),&\quad \gamma^{[2]}\gamma^{[3]}=-1,\\
-\mathrm{csch}(b\Delta^{[1]}\xi_1/\epsilon),&\quad \gamma^{[2]}\gamma^{[3]}=1,
\end{cases}
\label{eq:pure-1}
\end{equation}
and then $q^{[2]}(x,t)=q^{[3]}(x,t)=0$.  Clearly for each $t\in\mathbb{R}$ this solution is bounded only if $\gamma^{[2]}\gamma^{[3]}=-1$.
\item The type $2$ soliton is generated from the data $\lambda^{[2]}=a+ib$ with $a\in\mathbb{R}$ and $b>0$, and using parameters $x_2\in\mathbb{R}$ and $\phi_2\in\mathbb{R}$ (mod $2\pi$), the connection coefficient for $\mathbf{M}^+$ is written in the form $\beta_{31}^{[2]}=2ibe^{b\Delta^{[2]}x_2/\epsilon}e^{-ia\Delta^{[2]}x_2/\epsilon}e^{i\phi_2}$.  Defining a phase variable by $\xi_2:=x-x_2-c^{[2]}t$, the solution is given by $q^{[1]}(x,t)=0$, 
\begin{equation}
q^{[2]}(x,t)=
-\gamma^{[2]}e^{i\phi_2}e^{ia\Delta^{[2]}\xi_2/\epsilon}b\Delta^{[2]}\sqrt{\Delta^{[3]}\Delta^{[1]}}\begin{cases}\mathrm{sech}(b\Delta^{[2]}\xi_2/\epsilon),&\quad \gamma^{[1]}\gamma^{[3]}=1,\\
\mathrm{csch}(b\Delta^{[2]}\xi_2/\epsilon),&\quad \gamma^{[1]}\gamma^{[3]}=-1,
\end{cases}
\label{eq:pure-2}
\end{equation}
and $q^{[3]}(x,t)=0$.  Clearly for each $t\in\mathbb{R}$ this solution is bounded only if $\gamma^{[1]}\gamma^{[3]}=1$.
\item The type $3$ soliton is generated from the data $\lambda^{[3]}=a+ib$ with $a\in\mathbb{R}$ and $b>0$, and using parameters $x_3\in\mathbb{R}$ and $\phi_3\in\mathbb{R}$ (mod $2\pi$), the connection coefficient for $\mathbf{M}^+$ is written in the form $\beta_{21}^{[3]}=2ibe^{b\Delta^{[3]}x_3/\epsilon}e^{-ia\Delta^{[3]}x_3/\epsilon}e^{-i\phi_3}$.  Defining a phase variable by $\xi_3:=x-x_3-c^{[3]}t$, the solution is given by $q^{[1]}(x,t)=q^{[2]}(x,t)=0$ and
\begin{equation}
q^{[3]}(x,t)=
\gamma^{[3]}e^{i\phi_3}e^{-ia\Delta^{[3]}\xi_3/\epsilon}b\Delta^{[3]}\sqrt{\Delta^{[1]}\Delta^{[2]}}\begin{cases}
\mathrm{sech}(b\Delta^{[3]}\xi_3/\epsilon),&\quad \gamma^{[1]}\gamma^{[2]}=-1,\\
-\mathrm{csch}(b\Delta^{[3]}\xi_3/\epsilon),&\quad \gamma^{[1]}\gamma^{[2]}=1.
\end{cases}
\label{eq:pure-3}
\end{equation}
Clearly for each $t\in\mathbb{R}$ this solution is bounded only if $\gamma^{[1]}\gamma^{[2]}=-1$.
\item The type $\solsplit$ soliton is generated from the data $\lambda^{[\solsplit]}=a+ib$ with $a\in\mathbb{R}$ and $b>0$, and $\beta_{31}^{[\solsplit]}=2ibe^{b\Delta^{[2]}x_2/\epsilon}e^{-ia\Delta^{[2]}x_2/\epsilon}e^{i\phi_2}$ with $x_2,\phi_2\in\mathbb{R}$ and $\beta_{32}^{[\solsplit]}=2ibe^{b\Delta^{[1]}x_1/\epsilon}e^{-ia\Delta^{[1]}x_1/\epsilon}e^{-i\phi_1}$ with $x_1,\phi_1\in\mathbb{R}$.  We define $x_3$ by the relation 
\begin{equation}
\Delta^{[1]}x_1+\Delta^{[3]}x_3=\Delta^{[2]}x_2
\label{eq:phase-relations}
\end{equation}
and $\phi_3$ (mod $2\pi$) by the relation
\begin{equation}
\phi_1+\phi_2+\phi_3=0,
\label{eq:phi-relation-1}
\end{equation}
and define three phase variables by $\xi_j:=x-x_j-c^{[j]}t$.  Note that \eqref{eq:phase-relations} implies that 
\begin{equation}
\Delta^{[1]}\xi_1+\Delta^{[3]}\xi_3=\Delta^{[2]}\xi_2
\end{equation}
holds for all $(x,t)\in\mathbb{R}^2$.  
Then, 
\begin{multline}
q^{[1]}(x,t)=\gamma^{[1]}e^{i\phi_1}e^{-ia\Delta^{[1]}\xi_1/\epsilon}b\Delta^{[1]}
\sqrt{\Delta^{[2]}\Delta^{[3]}}\\
\cdot\begin{cases}
\left[\cosh(b\Delta^{[1]}\xi_1/\epsilon) - \tfrac{1}{2}\gamma^{[1]}\gamma^{[2]}
e^{-b\Delta^{[2]}\xi_2/\epsilon}e^{-b\Delta^{[3]}\xi_3/\epsilon}\right]^{-1},&\quad
\gamma^{[2]}\gamma^{[3]}=-1,\\
\left[-\sinh(b\Delta^{[1]}\xi_1/\epsilon) - \tfrac{1}{2}\gamma^{[1]}\gamma^{[2]}
e^{-b\Delta^{[2]}\xi_2/\epsilon}e^{-b\Delta^{[3]}\xi_3/\epsilon}\right]^{-1},&\quad
\gamma^{[2]}\gamma^{[3]}=1,
\end{cases}
\label{eq:split-1}
\end{multline}
\begin{multline}
q^{[2]}(x,t)=-\gamma^{[2]}e^{i\phi_2}e^{ia\Delta^{[2]}\xi_2/\epsilon}b\Delta^{[2]}\sqrt{\Delta^{[3]}\Delta^{[1]}}\\
\cdot\begin{cases}
\left[\cosh(b\Delta^{[2]}\xi_2/\epsilon)-\tfrac{1}{2}\gamma^{[2]}\gamma^{[3]}e^{-b\Delta^{[1]}\xi_1/\epsilon}e^{b\Delta^{[3]}\xi_3/\epsilon}\right]^{-1},&\quad\gamma^{[1]}\gamma^{[3]}=1,\\
\left[\sinh(b\Delta^{[2]}\xi_2/\epsilon)-\tfrac{1}{2}\gamma^{[2]}\gamma^{[3]}e^{-b\Delta^{[1]}\xi_1/\epsilon}e^{b\Delta^{[3]}\xi_3/\epsilon}\right]^{-1},&\quad\gamma^{[1]}\gamma^{[3]}=-1,
\end{cases}
\label{eq:split-2}
\end{multline}
and
\begin{multline}
q^{[3]}(x,t)=\gamma^{[3]}e^{i\phi_3}e^{-ia\Delta^{[3]}\xi_3/\epsilon}
b\Delta^{[3]}\sqrt{\Delta^{[1]}\Delta^{[2]}}\\
\cdot\begin{cases}
\left[\cosh(b\Delta^{[3]}\xi_3/\epsilon)+\tfrac{1}{2}\gamma^{[1]}\gamma^{[3]}e^{b\Delta^{[1]}\xi_1/\epsilon}e^{b\Delta^{[2]}\xi_2/\epsilon}\right]^{-1},&\quad\gamma^{[1]}\gamma^{[2]}=-1,\\
\left[-\sinh(b\Delta^{[3]}\xi_3/\epsilon)+\tfrac{1}{2}\gamma^{[1]}\gamma^{[3]}e^{b\Delta^{[1]}\xi_1/\epsilon}e^{b\Delta^{[2]}\xi_2/\epsilon}\right]^{-1},&\quad\gamma^{[1]}\gamma^{[2]}=1.
\end{cases}
\label{eq:split-3}
\end{multline}
\item The type $\solfuse$ soliton is generated from the data $\lambda^{[\solfuse]}=a+ib$ with $a\in\mathbb{R}$ and $b>0$, and $\beta_{21}^{[\solfuse]}=2ibe^{b\Delta^{[3]}x_3/\epsilon}e^{-ia\Delta^{[3]}x_3/\epsilon}e^{-i\phi_3}$ with $x_3,\phi_3\in\mathbb{R}$ and $\beta_{31}^{[\solfuse]}=2ibe^{b\Delta^{[2]}x_2/\epsilon}e^{-ia\Delta^{[2]}x_2/\epsilon}e^{i\phi_2}$ with $x_2,\phi_2\in\mathbb{R}$.  We define $x_1$  by the relation \eqref{eq:phase-relations} and define $\phi_1$ (mod $2\pi)$ by the relation
\begin{equation}
\phi_1+\phi_2+\phi_3=\pi.  
\end{equation}
Then recalling $\xi_j:=x-x_j-c^{[j]}t$,
\begin{multline}
q^{[1]}(x,t)=\gamma^{[1]}e^{i\phi_1}e^{-ia\Delta^{[1]}\xi_1/\epsilon}b\Delta^{[1]}\sqrt{\Delta^{[2]}\Delta^{[3]}}\\
\cdot\begin{cases}
\left[\cosh(b\Delta^{[1]}\xi_1/\epsilon) + \tfrac{1}{2}\gamma^{[1]}\gamma^{[3]}e^{b\Delta^{[2]}\xi_2/\epsilon}e^{b\Delta^{[3]}\xi_3/\epsilon}\right]^{-1},&\quad\gamma^{[2]}\gamma^{[3]}=-1,\\
\left[-\sinh(b\Delta^{[1]}\xi_1/\epsilon) + \tfrac{1}{2}\gamma^{[1]}\gamma^{[3]}e^{b\Delta^{[2]}\xi_2/\epsilon}e^{b\Delta^{[3]}\xi_3/\epsilon}\right]^{-1},&\quad\gamma^{[2]}\gamma^{[3]}=1,
\end{cases}
\label{eq:fuse-1}
\end{multline}
\begin{multline}
q^{[2]}(x,t)=-\gamma^{[2]}e^{i\phi_2}e^{ia\Delta^{[2]}\xi_2/\epsilon}b\Delta^{[2]}\sqrt{\Delta^{[3]}\Delta^{[1]}}\\
\cdot\begin{cases}
\left[\cosh(b\Delta^{[2]}\xi_2/\epsilon)-\tfrac{1}{2}\gamma^{[1]}\gamma^{[2]}e^{b\Delta^{[1]}\xi_1/\epsilon}e^{-b\Delta^{[3]}\xi_3/\epsilon}\right]^{-1},&\quad\gamma^{[1]}\gamma^{[3]}=1,\\
\left[\sinh(b\Delta^{[2]}\xi_2/\epsilon)-\tfrac{1}{2}\gamma^{[1]}\gamma^{[2]}e^{b\Delta^{[1]}\xi_1/\epsilon}e^{-b\Delta^{[3]}\xi_3/\epsilon}\right]^{-1},&\quad\gamma^{[1]}\gamma^{[3]}=-1,
\end{cases}
\label{eq:fuse-2}
\end{multline}
and
\begin{multline}
q^{[3]}(x,t)=\gamma^{[3]}e^{i\phi_3}e^{-ia\Delta^{[3]}\xi_3/\epsilon}b\Delta^{[3]}\sqrt{\Delta^{[1]}\Delta^{[2]}}\\
\cdot\begin{cases}
\left[\cosh(b\Delta^{[3]}\xi_3/\epsilon)-\tfrac{1}{2}\gamma^{[2]}\gamma^{[3]}e^{-b\Delta^{[1]}\xi_1/\epsilon}e^{-b\Delta^{[2]}\xi_2/\epsilon}\right]^{-1},&\quad\gamma^{[1]}\gamma^{[2]}=-1,\\
\left[-\sinh(b\Delta^{[3]}\xi_3/\epsilon)-\tfrac{1}{2}\gamma^{[2]}\gamma^{[3]}e^{-b\Delta^{[1]}\xi_1/\epsilon}e^{-b\Delta^{[2]}\xi_2/\epsilon}\right]^{-1},&\quad\gamma^{[1]}\gamma^{[2]}=1.
\end{cases}
\label{eq:fuse-3}
\end{multline}
\end{itemize}

\subsubsection\protect{\emph{Properties of solitons of type} $\solsplit$ \emph{and} $\solfuse$.}
\label{subsubsec:solitons-splitfuse}
Consider first the soliton of type $\solsplit$.  
The fields $\{q^{[k]}(x,t)\}_{k=1}^3$ are all proportional to $(1-\gamma^{[2]}\gamma^{[3]}e^{-2b\Delta^{[1]}\xi_1/\epsilon}+\gamma^{[1]}\gamma^{[3]}e^{-2b\Delta^{[2]}\xi_2/\epsilon})^{-1}$ via bounded nonvanishing factors.  Therefore, if the sign parameters satisfy $\gamma^{[2]}\neq \gamma^{[1]}=\gamma^{[3]}$, all three fields will be uniformly bounded and $q^{[j]}(\cdot,t)\in\mathscr{S}(\mathbb{R})$ for all $t\in\mathbb{R}$.  All other possible sign configurations imply singularities of all three fields for some $(x,t)\in\mathbb{R}^2$.  Indeed, consider the following remaining cases.
\begin{itemize}
\item
If $\gamma^{[1]}=\gamma^{[2]}\neq\gamma^{[3]}$, then dominant balance arguments show that there is a unique simple pole in all fields near $\xi_2=0$ when $|t|\gg 1$ provided that $e^{-2b\Delta^{[1]}\xi_1/\epsilon}\ll 1$ in the limit, i.e., that $t\to -\infty$; we denote this pole by $x=x_-(t)$.  Similarly, there is a unique simple pole in all fields near $\Delta^{[2]}\xi_2=\Delta^{[1]}\xi_1$ or equivalently near $\xi_3=0$ when $|t|\gg 1$ provided that $e^{-2b\Delta^{[1]}\xi_1/\epsilon}\gg 1$ in the limit, i.e., that $t\to +\infty$; we denote this pole by $x=x_+(t)$.  Invoking the implicit function theorem to continue the simple roots $x_\pm(t)$ of the denominator to finite $t$ shows that unique continuation is possible to all $t\in\mathbb{R}$, and hence $x_+(t)=x_-(t)$.  Therefore if $\gamma^{[1]}=\gamma^{[2]}\neq\gamma^{[3]}$, all three fields $q^{[j]}(x,t)$ exhibit a simple pole singularity at a unique point $x=x(t)\in\mathbb{R}$ for all $t\in\mathbb{R}$.  The location $x=x(t)$ of the singularity moves with velocity $c^{[2]}$ as $t\to -\infty$ and with velocity $c^{[3]}$ as $t\to +\infty$.
\item
Similarly, if $\gamma^{[1]}\neq\gamma^{[2]}=\gamma^{[3]}$, then for each $t\in\mathbb{R}$ there is a unique simple pole in all three fields $q^{[j]}(x,t)$ at a point $x=x(t)$ that moves with velocity $c^{[2]}$ as $t\to -\infty$ and with velocity $c^{[1]}$ as $t\to +\infty$.  
\item
Finally, if $\gamma^{[1]}=\gamma^{[2]}=\gamma^{[3]}$, then dominant balance arguments show that as $t\to +\infty$ there are two distinct simple poles of all three fields $q^{[j]}(x,t)$ near $\xi_1=0$ and near $\xi_3=0$, while for $t$ sufficiently negative there are no singularities and all three fields satisfy $q^{[j]}(\cdot,t)\in\mathscr{S}(\mathbb{R})$.  By the implicit function theorem, the continuations of the pole curves approximated by $\xi_1=0$ and $\xi_3=0$ from large positive $t$ will first collide at a point along the straight line 
\begin{equation}
\xi_3=\frac{\epsilon}{2b\Delta^{[3]}}\ln\left(\frac{\Delta^{[2]}}{\Delta^{[1]}}\right).
\end{equation}
\end{itemize}
Therefore, there exists $t_0\in\mathbb{R}$ for which the soliton of type $\solsplit$ corresponds to bounded fields $q^{[j]}(\cdot,t_0)$ only if $\gamma^{[2]}$ is distinct from both $\gamma^{[1]}$ and $\gamma^{[3]}$ or if $\gamma^{[1]}=\gamma^{[2]}=\gamma^{[3]}$.  In the former case, the fields lie in $\mathscr{S}(\mathbb{R})$ for all $t$, while in the latter case all three fields blow up in finite time at some $t_*>t_0$ with the birth of a double pole that then splits into two simple poles for $t>t_*$.  The elementary soliton solution of type $\solsplit$ therefore clearly exhibits the $L^2(\mathbb{R}_x)$ blowup that is suggested by the indefiniteness of the Manley-Rowe relations \eqref{eq:Manley-Rowe} for $\gamma^{[1]}=\gamma^{[2]}=\gamma^{[3]}$.

In the case that all three fields are bounded for all $t\in\mathbb{R}$, one can read off the large $|t|$ asymptotics from the explicit formulae \eqref{eq:split-1}--\eqref{eq:split-3}.  Letting $x-vt$ be bounded as $|t|\to\infty$, one sees that unless $v=c^{[k]}$, $q^{[k]}(x,t)$ decays exponentially to zero as $|t|\to\infty$.  Moreover, if $v=c^{[1]}$ then $q^{[1]}(x,t)$ decays exponentially to zero as $t\to -\infty$ but takes exactly the limiting form \eqref{eq:pure-1} as $t\to +\infty$.  If $v=c^{[2]}$ then $q^{[2]}(x,t)$ decays exponentially to zero as $t\to +\infty$ but takes exactly the limiting form \eqref{eq:pure-2} as $t\to -\infty$.  Finally, if $v=c^{[3]}$ then $q^{[3]}(x,t)$ decays exponentially to zero as $t\to -\infty$ but takes exactly the limiting form \eqref{eq:pure-3} as $t\to +\infty$.  These results are suggested by the notation $\solsplit$; indeed the solution represents a moderate-velocity solitary wave $q^{[2]}(x,t)$ of the form \eqref{eq:pure-2} that decays and splits as $t$ increases into a fast-moving solitary wave $q^{[1]}(x,t)$ of the form \eqref{eq:pure-1} and a slow-moving solitary wave $q^{[3]}(x,t)$ of the form \eqref{eq:pure-3}.  

Again considering the case $\gamma^{[2]}\neq\gamma^{[1]}=\gamma^{[3]}$ in which all fields are bounded,
we may consider certain degenerations of the type $\solsplit$ soliton.  For example, if $x_2,x_3\to -\infty$ in such a way that $x_1$ remains bounded (taking into account \eqref{eq:phase-relations}) then $q^{[2]},q^{[3]}\to 0$ while $q^{[1]}\to \gamma^{[1]}e^{i\phi_1}e^{-ia\Delta^{[1]}\xi_1/\epsilon}b\Delta^{[1]}\sqrt{\Delta^{[2]}\Delta^{[3]}}\mathrm{sech}(b\Delta^{[1]}\xi_1/\epsilon)$ uniformly for bounded $(x,t)\in\mathbb{R}^2$. In this way, the type $\solsplit$ soliton degenerates into the type $1$ soliton.  Similarly, if $x_1\to -\infty$ and $x_3\to +\infty$ such that $x_2$ remains bounded, then $q^{[1]},q^{[3]}\to 0$.  On the other hand, we have $q^{[2]}\to -\gamma^{[2]}e^{i\phi_2}e^{ia\Delta^{[2]}\xi_2/\epsilon}b\Delta^{[2]}\sqrt{\Delta^{[3]}\Delta^{[1]}}\mathrm{sech}(b\Delta^{[2]}\xi_2/\epsilon)$ uniformly for bounded $(x,t)\in\mathbb{R}^2$, exhibiting the degeneration of the type $\solsplit$ soliton to the type $2$ soliton.  In both of these cases, there is an obvious corresponding limit in the scattering domain, with $\beta^{[\solsplit]}_{31}\to 0$ yielding the degeneration to type $1$ and $\beta^{[\solsplit]}_{32}\to 0$ yielding the degeneration to type $2$.  However, a third possibility is to consider the limit $x_1,x_2\to +\infty$ in such a way that $x_3$ remains finite.  Then $q^{[1]},q^{[2]}\to 0$ while $q^{[3]}\to \gamma^{[3]}e^{i\phi_3}e^{-ia\Delta^{[3]}\xi_3/\epsilon}b\Delta^{[3]}\sqrt{\Delta^{[1]}\Delta^{[2]}}\mathrm{sech}(b\Delta^{[3]}\xi_3/\epsilon)$ uniformly for bounded $(x,t)\in\mathbb{R}^2$.  This is a degeneration of the type $\solsplit$ soliton to the type $3$ soliton; however, at the level of the spectral data the limit is more subtle as $\beta^{[\solsplit]}_{31}$ and $\beta^{[\solsplit]}_{32}$ both blow up.  The resolution of this apparent difficulty is to recall the identity \eqref{eq:type-solsplit-beta-relations}, which shows that the corresponding connection coefficients for the matrix $\mathbf{M}^-$ satisfy $\beta^{[\solsplit]}_{13}\to 0$ while $\beta^{[\solsplit]}_{12}$ has a finite limit; thus the divergence disappears upon working with $\mathbf{M}^-$ rather than $\mathbf{M}^+$.

We now consider the soliton of type $\solfuse$.  
The fields $\{q^{[k]}(x,t)\}_{k=1}^3$ are all proportional to $(1+\gamma^{[1]}\gamma^{[3]}e^{-2b\Delta^{[2]}\xi_2/\epsilon}-\gamma^{[1]}\gamma^{[2]}e^{-2b\Delta^{[3]}\xi_3/\epsilon})^{-1}$ via bounded nonvanishing factors.  Therefore, as in the case of the type $\solsplit$ soliton, unless either $\gamma^{[1]}=\gamma^{[3]}\neq\gamma^{[2]}$ or $\gamma^{[1]}=\gamma^{[2]}=\gamma^{[3]}$, all three fields exhibit a unique simple pole singularity at a point $x=x(t)$ well defined for all $t\in\mathbb{R}$.  If $\gamma^{[1]}=\gamma^{[2]}=\gamma^{[3]}$, $q^{[k]}(\cdot,t)\in\mathscr{S}(\mathbb{R})$, $k=1,2,3$, holds for $t$ sufficiently positive, but all fields contain singularities for $t$ sufficiently negative.  Finally, if $\gamma^{[1]}=\gamma^{[3]}\neq\gamma^{[2]}$ then $q^{[k]}(\cdot,t)\in\mathscr{S}(\mathbb{R})$ for all $t\in\mathbb{R}$, $k=1,2,3$, and all seminorms are uniformly bounded in time.  In the latter case, the type $\solfuse$ soliton resembles a superposition of solitons of types $1$ (fast-moving) and $3$ (slow-moving) for large negative $t$ but these solitons combine for finite $t$ and produce a single soliton of type $2$ (moderate velocity) for large positive $t$.  This is suggested by the notation $\solfuse$.  The type $\solfuse$ soliton can degenerate for bounded $(x,t)$ into solitons of types $1$, $2$, and $3$ by taking appropriate limits of the connection coefficients.

\subsubsection{Double-scaling limits of the $1+3$ two-pole solution}
\label{subsubsec:solitons-merging}
Suppose that $\gamma^{[1]}=\gamma^{[3]}\neq\gamma^{[2]}$ so that all reflectionless potentials are bounded for all $(x,t)\in\mathbb{R}^2$, and 
consider a reflectionless potential corresponding to two simple poles in $P$:  one pole $\lambda^{[1]}$ of type $1$ with 
 $\mathbf{M}^+$ connection coefficient $\beta_{32}^{[1]}$ and one pole $\lambda^{[3]}$ of type $3$ 
with $\mathbf{M}^+$ connection coefficient $\beta_{21}^{[3]}$.  We will show that by merging 
these two poles in appropriate ways we can obtain any of the five elementary 
solitons described in Appendix~\ref{subsubsec:solitons-onepole}.  Let 
\begin{equation}
	\lambda^{[1]} - \lambda^{[3]} = \delta e^{i \theta}
	\label{eq:lambda-delta-theta}
\end{equation}
for $\delta > 0$ and $\theta \in \R$ (mod $2\pi$) and write $\lambda^{[k]} - \lambda^{[k]*} = 2i b^{[k]}$, $k=1,3$.  We will consider $\delta\downarrow 0$ to merge the poles and we introduce a double-scaling limit by allowing the connection coefficients to scale with powers of $\delta$ by writing
\begin{gather}
	\beta_{32}^{[1]} = 2ib^{[1]} \delta^{m_1} \widehat{\beta}_{32}^{[1]} \quad\text{and}\quad
	\beta_{21}^{[3]} = 2ib^{[3]} \delta^{m_3} \widehat{\beta}_{21}^{[3]}
	\label{beta scales}
\end{gather}
for powers $m_1, m_3 \in \R$ to be determined. 
Using the procedure in Appendix~\ref{subsubsec:solitons-rhp}, we reconstruct the potentials $q^{[k]}=q^{[k]}(x,t)$, $k=1,2,3$, from $\mathbf{M}^+(x,t;\lambda)$ as
\begin{equation}\label{coefficients}
	\begin{aligned}
  		q^{[1]} &=\frac{
			\gamma^{[1]} \Delta^{[1]} \sqrt{\Delta^{[2]}\Delta^{[3]}}
			\delta^{m_1} 2b^{[1]} \widehat{\beta}^{[1]*}_{32}E^{[1]*} 
			\left( 1 +  \delta^{2m_3}| \widehat{\beta}^{[3]}_{21}E^{[3]}|^2
			- 2i b^{[3]} e^{i\theta} \delta^{2m_3-1}|\widehat{\beta}^{[3]}_{21}E^{[3]}|^2 \right)
		}
		{ 
			1 + \delta^{2m_1} | \widehat{\beta}^{[1]}_{32}E^{[1]}|^2 + 
			\delta^{2m_3} | \widehat{\beta}^{[3]}_{21}E^{[3]}|^2 +
			\delta^{2(m_1 + m_3 - 1)}  |\widehat{\beta}^{[1]}_{32}E^{[1]} \widehat{\beta}^{[3]}_{21}E^{[3]}|^2 
			\left( 4 b^{[1]} b^{[3]} + \delta^{2} \right)
	 	},    \\ 
	q^{[2]} &= \frac{ 
		i\gamma^{[2]} \Delta^{[2]} \sqrt{\Delta^{[1]}\Delta^{[3]}}	
		4 \delta^{m_1+m_3-1} e^{-i \theta} b^{[1]} b^{[3]}  
		\widehat{\beta}^{[1]}_{32}E^{[1]} \widehat{\beta}^{[3]}_{21}E^{[3]} 
		}
		{ 
			1 + \delta^{2m_1} | \widehat{\beta}^{[1]}_{32}E^{[1]}|^2 + 
			\delta^{2m_3} | \widehat{\beta}^{[3]}_{21}E^{[3]}|^2 +
			\delta^{2(m_1 + m_3 - 1)}  |\widehat{\beta}^{[1]}_{32}E^{[1]} \widehat{\beta}^{[3]}_{21}E^{[3]}|^2 
			\left( 4 b^{[1]} b^{[3]} + \delta^{2} \right)
		}, \\
	q^{[3]} &=\frac{
		\gamma^{[3]} \Delta^{[3]} \sqrt{\Delta^{[1]}\Delta^{[2]}}	
		\delta^{m_3} 2b^{[3]} \widehat{\beta}^{[3]*}_{21}E^{[3]*} 
		\left( 1 +  \delta^{2m_1}| \widehat{\beta}^{[1]}_{32}|^2
		+ 2i b^{[1]} e^{i\theta} \delta^{2m_1-1} |\widehat{\beta}^{[1]}_{32}E^{[1]}|^2 \right)
		}
		{ 
			1 + \delta^{2m_1} | \widehat{\beta}^{[1]}_{32}E^{[1]}|^2 + 
			\delta^{2m_3} | \widehat{\beta}^{[3]}_{21}E^{[3]}|^2 +
			\delta^{2(m_1 + m_3 - 1)}  |\widehat{\beta}^{[1]}_{32}E^{[1]} \widehat{\beta}^{[3]}_{21}E^{[3]}|^2 
			\left( 4 b^{[1]} b^{[3]} + \delta^{2} \right)
		},	
	\end{aligned}
\end{equation}
where the dependent variables enter via $E^{[k]}=E^{[k]}(x,t):=e^{i\lambda^{[k]}\Delta^{[k]}(x-c^{[k]}t)/\epsilon}$.
We now want to consider the limit of these solutions as $\delta \downarrow 0$ holding all other parameters fixed along with the independent variables $(x,t)$. The structure of the limiting solution depends on which of the four terms in the common denominator of each expression in \eqref{coefficients} is dominant;  let 
\begin{equation}
	m: = \min \{ 0, m_1, m_3, m_1+m_3-1 \}.
\end{equation}
The $(m_1,m_3)$-plane is then divided into four distinct regions on which $m$ is constant:
\begin{equation}
	\begin{aligned}
		A &:= \{ (m_1, m_3) \in \R^2 \,:\, m = 0 \}, \\
		B &:= \{ (m_1, m_3) \in \R^2 \,:\, m = m_1 \}, \\
		C &:= \{ (m_1, m_3) \in \R^2 \,:\, m = m_3 \}, \\
		D &:= \{ (m_1, m_3) \in \R^2 \,:\, m = m_1 + m_3 -1 \}. \\
	\end{aligned}
\end{equation}
As summarized in Figure~\ref{fig-scaling-parameters} in the interior of each region $\lim_{\delta \downarrow 0} q^{[k]}(x,t) = 0$, while on the boundaries between the regions elementary (single pole) solitons of different types emerge in the limit as $\delta \downarrow 0$. 
\begin{figure}[h]
\centering
\begin{tabular}{c c}
\begin{tikzpicture}[scale=2.5]
\path [fill=gray!10](0,2) -- (0,1) -- (1,0) -- (2,0)--(2,2)--(0,2);
\path [fill=gray!50](-1,1) -- (0,1) -- (1,0) -- (1,-1)--(-1,-1)--(-1,1);
\path [fill=gray!30](-1,1) -- (0,1) -- (0,2) -- (-1,2)--(-1,1);
\path [fill=gray!30](1,-1) -- (1,0) -- (2,0) -- (2,-1)--(1,-1);
\draw [dotted] (-0.1,0) -- (1,0);
\draw [dotted] (0,-0.1) -- (0,1);
\node [right] at (2.01,0) {$m_1$};
\node [above] at (0,2.01) {$m_3$};
\draw [fill] (0,0) circle (.1ex);
\node [below left] at (0,0) {$0$};
\draw [thick] (1,0) -- (0,1) node[circle, midway, draw=black, fill=white] {$2$};
\draw [thick] (0,1) -- (-1,1) node[circle, near end, draw=black, fill=white] {$3$};
\draw [thick][->] (0,1) -- (0,2.01) node[circle, near end, draw=black, fill=white] {$1$};
\draw [thick][->] (1,0) -- (2.01,0) node[circle, near end, draw=black, fill=white] {$3$};
\draw [thick] (1,0) -- (1,-1) node[circle, near end, draw=black, fill=white] {$1$};
\node [circle, draw=black, fill=white] at (1,0) {$\solfuse$};
\node [circle, draw=black, fill=white] at (0,1) {$\solsplit$};
\node at (1,1) {$A$};
\node at (-0.5,1.5) {$B$};
\node at (1.5,-0.5) {$C$};
\node at (-0.5,-0.5) {$D$};
\end{tikzpicture}\\ 
\end{tabular}
\caption{The type of elementary soliton that emerges in the double-scaling limit depends on the choice of exponents $m_1$ and $m_3$ in \eqref{beta scales}. The soliton type along each boundary curve in the $(m_1, m_3)$-plane is indicated above. 
Types $\protect\solfuse$ 
and $\solsplit$ occur only at the vertices $(0,1)$ and $(1,0)$ respectively. In each of the four open regions $A$, $B$, $C$, and $D$, $q^{[k]}(x,t) \to 0$ as $\delta\downarrow 0$ for $k=1,2,3$.
\label{fig-scaling-parameters}
}
\end{figure}
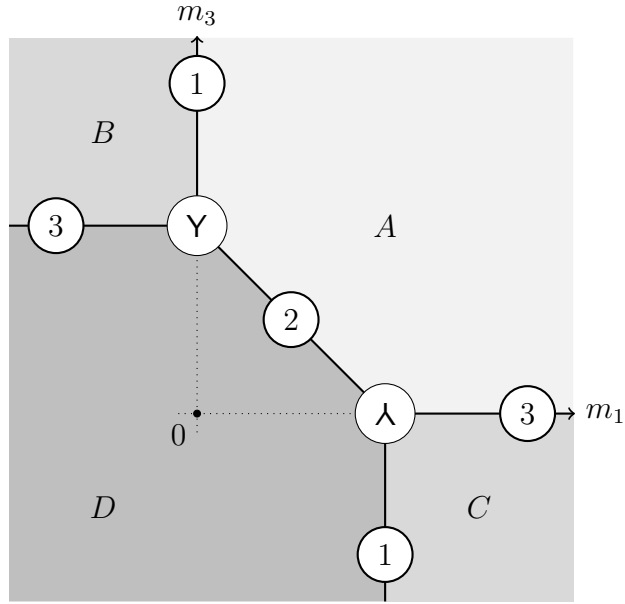
We describe these limits in detail below.

To study the solution in the interior of the four regions $A$, $B$, $C$, and $D$, we inspect \eqref{coefficients} and find the following.
\begin{itemize}
\item For $(m_1,m_3)$ in region $A$, 
		$q^{[1]} = \bigo{ \delta^{\min\{m_1, m_1+2m_3-1\}} }$,
		$q^{[2]} = \bigo{ \delta^{m_1 + m_3 - 1} }$, and
		$q^{[3]} = \bigo{ \delta^{\min\{m_3, 2m_1+m_3-1\}} }$.
\item For $(m_1,m_3)$ in region $B$,
		$q^{[1]} = \bigo{ \delta^{-m_1} }$,
		$q^{[2]} = \bigo{ \delta^{m_3 -m_1- 1} }$, and
		$q^{[3]} = \bigo{ \delta^{m_3 - 1} }$.	
\item For $(m_1,m_3)$ in region $C$,
		$q^{[1]} = \bigo{ \delta^{m_1-1} }$,
		$q^{[2]} = \bigo{ \delta^{m_1 -m_3- 1}  }$, and
		$q^{[3]} = \bigo{ \delta^{-m_3} }$.
\item For $(m_1,m_3)$ in region $D$, 
		$q^{[1]} = \bigo{ \delta^{\min\{1-m_1, 2 - m_1 - 2m_3\}} }$,
		$q^{[2]} = \bigo{ \delta^{1- m_1 - m_3} }$, and
		$q^{[3]} = \bigo{ \delta^{\min\{1-m_3, 2- 2m_1 -m_3\}} }$.
\end{itemize}
It is easy to check that in all four regions, the exponents in the above estimates are all strictly positive, and hence $q^{[k]}(x,t)\to 0$ as $\delta\downarrow 0$ for $k=1,2,3$.

Nontrivial limiting solutions exist on the shared boundaries between the various regions. To express these let 
\begin{equation}
	\begin{gathered}
	\lim_{\delta \downarrow 0} \lambda^{[1]} = \lim_{\delta \downarrow 0} \lambda^{[3]} 
	= \lambda = a + i b \qquad (a,b\in \R),
	\\
	\lim_{\delta \downarrow 0} b^{[1]} = \lim_{\delta \downarrow 0} b^{[3]} = \imag \{\lambda\} = b.
	\end{gathered}
\end{equation}
If we introduce parameters $x_k \in \R$ and $\phi_k$ (mod $2\pi$) $\in \R$, $k=1,3$, define related constants $x_2$ and $\phi_2$ by \eqref{eq:phase-relations} and \eqref{eq:phi-relation-1} respectively, 
and recall $\xi_k := x  - x_k- c^{[k]} t$ for $k=1,2,3$, then we have the limits
\begin{gather}
	\lim_{\delta\downarrow 0}\widehat{\beta}^{[1]}_{32}E^{[1]}(x,t) 
		= e^{-b\Delta^{[1]} \xi_1 / \epsilon} e^{i a \Delta^{[1]} \xi_1 / \epsilon} 
		e^{-i \phi_1}, \\
	\lim_{\delta\downarrow 0}\widehat{\beta}^{[3]}_{21}E^{[3]}(x,t) 
		= e^{-b\Delta^{[3]} \xi_3 / \epsilon} e^{i a \Delta^{[3]} \xi_3 / \epsilon} 
		e^{-i \phi_3},
\end{gather}
and the limit of the product is
\begin{gather}
	\lim_{\delta\downarrow 0}\widehat{\beta}^{[1]}_{32}E^{[1]}(x,t) \widehat{\beta}^{[3]}_{21}E^{[3]}(x,t) 
	= e^{-b\Delta^{[2]} \xi_2 / \epsilon} e^{i a \Delta^{[2]} \xi_2 / \epsilon} 
	e^{i \phi_2},
\end{gather}
which evolves in $(x,t)$ like a type-2 connection coefficient. The $\delta\downarrow 0$ limits of $q^{[k]}(x,t)$, $k=1,2,3$, 
when $(m_1,m_3)$ lies on the boundary between different regions are then as follows.  Recall the angle $\theta$ defined (mod $2\pi$) by \eqref{eq:lambda-delta-theta}.
\begin{itemize}
\item Type $1$ solitons emerge on the $AB$ and $CD$ boundaries: 
\begin{itemize}
	\item On the $AB$ boundary where $m_1=0$ and $m_3 > 1$ we have
	\begin{equation}
	\begin{gathered}	
	q^{[1]}(x,t) \to \gamma^{[1]} 
	e^{i \phi_1} e^{-i a \Delta^{[1]} \xi_1 / \epsilon} 
	b \Delta^{[1]} \sqrt{  \Delta^{[2]} \Delta^{[3]}} 
	\sech(b \Delta^{[1]} \xi_1/ \epsilon), \\
	q^{[2]}(x,t) \to 0,\quad  q^{[3]}(x,t)\to 0.
	\end{gathered}
	\end{equation}
	\item On the $CD$ boundary where $m_1=1$ and $m_3 < 0$ we have
	\begin{equation}
	\begin{gathered}	
	q^{[1]}(x,t) \to -i \gamma^{[1]} 
	e^{i(\theta+ \phi_1)} e^{-i a \Delta^{[1]} {\xi}_1 / \epsilon} 
	b \Delta^{[1]} \sqrt{  \Delta^{[2]} \Delta^{[3]}} 
	\sech( b \Delta^{[1]} \breve{\xi}_1/ \epsilon), \\
	q^{[2]}(x,t) \to 0,\quad  q^{[3]}(x,t) \to 0.
	\end{gathered}
	\end{equation}
	where $\breve{\xi}_k$ is defined for $k=1,2,3$ by the relation
	$2b e^{-b\Delta^{[k]}\xi_k/\eps} = e^{-b\Delta^{[1]} \breve{\xi}_k /\eps}$, essentially introducing a phase shift. 
\end{itemize}
\item 
Type $2$ solitons emerge on the $AD$ boundary where $1-m_1-m_3 = 0$ with $m_1,m_3>0$: 
	\begin{equation}
	\begin{gathered}	
	q^{[2]}(x,t) \to i \gamma^{[2]} 
	e^{i(\phi_2-\theta)} e^{i a \Delta^{[2]} \xi_2 / \epsilon} 
	b \Delta^{[2]} \sqrt{  \Delta^{[1]} \Delta^{[3]}} 
	\sech(b \Delta^{[2]} \breve \xi_2/ \epsilon), \\
	q^{[1]}(x,t) \to 0,\quad  q^{[3]}(x,t) \to 0.
	\end{gathered}
	\end{equation}
\item
Type $3$ solitons emerge on the $AC$ and $BD$ boundaries: 
\begin{itemize}
	\item On the $AC$ boundary where $m_1 > 1$ and $m_3 = 0$ we have 
	\begin{equation}
	\begin{gathered}	
	q^{[3]}(x,t) \to \gamma^{[3]} 
	e^{i \phi_3} e^{-i a \Delta^{[3]} \xi_3 / \epsilon} 
	b \Delta^{[3]} \sqrt{  \Delta^{[1]} \Delta^{[2]}} 
	\sech(b \Delta^{[3]} \xi_3/ \epsilon), \\
	q^{[1]}(x,t) \to 0,\quad q^{[2]}(x,t) \to 0.
	\end{gathered}
	\end{equation}
	\item On the $BD$ boundary where $m_1 < 0$ and $m_3 = 1$ we have
	\begin{equation}
	\begin{gathered}	
	q^{[3]}(x,t) \to -i \gamma^{[3]} 
	e^{i(\theta+ \phi_3)} e^{-i a \Delta^{[3]} {\xi}_3 / \epsilon} 
	b \Delta^{[3]} \sqrt{  \Delta^{[1]} \Delta^{[2]}} 
	\sech( b \Delta^{[3]} \breve{\xi}_3/ \epsilon), \\
	q^{[1]}(x,t) \to 0,\quad q^{[2]}(x,t) \to 0,
	\end{gathered}
	\end{equation}
	where $\breve{\xi}_3$ is as defined above.
\end{itemize}
\item 
At the $ABD$ vertex where $(m_1,m_3)=(0,1)$ we have a type $\solsplit$ soliton
\begin{equation}
	\begin{aligned}
	q^{[1]}(x,t) &\to \gamma^{[1]} 
		e^{i \phi_1} e^{-i a \Delta^{[1]} \xi_1 / \epsilon} 
		b \Delta^{[1]} \sqrt{  \Delta^{[2]} \Delta^{[3]}} 
		\left[
		\cosh(b \Delta^{[1]} \xi_1/\eps) + \tfrac{1}{2} 
		e^{-b \Delta^{[2]} \breve{\xi}_2/\eps} e^{-b \Delta^{[3]} \breve{\xi}_3/\eps}  
		\right]^{-1}, \\
	q^{[2]}(x,t) &\to i \gamma^{[2]} 
		e^{i(\phi_2-\theta)} e^{i a \Delta^{[2]} \xi_2 / \epsilon} 
		b \Delta^{[2]} \sqrt{  \Delta^{[1]} \Delta^{[3]}} 
		\left[
		\cosh(b \Delta^{[2]} \breve{\xi}_2/\eps) + \tfrac{1}{2} 
		e^{-b \Delta^{[1]} \xi_1/\eps} e^{b \Delta^{[3]} \breve{\xi}_3/\eps} 
		\right]^{-1}, \\
	q^{[3]}(x,t) &\to i \gamma^{[3]} 
		e^{i(\theta + \phi_3)} e^{-i a \Delta^{[3]} \xi_3 / \epsilon} 
		b \Delta^{[3]} \sqrt{  \Delta^{[1]} \Delta^{[2]}} 
		\left[
		\cosh(b \Delta^{[3]} \breve{\xi}_3/\eps) + \tfrac{1}{2} 
		e^{b \Delta^{[1]} \xi_1/\eps} e^{b \Delta^{[2]} \breve{\xi}_2/\eps} 
		\right]^{-1}.
	\end{aligned}
\end{equation}
\item
At the $ACD$ vertex where $(m_1,m_3)=(1,0)$ we have a type $\solfuse$ soliton
\begin{equation}
	\begin{aligned}
	q^{[1]}(x,t) &\to -i \gamma^{[1]} 
		e^{i(\theta + \phi_1)} e^{-i a \Delta^{[1]} \xi_1 / \epsilon} 
		b \Delta^{[1]} \sqrt{  \Delta^{[2]} \Delta^{[3]}} 
		\left[
		\cosh(b \Delta^{[1]} \breve{\xi}_1/\eps) + \tfrac{1}{2} 
		e^{b \Delta^{[2]} \breve{\xi}_2/\eps} e^{b \Delta^{[3]} \xi_3/\eps}  
		\right]^{-1}, \\
	q^{[2]}(x,t) &\to i \gamma^{[2]} 
		e^{i(\phi_2-\theta)} e^{i a \Delta^{[2]} \xi_2 / \epsilon} 
		b \Delta^{[2]} \sqrt{  \Delta^{[1]} \Delta^{[3]}} 
		\left[
		\cosh(b \Delta^{[2]} \breve{\xi}_2/\eps) + \tfrac{1}{2} 
		e^{b \Delta^{[1]} \breve{\xi}_1/\eps} e^{-b \Delta^{[3]} \xi_3/\eps}  
		\right]^{-1},  \\
	q^{[3]}(x,t) &\to \gamma^{[3]} 
		e^{i \phi_3} e^{-i a \Delta^{[3]} \xi_3 / \epsilon} 
		b \Delta^{[3]} \sqrt{  \Delta^{[1]} \Delta^{[2]}} 
		\left[
		\cosh(b \Delta^{[3]} \xi_3/\eps) + \tfrac{1}{2} 
		e^{-b \Delta^{[1]} \breve{\xi}_1/\eps} e^{-b \Delta^{[2]} \breve{\xi}_2/\eps}  
		\right]^{-1} .
	\end{aligned}
\end{equation}
\end{itemize}
Note that in the cases of the limits leading to solitons of types $1$, $2$, and $3$ in which $(m_1,m_3)$ lies along a line segment in the plane, the limiting solution is independent of the particular point chosen provided it is fixed on the interior of that segment.

\section{The Nonselfadjoint Zakharov-Shabat Problem}
\label{appC-nonselfadjoint-zs}
\subsection{Direct scattering theory for integrable potentials}
Let $\psi\in L^1(\mathbb{R};\mathbb{C})$ be a complex-valued integrable function, and let $\zeta\in\mathbb{C}$ be a complex spectral parameter.  The non-selfadjoint Zakharov-Shabat problem is the first-order linear system
\begin{equation}
\epsilon\frac{\dd \mathbf{w}}{\dd x}=\mathbf{Z}(x;\zeta)\mathbf{w},\quad\mathbf{Z}(x;\zeta):=\begin{pmatrix}-i\zeta & \psi(x)\\
-\psi(x)^* & i\zeta\end{pmatrix}
\label{eq:Zakharov-Shabat}
\end{equation}
governing an unknown $\mathbf{w}:\mathbb{R}\to\mathbb{C}^2$.  This differential 
equation was first derived by Zakharov and Shabat \cite{Zakharov:1972};  in 
particular it is part of the Lax pair for the cubic focusing nonlinear Schr\"odinger 
equation in one space dimension.  There exist two $2\times 2$ fundamental solution matrices denoted $\mathbf{W}^\pm(x;\zeta)$ defined for $\zeta\in\mathbb{R}$ whose columns are \emph{Jost solutions} of \eqref{eq:Zakharov-Shabat} and that satisfy the boundary conditions
\begin{equation}
\lim_{x\to\pm\infty}\mathbf{W}^\pm(x;\zeta)e^{i\zeta\sigma_3 x/\epsilon} = \mathbb{I},\quad\zeta\in\mathbb{R}.
\label{eq:ZS-Jost-BC}
\end{equation}
Both matrices $\mathbf{W}^\pm(x;\zeta)$ have unit determinant by Abel's Theorem, and they have the symmetry $\mathbf{W}^\pm(x;\zeta)^*=i\sigma_2\mathbf{W}^\pm(x;\zeta)(i\sigma_2)^{-1}$, where on the left-hand side we mean component-wise complex conjugation (no transpose).
The scattering matrix $\mathbf{S}^\mathrm{ZS}(\zeta)$ is defined in terms of these fundamental matrices by the equation
\begin{equation}
\mathbf{W}^+(x;\zeta)=\mathbf{W}^-(x;\zeta)\mathbf{S}^\mathrm{ZS}(\zeta),\quad x\in\mathbb{R},\quad \zeta\in\mathbb{R}.
\label{eq:ZS-scattering-matrix}
\end{equation}
The conjugation symmetry of the matrices $\mathbf{W}^\pm(x;\zeta)$ then implies that the scattering matrix can be written in the form
\begin{equation}
\mathbf{S}^\mathrm{ZS}(\zeta)=\begin{pmatrix}a(\zeta)^* & b(\zeta)^*\\-b(\zeta) & a(\zeta)\end{pmatrix},\quad\zeta\in\mathbb{R},
\label{eq:ZS-scattering-matrix-ab}
\end{equation}
and the unimodularity of $\mathbf{W}^\pm(x;\zeta)$ implies that $\mathbf{S}^\mathrm{ZS}(\zeta)$ has unit determinant, i.e., 
\begin{equation}
|a(\zeta)|^2+|b(\zeta)|^2=1,\quad \zeta\in\mathbb{R}.
\label{eq:ZS-ab-modulus}
\end{equation}
For each $x\in\mathbb{R}$, the first (second) column of $\mathbf{W}^-(x;\zeta)$ (of $\mathbf{W}^+(x;\zeta)$), denoted $\mathbf{w}^{-,1}(x;\zeta)$ ($\mathbf{w}^{+,2}(x;\zeta)$) is continuous for $\imag\{\zeta\}\ge 0$ and analytic for $\imag\{\zeta\}>0$, and it can be shown that the corresponding columns of the boundary condition \eqref{eq:ZS-Jost-BC} continue to hold true for $\imag\{\zeta\}>0$.  Since $a(\zeta)$ can be represented as a Wronskian determinant by
\begin{equation}
a(\zeta)=\det(\mathbf{w}^{-,1}(x;\zeta),\mathbf{w}^{+,2}(x;\zeta)),\quad\zeta\in\mathbb{R},
\label{eq:a-determinant}
\end{equation}
it follows that $a(\zeta)$ is the continuous boundary value of a function analytic in the upper half-plane $\imag\{\zeta\}>0$, and from the analysis of Volterra equations governing the Jost solutions, it can be shown that $a(\zeta)\to 1$ as $\zeta\to\infty$ with $\imag\{\zeta\}>0$.
The zeros of $a(\cdot)$ in the upper half-plane are the \emph{eigenvalues} of the Zakharov-Shabat problem \eqref{eq:Zakharov-Shabat}.  Indeed, if $\imag\{\zeta\}>0$ and $a(\zeta)=0$, then from \eqref{eq:a-determinant} it follows that there is a nonzero \emph{proportionality constant} $\tau\in\mathbb{C}$ corresponding to $\zeta$ such that $\mathbf{w}^{-,1}(x;\zeta)=\tau\mathbf{w}^{+,2}(x;\zeta)$ holds for all $x\in\mathbb{R}$.  Since $\mathbf{w}^{-,1}(x;\zeta)$ ($\mathbf{w}^{+,2}(x;\zeta)$) is a solution of \eqref{eq:Zakharov-Shabat} that decays as $x\to -\infty$ ($x\to +\infty$), the condition $a(\zeta)=0$ obviously implies the existence of a nontrivial solution $\mathbf{w}(x;\zeta)$ of \eqref{eq:Zakharov-Shabat} that decays rapidly to zero in both limits $x\to\pm\infty$.  Supposing that (i) there are a finite number of zeros of $a(\cdot)$ in the upper half-plane, say $\zeta_0,\dots,\zeta_{N-1}$ 
and (ii) all zeros of $a(\cdot)$ in the upper half-plane are simple, the function
\begin{equation}
\alpha(\zeta):=a(\zeta)\prod_{n=0}^{N-1}\frac{\zeta-\zeta_n^*}{\zeta-\zeta_n},\quad\imag\{\zeta\}\ge 0
\end{equation}
is clearly analytic and non-vanishing for $\imag\{\zeta\}>0$, and $\alpha(\zeta)\to 1$ as $\zeta\to\infty$ with $\imag\{\zeta\}>0$.  Moreover $|\alpha(\zeta)|^2=|a(\zeta)|^2$ whenever $\zeta\in\mathbb{R}$.  Letting $f(\zeta)$ be defined in terms of the principal branch of the logarithm as $f(\zeta):=\log(\alpha(\zeta))$ for $\imag\{\zeta\}>0$ and $f(\zeta):=-\log(\alpha(\zeta^*)^*)$ for $\imag\{\zeta\}<0$, we see that $f(\zeta)$ is analytic for $\zeta\in\mathbb{C}\setminus\mathbb{R}$.  From \eqref{eq:ZS-ab-modulus} we get $f_+(\zeta)-f_-(\zeta)=\log(1-|b(\zeta)|^2)$ for $\zeta\in\mathbb{R}$, where $f_\pm$ denotes the boundary value taken by $f$ on $\mathbb{R}$ from $\mathbb{C}_\pm$ (note that $|b(\zeta)|^2<1$ holds at all real $\zeta$ at which $a(\zeta)\neq 0$).  Therefore assuming that $a(\zeta)\neq 0$ for all $\zeta\in\mathbb{R}$, applying the Plemelj formula  gives
\begin{equation}
f(\zeta)=\frac{1}{2\pi i}\int_\mathbb{R} \frac{\log(1-|b(\xi)|^2)}{\xi-\zeta}\,\dd\xi,\quad \zeta\in\mathbb{C}\setminus\mathbb{R},
\end{equation}
from which we obtain the relation
\begin{equation}
a(\zeta)=\prod_{n=0}^{N-1}\frac{\zeta-\zeta_n}{\zeta-\zeta_n^*}\exp\left(\frac{1}{2\pi i}\int_\mathbb{R}\frac{\log(1-|b(\xi)|^2)}{\xi-\zeta}\,\dd\xi\right),\quad\imag\{\zeta\}>0
\label{eq:ZS-a-upper-half-plane}
\end{equation}
expressing the analytic continuation of $a(\cdot)$ from the real axis to the upper half-plane in terms of its zeros and the complementary function $b(\cdot)$ on the real axis.  

\subsection{Semiclassical direct scattering for Klaus-Shaw potentials}
\label{subapp-semiclassical-ZS}
Now suppose further that $\psi(x) = A(x)$ is a \emph{Klaus-Shaw potential}, i.e., $A:\mathbb{R}\to [0,+\infty)$ is a continuous nonnegative amplitude function that is nondecreasing on $(-\infty,x_0]$ and nonincreasing on $[x_0,+\infty)$ for some $x_0\in\mathbb{R}$, the global maximizer of $A$.  Klaus and Shaw \cite{KlausS02} proved that for such potentials, the eigenvalues in the upper half-plane are indeed finite in number, simple, and moreover they lie on the positive imaginary axis. Moreover, $a(\zeta)$ is nonzero for all real $\zeta$ with the possible exception of the origin $\zeta=0$, and $a(0)=0$ if and only if $\int_\mathbb{R}A(x)\,dx\in \pi(\mathbb{Z}+1/2)\epsilon$.  We avoid this transitional situation (in which an eigenvalue is born from the origin as the $L^1$ norm of $A/\epsilon$ increases) by assuming that, given $A$, $\epsilon$ is chosen so that for some $N=0,1,2,3,\dots$, 
\begin{equation}
\epsilon = \frac{1}{N\pi}\int_\mathbb{R}A(x)\,\dd x.
\label{eq:ZS-epsilon-assumption}
\end{equation}
In this situation, there are exactly $N$ strictly positive imaginary eigenvalues $\zeta_n=is_n$, $0<s_{N-1}<\cdots <s_1<s_0$, all simple.
Note that the assumption \eqref{eq:ZS-epsilon-assumption} allows the consideration of the limit $\epsilon\downarrow 0$ by the corresponding (discrete) limit $N\to\infty$.  For further details see \cite{Kamvissis:2003,Miller02}.  Considering the limit $N\to\infty$, we have the following \emph{approximate Zakharov-Shabat scattering data}:  
\begin{itemize}
\item $b(\zeta)=o(1)$ uniformly on $\mathbb{R}$ in the limit $N\to\infty$.
\item Let $\rho(s)$ be defined by 
\eq
\begin{split}
\rho(s):=&\frac{s}{\pi}\int_{x_-(s)}^{x_+(s)}\frac{\dd x}{\sqrt{A(x)^2-s^2}}\\
  =&-\frac{1}{\pi}\frac{\dd}{\dd s}
\int_{x_-(s)}^{x_+(s)}\sqrt{A(x)^2-s^2}\,\dd x,\quad 0<s<A_\mathrm{max}:=\max_{x\in\mathbb{R}}A(x),
\label{eq:ZS-density}
\end{split}
\endeq
where $x_-(s)<x_+(s)$ are the two roots $x$ of $A(x)^2-s^2$.
Then, the \emph{approximate eigenvalues} are $\zeta=i\widetilde{s}_{n}$, $n=0,\dots,N-1$, where $\widetilde{s}_{0},\dots,\widetilde{s}_{N-1}$ are determined uniquely by the Bohr-Sommerfeld quantization rule
\begin{equation}
\Psi(i\widetilde{s}_{n}) = (n+\tfrac{1}{2})\epsilon\pi = \frac{2n+1}{2N}\int_\mathbb{R}A(x)\,\dd x,\quad n=0,\dots,N-1
\label{eq:ZS-approximate-eigenvalues}
\end{equation}
with the \emph{phase integral} $\Psi$ being defined by
\begin{equation}
\Psi(is):=\pi\int_{s}^{A_\mathrm{max}}\rho(s')\,\dd s'= \int_{x_-(s)}^{x_+(s)}\sqrt{A(x)^2-s^2}\,\dd x,\quad 0<s<A_\mathrm{max}.
\label{eq:ZS-phase-integral}
\end{equation}
\item The proportionality constant $\tau$ associated with the eigenvalue best approximated by $\zeta=i\widetilde{s}_{n}$ is itself approximated by 
\begin{equation}
\tau\approx\widetilde{\tau}_{n} :=  i(-1)^Ke^{i(2K+1)\Psi(i\widetilde{s}_{n})/\epsilon}e^{\mu(i\widetilde{s}_{n})/\epsilon},\quad K\in\mathbb{Z},
\label{eq:ZS-proportionality-constant}
\end{equation}
where
\eq
\begin{split}
\mu(is):=&(x_+(s)+x_-(s))s + \int_{-\infty}^{x_-(s)}\left(\sqrt{s^2-A(x)^2}-s\right)\,\dd x \\
  & - \int_{x_+(s)}^{+\infty}\left(\sqrt{s^2-A(x)^2}-s\right)\,\dd x,\quad 0<s<A_\mathrm{max}.
\label{eq:mu-define}
\end{split}
\endeq
The value of the integer $K$ is arbitrary due to the definition \eqref{eq:ZS-approximate-eigenvalues} of the numbers $\{\widetilde{s}_{n}\}_{n=0}^{N-1}$ (in fact, $\widetilde{\tau}_{n}e^{-\mu(i\widetilde{s}_{n})/\epsilon}=(-1)^{n+1}$ regardless of the value of $K\in\mathbb{Z}$).  The advantage of interpolating the factor $(-1)^{n+1}$ using the phase integral $\Psi$ is that, for each choice of $K\in\mathbb{Z}$, $\widetilde{\tau}_{n}$ may be considered to be an $n$-independent analytic function of $s$ evaluated at $s=\widetilde{s}_{n}$.
\end{itemize}
Note that according to \eqref{eq:ZS-a-upper-half-plane}, since $b(\cdot)$ is negligible on the real axis, the function $a(\zeta)$ is approximated for $\imag\{\zeta\}>0$ by a Blaschke product:
\begin{equation}
a(\zeta)\approx\widetilde{a}(\zeta):=\prod_{n=0}^{N-1}\frac{\zeta-i\widetilde{s}_{n}}{\zeta+i\widetilde{s}_{n}},
\label{eq:a-tilde-define}
\end{equation}
while if $\zeta$ is the eigenvalue best approximated by $i\widetilde{s}_{m}$,
\begin{equation}
a'(\zeta)\approx \widetilde{a}'(i\widetilde{s}_{m})=\frac{1}{i}\frac{\prod_{n\neq m}(\widetilde{s}_{m}-\widetilde{s}_{n})}{\prod_{n=0}^{N-1}(\widetilde{s}_{m}+\widetilde{s}_{n})}.
\end{equation}
The above asymptotic formulae may be obtained by applying turning point theory (for $A$ Klaus-Shaw there are exactly two real simple turning points $x=x_\pm(s)$ for $\zeta=is$ with $0<s<A_\mathrm{max}$, and the real axis can be covered by two overlapping intervals, each containing precisely one of the turning points, on which Langer transformations can be used to map the Zakharov-Shabat problem \eqref{eq:Zakharov-Shabat} to a controllable perturbation of the Airy equation).  This theory is suitably robust to make the above approximations rigorous as long as $s$ is confined to an arbitrary closed subinterval of $(0,A_\mathrm{max})$.

Still considering the case of $\psi(x)=A(x)$ with $A:\mathbb{R}\to\mathbb{R}_+$, we may obtain approximate formulae valid as $\epsilon\downarrow 0$ for the Jost solutions $\mathbf{w}^{-,1}(x;\zeta)$ and $\mathbf{w}^{+,2}(x;\zeta)$ in the complementary situation that $\zeta\in\mathbb{C}_+\setminus [0,iA_\mathrm{max}]$, i.e., for $\zeta$ bounded away from the eigenvalue locus.  In this situation the WKB method without turning points applies, i.e., solutions $\mathbf{w}(x;\zeta)$ of \eqref{eq:Zakharov-Shabat} are approximated by $\mathbf{w}(x;\zeta)\approx \mathbf{w}_0(x;\zeta)e^{f(x;\zeta)/\epsilon}$, where $f'(x;\zeta)$ is an eigenvalue of the coefficient matrix $\mathbf{Z}(x;\zeta)$ in \eqref{eq:Zakharov-Shabat} with eigenvector $\mathbf{w}_0(x;\zeta)$ normalized so that $\mathbf{w}_0'(x;\zeta)$ is in the column space of the singular matrix $\mathbf{Z}(x;\zeta)-f'(x;\zeta)\mathbb{I}$.  This approximation can be proved to be accurate given an initial value of $\mathbf{w}(x_0;\zeta)$ in an interval with endpoint $x_0$ on which the real part of the exponent $f(x;\zeta)$ is strictly increasing in the direction away from $x_0$.  By an easy generalization of this analysis to allow $x_0\to \pm\infty$, the Jost solutions $\mathbf{w}^{-,1}(x;\zeta)$ and $\mathbf{w}^{+,2}(x;\zeta)$ can be rigorously approximated on the respective intervals $(-\infty,x_1]$ and $[x_1,+\infty)$ for any $x_1\in\mathbb{R}$ with the use of the respective eigenvalues $f'(x;\zeta)=(-\zeta^2-A(x)^2)^{1/2}$ and $f'(x;\zeta)=-(-\zeta^2-A(x)^2)^{1/2}$ (principal branch of the square root) which have respectively positive and negative real parts on the corresponding semi-infinite intervals.
These approximations are the following:
\begin{equation}
\begin{split}
\mathbf{w}^{-,1}(x;\zeta)\approx\widetilde{\mathbf{w}}^{-,1}(x;\zeta)&:=N(x;\zeta)\begin{pmatrix}
i\zeta-\left(-\zeta^2-A(x)^2\right)^{1/2}\\A(x)\end{pmatrix} e^{-i\zeta x/\epsilon}e^{\nu^-(x;\zeta)/\epsilon},\\
\mathbf{w}^{+,2}(x;\zeta)\approx\widetilde{\mathbf{w}}^{+,2}(x;\zeta)&:=N(x;\zeta)\begin{pmatrix}A(x)\\i\zeta-\left(-\zeta^2-A(x)^2\right)^{1/2}\end{pmatrix}
e^{i\zeta x/\epsilon}e^{\nu^+(x;\zeta)/\epsilon},
\end{split}
\label{eq:approximate-Josts}
\end{equation}
for $x\in\mathbb{R}$ and $\zeta\in\mathbb{C}_+\setminus[0,iA_\mathrm{max}]$, where the complex scalar normalizing factor $N(x;\zeta)$ is given by
\begin{equation}
N(x;\zeta):=-\left(2\left(-\zeta^2-A(x)^2\right)^{1/2}\left(\left(-\zeta^2-A(x)^2\right)^{1/2}-i\zeta\right)\right)^{-1/2},
\label{eq:approximate-Josts-norm}
\end{equation}
and where
\begin{equation}
\begin{split}
\nu^-(x;\zeta)&:=\int_{-\infty}^x\left[\left(-\zeta^2-A(y)^2\right)^{1/2}+i\zeta\right]\,\dd y,\\
\nu^+(x;\zeta)&:=\int_x^{+\infty}\left[\left(-\zeta^2-A(y)^2\right)^{1/2}+i\zeta\right]\,\dd y.
\end{split}
\label{eq:mu-plus-minus-define}
\end{equation}
We may observe that, at any value of $x\in\mathbb{R}$ at which $A(x)=0$, significant simplification occurs:
\begin{equation}
A(x)=0\;\implies\; \widetilde{\mathbf{w}}^{-,1}(x;\zeta)=\begin{pmatrix}1\\0\end{pmatrix}e^{-i\zeta x/\epsilon}e^{\nu^-(x;\zeta)/\epsilon}\;\text{and}\;
\widetilde{\mathbf{w}}^{+,2}(x;\zeta)=\begin{pmatrix}0\\1\end{pmatrix}e^{i\zeta x/\epsilon}e^{\nu^+(x;\zeta)/\epsilon}.
\label{eq:approximate-Josts-simplify}
\end{equation}
Furthermore, observe that if $x>\mathrm{supp}(A)$ (respectively $x<\mathrm{supp}(A)$), then $\nu^-(x;\zeta)$ (respectively $\nu^+(x;\zeta)$) coincides with a complete integral:
\begin{equation}
L(\zeta):=-\nu^-(+\infty;\zeta)=-\nu^+(-\infty;\zeta)=\int_\mathbb{R}\left[(-\zeta^2-A(y)^2)^{1/2}+i\zeta\right]\,\dd y.
\label{eq:Lfunc-def}
\end{equation}
The vectors $\mathbf{w}^{-,1}(x;\zeta)e^{i\zeta x/\epsilon}e^{-\nu^-(x;\zeta)/\epsilon}$ and $\mathbf{w}^{+,2}(x;\zeta)e^{-i\zeta x/\epsilon}e^{-\nu^+(x;\zeta)/\epsilon}$ actually have complete asymptotic expansions in powers of $\epsilon$ provided that $A(x)$ is infinitely differentiable, and the approximations \eqref{eq:approximate-Josts} capture just the leading term.   In this case, for every $p>0$,
\begin{gather}
w_2^{-,1}(x;\zeta)e^{i\zeta x/\epsilon}e^{L(\zeta)/\epsilon} = \bigo{\epsilon^p},\quad\epsilon\downarrow 0,\quad x>\mathrm{supp}(A)
\intertext{and}
w_1^{+,2}(x;\zeta)e^{-i\zeta x/\epsilon}e^{L(\zeta)/\epsilon}=\bigo{\epsilon^p},\quad\epsilon\downarrow 0,\quad x<\mathrm{supp}(A)
\end{gather}
hold for each $\zeta\in\mathbb{C}_+$ bounded away from the eigenvalue locus.  Thus, these quantities are small \emph{beyond all orders} in $\epsilon$.  In the case of potentials $A$ with only finitely-many derivatives one generally only gets a finite order of vanishing.  In both cases it is challenging to calculate a leading term (the quantities on the left-hand side are generally nonzero).

\end{document}